\documentclass[11pt]{article}
\usepackage[utf8]{inputenc}

\usepackage[
    backend=biber,
    style=numeric-comp,
    sortcites,
    sorting=none,
    giveninits=true,
    natbib,
    hyperref=true,
    maxbibnames=99,
    doi=false,isbn=false,url=false,eprint=false
]{biblatex}

\usepackage{comment}
\usepackage{placeins}

\usepackage{xargs}
\usepackage[colorinlistoftodos,prependcaption,textsize=tiny]{todonotes}
\newcommandx{\note}[2][1=]{\todo[linecolor=blue,backgroundcolor=blue!25,bordercolor=blue,#1]{#2}}

\usepackage[ruled,linesnumbered]{algorithm2e}

\let\oldnl\nl
\newcommand{\nonl}{\renewcommand{\nl}{\let\nl\oldnl}}

\newlength\mylen
\newcommand\myinput[1]{\nonl
  \settowidth\mylen{\KwIn{}}%
  \setlength\hangindent{\mylen}%
  \hspace*{\mylen}#1\\}

\usepackage{xcolor}
\usepackage{graphicx,amsfonts,amsthm,amssymb}
\usepackage[hidelinks]{hyperref}
\usepackage{amsmath,mathrsfs}
\usepackage{wrapfig,lipsum,booktabs}
\usepackage{subcaption}
\usepackage{appendix}
\usepackage{enumitem}
\usepackage{multirow}
\usepackage{optidef}
\usepackage{cleveref}

\newcommand{\IP}{{\rm I}\kern-0.18em{\rm P}}
\newcommand{\1}{{\rm 1}\kern-0.24em{\rm I}}
\newcommand{\E}{{\rm I}\kern-0.18em{\rm E}}
\newcommand{\R}{{\rm I}\kern-0.18em{\rm R}}

\newcommand{\cD}{\mathcal{D}}

\newcommand{\tF}{\Tilde{F}}
\newcommand{\tf}{\Tilde{f}}

\newcommand{\cA}{\mathcal{A}}
\newcommand{\cI}{\mathcal{I}}

\newcommand{\htau}{{\widehat{\tau}}}
\newcommand{\hDelta}{{\widehat{\Delta}}}

\renewcommand{\P}[1]{\mathbb{P}\left[#1\right]}
\newcommand{\EV}[1]{\mathbb{E}\left[#1\right]}
\newcommand{\I}[1]{\mathbb{I}\left[#1\right]}

\allowdisplaybreaks

\newcommand\independent{\protect\mathpalette{\protect\independenT}{\perp}}
\def\independenT#1#2{\mathrel{\rlap{$#1#2$}\mkern2mu{#1#2}}}

\newtheorem{theorem}{Theorem}

\newtheorem{assumption}{Assumption}

\newtheorem{lemma}{Lemma}
\newtheorem{proposition}{Proposition}
\newtheorem{corollary}{Corollary}
\newtheorem{definition}{Definition}

\title{Adaptive conformal classification with noisy labels}
\author{Matteo Sesia\thanks{Department of Data Sciences and Operations, University of Southern California, Los Angeles, California.}, Y.~X.~Rachel Wang\thanks{School of Mathematics and Statistics, University of Sydney, Sydney, Australia.}, Xin Tong\footnotemark[1]}
\date{\today}

\begin{document}

\maketitle

\begin{abstract}
This paper develops novel conformal prediction methods for classification tasks that can automatically adapt to random label contamination in the calibration sample, leading to more informative prediction sets with stronger coverage guarantees compared to state-of-the-art approaches. This is made possible by a precise characterization of the effective coverage inflation (or deflation) suffered by standard conformal inferences in the presence of label contamination, which is then made actionable through new calibration algorithms. Our solution is flexible and can leverage different modeling assumptions about the label contamination process, while requiring no knowledge of the underlying data distribution or of the inner workings of the machine-learning classifier. The advantages of the proposed methods are demonstrated through extensive simulations and an application to object classification with the CIFAR-10H image data set.
\end{abstract}

\section{Introduction} \label{sec:intro}

\subsection{Background and motivation}

Conformal inference \citep{vovk2005algorithmic} is a versatile and increasingly popular framework for estimating the uncertainty of predictions output by any supervised learning model, including for example modern deep neural networks for multi-class classification.
Its two key strengths are that: (1) it requires no parametric assumptions about the data distribution, making it relevant to a variety of real-world applications; and (2) it can accommodate arbitrarily complex {\em black-box} predictive models, mitigating the risk of early obsolescence in the rapidly evolving field of machine learning.
In a nutshell, conformal inference is able to transform the output of any model into a relatively informative and well-calibrated {\em prediction set} for the unknown label of a future test point, while providing precise coverage guarantees in finite samples.
Intuitively, this is achieved by carefully leveraging the empirical distribution of suitable residuals (or {\em conformity scores}) evaluated on held-out data that were not utilized for training.
Notably, these guarantees can be established assuming only that the calibration data are {\em exchangeable} (or, for simplicity, independent and identically distributed) random samples from the population of interest.
While this framework can provide useful uncertainty estimates while relying on weaker assumptions compared to classical parametric modeling, existing conformal prediction methods are not always fully satisfactory.
One limitation that we aim to address in this paper is the reliance on the assumption that the calibration data are all labeled correctly, which is often unrealistic.

In fact, it may be expensive or even outright unfeasible to acquire accurately labeled (or {\em clean}) data in many applications, even if there is an abundance of lower-quality observations with imperfect labels \citep{natarajan2013learning,sukhbaatar2014training,song2022learning,yao2022asymmetric}, which we also call {\em noisy} or {\em contaminated}. 
For example, the Amazon Mechanical Turk is a crowdsourcing platform that allows researchers and organizations to assign labels to large-scale unsupervised data by leveraging a global workforce of remote human annotators \citep{sorokin2008utility}.
This platform is utilized across diverse fields, including image recognition and natural language processing.
Crowdsourcing tends to be fast and cost-effective, but it raises concerns about poor annotation quality and its impacts on the reliability of downstream analyses \citep{ipeirotis2010quality,snow2008cheap,kennedy2020shape,aguinis2021mturk}.
Other applications motivating our work are those in which data with labels carrying sensitive personal information are randomly anonymized to safeguard user privacy \citep{warner1965randomized,evfimievski2003limiting,kasiviswanathan2011can,ghazi2021deep,angelopoulos2022private}.

Although significant efforts have been dedicated to {\em training} predictive models using data with noisy labels \citep{natarajan2013learning,sukhbaatar2014training,northcutt2021pervasive,algan2021image,song2022learning}, the challenge of {\em calibrating} those models using conformal inference has only recently begun to receive some attention \citep{cauchois2022predictive,einbinder2022conformal,barber2022conformal} and is still understudied. 
This paper aims to help fill this gap.

\subsection{Preview of our contributions and paper outline} \label{sec:intro-preview}

This paper addresses the questions of whether and how conformal inference should account for the possible presence of incorrectly labeled calibration samples, focusing on the task of constructing prediction sets for multi-class classification.
After recalling the relevant background, we begin in Section~\ref{sec:preliminaries} by carefully studying the impact of label contamination on the effective coverage achieved by standard conformal prediction sets. 
Our analysis shows and quantifies precisely how label contamination may lead to prediction sets that are either too liberal or too conservative.
The practical significance of these results becomes clear in Section~\ref{sec:methods}, where we develop and study a novel calibration method that can automatically adapt to label contamination, without assuming any knowledge of the data distribution or of the classifier.
As we shall see, our method can produce more informative prediction sets with more robust coverage guarantees compared to standard approaches.
Initially, we assume the label contamination process is known and belongs to a broad family encompassing most of the typical models from the related literature on learning from noisy labels \citep{natarajan2013learning, scott2013classification,ghosh2017robust}.
Then, we extend our solution to accommodate models that may depend on unknown parameters, and we explain how to estimate those parameters.

In Section~\ref{sec:empirical}, we demonstrate the practical performance of our methods through extensive numerical experiments and an analysis of the CIFAR-10H image classification data set \citep{peterson2019human}.
A preview of some results is given here by Figure~\ref{fig:exp-cifar10-marginal}, which demonstrates that our methods are practical and produce smaller (more informative) prediction sets compared to standard conformal inference techniques, while maintaining valid coverage.
Finally, Section~\ref{sec:discussion} concludes with a discussion and some ideas for future work.

Some content is deferred to the Appendix, for lack of space.
Section~\ref{app:review-aps} reviews the standard conformal prediction approach.
Section~\ref{app:methods} provides additional algorithmic tables outlining the proposed methods.
Section~\ref{app:simplified-methods} describes some simplified implementations of our methods under more restricted contamination models.
Section~\ref{app:worst-case} discusses the relation between our work and existing theoretical worst-case results about the behavior  of standard conformal predictions applied to non-exchangeable data \citep{barber2022conformal}.
Section~\ref{app:extensions-theory} extends our theoretical results from Section~\ref{sec:preliminaries} to study marginal instead of label-conditional coverage.
Section~\ref{app:extensions} extends our methods from Section~\ref{sec:methods} to construct adaptive prediction sets satisfying other types of theoretical guarantees, such as marginal coverage and calibration-conditional coverage \citep{vovk2012conditional}.
Further extensions encompassing even stronger guarantees, such as equalized coverage over protected categories \citep{romano2019malice}, would also be possible but are omitted for length-related reasons.
Section~\ref{app:proofs} contains all mathematical proofs.
Section~\ref{app:figures} describes additional numerical results, and Section~\ref{app:experiments} provides further technical details about our experiments.

\begin{figure}[!htb]
\centering 
\includegraphics[width=0.85\linewidth]{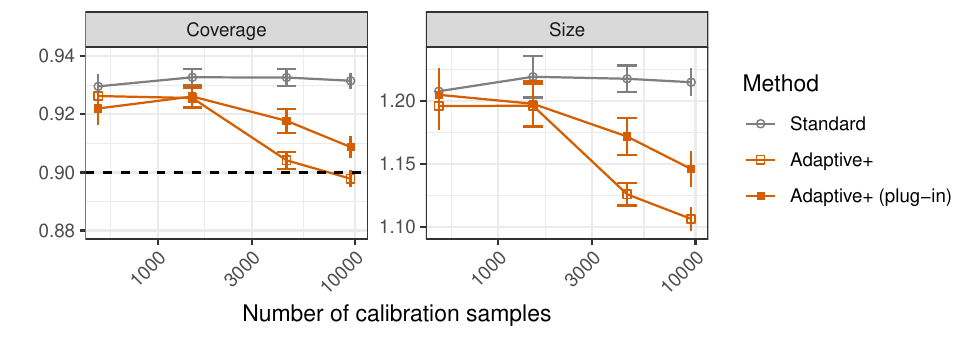}
\caption{Performance of the proposed adaptive conformal prediction methods on CIFAR-10H image data with imperfect labels.
The results are shown as a function of the calibration sample size and compared to the performance of standard conformal predictions, which are too conservative due to the presence of inaccurate labels.
The dashed line indicates the nominal 90\% marginal coverage level.}
\label{fig:exp-cifar10-marginal}
\end{figure}

\subsection{Related work} \label{sec:related-work}

Conformal inference was pioneered by Vovk and collaborators \citep{vovk2005algorithmic} and has become a very active area of research \citep{lei2013distribution,lei2014distribution,lei2018distribution}, with applications including outlier detection \citep{smith2015conformal,guan2019prediction,Liang_2022_integrative_p_val,bates2021testing,bashari2023derandomized}, regression \citep{romano2019conformalized,sesia2020comparison,sesia2021conformal}, and classification \citep{hechtlinger2018cautious,sadinle2019least,cauchois2020knowing,romano2020classification,einbinder2022training}.

Numerous works have studied the robustness of conformal inferences to the breakdown of standard data exchangeability assumptions. 
However, these investigations largely concentrated on different challenges, such as covariate shift \citep{tibshirani2019conformal}, repeated sampling \citep{sesia2023conformal}, time-series dependencies \citep{gibbs2021adaptive}, and label shift \citep{podkopaev2021distribution,si2023pac}.
We refer to Section~\ref{subsec:linear-contam-model} for further details on the crucial distinction between label {\em shift} and the label {\em contamination} problem considered in this paper.

\citet{barber2022conformal} conducted a general theoretical study of conformal inference with non-exchangeable data.
Although their goals and approaches differ from ours, some of their insights also apply to scenarios involving label contamination. 
Our findings are consistent with theirs, albeit somewhat weaker, when viewed through the lens of a {\em worst-case} analysis of traditional methods.
However, our goal is not to refine their {\em theoretical} analysis. 
On the contrary, we make a completely separate methodological contribution---we develop novel methods for constructing more informative prediction sets in the presence of {\em random} label contamination.

Our paper is more closely related to \citet{cauchois2022predictive} and \citet{einbinder2022conformal}, which also considered conformal prediction sets calibrated using imperfectly labeled data.
However, our research is distinct and involves many novelties. 
The problem of \citet{cauchois2022predictive} is different because they allow the calibration data to contain imperfect labels but seek to predict analogous imperfect labels for the test data. By contrast, we aim to predict the true (clean) labels.
Our perspective is more aligned with that of \citet{einbinder2022conformal}, but their efforts are focused on establishing the conservativeness of standard conformal inferences to certain forms of random label contamination, and they do not attempt to mitigate that (often excessive) conservativeness or to protect against other forms of contamination.
Our theoretical analyses take a different approach, leading to more general and quantitative results that hold for a broader class of contamination models.
Then, our main contribution is to develop novel methods that can automatically adapt to label contamination.
To the best of our knowledge, this is the first work to propose conformal prediction methods that are adaptive to label contamination.

\section{Preliminaries} \label{sec:preliminaries}

\subsection{Problem statement} \label{sec:setup}

Consider $n+1$ data points $(X_i,Y_i,\tilde{Y}_i)$, for $i \in [n+1] = \{1,\ldots,n+1\}$, where $X_i \in \mathbb{R}^d$ is a feature vector, $Y_i \in [K]$ is a latent categorical label, and $\tilde{Y}_i \in [K]$ is an observable label that we interpret as a contaminated version of $Y_i$.
Assume the data are i.i.d.~random samples from some unknown distribution. 
A weaker {\em exchangeability} assumption often turns out to be sufficient in the related literature, but we focus on i.i.d.~data in this paper.
The problem we consider is that of constructing informative conformal prediction sets for the true label $Y_{n+1}$ of a test point with features $X_{n+1}$, leveraging the available observations $(X_i, \tilde{Y}_i)$ indexed by $[n]$, to which we collectively refer as $\mathcal{D}$.
Of course, in order to establish precise coverage guarantees, it will be necessary to introduce some assumptions about the relation between the true and contaminated labels, as we shall see.
Our setup thus extends the standard conformal prediction framework for uncontaminated data, which corresponds for the special case in which $\tilde{Y} = Y$ almost surely.

\subsection{Relevant technical background} \label{sec:background}

This section recalls some relevant background on standard conformal classification methods, which ignore the possibility of label contamination.
Our notation is inspired by that of \citet{romano2020classification} but involves some helpful adaptations.
Note that we will refer to standard conformal prediction approaches as working with observations $(X, \tilde{Y})$ to keep our notation consistent, although the prior literature typically did not distinguish between $Y$ and $\tilde{Y}$.

Let $\mathcal{A}$ denote any learning algorithm for classification; e.g., a logistic regression model or a neural network.
The role of $\mathcal{A}$, which is treated as a black-box throughout the paper, is to train a model $\hat{\pi}$ that estimates the distribution of $\tilde{Y} \mid X$ using the data in $\mathcal{D}$.
For any $k \in [K]$ and $x \in \mathbb{R}^d$, let $\hat{\pi}(x,k) \in [0,1]$ be the estimated probability of $\tilde{Y} = k \mid X=x$.
In the following, it will sometimes be convenient to assume that the distribution of $\hat{\pi}(X,k)$ is continuous.
This condition could be relaxed at the cost of a more involved notation, but it can be made realistic by simply adding a small amount of independent noise to the output of the model.
Further, we assume that $\hat{\pi}$ is normalized; i.e., $\sum_{k=1}^{K} \hat{\pi}(x,k) = 1$.
Aside from these requirements, we allow $\hat{\pi}$ to be anything and do not necessarily expect it to model $\tilde{Y} \mid X$ accurately.
Thus, typical off-the-shelf classifiers can provide suitable statistics $\hat{\pi}$ by default. 
For example, one may choose $\hat{\pi}$ to be the output of the final soft-max layer of a deep neural network.

This paper studies how to translate such a black-box model into a reliable prediction set for $Y_{n+1}$ given $X_{n+1}$.
But first we need to review the standard conformal approach to the special case where $Y=\tilde{Y}$.
The first key notion that we recall is that of a {\em prediction function}.

\begin{definition}[Prediction function] \label{def:pred-function}
Let $\mathcal{C}$ be a set-valued function, whose form may depend on the model $\hat{\pi}$, that takes as input $x \in \mathbb{R}^d$ and $\tau = (\tau_1, \ldots, \tau_K) \in [0,1]^K$, and returns as output a subset of $[K]$.
We say that $\mathcal{C}$ is a prediction function if it is monotone increasing with respect to each element of $\tau$, 
is such that the event $k \in \mathcal{C}(x, \tau)$ depends on $\tau$ only through its $k$-th element $\tau_k$, and
satisfies $k \in \mathcal{C}(X, \tau)$ whenever $\tau_k=1$, for any $k \in [K]$.
\end{definition}
Note that the dependence of a prediction function $\mathcal{C}$ on $\hat{\pi}$ will typically be kept implicit unless otherwise necessary to avoid ambiguity; i.e., $\mathcal{C}(x, \tau) = \mathcal{C}(x, \tau; \hat{\pi}) \subseteq [K]$.

For any prediction function $\mathcal{C}$, we define the {\em conformity score} function $\hat{s}: \mathbb{R}^d \times [K] \mapsto [0,1]$, also implicitly depending on $\hat{\pi}$, as that function which outputs the smallest value of $\tau_k$ allowing the label $k$ to be contained in the set $\mathcal{C}(x, \tau_k)$. That is,
\begin{align}  \label{eq:conf-scores}
  \hat{s}(x, k)
  = \inf \left\{ \tau_k \in [0,1] : k \in \mathcal{C}(x, \tau_k) \right\}.
\end{align}
Note that the short-hand $\mathcal{C}(X, \tau_k)$ above is a slight abuse of notation, since $\mathcal{C}$ was defined for a vector-valued input $\tau$, but it does not introduce any ambiguity. 

An example of $\mathcal{C}$ is the function that outputs the set of all labels $k \in [K]$ for which the estimated conditional probability of $\tilde{Y}=k \mid X=x$ is sufficiently large; i.e.,
\begin{align} \label{eq:pred-function-hps}
  \mathcal{C}(x, \tau; \hat{\pi})
  & := \{ k \in [K] : \hat{\pi}(x,k) \geq 1 - \tau_k\}.
\end{align}
The associated scores are  $\hat{s}(x,k) = 1-\hat{\pi}(x,k)$. These are sometimes called {\em homogeneous conformity scores} because the prediction function in~\eqref{eq:pred-function-hps} is not designed to account for heteroscedasticity in the distribution of $Y \mid X$  \citep{cauchois2020knowing}.
While this choice of $\mathcal{C}$ is convenient to keep the notation simple, all of our results also extend to other prediction functions, including those associated with the more sophisticated {\em generalized inverse quantile} conformity scores of \citet{romano2020classification}, which we review in Section~\ref{app:review-aps}.

A standard implementation of conformal inference begins by randomly splitting the data into two disjoint subsets, $\mathcal{D}^{\text{train}}$ and $\mathcal{D}^{\mathrm{cal}}$, such that $\mathcal{D} = \mathcal{D}^{\text{train}} \cup \mathcal{D}^{\mathrm{cal}}$.
The model $\hat{\pi}$ is trained using the data in $\mathcal{D}^{\text{train}}$. The held-out data in $\mathcal{D}^{\mathrm{cal}}$ are utilized to compute conformity scores $\hat{s}(X_i, k)$ via~\eqref{eq:conf-scores}, for all $k \in [K]$ and $i \in \mathcal{D}^{\mathrm{cal}}$, according to the desired prediction function $\mathcal{C}$.
These scores are then utilized to calibrate a prediction set for a new test point with features $X_{n+1}$ as follows. 
For each $k \in [K]$, define $\mathcal{D}_k^{\mathrm{cal}} = \{ i \in \mathcal{D}^{\mathrm{cal}} : \tilde{Y}_i = k\}$, $n_k = |\mathcal{D}_k^{\mathrm{cal}}|$, and $\hat{\tau} = (\hat{\tau}_1, \ldots, \hat{\tau}_K)$, where
\begin{align} \label{eq:def-tau-hat}
  \hat{\tau}_k = \lceil (1+n_k)\cdot(1-\alpha) \rceil \text{-th smallest value in } \{\hat{s}(X_i,k)\}_{i \in \mathcal{D}_k^{\mathrm{cal}}},
\end{align}
and $\alpha \in (0,1)$ is the desired significance level.
The prediction set for $Y_{n+1}$ is given by $\hat{C}(X_{n+1}) = \mathcal{C}(X_{n+1}, \hat{\tau}; \hat{\pi})$.
See Algorithm~\ref{alg:standard-lab-cond} in Section~\ref{app:standard-lab-cond} for a summary of this method.
This procedure makes it possible to prove that $\hat{C}(X_{n+1})$ has \textit{label-conditional coverage} \citep{vovk2003mondrian},
\begin{align} \label{eq:def-lab-cond-coverage}
  \P{ Y_{n+1} \in \hat{C}(X_{n+1}) \mid Y_{n+1} = k} \geq 1-\alpha, \qquad \text{ for all } k \in [K],
\end{align}
as long as $Y = \tilde{Y}$ almost surely. See Proposition~\ref{prop:standard-coverage-label} for a formal statement of this result, which also provides an almost-matching coverage upper bound.
Note that the probability in~\eqref{eq:def-lab-cond-coverage} is taken with respect to $X_{n+1}$ and the data in $\mathcal{D}$, both of which are random.
In the following sections, we will study the behavior of this method while allowing $\tilde{Y} \neq Y$.

In the meantime, we recall that an alternative standard approach is to construct prediction sets $\hat{C}(X_{n+1})$ satisfying the following weaker notion of {\em marginal coverage}:
\begin{align} \label{eq:def-marg-coverage}
  \P{ Y_{n+1} \in \hat{C}(X_{n+1})} \geq 1-\alpha.
\end{align}
As long as $Y = \tilde{Y}$ almost surely, prediction sets with marginal coverage can be obtained by simply replacing the subset $\mathcal{D}_k^{\mathrm{cal}}$ with $\mathcal{D}^{\mathrm{cal}}$ in Algorithm~\ref{alg:standard-lab-cond}. See Algorithm~\ref{alg:standard-marg} and Proposition~\ref{prop:standard-coverage-marginal} in Section~\ref{app:standard-marg} for further details.
While marginal coverage~\eqref{eq:def-marg-coverage} is not as strong as label-conditional coverage~\eqref{eq:def-lab-cond-coverage}, it is a useful notion because it is easier to achieve using smaller, and hence more informative, prediction sets when limited data are available.
Therefore, we will study how to efficiently control both~\eqref{eq:def-lab-cond-coverage} and~\eqref{eq:def-marg-coverage} using contaminated data, leaving it to practitioners to determine which type of guarantee is most appropriate for a given application.
However, due to lack of space, most details about marginal coverage are presented in the Appendix.




\subsection{General coverage bounds under label contamination} \label{sec:general-coverage-bounds}

This section analyzes theoretically the behavior of standard conformal classification approaches applied with contaminated calibration data.
In particular, we demonstrate how label noise can cause the effective coverage of these methods to be either inflated or deflated, as made precise by an explicit factor that depends on the distribution of the conformity scores. 
For simplicity, we focus here on label-conditional coverage~\eqref{eq:def-lab-cond-coverage}, studying the behavior of Algorithm~\ref{alg:standard-lab-cond}.
These results will be extended to study marginal coverage~\eqref{eq:def-marg-coverage} in Section~\ref{app:extensions-theory}.

It is worth emphasizing that the results presented in this section require no assumptions on the contamination process and encompass a wide range of scenarios in which label noise leads to over-coverage or under-coverage; thus, our analysis is more general than that of \citet{einbinder2022conformal}, which focused on establishing conservativeness under a narrower random corruption model.
That being said, it will become clear in Section~\ref{sec:methods} that some assumptions about the contamination model are useful to achieve the more ambitious goal of developing practical conformal prediction methods that can automatically adapt to label noise.

The following notation will be helpful. For any $k,l \in [K]$ and $t \in \mathbb{R}$, define
  \begin{align} \label{eq:cdf-scores}
\begin{split}
    F_l^{k}(t) & := \P{\hat{s}(X,k) \leq t \mid Y=l, \mathcal{D}^{\text{train}}},  \\
    \tilde{F}_l^{k}(t) & := \P{\hat{s}(X,k) \leq t \mid \tilde{Y}=l, \mathcal{D}^{\text{train}}}.
  \end{split}
  \end{align}
In words, $F_l^{k}$ is the cumulative distribution function of $\hat{s}(X,k)$, based on a fixed function $\hat{s}$ and a random sample $X$ from the distribution of $X \mid Y=l$, namely $P_l$. Analogously, $\tilde{F}_l^{k}$ is the cumulative distribution function corresponding to a random sample $X$ from the distribution of $X \mid \tilde{Y}=l$, namely $\tilde{P}_l$.
For any $k \in [K]$ and $t \in \mathbb{R}$, define also $\Delta_k(t)$ as
  \begin{align} \label{eq:delta}
    \Delta_k(t) & := F_k^{k}(t) - \tF_k^{k}(t).
  \end{align}
We refer to $\Delta_k(t)$ as the {\em coverage inflation factor} because its expected value controls the discrepancy between the real and nominal coverage of the prediction sets output by Algorithm~\ref{alg:standard-lab-cond}, as established by the next result.
Note that $\Delta_k(t)$ may be either positive or negative.

\begin{theorem}\label{thm:coverage-lab-cond}
Suppose $(X_i,Y_i,\tilde{Y}_i)$ are i.i.d.~for all $i \in [n+1]$.
Fix any prediction function $\mathcal{C}$ satisfying Definition~\ref{def:pred-function}, and let $\hat{C}(X_{n+1})$ indicate the prediction set output by Algorithm~\ref{alg:standard-lab-cond} applied using the corrupted labels $\tilde{Y}_i$ instead of the clean labels $Y_i$, for all $i \in [n]$.
Then,
\begin{align} \label{eq:prop-label-cond-coverage-lower}
  \P{Y_{n+1} \in \hat{C}(X_{n+1}) \mid Y_{n+1} = k} \geq 1 - \alpha + \EV{\Delta_k(\htau_k)}.
\end{align}
Further, if the conformity scores $\hat{s}(X_i,\tilde{Y}_i)$ used by Algorithm~\ref{alg:standard-lab-cond} are almost-surely distinct,
\begin{align} \label{eq:prop-label-cond-coverage-upper}
  \P{Y_{n+1} \in \hat{C}(X_{n+1}) \mid Y_{n+1} = k} \leq 1 - \alpha + \frac{1}{n_k+1} + \EV{\Delta_k(\htau_k)}.
\end{align}
\end{theorem}
Above, $\EV{\Delta_k(\htau_k)}$ denotes an expected value taken with respect to the randomness in all data in $\mathcal{D}$, including the calibration samples $\mathcal{D}^{\mathrm{cal}}$ upon which $\htau_k$ depends.
Note that it would also possible to obtain slightly stronger versions of~\eqref{eq:prop-label-cond-coverage-lower} and~\eqref{eq:prop-label-cond-coverage-upper} in which the data $\mathcal{D}$ are treated as fixed and $\EV{\Delta_k(\htau_k)}$ is replaced by $\Delta_k(\htau_k)$; see the proof of Theorem~\ref{thm:coverage-lab-cond}.
In any case, this result tells us that standard conformal inferences may be either overly conservative or too liberal when dealing with contaminated data, depending on whether $\EV{\Delta_k(\htau_k)}$ is positive or negative.

This result, which we extend in Section~\ref{app:extensions-theory} to study marginal coverage, serves as the starting point of our methodology.
As a preliminary step towards the goal of achieving {\em tight} coverage at level $1-\alpha$, the next subsection introduces additional modeling assumptions about the label contamination process and sheds more light onto the situations in which standard conformal inferences can be guaranteed to be conservative.

\subsection{Coverage lower bounds under a general linear contamination model}
\label{subsec:linear-contam-model}

To obtain more interpretable and actionable expressions for the coverage bounds presented above, it is necessary to introduce some assumptions about the relation between the latent labels $Y$ and the observable labels $\tilde{Y}$.
Fortunately, significant progress can be made by simply assuming that $\tilde{Y}$ is conditionally independent of $X$ given $Y$.
This corresponds to a widely used class of label contamination models \citep{natarajan2013learning, ghosh2017robust}.
\begin{assumption}\label{assumption:linear-contam}
$\tilde{Y} \independent X \mid Y$.
\end{assumption}

Assumption~\ref{assumption:linear-contam} gives rise to a convenient relation between the distribution of $X \mid \tilde{Y}=k$, namely $\tilde{P}_k$, and the distributions of $X \mid Y=l$, namely $P_l$, for all $k,l \in [K]$.
In plain words, for any $k \in [K]$, the distribution $\tilde{P}_k$ is a linear mixture of the distributions $P_l$ for all $l \in [K]$.

\begin{proposition}\label{prop:indep-linear}
  Let $M$ be a $K \times K$ matrix such that $M_{kl} = \mathbb{P}[Y=l \mid \tilde{Y}=k]$ for any $k,l \in [K]$.
  Then, Assumption~\ref{assumption:linear-contam} implies that 
\begin{align}       \label{eq:contam_model}
  \tilde{P}_k = \sum_{l=1}^{K} M_{kl} P_l.
\end{align}
\end{proposition}

The mixture model in~\eqref{eq:contam_model} extends the classical binary class-dependent noise model in \citet{scott2013classification} to the multi-class setting. 
Further, combined with Bayes' theorem, Assumption~\ref{assumption:linear-contam} leads to: $\mathbb{P}[\tilde{Y}=k \mid  X, Y=l] = M_{kl}\tilde{\rho}_k/\rho_l$, where $\rho_k := \mathbb{P}[Y = k]$ and $\tilde{\rho}_k = \mathbb{P}[\tilde{Y} = k]$ denote the marginal distributions of the clean and corrupted labels, respectively.
If $M_{kl}=\eta / (K-1)$ for all $k\neq l$, where $0\leq\eta\leq1$, and $\rho_k=\tilde{\rho}_k= 1/K$ for all $k\in[K]$, this model reduces to the simple homogeneous noise setting \citep{ghosh2017robust}. 

We remark that it should now be clear how our problem is distinct from that of label shift \citep{podkopaev2021distribution,si2023pac}.
The latter refers to situations in which the marginal distribution of $\tilde{Y}$ may differ from that of $Y$ but $\tilde{P}_k = P_k$ for all $k \in [K]$, while label contamination generally leads to $\tilde{P}_k \neq P_k$. 
In this sense, label contamination is more challenging because the distribution of $X \mid \tilde{Y}$ used in training and calibration differs from that of  $X \mid Y$ used in testing. 
At the same time, our problem is also distinct from {\em covariate shift} \citep{tibshirani2019conformal} because in our case the {\em marginal} distribution of $X$ does not vary.

The first useful implication of Proposition~\ref{prop:indep-linear} is that standard conformal inferences under label contamination are often conservative.


\begin{corollary} \label{cor:coverage-cond}
Consider the same setting of Theorem~\ref{thm:coverage-lab-cond} and assume Assumption~\ref{assumption:linear-contam} holds.
Suppose also that the cumulative distribution functions of the scores \eqref{eq:cdf-scores} satisfy
    \begin{equation} \label{eq:assump_scores-cond}
        \max_{l \neq k} F_l^{k}(t)  \leq F_k^{k}(t),
      \end{equation}
      for all $t\in\mathbb{R}$ and $k\in[K]$.
Then, $\Delta_k(\hat{\tau}_k) \geq 0$ almost-surely for any $k \in [K]$, and hence the predictions sets $\hat{C}(X_{n+1})$ output by Algorithm~\ref{alg:standard-lab-cond} satisfy~\eqref{eq:def-lab-cond-coverage}.
\end{corollary}

The ``stochastic dominance'' condition in \eqref{eq:assump_scores-cond} states that the trained model tends to assign smaller scores $\hat{s}(X,k)$ when $Y=k$.
If $\hat{s}(X,k) = 1- \hat{\pi}(X,k)$, 
this is intuitively equivalent to requiring the model to estimate the distribution of $Y \mid X$ sufficiently accurately as to at least preserve the relative ranking of the most likely labels, on average.

In summary, Corollary~\ref{cor:coverage-cond}, together with its extension to marginal coverage presented in Section~\ref{app:extensions-theory}, provides a lower bound that highlights a certain robustness of standard conformal inferences to label contamination, consistently with \citet{einbinder2022conformal}. 
However, these results are not yet fully satisfactory for at least two reasons. Firstly, it is unclear how to check whether the condition in~\eqref{eq:assump_scores-cond} holds in practice.
Secondly, even if~\eqref{eq:assump_scores-cond} is satisfied, one may be concerned that standard conformal inferences can be too conservative under label contamination, leading to unnecessarily large prediction sets.
This is why we develop in the next section novel methods that can automatically adapt to label contamination, producing more informative prediction sets that rigorously guarantee coverage at the desired level.

We conclude this section by noting that Theorem~\ref{thm:coverage-lab-cond} could also be applied to derive theoretical {\em worst-case} coverage bounds for standard conformal prediction sets calibrated with contaminated data, even without any (empirical) information about the conformity score distribution. 
Section~\ref{app:worst-case} elaborates on this subject and discusses its connection to the elegant theoretical insights of \citet{barber2022conformal}.
However, it is worth emphasizing that such worst-case bounds have limited practical relevance in the context of this paper.
This is because worst-case bounds, while theoretically interesting, do not offer actionable guidance on how to enhance the informativeness of standard conformal prediction sets in the face of random label contamination.

\section{General methodology} \label{sec:methods}

\subsection{Adaptive coverage under a known label contamination model} \label{sec:method-known-noise}

We present a method for constructing prediction sets that automatically adapt to label contamination. 
For simplicity, we begin by focusing on label-conditional coverage assuming that the contamination model is known. 
Subsequently, we will extend similar ideas to accommodate unknown contamination models, and to provide other types of coverage guarantees.
The assumption of a known contamination model is convenient and provides a useful stepping stone for our next developments.
Further, this assumption  is well justified in several interesting applications, such as those involving controlled label randomization designed to ensure label-differential privacy \citep{evfimievski2003limiting,kasiviswanathan2011can,ghazi2021deep}.

\subsubsection{A plug-in estimate for the coverage inflation factor}

Our method leverages Assumption~\ref{assumption:linear-contam} through Proposition~\ref{prop:indep-linear}, which makes it possible to write the inflation factor $\Delta_k(t)$ in~\eqref{eq:delta} in terms of quantities that are either known or estimable.
 In fact, if $M$ in~\eqref{eq:contam_model} admits a matrix inverse $V = M^{-1}$, the factor $\Delta_k(t)$ can be expressed as:
\begin{align} \label{eq:delta-2}
  \Delta_k(t)
  & = (V_{kk}-1) \tF_k^{k}(t) + \sum_{l \neq k} V_{kl} \tF_l^{k}(t).
\end{align}
This expression only depends on $V$, which is assumed to be known, and on the distributions of the scores computed from $\tilde{Y}$, which are observable. 
This suggests it may be possible to estimate $\Delta_k(t)$ from the data and then leverage Theorem~\ref{thm:coverage-lab-cond} to obtain an adaptive prediction method with tighter coverage guarantees compared to the standard approach studied in Section~\ref{sec:preliminaries}. 
In the following, we develop such a method and establish both upper and lower bounds for its coverage, assuming that $V$ is known. The problem of estimating $V$ will be addressed later.

Our method begins by randomly splitting the labeled data into two disjoint subsets, $\mathcal{D}^{\text{train}}$ and $\mathcal{D}^{\mathrm{cal}}$, similarly to standard conformal inference. 
The observations in $\mathcal{D}^{\text{train}}$ are used to train the model $\hat{\pi}$, while those in $\mathcal{D}^{\mathrm{cal}}$ are used to compute conformity scores $\hat{s}(X_i, k)$ via~\eqref{eq:conf-scores}, for all $k \in [K]$ and $i \in \mathcal{D}^{\mathrm{cal}}$, according to the desired prediction function $\mathcal{C}$.
For any $k,l \in [K]$, let $\hat{F}_l^{k}$ denote the empirical cumulative distribution function of $\hat{s}(X_i,k)$ for $i \in \mathcal{D}_l^{\mathrm{cal}} = \{ i \in \mathcal{D}^{\mathrm{cal}} : \tilde{Y}_i = l\}$; i.e.,
\begin{align} \label{eq:F-hat}
  \hat{F}_l^{k}(t) := \frac{1}{n_l} \sum_{i \in \mathcal{D}_l^{\mathrm{cal}}} \I{\hat{s}(X_i,k) \leq t},
\end{align}
where $n_l = |\mathcal{D}_l^{\mathrm{cal}}|$. In other words, $\hat{F}_l^{k}(t)$ intuitively estimates $\tilde{F}_l^{k}$.
If the matrix $V$ is known, one can leverage the $\hat{F}_l^{k}$ functions to compute a plug-in estimate of $\Delta_k(t)$:
\begin{align} \label{eq:delta-hat}
  \hat{\Delta}_k(t)
  & := (V_{kk}-1) \hat{F}_k^{k}(t) + \sum_{l \neq k} V_{kl} \hat{F}_l^{k}(t).
\end{align}

If $\hat{\Delta}_k$ could estimate $\Delta_k$ accurately, one would guess from Theorem~\ref{thm:coverage-lab-cond} that Algorithm~\ref{alg:standard-lab-cond}---the standard conformal inference method that ignores label contamination---leads to an effective coverage close to $1 - \alpha + \hat{\Delta}_k(\hat{\tau}_k)$, where $\hat{\tau}_k$ is the data-driven calibration parameter computed via~\eqref{eq:def-tau-hat}.
This suggests adjusting the nominal significance level to something close to $\alpha - \hat{\Delta}_k(\hat{\tau}_k)$ to achieve $1-\alpha$ coverage. We will now translate this intuition into a rigorous method.

\subsubsection{The adaptive calibration algorithm}

For any $k \in [K]$, define the set $\hat{\mathcal{I}}_k \subseteq [n_k]$ as
\begin{align} \label{eq:Ik-set}
  \hat{\mathcal{I}}_k := \left\{i \in [n_k] : \frac{i}{n_k} \geq 1 - \alpha - \hDelta_k(S^k_{(i)}) + \delta(n_k,n_{*})  \right\},
\end{align}
where $S^k_{(i)}$, for $i \in [n_k]$, are the ascending order statistics of $\{\hat{s}(X_j, k)\}_{j \in \mathcal{D}_j^{\mathrm{cal}}}$, while $n_{*} := \min_{k \in [K]} n_k$, and $\delta(n_k,n_{*})$ is a correction factor specified later. 
Then, our threshold $\hat{\tau}_k$ is:
\begin{align} \label{eq:Ik-set-tauk}
  \hat{\tau}_k
  & = \begin{cases}
    S^k_{(\hat{i}_k)} \text{ where } \hat{i}_k = \min\{i \in \hat{\mathcal{I}_k}\}, & \text{if } \hat{\mathcal{I}}_k \neq \emptyset, \\
    1, & \text{if } \hat{\mathcal{I}}_k = \emptyset.
  \end{cases}
\end{align}
Finally, the adaptive prediction set output by our proposed method is $\hat{C}(X_{n+1}) = \mathcal{C}(X, \hat{\tau})$, where $\hat{\tau}=(\hat{\tau}_1, \ldots,\hat{\tau}_K)$.
This procedure is outlined by Algorithm~\ref{alg:correction}. 

\begin{algorithm}[!ht]
\DontPrintSemicolon

\KwIn{Data set $\{(X_i, \tilde{Y}_i)\}_{i=1}^{n}$ with corrupted labels $\tilde{Y}_i \in [K]$.}
\myinput{
The inverse $V$ of the matrix $M$ in \eqref{eq:contam_model}.}
\myinput{Unlabeled test point with features $X_{n+1}$.}
\myinput{Machine learning algorithm $\mathcal{A}$ for training a $K$-class classifier.}
\myinput{Prediction function $\mathcal{C}$ satisfying Definition~\ref{def:pred-function}; e.g., \eqref{eq:pred-function-hps}.}
\myinput{Desired coverage level $1-\alpha \in (0,1)$.}

Randomly split $[n]$ into two disjoint subsets, $\mathcal{D}^{\text{train}}$ and $\mathcal{D}^{\mathrm{cal}}$.\;
Train the classifier $\mathcal{A}$ on the data in $\mathcal{D}^{\text{train}}$. \;
Compute conformity scores $\hat{s}(X_i,k)$ using \eqref{eq:conf-scores} for all $i \in \mathcal{D}^{\mathrm{cal}}$ and all $k \in [K]$.\;
Define the empirical CDF $\hat{F}_l^{k}$ of $\{\hat{s}(X_i,k) : i \in \mathcal{D}_l^{\mathrm{cal}} \}$ for all $l \in [K]$,  as in \eqref{eq:F-hat}. \;
\For{$k=1, \dots, K$}{
  Define $\mathcal{D}_k^{\mathrm{cal}} = \{ i \in \mathcal{D}^{\mathrm{cal}} : \tilde{Y}_i = k\}$ and $n_k = |\mathcal{D}_k^{\mathrm{cal}}|$.\;
  Sort $\{\hat{s}(X_i,k) : i \in \mathcal{D}_k^{\mathrm{cal}} \}$ into $(S^k_{(1)}, S^k_{(2)}, \dots, S^k_{(n_k)})$, in ascending order. \;
  Compute $\hat{F}_l^{k}(S^k_{(i)})$ for all $i \in [n_k]$ and $l \in [K]$. \;
  Compute $\hDelta_k(S^k_{(i)})$ for all $i \in [n_k]$, as in \eqref{eq:delta-hat}. \;
  Compute $\delta(n_k,n_{*})$ using \eqref{eq:delta-constant}, based on a Monte Carlo estimate of $c(n_k)$ in~\eqref{eq:define-cn}.\;
  Construct the set $\hat{\mathcal{I}}_k \subseteq [K]$ as in \eqref{eq:Ik-set}.\;
  Evaluate $\hat{\tau}_k$ based on $\hat{\mathcal{I}}_k$ as in \eqref{eq:Ik-set-tauk}.

}
Evaluate $\hat{C}(X_{n+1}) = \mathcal{C}(X_{n+1}, \hat{\tau}; \hat{\pi})$, where $\hat{\tau} = (\hat{\tau}_1, \ldots, \hat{\tau}_K)$.

\nonl
\textbf{Output: } Conformal prediction set $\hat{C}(X_{n+1})$ for $Y_{n+1}$.

\caption{Adaptive classification under a known label contamination model}
\label{alg:correction}
\end{algorithm}

If we ignored the noise and finite-sample correction terms (i.e., imagining that $\hDelta_k(S^k_{(i)})=0$ and $\delta(n_k,n_{*})=0$), the threshold $\hat{\tau}_k$ in~\eqref{eq:Ik-set-tauk} would intuitively reduce to the $\lceil (1-\alpha)n_k \rceil $-th smallest value among the conformity scores for the calibration points with label $k$, similarly to the standard method reviewed in Section~\ref{sec:background}.
In general, though, the more complicated form of $\hat{\tau}_k$ in~\eqref{eq:Ik-set-tauk} is designed to approximately cancel the unknown coverage inflation factor $\Delta_k(\htau_k)$ arising when the standard conformal inference method is applied to contaminated data, as described by Theorem~\ref{thm:coverage-lab-cond}.
The main purpose of $\delta(n_k,n_{*})$ in~\eqref{eq:Ik-set} is to account for possible random errors in the estimation of the unknown function $\Delta_k(\cdot)$ through $\hDelta_k(\cdot)$, allowing us to obtain finite-sample guarantees.
The exact form of this correction term is discussed next.

For any $k \in [K]$, let $U_1, \ldots, U_{n_k}$ be i.i.d.~uniform random variables on $[0,1]$, and denote their order statistics as $U_{(1)}, \ldots, U_{(n_k)}$.
Then, define
\begin{align} \label{eq:define-cn}
  c(n_k)
  & := \EV{ \sup_{i \in [n_k]} \left\{ \frac{i}{n_k} - U_{(i)} \right\} },
\end{align}
and 
\begin{align} \label{eq:delta-constant}
  \delta(n_k,n_{*})
  := c(n_k) + \frac{2 \sum_{l\neq k}|V_{kl}| }{\sqrt{n_{*}}} \min \left\{ K \sqrt{\frac{\pi}{2}} , \frac{1}{\sqrt{n_{*}}} + \sqrt{\frac{\log(2K) + \log(n_{*})}{2}} \right\}.
\end{align}
We know from classical results in empirical process theory that $c(n_k)$ in~\eqref{eq:define-cn} scales as $1/\sqrt{n_{k}}$ if $n_{k}$ is large, and thus the overall correction term $\delta(n_k,n_{*})$ tends to vanish as $1/\sqrt{n_{*}}$ in the large-sample limit.
Note that the constant $c(n_k)$ in \eqref{eq:define-cn} will be assumed henceforth to be known because it can be easily estimated up to arbitrary precision via a Monte Carlo simulation of $n_k$ independent standard uniform random variables.
Combined with the adaptive nature of our threshold $\hat{\tau}_k$ in~\eqref{eq:Ik-set-tauk}, this finite-sample correction allows Algorithm~\ref{alg:correction} to enjoy a stronger coverage guarantee under label contamination compared to standard conformal prediction.

\begin{theorem} \label{thm:algorithm-correction}
Suppose $(X_i,Y_i,\tilde{Y}_i)$ are i.i.d.~for all $i \in [n+1]$, and that Assumption~\ref{assumption:linear-contam} holds.
Fix any prediction function $\mathcal{C}$ satisfying Definition~\ref{def:pred-function}, and let $\hat{C}(X_{n+1})$ indicate the prediction set output by Algorithm~\ref{alg:correction} based on the inverse $V$ of the model matrix $M$ in the label contamination model~\eqref{eq:contam_model}.
Then, $\mathbb{P}[Y_{n+1} \in \hat{C}(X_{n+1}) \mid Y = k] \geq 1 - \alpha$ for all $k \in [K]$.
\end{theorem}

Intuitively, this says that Algorithm~\ref{alg:correction} provides valid prediction sets at level $1-\alpha$ despite the label noise.
Crucially, this does not require any assumptions about the classifier's accuracy, in contrast with the potentially more delicate behavior of standard conformal inferences (Algorithm~\ref{alg:standard-lab-cond}); i.e., see Theorem~\ref{thm:coverage-lab-cond} and Corollary~\ref{cor:coverage-cond}.
Further, under some additional regularity conditions, it can be proved that Algorithm~\ref{alg:correction} is not overly conservative, as discussed next.

\subsubsection{A coverage upper bound}

The assumptions needed for our coverage upper bound are stated here and explained below.
\begin{assumption} \label{assumption:regularity-dist}
  For all $k,l \in [K]$, the cumulative distribution functions $\tF^{k}_{l}$ are differentiable on the interval $(0,1)$, and the corresponding densities $\tf^{k}_{l}$ are uniformly bounded with $\|\tf^{k}_{l}\|_{\infty} \leq f_{\max}$, for some $f_{\max} > 0$. Further, $f_{\min} := \min_{k \in [K]} \inf_{t\in(0,1)} \tf^k_k(t) > 0$.
\end{assumption}

\begin{assumption} \label{assumption:consitency-scores}
For any $k \in [K]$, the cumulative distribution function $\tF^{k}_{k}$ satisfies
\begin{align*}
  \max_{l \neq k} \tilde{F}^k_l(t) \leq \tF^{k}_{k}(t), \qquad \forall t \in \mathbb{R}.
\end{align*}
\end{assumption}

\begin{assumption} \label{assumption:regularity-dist-delta}
The coverage inflation factor $\Delta_k(t)$ is bounded from below by:
\begin{align*}
  \inf_{t \in (0,1)} \Delta_k(t) \geq - \alpha + c(n_k) + \frac{2 \sum_{l\neq k}|V_{kl}|}{\sqrt{n_{*}}} \left( \frac{1}{\sqrt{n_{*}}} + 2 \sqrt{\frac{\log(2K) + \log(n_{*})}{2}} \right).
\end{align*}
\end{assumption}

Assumption~\ref{assumption:regularity-dist} merely requires that the distribution of the conformity scores should be continuous with bounded density; this can be ensured in practice by adding a small amount of random noise to the scores computed by any classifier.
Assumption~\ref{assumption:consitency-scores} simply states that the classifier tends to assign smaller scores $\hat{s}(X,k)$ to data points with $\tilde{Y}=k$.
This may be reminiscent of the stochastic dominance condition in Corollary~\ref{cor:coverage-cond}, although it is different and arguably weaker.
In fact, the classifier is trained on data with corrupted labels. Therefore, as long as it can achieve non-trivial prediction accuracy, it should assign smaller scores $\hat{s}(X,k)$ when $\tilde{Y}=k$.

Assumption~\ref{assumption:regularity-dist-delta} looks slightly more involved, but it is also quite realistic. For example, it is always satisfied in the large-sample limit, $n_{*} \to \infty$, if the stochastic dominance condition defined in~\eqref{eq:assump_scores-cond} holds, because in that case $\inf_{t \in (0,1)} \Delta_k(t) \geq 0$.
Further, as discussed in more detail in Section~\ref{app:simplified-methods}, Assumption~\ref{assumption:regularity-dist-delta} can also be satisfied if the stochastic dominance condition in~\eqref{eq:assump_scores-cond} does not hold, as long as some additional assumptions are imposed on the label contamination model.
Under this setup, a finite-sample upper bound for the coverage of the conformal prediction sets output by Algorithm~\ref{alg:correction} is established below.

\begin{theorem} \label{thm:algorithm-correction-upper}
Under the setup of Theorem~\ref{thm:algorithm-correction}, let $\hat{C}(X_{n+1})$ be the prediction set output by Algorithm~\ref{alg:correction} based on the 
inverse $V$ of the matrix $M$ in~\eqref{eq:contam_model}. Suppose also that Assumptions~\ref{assumption:regularity-dist}--\ref{assumption:regularity-dist-delta} hold. Then, $\mathbb{P}[Y_{n+1} \in \hat{C}(X_{n+1}) \mid Y = k] \leq  1 - \alpha + \varphi_k(n_k,n_{*})$ for all $k \in [K]$,
where
  \begin{align*}
    \varphi_k(n_k,n_{*})
    & = 2 \delta(n_k,n_{*}) + \frac{1}{n_{*}} + \frac{1 + 2 \sum_{l\neq k} |V_{kl}| \frac{f_{\max}}{f_{\min}} \sum_{j=1}^{n_k+1} \frac{1}{j}}{n_k} + \frac{V_{kk} + \sum_{l \neq k} |V_{kl}|}{n_k+1}.
  \end{align*}
\end{theorem}

Thus, the sets output by Algorithm~\ref{alg:correction} are {\em asymptotically tight} because $\varphi_k(n_k,n_{*}) \to 0$ as $n_{*} \to \infty$.
While this is already encouraging about the efficiency of Algorithm~\ref{alg:correction}, our method can be further refined to produce even more informative prediction sets that remain valid in those (rather common) scenarios in which standard conformal inferences are too conservative.

\subsubsection{Boosting power with more optimistic calibration} \label{sec:boosting-optimist}

We know from Corollary~\ref{cor:coverage-cond} that even the standard conformal inference approach of Algorithm~\ref{alg:standard-lab-cond} is conservative under the (relatively mild) stochastic dominance condition in \eqref{eq:assump_scores-cond}.
This motivates us to devise a hybrid method that can outperform both Algorithm~\ref{alg:correction} and Algorithm~\ref{alg:standard-lab-cond}, while retaining guaranteed coverage under a slightly stronger version of \eqref{eq:assump_scores-cond}. 
Intuitively, the idea is to adaptively choose between Algorithm~\ref{alg:correction} and Algorithm~\ref{alg:standard-lab-cond} depending on which approach leads to a lower (less conservative) calibrated threshold.
In other words, we propose to apply Algorithm~\ref{alg:correction} with the set $\hat{\mathcal{I}}_k$ in~\eqref{eq:Ik-set} replaced by
\begin{align} \label{eq:Ik-set-optimist}
  \hat{\mathcal{I}}_k := \left\{i \in [n_k] : \frac{i}{n_k} \geq 1 - \alpha - \max\left\{\hDelta_k(S^k_{(i)}) - \delta(n_k,n_{*}), -\frac{1-\alpha}{n_k} \right\}  \right\}.
\end{align}
Perhaps surprisingly, this somewhat greedy approach typically produces valid predictions.

\begin{proposition} \label{thm:algorithm-correction-optimistic}
Under the setup of Theorem~\ref{thm:algorithm-correction}, assume also that $\inf_{t\in \R}\Delta_k(t) \geq \delta(n_k, n_{*}) - (1-\alpha)/n_k$.
If $\hat{C}(X_{n+1})$ is the prediction set output by Algorithm~\ref{alg:correction} applied with the set $\hat{\mathcal{I}}_k$ defined in~\eqref{eq:Ik-set-optimist} instead of~\eqref{eq:Ik-set}, then $\mathbb{P}[ Y_{n+1} \in \hat{C}(X_{n+1}) \mid Y_{n+1} = k] \geq 1-\alpha$ for all $k \in [K]$.
\end{proposition}

The additional assumption of Proposition~\ref{thm:algorithm-correction-optimistic}, $\inf_{t\in \R}\Delta_k(t) \geq \delta(n_k, n_{*}) - (1-\alpha)/n_k$, is stronger than the stochastic dominance condition in~\eqref{eq:assump_scores-cond}, but it is not unrealistic. When the calibration set size $n_k$ is sufficiently large to make $\delta(n_k, n_{*})$ small, this assumption is closely related to~\eqref{eq:assump_scores-cond}, which implies $\inf_{t\in \R}\Delta_k(t) \geq 0$; see the proof of Corollary~\ref{cor:coverage-cond}.
In fact, Section~\ref{sec:empirical} will show that the hybrid method described in this section tends to work very well in practice.

\subsection{Adaptive coverage under a bounded label contamination model} \label{sec:method-bounded-noise}

We now extend Algorithm~\ref{alg:correction} by relaxing the assumption that the matrix $M$ in~\eqref{eq:contam_model} is fully known.
In particular, we assume only that $M$ is invertible and a joint confidence region $[\hat{V}^{\mathrm{low}}, \hat{V}^{\mathrm{upp}}]$ is available for the off-diagonal entries of $V = M^{-1}$, with $\hat{V}^{\mathrm{low}}, \hat{V}^{\mathrm{upp}} \in \mathbb{R}^{K \times K}$, such that $V_{kl} \in [\hat{V}_{kl}^{\mathrm{low}}, \hat{V}_{kl}^{\mathrm{upp}}]$ with probability at least $1-\alpha_V$ simultaneously for all $l \neq k$, at some significance level $\alpha_V \in (0,1)$.
Here, it is understood that the matrices $\hat{V}^{\mathrm{low}}$ and $\hat{V}^{\mathrm{upp}}$ are independent of the data utilized to calibrate our conformal inferences.
It should be anticipated that there will be some trade-offs involved in the choice of $\alpha_V$, which should generally not exceed the desired level $\alpha$ of the output conformal prediction sets, but this matter will become clearer later.

To simplify the notation in the following, it is helpful to define $\hat{\delta}^{(V)}_{kl} := \hat{V}^{\mathrm{upp}}_{kl} - \hat{V}^{\mathrm{low}}_{kl}$ for all $l \neq k$.
Further, it is useful to imagine that a (possibly very conservative) deterministic upper bound $\bar{V}^{\mathrm{upp}} \in \mathbb{R}^{K \times K}$ for the off-diagonal entries of $V$ is also known a priori, such that  $\max\{|\hat{V}_{kl}^{\mathrm{low}}|, |\hat{V}_{kl}^{\mathrm{upp}}|,|V_{kl}|\} \leq |\bar{V}_{kl}^{\mathrm{upp}}|$ almost-surely for all $l \neq k$.
We refer to Section~\ref{app:simplified-methods} for concrete examples of $\hat{V}^{\mathrm{low}}, \hat{V}^{\mathrm{upp}}$ and $\bar{V}^{\mathrm{upp}}$ corresponding to two special cases in which more specific knowledge about the structure of the matrix $V$ is also available.
In the meantime, here we continue describing our method in generality, without additional constraints on $V$.

For ease of notation, let us define also
\begin{align} \label{eq:def-zeta-k}
& \hat{\delta}^{(V)}_{k*} := \max_{l \neq k} \hat{\delta}^{(V)}_{kl},
& \hat{\zeta}_k := \max_{l \neq k} \left( \hat{V}^{\mathrm{upp}}_{kl} - V_{kl} \right) - \min_{l \neq k} \left( \hat{V}^{\mathrm{upp}}_{kl} - V_{kl} \right).
\end{align}
Intuitively, $\hat{\delta}^{(V)}_{k*}$ represents the width of the simultaneous confidence band for the off-diagonal entries of $V$ at its widest point, while $\hat{\zeta}_k$ quantifies the uniformity of that confidence band. In particular, $\hat{\zeta}_k=0$ if $\hat{V}^{\mathrm{low}}_{kl}$, $\hat{V}^{\mathrm{upp}}_{kl}$, and $V_{kl}$ are constant for all $l \neq k$, as one should expect to be true under the special label contamination model discussed in Section~\ref{app:RR-model}, for example.
Further, let $\hat{\zeta}^{\mathrm{upp}}_k$ denote a $1-\alpha_V$ upper confidence bound for $\hat{\zeta}_k$. This can generally be extracted directly from $[\hat{V}^{\mathrm{low}}, \hat{V}^{\mathrm{upp}}]$, since $\hat{\zeta}_k$ is a known functional of $\hat{V}_{kl}$ and $V_{kl}$ for $l \neq k$, but it could also be informed by additional prior knowledge about the structure of the matrix $V$. We refer to Section~\ref{app:BRR-model} for some concrete examples on how to compute $\hat{\zeta}_k$ in practice.

Then, our solution consists of applying Algorithm~\ref{alg:correction} after replacing $\hat{\Delta}_k(t)$ in~\eqref{eq:delta-hat} with
\begin{align} \label{eq:delta-hat-k-epsilon-bound}
\begin{split}
  \hat{\Delta}_k^{\mathrm{ci}}(t)
  & := \sum_{l \neq k} \hat{V}^{\mathrm{upp}}_{kl} \left( \hat{F}_l^{k}(t) - \hat{F}_k^{k}(t) \right) -  \hat{\delta}^{(V)}_{k*} (K-1) \left| \hat{F}_k^{k}(t) - \frac{\sum_{l \neq k}  \hat{F}_l^{k}(t)}{K-1} \right| \\
  & \qquad - |\hat{\zeta}^{\mathrm{upp}}_k| \sum_{l \neq k} \left| \hat{F}_l^{k}(t) - \hat{F}_k^{k}(t) \right|,
\end{split}
\end{align}
the correction factor $\delta(n_k,n_{*})$ in \eqref{eq:delta-constant} with 
\begin{align} \label{eq:delta-constant-epsilon-ci}
\begin{split}
  \delta^{\mathrm{ci}}(n_k, n_*)
  & := c(n_k) + \frac{2 \sum_{l \neq k} \left(|\hat{V}^{\mathrm{upp}}_{kl}| + \hat{\delta}^{(V)}_{kl}\right)}{\sqrt{n_{*}}} \min \left\{ K \sqrt{\frac{\pi}{2}} , \frac{1}{\sqrt{n_{*}}} + \sqrt{\frac{\log(2K) + \log(n_{*})}{2}} \right\} \\
  & \qquad + 2 \alpha_V \sum_{l \neq k} |\bar{V}_{kl}^{\mathrm{upp}}|,
\end{split}
\end{align}
and the set $\hat{\mathcal{I}}_k \subseteq [K]$ in \eqref{eq:Ik-set} with
\begin{align} \label{eq:Ik-set-ci}
  \hat{\mathcal{I}}^{\mathrm{ci}}_k := \left\{i \in [n_{k}] : \frac{i}{n_{k}} \geq 1 - \alpha - \hDelta^{\mathrm{ci}}_k(S_{(i)}) + \delta^{\mathrm{ci}}(n_k,n_{*})  \right\}.
\end{align}
This method, outlined by Algorithm~\ref{alg:correction-ci}, provably achieves label-conditional coverage.

\begin{algorithm}[!ht]
\DontPrintSemicolon

\KwIn{Data set $\{(X_i, \tilde{Y}_i)\}_{i=1}^{n}$ with corrupted labels $\tilde{Y}_i \in [K]$.}
\myinput{A $1-\alpha_V$ joint confidence region $[\hat{V}^{\mathrm{low}}, \hat{V}^{\mathrm{upp}}]$ for the off-diagonal elements of $V$.}
\myinput{Constants $\bar{V}_{kl}^{\mathrm{upp}}$: $\max\{|\hat{V}_{kl}^{\mathrm{upp}}|, |V_{kl}|\} \leq |\bar{V}_{kl}^{\mathrm{upp}}|$ almost-surely for all $l \neq k$.}
\myinput{Unlabeled test point with features $X_{n+1}$.}
\myinput{Machine learning algorithm $\mathcal{A}$ for training a $K$-class classifier.}
\myinput{Prediction function $\mathcal{C}$ satisfying Definition~\ref{def:pred-function}; e.g., \eqref{eq:pred-function-hps}.}
\myinput{Desired significance level $\alpha \in (0,1)$.}

Randomly split $[n]$ into two disjoint subsets, $\mathcal{D}^{\text{train}}$ and $\mathcal{D}^{\mathrm{cal}}$.\;
Train the classifier $\mathcal{A}$ on the data in $\mathcal{D}^{\text{train}}$. \;
Compute conformity scores $\hat{s}(X_i,k)$ using \eqref{eq:conf-scores} for all $i \in \mathcal{D}^{\mathrm{cal}}$ and all $k \in [K]$.\;
Define the empirical CDF $\hat{F}_l^{k}$ of $\{\hat{s}(X_i,k) : i \in \mathcal{D}_l^{\mathrm{cal}} \}$ for all $l \in [K]$,  as in \eqref{eq:F-hat}. \;
\For{$k=1, \dots, K$}{
  Define $\mathcal{D}_k^{\mathrm{cal}} = \{ i \in \mathcal{D}^{\mathrm{cal}} : \tilde{Y}_i = k\}$ and $n_k = |\mathcal{D}_k^{\mathrm{cal}}|$.\;
  Sort $\{\hat{s}(X_i,k) : i \in \mathcal{D}_k^{\mathrm{cal}} \}$ into $(S^k_{(1)}, S^k_{(2)}, \dots, S^k_{(n_k)})$, in ascending order. \;
  Compute $\hat{F}_l^{k}(S^k_{(i)})$ for all $i \in [n_k]$ and $l \in [K]$. \;
  Compute $\hDelta_k^{\mathrm{ci}}(S^k_{(i)})$ for all $i \in [n_k]$, as in \eqref{eq:delta-hat-k-epsilon-bound}. \;
  Compute $\delta^{\mathrm{ci}}(n_k, n_*)$ using \eqref{eq:delta-constant-epsilon-ci}, based on a Monte Carlo estimate of $c(n_k)$ in~\eqref{eq:define-cn}.\;
  Construct the set $\hat{\mathcal{I}}^{\mathrm{ci}}_k \subseteq [K]$ as in \eqref{eq:Ik-set-ci}, based on $\hDelta_k^{\mathrm{ci}}$ and $\delta^{\mathrm{ci}}$.\;
  Evaluate $\hat{\tau}_k$ by applying \eqref{eq:Ik-set-tauk} with $\hat{\mathcal{I}}^{\mathrm{ci}}_k$ instead of $\hat{\mathcal{I}}_k$.

}
Evaluate $\hat{C}^{\mathrm{ci}}(X_{n+1}) = \mathcal{C}(X_{n+1}, \hat{\tau}; \hat{\pi})$, where $\hat{\tau} = (\hat{\tau}_1, \ldots, \hat{\tau}_K)$.

\nonl
\textbf{Output: } Conformal prediction set $\hat{C}^{\mathrm{ci}}(X_{n+1})$ for $Y_{n+1}$.

\caption{Adaptive classification under a bounded label contamination model}
\label{alg:correction-ci}
\end{algorithm}

\begin{theorem} \label{thm:algorithm-correction-ci}
Suppose $(X_i,Y_i,\tilde{Y}_i)$ are i.i.d.~for all $i \in [n+1]$.
Assume the general linear mixture contamination model described in Section~\ref{subsec:linear-contam-model} holds, with $V = M^{-1}$.
Fix any prediction function $\mathcal{C}$ satisfying Definition~\ref{def:pred-function}, and let $\hat{C}^{\mathrm{ci}}(X_{n+1})$ indicate the prediction set output by Algorithm~\ref{alg:correction-ci} based on an independent simultaneous confidence region $[\hat{V}^{\mathrm{low}}, \hat{V}^{\mathrm{upp}}]$ such that $V_{kl} \in [\hat{V}_{kl}^{\mathrm{low}}, \hat{V}_{kl}^{\mathrm{upp}}]$ for all $l \neq k$ with probability at least $1-\alpha_V$.
Assume also that $\max\{|\hat{V}_{kl}^{\mathrm{upp}}|, |V_{kl}|\} \leq |\bar{V}_{kl}^{\mathrm{upp}}|$ almost-surely for all $l \neq k$, for some known constants $|\bar{V}_{kl}^{\mathrm{upp}}| > 0$.
Then, $\mathbb{P}[Y_{n+1} \in \hat{C}^{\mathrm{ci}}(X_{n+1}) \mid Y = k] \geq 1 - \alpha$ for all $k \in [K]$.
\end{theorem}

It is now clear that the level $\alpha_V$ of the confidence region for $V$ affects the magnitude of the finite-sample correction term in~\eqref{eq:delta-constant-epsilon-ci} through the product $\alpha_V \sum_{l \neq k} |\bar{V}_{kl}^{\mathrm{upp}}|$. Therefore, there is an important trade-off in the choice of $\alpha_V$, because smaller values of the latter tend to lead to larger upper bounds $|\bar{V}_{kl}^{\mathrm{upp}}|$ and $|\hat{V}_{kl}^{\mathrm{upp}}|$.
In practice, we have observed that choosing a relatively small value such as $\alpha_V=0.01$ often works well in practice, although it may not be optimal.

The following theorem formalizes the intuition that the prediction sets computed by Algorithm~\ref{alg:correction-ci} tend to be more informative if the confidence bands for $V_{kl}$ are tighter.
This result naturally extends Theorem~\ref{thm:algorithm-correction-upper} to provide a coverage upper bound for the conformal prediction sets output by Algorithm~\ref{alg:correction-ci}. Similarly to Theorem~\ref{thm:algorithm-correction-upper}, three technical conditions are needed: Assumptions \ref{assumption:regularity-dist}, \ref{assumption:consitency-scores} and~\ref{assumption:regularity-dist-delta-ci}, with the latter being a suitable variation of Assumption~\ref{assumption:regularity-dist-delta}.

\begin{assumption} \label{assumption:regularity-dist-delta-ci}
The coverage inflation factor $\Delta_k(t)$ is almost-surely bounded by:
\begin{align} \label{eq:regularity-dist-delta-ci}
\begin{split}
  \inf_{t \in (0,1)} \Delta_k(t) 
  & \geq - \alpha + c(n_k) + \frac{2 \sum_{l \neq k} \left( |\hat{V}^{\mathrm{upp}}_{kl}| + \hat{\delta}^{(V)}_{kl} \right) }{\sqrt{n_{*}}} \left( \frac{1}{\sqrt{n_{*}}} + \sqrt{\frac{\log(2K) + \log(n_{*})}{2}} \right) \\
  & \qquad + 2 \alpha_V \sum_{l \neq k} |\bar{V}_{kl}^{\mathrm{upp}}| + (K-1) \left( \hat{\delta}^{(V)}_{k*}  + |\hat{\zeta}^{\mathrm{upp}}_k| \right).
\end{split}
\end{align}
\end{assumption}

The difference between Assumption~\ref{assumption:regularity-dist-delta-ci} and Assumption~\ref{assumption:regularity-dist-delta} is that the upper bound for $\Delta_k(t) $ imposed by the latter is a random variable that depends on the confidence region $[\hat{V}^{\mathrm{low}}, \hat{V}^{\mathrm{upp}}]$.
In the limit of $n_{*} \to \infty$, Assumption~\ref{assumption:regularity-dist-delta-ci} becomes approximately equivalent to 
\begin{align} \label{eq:epsilon-bound-asymptotically-tight-ci}
  \sum_{l \neq k} |V_{kl}| + (K-1) \left( \hat{\delta}^{(V)}_{k*}  + |\hat{\zeta}^{\mathrm{upp}}_k| \right)
  \leq \alpha - 2 \alpha_V \sum_{l \neq k} |\bar{V}_{kl}^{\mathrm{upp}}|.
\end{align}
This means that Assumption~\ref{assumption:regularity-dist-delta-ci} is often realistic, similarly to Assumption~\ref{assumption:regularity-dist-delta}, as long as $\alpha_V \leq \alpha$, the contaminated labels are not too different from the true labels, and the confidence region $[\hat{V}^{\mathrm{low}}, \hat{V}^{\mathrm{upp}}]$ is not too wide. 
We refer to Section~\ref{app:RR-model} for further details on the interpretation of~\eqref{eq:epsilon-bound-asymptotically-tight-ci}, which can be simplified under more specific label contamination models.
Moreover, we remark that Assumption~\ref{assumption:regularity-dist-delta-ci} is always satisfied in the large-sample limit if the stochastic dominance condition in~\eqref{eq:assump_scores-cond} holds (i.e., $\inf_{t \in (0,1)} \Delta_k(t) \geq 0$) and sufficiently tight confidence bounds $[\hat{V}^{\mathrm{low}}, \hat{V}^{\mathrm{upp}}]$ for a small enough $\alpha_V$ are available. 
Then, a finite-sample upper bound for the coverage of the conformal prediction sets output by Algorithm~\ref{alg:correction-ci} is established below.

\begin{theorem} \label{thm:algorithm-correction-upper-ci}
Under the setup of Theorem~\ref{thm:algorithm-correction-ci}, assume the general contamination model from Section~\ref{subsec:linear-contam-model} holds. Let $\hat{C}^{\mathrm{ci}}(X_{n+1})$ indicate the prediction set output by Algorithm~\ref{alg:correction-ci} based on an independent simultaneous confidence region $[\hat{V}^{\mathrm{low}}, \hat{V}^{\mathrm{upp}}]$ such that $V_{kl} \in [\hat{V}_{kl}^{\mathrm{low}}, \hat{V}_{kl}^{\mathrm{upp}}]$ for all $l \neq k$ with probability at least $1-\alpha_V$.
Assume also that $\max\{|\hat{V}_{kl}^{\mathrm{upp}}|, |V_{kl}|\} \leq |\bar{V}_{kl}^{\mathrm{upp}}|$ almost-surely for all $l \neq k$, for some known constants $|\bar{V}_{kl}^{\mathrm{upp}}| > 0$.
Suppose also that Assumptions~\ref{assumption:regularity-dist}, \ref{assumption:consitency-scores}, and~\ref{assumption:regularity-dist-delta-ci} hold. Then, $\mathbb{P}[Y_{n+1} \in \hat{C}^{\mathrm{ci}}(X_{n+1}) \mid Y = k] \leq  1 - \alpha + \varphi_k^{\mathrm{ci}}(n_k,n_{*}),$ for all $k \in [K]$, where
  \begin{align*}
    \varphi_k^{\mathrm{ci}}(n_k,n_{*})
  & = \frac{1}{n_{*}} + \left( 1 + 4 \sum_{l \neq k} |\bar{V}_{kl}^{\mathrm{upp}}| \right) \alpha_V + 2c(n_k) + (K-1) \left( 2 \EV{\hat{\delta}^{(V)}_{k*}}  + \EV{|\hat{\zeta}^{\mathrm{upp}}_k|} \right)  \\
    & \quad + \frac{4 \sum_{l \neq k} \EV{ |\hat{V}^{\mathrm{upp}}_{kl}| + \hat{\delta}^{(V)}_{kl} } }{\sqrt{n_{*}}} \min \left\{ K \sqrt{\frac{\pi}{2}} , \frac{1}{\sqrt{n_{*}}} + \sqrt{\frac{\log(2K) + \log(n_{*})}{2}} \right\}\\
  & \quad + \frac{2}{n_k} \left[ 1 + \sum_{l \neq k} |V_{kl}| + \sum_{l\neq k}|V_{kl}| \cdot \frac{f_{\max}}{f_{\min}} \cdot \sum_{j=1}^{n_k+1} \frac{1}{j} \right].
  \end{align*}

\end{theorem}

The interpretation of Theorem~\ref{thm:algorithm-correction-upper-ci} is similar to that of Theorem~\ref{thm:algorithm-correction-upper}, although now the unknown nature of the contamination model necessarily introduces some slack in the coverage upper bound. In particular, note that $\varphi_k^{\mathrm{ci}}(n_k,n_{*})$ converges to a finite quantity as $n_{*} \to \infty$, but our prediction sets can still be (approximately) tight if $\bar{V}_{kl}^{\mathrm{upp}}$ is finite for all $l \neq k$, $\alpha_V$ is small, and the expected lengths of all confidence intervals $[\hat{V}_{kl}^{\mathrm{low}}, \hat{V}_{kl}^{\mathrm{upp}}]$ are also small for all $l \neq k$.

We conclude this section by noting that the power of Algorithm~\ref{alg:correction-ci} can be further boosted without losing the coverage guarantee, as long as a relatively mild ``optimistic'' condition on the coverage inflation factor in~\eqref{eq:delta} holds.
Concretely, we propose to apply Algorithm~\ref{alg:correction-ci} based on
\begin{align} \label{eq:Ik-set-ci-optimist}
  \hat{\mathcal{I}}^{\mathrm{ci}}_k := \left\{i \in [n_{k}] : \frac{i}{n_{k}} \geq 1 - \alpha - \max\left\{ \hDelta^{\mathrm{ci}}_k(S_{(i)}) - \delta^{\mathrm{ci}}(n_k,n_{*}) , - \frac{1-\alpha}{n_k} \right\} \right\},
\end{align}
instead of the more conservative option described above in~\eqref{eq:Ik-set-ci}.
This leads to an optimistic variation of Algorithm~\ref{alg:correction-ci}, analogous to the extension of Algorithm~\ref{alg:correction} presented earlier in Section~\ref{sec:boosting-optimist}, which still enjoys similar coverage properties.

\begin{proposition} \label{thm:algorithm-correction-ci-optimist}
Under the setup of Theorem~\ref{thm:algorithm-correction-ci}, assume also that $\inf_{t\in \R}\Delta_k(t) \geq \delta^{\mathrm{ci}}(n_k, n_{*}) - (1-\alpha)/n_k$.
If $\hat{C}(X_{n+1})$ is the set output by Algorithm~\ref{alg:correction-ci} applied with the definition of $\hat{\mathcal{I}}^{\mathrm{ci}}_k$ in~\eqref{eq:Ik-set-ci-optimist} instead of~\eqref{eq:Ik-set-ci}, then $\mathbb{P}[Y_{n+1} \in \hat{C}^{\mathrm{ci}}(X_{n+1}) \mid Y = k] \geq 1 - \alpha$ for any $k \in [K]$.
\end{proposition}

\subsection{A general method for fitting the label contamination model} \label{sec:method-noise-estim}

We now shift our focus to the estimation of the contamination model.
On the one hand, practitioners sometimes have prior information about the matrix $M$ in \eqref{eq:contam_model}.
For example, they might have knowledge of the labeling processes \citep{sorokin2008utility} or previous experiences with related data.
Further, the contamination model can be known in applications involving differential privacy  \citep{warner1965randomized,evfimievski2003limiting,kasiviswanathan2011can,ghazi2021deep}.
On the other hand, it is interesting to also consider a data-driven approach that can estimate $M$ leveraging some observations of both contaminated and clean data.

Of course, if clean data are available, one may think of circumventing the problem studied in this paper by calibrating the conformal inferences without using the contaminated samples.
However, data with high-quality labels are often scarce, and they may not always be available in sufficient numbers to reliably calibrate the $K$ thresholds $\hat{\tau}_k$ needed to guarantee label-conditional coverage \citep{vovk2012conditional}, especially if $K$ is large. Further, even more abundant clean calibration data may be needed if one aims to achieve stronger guarantees such as equalized coverage over protected categories \citep{romano2019malice}.
By contrast, our label contamination model can easily incorporate well-justified structural constraints that reduce the number of parameters to be estimated, making our estimation task seem manageable even with a small clean data set.

With this premise, we begin to tackle the estimation of $V = M^{-1}$ from a general perspective, avoiding for the time being any additional assumptions about the structure of $M$.
Subsequently, in Section~\ref{app:simplified-methods}, our attention will turn to more specific contamination models that can be fitted quite accurately even with a very limited amount of clean data.

Let $\mathcal{D}^{1}$ denote a random contaminated data set, independent of $\mathcal{D}$ but identically distributed, with $n^{1}=|\mathcal{D}^{1}| < n$.
Similarly, let $\mathcal{D}^{0}$ denote a (smaller) clean data set containing i.i.d.~pairs of observations $(X_i^{0},Y_i^{0})$, for $i \in [n^{0}]$, independently drawn from the same distribution $P_{XY}$ corresponding to the data in $\mathcal{D}$.
First, randomly partition $\mathcal{D}^1$ into two disjoint subsets, $\mathcal{D}^1_{a}$ and $\mathcal{D}^1_{b}$.
The data in $\mathcal{D}^1_{a}$ are utilized to train a $K$-class classifier, possibly with the same machine learning algorithm utilized to compute the conformity scores.
 Let $\hat{f}(X) \in [K]$ denote the most likely label predicted by this classifier for a new sample with features $X$.
Then, define the matrices $\tilde{Q} \in [0,1]^{K \times K}$ and $Q \in [0,1]^{K \times K}$ such that, for any $l,k \in [K]$,
\begin{align} \label{eq:Q-def}
    & \tilde{Q}_{lk} := \P{ \hat{f}(X) = k \mid \tilde{Y}=l, \hat{f} }, 
    & Q_{lk} := \P{ \hat{f}(X) = k \mid Y=l, \hat{f} }.
\end{align}
As proved in Section~\ref{app:proofs}, the matrices $Q$ and $\tilde{Q}$ satisfy the following {\em estimating equation}:
\begin{align} \label{eq:Q-M}
  \tilde{Q} = M Q.
\end{align}
Thus, as long as $\tilde{Q}$ is invertible, $V := M^{-1}$ is given by
\begin{align} \label{eq:Q-V}
  V = Q \tilde{Q}^{-1}.
\end{align}

Equation~\eqref{eq:Q-V} suggests the following strategy for estimating $V$.
First, note that the data in $\mathcal{D}^{1}_{b}$ can be utilized to compute an intuitive point estimate of $\tilde{Q}$; that is, for each $l,k \in [K]$,
\begin{align} \label{eq:Q-tilde-estim}
    \tilde{Q}_{lk}
    & \approx \frac{\sum_{i \in \mathcal{D}^{1}_{b}} \I{\hat{f}(X_i)=k, \tilde{Y}=l}}{\sum_{i \in \mathcal{D}^{1}_{b}} \I{\tilde{Y}=l}} .
\end{align}
Similarly, the data in $\mathcal{D}^{0}$ can be utilized to obtain an intuitive point estimate of $Q$; i.e., 
\begin{align} \label{eq:Q-estim}
    Q_{lk}
    & \approx \frac{\sum_{i \in \mathcal{D}^{0}} \I{\hat{f}(X_i)=k, Y=l}}{\sum_{i \in \mathcal{D}^{0}} \I{Y=l}} .
\end{align}

If the contaminated observations are relatively abundant (i.e., $|\mathcal{D}^{1}| \gg |\mathcal{D}^{0}|$), Equation~\eqref{eq:Q-tilde-estim} provides an estimate of $\tilde{Q}$ with low variance compared to that of $Q$ in~\eqref{eq:Q-estim}.
Therefore, the leading source of uncertainty in $V$ comes from $Q$ and is due to the unknown joint distribution of $(\hat{f}(X), Y)$ conditional on $\hat{f}$.
This is a multinomial distribution with $K^2$ categories and event probabilities equal to $\lambda_{lk} = \P{ \hat{f}(X) = k, Y=l  \mid \hat{f} }$,
for all $l,k \in [K]$. Then, since $Q_{lk} = \lambda_{lk} / \sum_{s=1}^{K} \lambda_{ls}$ for any $l,k \in [K]$, each element of the matrix $V$ in \eqref{eq:Q-V} can be written as a function of the multinomial parameter vector $\lambda = (\lambda_{lk})_{l,k \in [K]}$ and of other quantities (i.e., $\tilde{Q}$) that are already known with relatively high accuracy:
\begin{align} \label{eq:multinomial-V}
  V_{lk} =  \frac{\sum_{s=1}^{K} \lambda_{ls} (\tilde{Q}^{-1})_{sk}}{\sum_{s=1}^{K} \lambda_{ls}}.
\end{align}

Thus, a simultaneous confidence region $[\hat{V}^{\mathrm{low}},\hat{V}^{\mathrm{upp}}]$ for the off-diagonal entries of $V$ can be obtained by applying standard parametric bootstrap techniques for multinomial parameters; e.g., see \citet{sison1995simultaneous}.
Finally, an adaptive prediction set $\hat{C}^{\mathrm{ci}}(X_{n+1})$ for $Y_{n+1}$ can be constructed by applying Algorithm~\ref{alg:correction-ci} based on $[\hat{V}^{\mathrm{low}},\hat{V}^{\mathrm{upp}}]$.
This two-step procedure is summarized by Algorithm~\ref{alg:correction-estimation} in Section~\ref{app:methods}.

Alternatively, one could seek only a point estimate $\hat{V}$ of $V$, by replacing the multinomial parameters in~\eqref{eq:multinomial-V} with their standard maximum-likelihood estimates. 
Then, it seems intuitive to construct prediction sets by applying Algorithm~\ref{alg:correction} with the {\em plug-in} estimate $\hat{V}$ instead of $V$. 
Section~\ref{sec:empirical} will demonstrate that this partly heuristic approach, outlined by Algorithm~\ref{alg:correction-estimation-plugin} in Section~\ref{app:methods}, tends to work well in practice and often leads to more informative prediction sets compared to Algorithm~\ref{alg:correction-estimation}. 
To conclude this section, we note that the estimation process for the label contamination model, as covered here, can be streamlined and made more practically feasible in cases with limited clean data. This simplification is achieved by introducing additional assumptions about the structure of the matrix $M$, as discussed in Section~\ref{app:simplified-methods}.

\section{Empirical demonstrations} \label{sec:empirical}

Sections~\ref{sec:empirical-known}--\ref{sec:empirical-ci} demonstrate the usage of our methods on simulated data.
Section~\ref{sec:empirical-known} applies methods from Section~\ref{sec:method-known-noise} and their extensions from Section~\ref{app:extensions}, assuming a known contamination model.
Section~\ref{sec:empirical-bounded} applies methods from Section~\ref{sec:method-bounded-noise}, which are useful when the contamination model is unknown.
Section~\ref{sec:empirical-ci} focuses on the estimation of the contamination model, applying methods from Section~\ref{sec:method-noise-estim}.
Finally, Section~\ref{sec:empirical-cifar} presents an application to image data.

\subsection{Simulations under a known label contamination model} \label{sec:empirical-known}

We begin by demonstrating the performance of our methods on synthetic data.
For this purpose, we simulate classification data with $K=4$ labels and $d=50$ features from a Gaussian mixture distribution using the standard {\em make\_classification} function from the Scikit-Learn Python package \citep{scikit-learn}.
This function creates $2K$ clusters of points normally distributed, with unit variance, about the vertices of a 25-dimensional hypercube with sides of length 2, and then randomly assigns an equal number of clusters to each of the $K$ classes.
Note that this leads to uniform label frequencies; i.e., $\rho_k := \P{Y=k} = 1/K$ for all $k \in [K]$.
We refer to \citet{scikit-learn} and \citet{guyondesign} for further details about the data-generating process.
The results of additional experiments based on different data distributions are in Section~\ref{app:figures}.
Conditional on the simulated data, the contaminated labels $\tilde{Y}$ are generated following a {\em randomized response}  model \citep{warner1965randomized}, an intuitive special case of the linear mixture model from Section~\ref{subsec:linear-contam-model}.
Specifically, $\mathbb{P}[\tilde{Y}=k \mid  X, Y=l] = (1-\epsilon) \I{k=l} + \epsilon/K$, which corresponds to $M_{kl} = \mathbb{P}[\tilde{Y}=k \mid  X, Y=l] \cdot \rho_l / \tilde{\rho}_k$, for all $l,k \in [K]$, considering a range of values for $\epsilon \in [0,0.2]$.
Additional experiments based on different contamination processes will be presented later.

A random forest classifier implemented by Scikit-Learn is trained on $10,000$ independent observations with contaminated labels generated as described above.
The classifier is then applied to an independent and identically distributed calibration data set, whose labels are also similarly contaminated, in order to construct generalized inverse quantile conformity scores with the recipe of \citet{romano2020classification}, reviewed in Section~\ref{app:adaptive-scores}.
These scores are transformed into prediction sets for 2000 independent unlabeled test points following three alternative approaches.
The first one is the standard conformal inference approach, which seeks 90\% label-conditional coverage while ignoring the presence of label contamination.
The second approach, called {\em Adaptive}, is the method outlined by Algorithm~\ref{alg:correction}, which we also apply with $\alpha=0.1$.
The third approach, called {\em Adaptive+}, is the optimistic variation of the {\em Adaptive} approach, as described in Section~\ref{sec:boosting-optimist}.
Both the {\em Adaptive} and {\em Adaptive+} methods are applied assuming perfect knowledge of $M$.

\begin{figure}[!htb]
\centering
\includegraphics[width=0.9\linewidth]{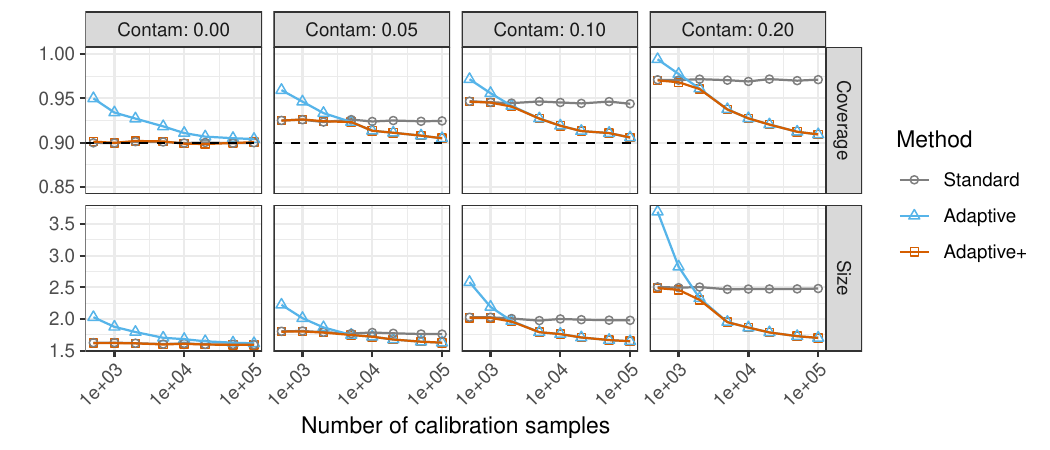}
\caption{Performances of different conformal prediction methods on simulated data with random label contamination of varying strength, as a function of the number of calibration samples. 
The dashed horizontal line indicates the 90\% nominal label-conditional coverage level.}
\label{fig:exp-synthetic-1-lab-cond-K4-ncal}
\end{figure}

Figure~\ref{fig:exp-synthetic-1-lab-cond-K4-ncal} compares the performances of the prediction sets obtained with the three alternative methods, measured in terms empirical coverage---the average proportion of test points for which the true label is contained in the prediction set---and average size.
The results are shown as a function of the number of calibration data points and of the contamination parameter $\epsilon$, averaging over 25 independent repetitions of each experiment.
Unsurprisingly, the standard conformal prediction sets are overly conservative if $\epsilon>0$ and their size does not change significantly as the number of calibration samples increases.
By contrast, the {\em Adaptive} and {\em Adaptive+} methods  tend to produce more informative prediction sets as the calibration sample grows.
Further, the {\em Adaptive+} sets can be smaller than both the standard and {\em Adaptive} sets, as long as the number of calibration samples is large enough.
Overall, these experiments demonstrate that our {\em Adaptive} and {\em Adaptive+} methods are effective at constructing more informative prediction sets with valid coverage in the presence of label contamination, and that the optimistic {\em Adaptive+} version is preferable in practice, even though its theoretical guarantee relies on slightly stronger technical assumptions on the coverage inflation factor (ref.~Section~\ref{sec:boosting-optimist}).

\textbf{Additional numerical results.}
Figures~\ref{fig:exp-synthetic-1-lab-cond-K4-ncal_lc}--\ref{fig:exp-synthetic-1-lab-cond-K4-ncal-cc} in Section~\ref{app:figures-known} present the results of further experiments with similar conclusions.
Figure~\ref{fig:exp-synthetic-1-lab-cond-K4-ncal_lc} reports additional performance metrics from the experiments of Figure~\ref{fig:exp-synthetic-1-lab-cond-K4-ncal}, stratifying the results based on the true label of the test point. 
This confirms that all methods under comparison achieve 90\% label-conditional coverage.
Figure~\ref{fig:exp-synthetic-1-lab-cond-K4-ncal-epsilon} gives an alternative view of these experiments, by plotting the results as a function of $\epsilon$ separately for different calibration sample sizes.
Figure~\ref{fig:exp-synthetic-1-lab-cond-ncal-epsilon} presents results from experiments analogous to those in Figure~\ref{fig:exp-synthetic-1-lab-cond-K4-ncal-epsilon}, but fixing $\epsilon=0.1$ and varying instead the number of labels $K$.
Figure~\ref{fig:exp-synthetic-1-lab-cond-K4-ncal-models} presents additional results from experiments analogous to those in Figure~\ref{fig:exp-synthetic-1-lab-cond-K4-ncal}, fixing $\epsilon=0.1$ but utilizing different types of classifiers, namely a support vector machine and a neural network.

\textbf{The effect of the data distribution.} 
The robustness of our results to different data distributions is demonstrated by Figures~\ref{fig:exp-synthetic-1-lab-cond-K4-ncal_synthetic2}--\ref{fig:exp-synthetic-1-lab-cond-K4-ncal_synthetic3}, which report on experiments similar to those of Figure~\ref{fig:exp-synthetic-1-lab-cond-K4-ncal}.
Figure~\ref{fig:exp-synthetic-1-lab-cond-K4-ncal_synthetic2} considers data with $K=4$ classes simulated from a logistic  model with random parameters, which is inspired by \citet{romano2020classification} and described as follows.
The features follow a standard multivariate Gaussian distribution of dimensions $d = 50$, and $Y \mid X$ is multinomial with weights proportional to $\exp[(X^T W)_k )]$, for all $k \in [K]$, where $W \in \mathbb{R}^{d \times K}$ is an independent standard normal random vector.
Figure~\ref{fig:exp-synthetic-1-lab-cond-K4-ncal_synthetic3} also involves synthetic data with $K=4$ classes, but those data are generated from the heteroscedastic decision-tree model defined in Section~\ref{app:details-experiments-tree}, which is borrowed from \citet{romano2020classification}.

\textbf{The effect of the label contamination process.}
The robustness of our results to different data distributions is demonstrated by Figures~\ref{fig:exp-synthetic-1-lab-cond-K4-ncal_block}--\ref{fig:exp-synthetic-1-lab-cond-K4-ncal_random}, which report on experiments based on data simulated from a logistic  model with random parameters, as in Figure~\ref{fig:exp-synthetic-1-lab-cond-K4-ncal_synthetic2}.
However, now the labels are contaminated differently.
In Figure~\ref{fig:exp-synthetic-1-lab-cond-K4-ncal_block}, the contamination process is described by a transition matrix $T \in [0,1]^{K \times K}$, with $T_{kl} = \mathbb{P}[\tilde{Y}=k \mid  X, Y=l]$ for all $l,k \in [K]$, given by $T = (1-\epsilon)I_{K} + \epsilon/K \cdot B_{K,2}$, where $B_{K,2}$ is a block-diagonal matrix with $K/2$ constant blocks equal to $J_2$---the $2 \times 2$ matrix of ones.
In Figure~\ref{fig:exp-synthetic-1-lab-cond-K4-ncal_random}, $T = (1-\epsilon)I_{K} + \epsilon/K \cdot U_K$, where $U_K$ is a matrix of i.i.d.~uniform random numbers on $[0,1]$, standardized to have its columns sum to one.

\textbf{Prediction sets targeting other notions of coverage.}
Figures~\ref{fig:exp-synthetic-1-marginal-K4-ncal}--\ref{fig:exp-synthetic-1-marginal-K4-ncal-epsilon} present results from experiments analogous to those in Figures~\ref{fig:exp-synthetic-1-lab-cond-K4-ncal} and~\ref{fig:exp-synthetic-1-lab-cond-K4-ncal-epsilon}, respectively, with the difference that all methods under comparison are applied to seek 90\% marginal coverage instead of 90\% label-conditional coverage. Within our adaptive framework, this is obtained by replacing Algorithm~\ref{alg:correction} with Algorithm~\ref{alg:correction-marg}, as explained in~Section~\ref{app:methods-marginal}.
Then, Figures~\ref{fig:exp-synthetic-1-marginal-K4-ncal_synthetic2_uniform_estim-rho}--\ref{fig:exp-synthetic-1-marginal-K4-ncal_synthetic2_random_estim-rho} demonstrate the robustness of Algorithm~\ref{alg:correction-marg} to the empirical estimation of the contaminated label frequencies $\tilde{\rho}$ from the available data.
Finally, Figure~\ref{fig:exp-synthetic-1-lab-cond-K4-ncal-cc} reports on experiments in which the goal is to achieve valid coverage conditional on the calibration data.
Within our adaptive framework, this is obtained by replacing Algorithm~\ref{alg:correction} with Algorithm~\ref{alg:correction-cc}, as explained in~Section~\ref{app:methods-cal-cond}.

\subsection{Simulations under a bounded label contamination model} \label{sec:empirical-bounded}

We now apply Algorithm~\ref{alg:correction-ci}, which does not require perfect knowledge of the model matrix $M$.
As a starting point, we focus on a contamination process described by the randomized response  model \citep{warner1965randomized} with an unknown parameter.
As explained in Section~\ref{app:RR-model}, this model streamlines the implementation of Algorithm~\ref{alg:correction-ci}, which in this case requires as input only a confidence interval for a scalar parameter $\epsilon \in [0,1)$; see Section~\ref{app:RR-model} for further details.

Figure~\ref{fig:exp-synthetic-1-bounded-ncal-eps0.2-lower} compares the performance of Algorithm~\ref{alg:correction-ci} ({\em Adaptive}) and its optimistic variation ({\em Adaptive+}), from Section~\ref{sec:method-bounded-noise}, to that of the standard conformal method that ignores label contamination, using synthetic data similar to those in Figure~\ref{fig:exp-synthetic-1-lab-cond-K4-ncal}.
The differences are that now the number of possible labels is varied, $K \in \{2,4,8\}$, the number of calibration samples is 10,000, and the true noise parameter is $\epsilon=0.2$.
The {\em Adaptive} (and {\em Adaptive+}) prediction sets are constructed by applying Algorithm~\ref{alg:correction-ci} (and its optimistic variation) based on a 99\% confidence interval for $\epsilon$ whose lower bound is varied, while the upper bound is fixed to $0.2$.
The results show that our prediction sets always achieve valid label-conditional coverage and become increasingly informative as the lower bound for $\epsilon$ increases, as anticipated by our theory.

\begin{figure}[!htb]
\centering
\includegraphics[width=0.9\linewidth]{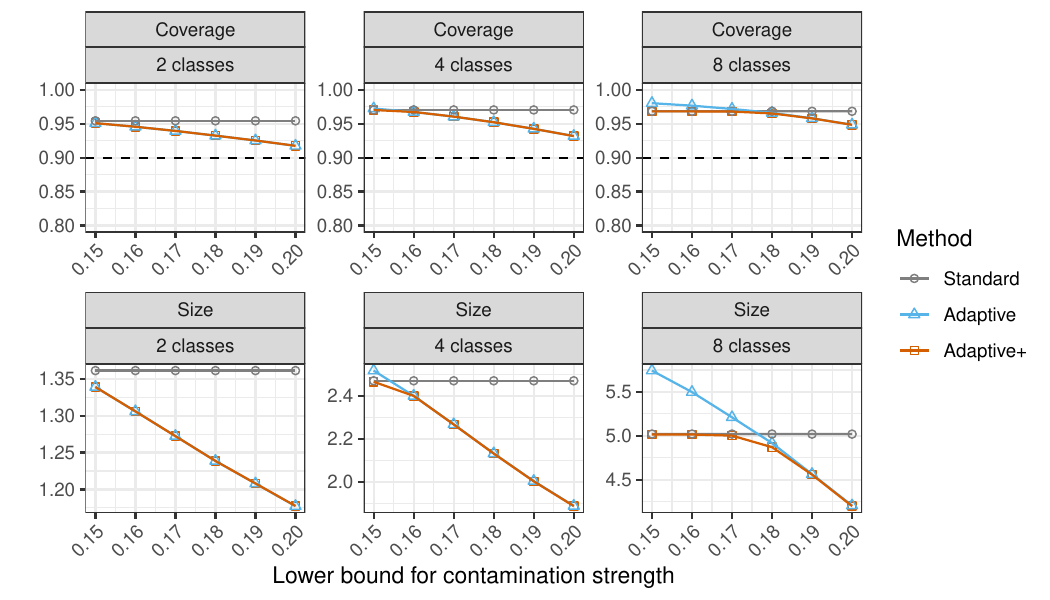}
\caption{Performances of different conformal prediction methods on simulated data with contaminated labels from a randomized response model. 
The results are shown as a function of the known lower bound for the label noise parameter $\epsilon=0.2$ and of the number of possible labels.
The number of calibration samples is 10,000.
Other details are as in Figure~\ref{fig:exp-synthetic-1-lab-cond-K4-ncal}.}
\label{fig:exp-synthetic-1-bounded-ncal-eps0.2-lower}
\end{figure}

Additional results with qualitatively consistent conclusions are presented by Figures~\ref{fig:exp-synthetic-1-lab-cond-K4-ncal-upper}--\ref{fig:exp-synthetic-1-bounded-K4-ncal-eps0.1} in Section~\ref{app:figures-bounded-RR}.
Figure~\ref{fig:exp-synthetic-1-lab-cond-K4-ncal-upper} reports on experiments that differ from those in Figure~\ref{fig:exp-synthetic-1-bounded-ncal-eps0.2-lower} in that the upper bound for $\epsilon$ is varied while the lower bound is fixed equal to the true parameter value.
Figures~\ref{fig:exp-synthetic-1-bounded-K4-ncal-eps0.2}--\ref{fig:exp-synthetic-1-bounded-K4-ncal-eps0.1} report on the results of experiments similar to those of Figure~\ref{fig:exp-synthetic-1-bounded-ncal-eps0.2-lower}, respectively fixing $\epsilon = 0.2$ and $\epsilon = 0.1$, while varying the corresponding lower confidence bound and the calibration set size.
Overall, these results confirm that tighter bounds for $\epsilon$ generally allow Algorithm~\ref{alg:correction-ci} to construct more informative prediction sets, and that the advantage of our adaptive method is more noticeable when the contaminated data are abundant.

Section~\ref{app:figures-bounded-BRR} presents the results of similar experiments in which the label contamination process is more complex, following a {\em two-level} extension of the randomized response model.
This model is discussed in Section~\ref{app:BRR-model} and can accommodate potential label hierarchies using two distinct parameters, $\epsilon \in [0,1)$ and $\nu \in [0,1]$. As explained in Section~\ref{app:BRR-model}, Algorithm~\ref{alg:correction-ci} also simplifies under a two-level randomized response model, and in this case it requires as input only a pair of simultaneously valid confidence intervals for $\epsilon$ and $\nu$.
Figures~\ref{fig:exp-synthetic-1-bounded-BRR-K4-eps0.2-ncal10000-nuse0}--\ref{fig:exp-synthetic-1-bounded-BRR-K4-eps0.2-ncal100000-nuse0.02} report on the performance of our method as a function of the width of the confidence interval for $\epsilon$, using different values of $\nu$, confidence intervals for $\nu$, and calibration set sizes.
Figures~\ref{fig:exp-synthetic-1-bounded-BRR-K4-eps0.2-ncal10000-epsne0}--\ref{fig:exp-synthetic-1-bounded-BRR-K4-eps0.2-ncal100000-epsne0.02} report on similar experiments in which the width of the confidence interval for $\nu$ is varied, using different values of $\nu$, confidence intervals for $\epsilon$, and calibration set sizes.
Overall, these results support the previous conclusions: Algorithm~\ref{alg:correction-ci} leads to more informative prediction sets compared to standard conformal methods even if the label contamination model is unknown. Further, its advantage tends to grow as the contaminated calibration set becomes larger.

\subsection{Robustness to model estimation and mis-specification} \label{sec:empirical-ci}

We now examine the performance of the methods described in Section~\ref{sec:method-noise-estim} for estimating the contamination model using an independent ``model-fitting'' data set containing both clean and contaminated labels.
Again, we begin by focusing on the randomized response  model.
More challenging estimation settings will be considered subsequently.

To begin, we focus on a randomized response model with an unknown scalar parameter $\epsilon$.
We compare the performance of three alternative implementations of our {\em Adaptive+} method to the standard conformal approach, on synthetic data with $K=2$ labels, similar to those utilized for Figure~\ref{fig:exp-synthetic-1-bounded-ncal-eps0.2-lower}.
The first implementation of our method ({\em Adaptive+}) consists of applying the optimistic version of Algorithm~\ref{alg:correction} using perfect {\em oracle} knowledge of the matrix $M$, as written explicitly in Section~\ref{app:RR-model}, based on the correct value $\epsilon=0.2$. This corresponds to the general method presented in Section~\ref{sec:empirical-known}.
The second implementation of our method, which we call {\em Adaptive+ (plug-in)}, consists of applying the optimistic version of Algorithm~\ref{alg:correction} using an approximate version of $M$ obtained by replacing the unknown noise parameter $\epsilon$ with an intuitive point estimate $\hat{\epsilon}$ calculated from the model-fitting data as explained in Section~\ref{app:RR-model}---the latter simplifies the more general procedure described in Section~\ref{sec:method-noise-estim} for the special case of the randomized response model.
The third implementation of our method, which we call {\em Adaptive+ (CI)}, consists of applying the optimistic version of Algorithm~\ref{alg:correction-ci} using a 99\% bootstrap confidence interval for $\epsilon$.
This confidence interval is produced by a specialized version of the general method from Section~\ref{sec:method-noise-estim}, as explained in Section~\ref{app:RR-model}.
For simplicity, Algorithm~\ref{alg:correction-ci} is applied assuming a known fixed upper bound for $\epsilon=0.2$ equal to $\bar{\varepsilon}=0.2$, so that the bootstrap is effectively only needed to estimate the lower confidence bound.

Figure~\ref{fig:exp-synthetic-1-lab-cond_ci_K2} reports on the performance of all methods as a function of the size and composition of the model-fitting data set. The results show that the heuristic {\em Adaptive+ (plug-in)} method performs very similarly to the ideal {\em Adaptive+} method based on oracle knowledge of the true contamination parameter $\epsilon$.
By contrast, the {\em Adaptive+ (CI)} tends to be more conservative and can lead to prediction sets that are significantly more informative compared to the standard conformal inference benchmark only if the number of clean samples in the model-fitting data set is large.
Figure~\ref{fig:exp-synthetic-1-lab-cond_ci_bounds_K2} in Section~\ref{app:figures-ci-RR} plots explicitly the average upper and lower bounds of the bootstrap confidence intervals estimated by the {\em Adaptive+ (CI)} in the experiments of Figure~\ref{fig:exp-synthetic-1-lab-cond_ci_K2}.

\begin{figure}[!htb]
\centering
\includegraphics[width=0.9\linewidth]{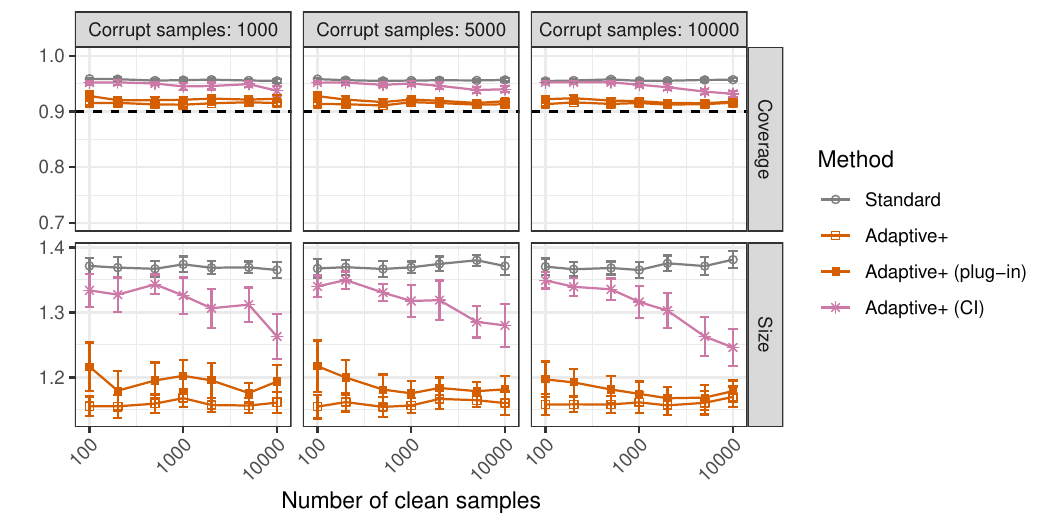}
\caption{Performances of different conformal prediction methods, as a function of the numbers of clean and corrupted samples used to fit the unknown parameter of a randomized response model for the contamination process.
The number of possible classes is $K=2$.
Other details are as in Figure~\ref{fig:exp-synthetic-1-bounded-ncal-eps0.2-lower}.}
\label{fig:exp-synthetic-1-lab-cond_ci_K2}
\end{figure}

Further results from related experiments, with qualitatively similar conclusions, are presented by Figures~\ref{fig:exp-synthetic-1-lab-cond_ci_K2_emax0.25}--\ref{fig:exp-synthetic-1-lab-cond_point_K8} in Section~\ref{app:figures-ci-RR}.
Figures~\ref{fig:exp-synthetic-1-lab-cond_ci_K2_emax0.25} and~\ref{fig:exp-synthetic-1-lab-cond_ci_bounds_K2_emax0.25} report on results analogous to those in Figures~\ref{fig:exp-synthetic-1-lab-cond_ci_K2} and~\ref{fig:exp-synthetic-1-lab-cond_ci_bounds_K2}, respectively, with the only difference that now the fixed upper bound for $\epsilon=0.2$ utilized by the {\em Adaptive+ (CI)} method is set equal to $\bar{\varepsilon}=0.25$.
Figures~\ref{fig:exp-synthetic-1-lab-cond_point_K2}--\ref{fig:exp-synthetic-1-lab-cond_point_K8} further demonstrate the power and robustness of the {\em Adaptive+ (plug-in)} method in experiments with different values of the true noise parameter $\epsilon$ and different numbers of possible labels.

Section~\ref{app:figures-ci-BRR} details further experiments conducted within a more general contamination process, described by a two-level randomized response model.
Specifically, Figure~\ref{fig:exp-synthetic-1-lab-cond_eps0.1_K4_BRR} focuses on synthetic data with $K=4$ classes and summarizes the performances of prediction sets obtained with different methods, while Figures~\ref{fig:exp-synthetic-1-lab-cond_ci_K4_emax0.25_BRR}--\ref{fig:exp-synthetic-1-lab-cond_ci_nu_K4_emax0.25_BRR} plot the corresponding estimated confidence intervals for the unknown parameters of the contamination model.
Figures~\ref{fig:exp-synthetic-1-lab-cond_eps0.1_K8_BRR}--\ref{fig:exp-synthetic-1-lab-cond_ci_nu_K8_emax0.25_BRR} report on similar results for data with $K=8$ classes.
Overall, our {\em Adaptive+ (plug-in)} method performs well across all scenarios considered, even though it ignores the true contamination model parameters. Notably, this partly heuristic implementation of Algorithm~\ref{alg:correction} attains valid coverage empirically and leads to more informative predictions compared to the more rigorous {\em Adaptive+ (CI)} approach of Algorithm~\ref{alg:correction-ci}.
Moreover, our estimators of the contamination model parameters (Section~\ref{sec:method-noise-estim}) are confirmed to be consistent, in the sense that their estimated confidence intervals become narrower as the number of samples available for estimation increases.

Finally, Figures~\ref{fig:exp-synthetic-1-lab-cond_ci_K4_block}--\ref{fig:exp-synthetic-1-lab-cond_ci_K4_random} in Section~\ref{app:figures-robust} describe experiments demonstrating the robustness of our methods to possible mis-specification of the label contamination process, going beyond the estimation of the parameters in a (generalized) randomized response model.

\FloatBarrier
\subsection{Demonstrations with CIFAR-10 image data} \label{sec:empirical-cifar}

This section demonstrates the use of our methods in a object classification application based on real-world 32x32 color images.
As anticipated in Section~\ref{sec:intro}, we focus on the CIFAR-10H data set \citep{peterson2019human}, a variation of the larger CIFAR-10 data set \citep{krizhevsky2020cifar} that includes imperfect labels assigned by approximately 50 independent human annotators via the Amazon Mechanical Turk, for a subset of 10,000 images.
Each image depicts an object belonging to one of 10 possible classes: airplane, car, bird, cat, deer, dog, frog, horse, ship, or truck.
Since the individual annotators do not always agree on the content of each image, we can think of their labels as being a randomly contaminated version of the corresponding ``true'' labels contained in the original CIFAR-10 data set.
Our goal is to construct informative prediction sets for the true labels, using a conformal predictor calibrated on the contaminated data.
For simplicity, we work with a slightly modified version of the CIFAR-10H data in which each image has a single corrupted label $\tilde{Y}$, randomly sampled from a multinomial distribution whose weights are equal to the relative label frequencies assigned to that image by different human annotators.
Note that these corrupted labels coincide with the true CIFAR-10 labels approximately 95.4\% of the time; see Figure~\ref{fig:cifar10-demo}~(a) for a visualization of some images for which the labels do not match.

A ResNet-18 convolutional neural network serves as base classifier; this was implemented by the PyTorch Python package \citep{paszke2019pytorch} and pre-trained using the 50,000 CIFAR-10 images excluded from the CIFAR-10H data set.
The output of the final soft-max layer of the deep neural network provides estimates of the class probabilities for any new given image, and from that we calculate conformity scores with the recipe reviewed in Section~\ref{app:adaptive-scores}.
The conformal predictor is then calibrated using three alternative methods, based on a random subset of the 10,000 CIFAR-10H images whose size is varied as a control parameter.

Due to the larger number of labels, $K=10$, and to help highlight the flexibility of our methods, we focus first on marginal~\eqref{eq:def-marg-coverage} instead of label-conditional~\eqref{eq:def-lab-cond-coverage} coverage.
The first method considered is Algorithm~\ref{alg:standard-marg}, the standard conformal approach that seeks marginal coverage while ignoring label contamination.
The second method, {\em Adaptive+}, is the optimistic variation of our Algorithm~\ref{alg:correction-marg} from Section~\ref{app:methods-marginal}, which we apply imagining that the label contamination process follows the randomized response  model described in Section~\ref{app:RR-model}. 
This model assumes that the observed labels are conditionally independent of the image features, given the true labels, and that an incorrect label is as likely as any other incorrect label. Despite the potential misalignment with the true data generating process, this model is useful for our purposes. 
Specifically, it simplifies the matrix $M$, allowing us to characterize it using a single noise parameter, $\epsilon \in [0,1)$, which is fixed here as $\epsilon = 0.051$. 
This $\epsilon$ value is selected to match the mean fraction of CIFAR-10H samples where $Y \neq \tilde{Y}$, roughly $0.046 = \epsilon(1-1/K)$. 
The third method, {\em Adaptive+ (plug-in)}, differs from {\em Adaptive+} in that it utilizes a plug-in estimate of $\epsilon$ obtained via maximum-likelihood.
This estimate is evaluated as explained in Section~\ref{app:RR-model}, by applying the same pre-trained ResNet-18 convolutional neural network to a smaller independent data set containing both clean and corrupted data in equal proportions.
In particular, the number of clean images used to estimate $\epsilon$ is 10\% of the total number of corrupted calibration images.

Figure~\ref{fig:exp-cifar10-marginal}, previewed in Section~\ref{sec:intro-preview}, reports on the prediction sets constructed by the three methods for a random test set 500 CIFAR-10H images, varying the size of the calibration sample between 500 and 9500. 
All experiments are independently repeated 50 times, using different random splits of the CIFAR-10H data into calibration, model-fitting, and test subsets.
The results show that the standard method is overly conservative, while our adaptive approaches are able to achieve valid coverage with increasingly more informative prediction sets as the size of the calibration sample grows.
See Figure~\ref{fig:cifar10-demo}~(b) for a visualization of some concrete examples in which our {\em Adaptive+} method leads to more informative prediction sets compared to the standard conformal inference benchmark.
Finally, Figure~\ref{fig:exp-cifar10-lab-cond} in Section~\ref{app:figures-cifar10} presents analogous results from similar experiments in which we target label-conditional coverage, using Algorithms~\ref{alg:correction} and~\ref{alg:standard-lab-cond} instead of Algorithms~\ref{alg:correction-marg} and \ref{alg:standard-marg}, respectively.
Note that Algorithm~\ref{alg:correction} requires a larger data set in this case, compared to the marginal coverage setting, in order to produce prediction sets that are significantly more informative compared to those of Algorithm~\ref{alg:standard-lab-cond}.
The reason for this is the stronger nature of the label-conditional guarantee, which effectively diminishes the usable sample size for both Algorithm~\ref{alg:correction} and Algorithm~\ref{alg:standard-lab-cond} by a factor of $K=10$.

\FloatBarrier

\section{Discussion} \label{sec:discussion}

This paper studied in-depth the problem of conformal classification in the presence of calibration data with contaminated labels.
Our research contributes to the growing literature on conformal inference beyond exchangeability \citep{barber2022conformal}, addressing an open practical problem.
A key innovation of our approach is its capacity to automatically adjust to random label contamination, resulting in prediction sets that offer more robust coverage guarantees and are often more informative compared to those given by state-of-the-art approaches. 
Additionally, our framework is highly flexible, enabling several variations of our methodology that target different coverage metrics and can accommodate varying degrees of understanding of the label contamination process. 
These adaptations hint at potential for further expansion in future research. Both theoretical and practical findings presented in this paper underscore the enhanced benefits of our adaptive approach when employed with larger sets of contaminated calibration data. This aspect renders our research particularly applicable to real-world situations where there are abundant data of modest quality, but precise labels are scarce.

This research opens several opportunities for future work.
For example, it may be interesting to study possible extensions of our methods that can be applied with regression data, or even with other types of more complex data for which conformal inference has already been utilized, including causal inference \citep{lei2021conformal}, survival analysis \citep{candes2021conformalized}, and matrix completion \citep{gui2023conformalized}.
Alternatively, it may be possible to account for label contamination in the context of more sophisticated conformal prediction frameworks such as full-conformal inference \citep{vovk2005algorithmic} and cross-validation+ \citep{barber2019predictive}, which are more computationally expensive but can make more efficient use of limited observations. 
Finally, future research might be able to uncover further methodological opportunities by combining the ideas presented in this paper with the theoretical analyses of \cite{barber2022conformal}.

A software implementation of the methods presented in this paper is available online at \url{https://github.com/msesia/conformal-label-noise}.

\section*{Acknowledgements}
M.~S.~was supported by NSF grant DMS 2210637 and by an Amazon Research Award.
We are grateful to two anonymous referees for their constructive feedback about an earlier version of this manuscript.
We also thank Adel Javanmard for helpful suggestions.

\clearpage
\printbibliography

\clearpage
\appendix

\renewcommand{\thesection}{A\arabic{section}}
\renewcommand{\theequation}{A\arabic{equation}}
\renewcommand{\thetheorem}{A\arabic{theorem}}
\renewcommand{\theproposition}{A\arabic{proposition}}
\renewcommand{\thelemma}{A\arabic{lemma}}
\renewcommand{\thetable}{A\arabic{table}}
\renewcommand{\thefigure}{A\arabic{figure}}
\renewcommand{\thealgocf}{A\arabic{algocf}}
\setcounter{figure}{0}
\setcounter{table}{0}
\setcounter{proposition}{0}
\setcounter{theorem}{0}
\setcounter{lemma}{0}
\setcounter{algocf}{0}

\section{Review of standard conformal classification methods} \label{app:review-aps}

\subsection{Conformity scores based on generalized inverse quantiles} \label{app:adaptive-scores}

We briefly review here the construction of the generalized inverse quantile conformity scores proposed by~\cite{romano2020classification}, upon which we rely in the empirical demonstrations of Section~\ref{sec:empirical}.
These conformity scores are more involved compared to the classical homogeneous scores defined in~\eqref{eq:pred-function-hps}, but they have the advantage of leading to more flexible prediction sets that can account for possible heteroscedasticity in the distribution of $Y \mid X$. We refer to \cite{romano2020classification} or \cite{cauchois2020knowing} for further information about the limitations of the scores in~\eqref{eq:pred-function-hps}.

For any $x \in \mathbb{R}^{d}$ and $t \in [0,1]$, define
\begin{align} \label{eq:oracle-threshold}
  \hat{Q}(x, \hat{\pi}, t) & = \min \{ k \in \{1,\ldots,K\} : \hat{\pi}_{(1)}(x) + \hat{\pi}_{(2)}(x) + \ldots + \hat{\pi}_{(k)}(x) \geq t \},
\end{align}
where $\hat{\pi}_{(1)}(x) \geq \ldots \geq \hat{\pi}_{(K)}(x)$ are the descending order statistics of $\hat{\pi}(x,1), \ldots, \hat{\pi}(x,K)$.
Intuitively, $\hat{Q}(x, \hat{\pi}, \cdot)$ may be seen as a generalized quantile function.
Similarly, let $\hat{r}(x, \hat{\pi}, k)$ denote the rank of $\hat{\pi}(x,k)$ among $\hat{\pi}(x,1), \ldots, \hat{\pi}(x,K)$.
With this notation, one can also define a corresponding generalized cumulative distribution function:
\begin{align*}
  \hat{\Pi}(x,\hat{\pi},k) = \hat{\pi}_{(1)}(x) + \hat{\pi}_{(2)}(x) + \ldots + \hat{\pi}_{(\hat{r}(x, \hat{\pi}, k))}(x).
\end{align*}
Then, the function $\mathcal{C}$ proposed by~\cite{romano2020classification} can be written as:
\begin{align}  \label{eq:pred-sets-aps}
  \mathcal{C}(x; \hat{\pi}, \tau)
  & = \{ k \in [K] : \hat{r}(x, \hat{\pi}, k) \leq \hat{Q}(x, \hat{\pi}, \tau_k)\},
\end{align}
and the corresponding conformity scores defined in~\eqref{eq:conf-scores} can be evaluated efficiently by noting that $\hat{s}(x,k) = \hat{\Pi}(x,k)$; see \cite{romano2020classification} for further details.

The prediction function defined in~\eqref{eq:pred-sets-aps} may be understood by noting that, if $\tau = (\tau_0, \ldots, \tau_0)$ for some $\tau_0 \in [0,1]$, the output of $\mathcal{C}(x; \hat{\pi}, \tau)$ is the list of most likely classes according to $\hat{\pi}(x)$ up until the first label $l$ for which $\hat{\Pi}(x,\hat{\pi},l) \geq \tau_0$.
Therefore, in the ideal case where $\hat{\pi}(x,k) = \P{Y = k \mid X=x}$, one can verify that $\mathcal{C}(x; \hat{\pi}, \tau)$ is the smallest possible (deterministic) prediction set for $Y$ with perfect object-conditional coverage at level $\tau_0$, i.e., satisfying $\P{Y \in \mathcal{C}(x; \hat{\pi}, \tau) \mid X=x} \geq \tau_0$.
Note that \cite{romano2020classification} also developed a more powerful randomized version of \eqref{eq:pred-sets-aps} that enjoys similar theoretical properties while being able to produce even more informative prediction sets.
The results of this paper can also seamlessly accommodate such additional randomness in $\mathcal{C}$, and indeed that is the practical approach followed in the empirical demonstrations of Section~\ref{sec:empirical}, but we choose not to review such extension explicitly here to avoid making the notation too cumbersome.

\clearpage
\subsection{Prediction sets with label-conditional coverage} \label{app:standard-lab-cond}

\begin{algorithm}[H]
\DontPrintSemicolon

\KwIn{Data set $\{(X_i, \tilde{Y}_i)\}_{i=1}^{n}$, with observable labels $\tilde{Y}_i \in [K]$.}
\myinput{Unlabeled test point with features $X_{n+1}$.}
\myinput{Machine learning algorithm $\mathcal{A}$ for training a $K$-class classifier.}
\myinput{Prediction function $\mathcal{C}$ satisfying Definition~\ref{def:pred-function}; e.g., \eqref{eq:pred-function-hps}.}
\myinput{Desired coverage parameter $\alpha$.}

Randomly split $[n]$ into two disjoint subsets, $\mathcal{D}^{\text{train}}$ and $\mathcal{D}^{\mathrm{cal}}$.\;
Train the classifier $\mathcal{A}$ on the data in $\mathcal{D}^{\text{train}}$. \;
\For{$k=1, \dots, K$}{
  Define $\mathcal{D}_k^{\mathrm{cal}} = \{ i \in \mathcal{D}^{\mathrm{cal}} : Y_i = k\}$ and $n_k = |\mathcal{D}_k^{\mathrm{cal}}|$.\;
  Compute $\hat{s}(X_i, k)$ using \eqref{eq:conf-scores}, for all $i \in \mathcal{D}_k^{\mathrm{cal}}$. \;
  Define $\hat{\tau}_k$ as the $\lceil (1+n_k)\cdot(1-\alpha) \rceil$ smallest value in $ \{\hat{s}(X_i,k)\}_{i \in \mathcal{D}_k^{\mathrm{cal}}}$.\;
}
Evaluate $\hat{C}(X_{n+1}) = \mathcal{C}(X_{n+1}, \hat{\tau}; \hat{\pi})$, where $\hat{\tau} = (\hat{\tau}_1, \ldots, \hat{\tau}_K)$.

\nonl
\textbf{Output: } Conformal prediction set $\hat{C}(X_{n+1})$ for $Y_{n+1}$, satisfying \eqref{eq:def-lab-cond-coverage}.\;

\caption{Standard conformal classification with label-conditional coverage.} \label{alg:standard-lab-cond}
\end{algorithm}

\begin{proposition}[e.g., from \cite{vovk2003mondrian} or \cite{romano2020classification}] \label{prop:standard-coverage-label}
If the data pairs $(X_i,\tilde{Y}_i)$, for all $i \in [n+1]$, are exchangeable random samples from some joint distribution, the prediction set $\hat{C}(X_{n+1})$ output by Algorithm~\ref{alg:standard-lab-cond} has label-conditional coverage \eqref{eq:def-lab-cond-coverage} for the observable labels $\tilde{Y}$; i.e.,
\begin{align*}
  \P{ \tilde{Y}_{n+1} \in \hat{C}(X_{n+1}) \mid \tilde{Y}_{n+1} = k} \geq 1-\alpha, \qquad \text{ for all } k \in [K].
\end{align*}
Further, if all scores $\hat{s}(X_i,\tilde{Y}_i)$ computed by Algorithm~\ref{alg:standard-lab-cond} are almost-surely distinct,
\begin{align*}
  \P{ \tilde{Y}_{n+1} \in \hat{C}(X_{n+1}) \mid \tilde{Y}_{n+1} = k} \leq 1-\alpha + \frac{1}{n_k+1}, \qquad \text{ for all } k \in [K],
\end{align*}
where $n_k = |\mathcal{D}_k^{\mathrm{cal}}|$ is the number of data points with label $k$ in the calibration set.
\end{proposition}

In words, this result tells us that Algorithm~\ref{alg:standard-lab-cond} is always guaranteed to achieve valid label-conditional coverage for the observable labels $\tilde{Y}$, regardless of which prediction function $\mathcal{C}$ is employed.
Note that in the standard setting without label contamination, the observable labels $\tilde{Y}$ are simply assumed to be always equal to $Y$.

\subsection{Prediction sets with marginal coverage} \label{app:standard-marg}

\begin{algorithm}[H]
\DontPrintSemicolon

\KwIn{Data set $\{(X_i, \tilde{Y}_i)\}_{i=1}^{n}$ with observable labels $\tilde{Y}_i \in [K]$.}
\myinput{Unlabeled test point with features $X_{n+1}$.}
\myinput{Machine learning algorithm $\mathcal{A}$ for training a $K$-class classifier.}
\myinput{Prediction function $\mathcal{C}$ satisfying Definition~\ref{def:pred-function}; e.g., \eqref{eq:pred-function-hps}.}
\myinput{Desired coverage parameter $\alpha$.}

Randomly split $[n]$ into two disjoint subsets, $\mathcal{D}^{\text{train}}$ and $\mathcal{D}^{\mathrm{cal}}$, defining $n_{\mathrm{cal}} = |\mathcal{D}^{\mathrm{cal}}|$.\;
Train the classifier $\mathcal{A}$ on the data in $\mathcal{D}^{\text{train}}$. \;
Compute $\hat{s}(X_i, \tilde{Y}_i)$ using \eqref{eq:conf-scores}, for all $i \in \mathcal{D}^{\mathrm{cal}}$.\;
Define $\hat{\tau}$ as the $\lceil (1+n_{\mathrm{cal}})\cdot(1-\alpha) \rceil$ smallest value in $ \{\hat{s}(X_i,\tilde{Y}_i)\}_{i \in \mathcal{D}^{\mathrm{cal}}}$.\;
Evaluate $\hat{C}(X_{n+1}) = \mathcal{C}(X_{n+1}, \hat{\tau};  \hat{\pi})$, where $\hat{\tau} = (\hat{\tau}_0, \ldots, \hat{\tau}_0)$.

\nonl
\textbf{Output: } Conformal prediction set $\hat{C}(X_{n+1})$ for $\tilde{Y}_{n+1}$, satisfying \eqref{eq:def-marg-coverage}.\;

\caption{Standard conformal classification with marginal coverage.} \label{alg:standard-marg}
\end{algorithm}

\begin{proposition}[e.g., from \cite{lei2013distribution} or \cite{romano2020classification}] \label{prop:standard-coverage-marginal}
If the data pairs $(X_i,\tilde{Y}_i)$, for all $i \in [n+1]$, are exchangeable random samples from some joint distribution, the prediction set $\hat{C}(X_{n+1})$ output by Algorithm~\ref{alg:standard-lab-cond} has marginal coverage \eqref{eq:def-marg-coverage} for the observable labels $\tilde{Y}$; i.e.,
\begin{align*}
  \P{ \tilde{Y}_{n+1} \in \hat{C}(X_{n+1}) } \geq 1-\alpha.
\end{align*}
Further, if all scores $\hat{s}(X_i,\tilde{Y}_i)$ computed by Algorithm~\ref{alg:standard-lab-cond} are almost-surely distinct,
\begin{align*}
  \P{ \tilde{Y}_{n+1} \in \hat{C}(X_{n+1})} \leq 1-\alpha + \frac{1}{n_{\mathrm{cal}}+1},
\end{align*}
where $n_{\mathrm{cal}} = |\mathcal{D}^{\mathrm{cal}}|$ is the number of data points in the calibration set.
\end{proposition}

\clearpage

\section{Additional methodological details} \label{app:methods}

\begin{algorithm}[!ht]
\DontPrintSemicolon

\KwIn{Data set $\mathcal{D} = \{(X_i, \tilde{Y}_i)\}_{i=1}^{n}$ with corrupted labels $\tilde{Y}_i \in [K]$.}
\myinput{Independent data set $\mathcal{D}^1 = \{(X_i^{1}, \tilde{Y}^{1}_i)\}_{i=1}^{n^1}$ with corrupted labels $\tilde{Y}^{1}_i \in [K]$.}
\myinput{Clean data set $\mathcal{D}^0 = \{(X_i^{0}, Y^{0}_i)\}_{i=1}^{n^0}$ with labels $Y^{0}_i \in [K]$.}
\myinput{Unlabeled test point with features $X_{n+1}$.}
\myinput{Machine learning algorithm $\mathcal{A}$ for training a $K$-class classifier.}
\myinput{Prediction function $\mathcal{C}$ satisfying Definition~\ref{def:pred-function}; e.g., (\ref{eq:pred-function-hps}).}
\myinput{Desired significance level $\alpha \in (0,1)$.}
\myinput{Desired significance level $\alpha_V \in (0,\alpha)$ for the estimation of $V$.}

Randomly split the data in $\mathcal{D}^1$ into two disjoint subsets, $\mathcal{D}^{1}_{a}$ and $\mathcal{D}^{1}_{b}$.\;
Train the classifier $\mathcal{A}$ on the data in $\mathcal{D}^{1}_{a}$. \;
Assess the accuracy of $\mathcal{A}$ to predict $\tilde{Y}$ in $\mathcal{D}^{1}_{b}$ and $Y$ in $\mathcal{D}^{0}$, using (\ref{eq:Q-tilde-estim}) and (\ref{eq:Q-estim}). \;
Construct a $1-\alpha_V$ joint confidence region $[\hat{V}^{\mathrm{low}}, \hat{V}^{\mathrm{upp}}]$ for the off-diagonal entries of $V := M^{-1}$ using (\ref{eq:multinomial-V}) and the parametric bootstrap, as explained in Section~\ref{sec:method-noise-estim}.\;
Apply Algorithm~\ref{alg:correction-ci} to construct a conformal prediction set $\hat{C}^{\mathrm{ci}}(X_{n+1})$ for $Y_{n+1}$.
\nonl
\textbf{Output: } Conformal prediction set $\hat{C}^{\mathrm{ci}}(X_{n+1})$ for $Y_{n+1}$.
\caption{Adaptive classification with contamination model fitting (CI)}
\label{alg:correction-estimation}
\end{algorithm}

\begin{algorithm}[!ht]
\DontPrintSemicolon

\KwIn{Data set $\mathcal{D} = \{(X_i, \tilde{Y}_i)\}_{i=1}^{n}$ with corrupted labels $\tilde{Y}_i \in [K]$.}
\myinput{Independent data set $\mathcal{D}^1 = \{(X_i^{1}, \tilde{Y}^{1}_i)\}_{i=1}^{n^1}$ with corrupted labels $\tilde{Y}^{1}_i \in [K]$.}
\myinput{Clean data set $\mathcal{D}^0 = \{(X_i^{0}, Y^{0}_i)\}_{i=1}^{n^0}$ with labels $Y^{0}_i \in [K]$.}
\myinput{Unlabeled test point with features $X_{n+1}$.}
\myinput{Machine learning algorithm $\mathcal{A}$ for training a $K$-class classifier.}
\myinput{Prediction function $\mathcal{C}$ satisfying Definition~\ref{def:pred-function}; e.g., (\ref{eq:pred-function-hps}).}
\myinput{Desired significance level $\alpha \in (0,1)$.}

Randomly split the data in $\mathcal{D}^1$ into two disjoint subsets, $\mathcal{D}^{1}_{a}$ and $\mathcal{D}^{1}_{b}$.\;
Train the classifier $\mathcal{A}$ on the data in $\mathcal{D}^{1}_{a}$. \;
Assess the accuracy of $\mathcal{A}$ to predict $\tilde{Y}$ in $\mathcal{D}^{1}_{b}$ and $Y$ in $\mathcal{D}^{0}$, using (\ref{eq:Q-tilde-estim}) and (\ref{eq:Q-estim}). \;
Calculate a point-estimate for the off-diagonal entries of $V := M^{-1}$ using (\ref{eq:multinomial-V}), as explained in Section~\ref{sec:method-noise-estim}.\;
Apply Algorithm~\ref{alg:correction} to construct a conformal prediction set $\hat{C}(X_{n+1})$ for $Y_{n+1}$.
\nonl
\textbf{Output: } Conformal prediction set $\hat{C}(X_{n+1})$ for $Y_{n+1}$.
\caption{Adaptive classification with contamination model fitting (plug-in)}
\label{alg:correction-estimation-plugin}
\end{algorithm}

\clearpage

\section{Simplified methods for special contamination models} \label{app:simplified-methods}

This section illustrates two specific instances of the general label contamination model introduced in Section~\ref{subsec:linear-contam-model}. 
These are particular cases in which the implementation of the adaptive conformal prediction methods developed in this paper can be significantly streamlined, and the estimation of their parameters, as discussed in Section~\ref{sec:method-noise-estim}, is also simplified.
We begin in Section~\ref{app:RR-model} with an exploration of label contamination processes described by the one-parameter randomized response model \citep{warner1965randomized}.
Following this, Section~\ref{app:BRR-model} expands to consider more complex processes described by a two-parameter extension of the randomized response model, which is specifically designed to accommodate potential label hierarchies.

\subsection{The randomized response model} \label{app:RR-model}

\subsubsection{Model description} \label{app:simplified-model}

Suppose that the relation between $X,Y$ and $\tilde{Y}$ satisfies
\begin{align}  \label{eq:def-RR}
  \P{\tilde{Y} = k \mid X, Y=l}  = (1-\epsilon) \I{k=l} + \frac{\epsilon}{K},
\end{align}
for all $l,k \in [K]$, where $\epsilon \in [0,1)$ is a scalar parameter controlling the amount of random label {\em noise}.
This setup corresponds to the classical {\em randomized response}  model of \citet{warner1965randomized}, and it has recently found many relevant applications in the context of differential privacy \citep{duchi2013local,kairouz2016discrete,ghazi2021deep}.
In particular, a well-known technique for achieving $(\varepsilon,0)$-{\em label differential privacy} consists of replacing each individual observation of $Y$ with a noisy label $\tilde{Y}$ according  to the randomized response model defined in~\eqref{eq:def-RR}, with the parameter $\epsilon$ given by
\begin{align*}
  \epsilon = \frac{K}{e^{\varepsilon} + K -1}.
\end{align*}
We refer to \cite{ghazi2021deep} for a formal definition of label differential privacy and a proof of this result.

In the notation of Section~\ref{subsec:linear-contam-model}, the randomized response model described above corresponds to a matrix $M$ (\ref{eq:contam_model}) with the form 
\begin{align} \label{eq:M-rr}
  M_{kl} =  \frac{(1-\epsilon) \rho_k}{(1-\epsilon)\rho_k + \epsilon/K} \I{k=l} + \frac{\rho_l \cdot \epsilon/K }{(1-\epsilon)\rho_k + \epsilon/K},
\end{align}
for any $k,l \in [K]$, while the contaminated label frequencies become
\begin{align} \label{eq:tilde-rho-rr}
  \tilde{\rho}_k = (1-\epsilon) \rho_k + \frac{\epsilon}{K}.
\end{align}
Note that Equation~\eqref{eq:M-rr} follows from~\eqref{eq:def-RR} with a straightforward application of Bayes' rule, and then Equation~\eqref{eq:tilde-rho-rr} is easily obtained by recalling that $\tilde{\rho}_k := \mathbb{P}[\tilde{Y} = k]$ for all $k \in [K]$.

As long as $\epsilon < 1$ and $\rho_k > 0$ for all $k \in [K]$, this matrix $M$ can be inverted analytically by applying the Sherman-Morrison formula, which leads to $V=M^{-1}$ with
\begin{align} \label{eq:V-rr}
\begin{split}
  V_{kl} 
  & = \left( 1 + \frac{\epsilon}{1-\epsilon} \cdot \frac{1/K}{\rho_k} \right) \I{k=l} - \frac{\epsilon}{(1-\epsilon)K} \cdot \frac{ (1-\epsilon) \rho_l + \epsilon/K }{\rho_k} \\
  & = \frac{\tilde{\rho}_k}{\tilde{\rho}_k-\epsilon/K} \I{k=l} - \frac{\epsilon}{K} \cdot \frac{ \tilde{\rho}_l}{\tilde{\rho}_k-\epsilon/K}.
\end{split}
\end{align}

The relatively simple structure of this matrix $V$ makes the randomized response model particularly interesting to focus on.
In particular, we will see below how this model leads to an easier-to-interpret version of the general methodology presented in Section~\ref{sec:methods}.
Further, the scalar nature of the unknown parameter $\epsilon \in [0,1)$ simplifies the task of empirically fitting the model given a limited amount of clean data, thus overcoming the main practical limitation of the general estimation approach described in Section~\ref{sec:method-noise-estim}.

\subsubsection{Adaptive coverage under known label noise}

We begin by showing how the adaptive conformal prediction methodology from Section~\ref{sec:method-known-noise} simplifies under a randomized response model for the label contamination process.
First, note that the plug-in estimate $\hat{\Delta}_k(t)$ of the coverage inflation factor $\Delta_k(t)$ utilized by Algorithm~\ref{alg:correction}, originally defined in~\eqref{eq:delta-hat}, simplifies to:
\begin{align} \label{eq:Delta-hat-RR}
\begin{split}
  \hat{\Delta}_k(t)
    & = \frac{\epsilon(1-\tilde{\rho}_k)}{K \tilde{\rho}_k-\epsilon}  \left[ \hat{F}_k^{k}(t) 
    - \frac{\sum_{l \neq k} \tilde{\rho}_l \hat{F}_l^{k}(t)}{\sum_{l \neq k} \tilde{\rho}_l}  \right],
  \end{split}
\end{align}
while the finite-sample correction term $\delta(n_k,n_{*})$ defined in~\eqref{eq:delta-constant} simplifies to
\begin{align}
  \delta(n_k,n_{*})
  & = c(n_k)
    + \frac{2\epsilon(1-\tilde{\rho}_k)}{(K \tilde{\rho}_k-\epsilon) \sqrt{n_{*}} } 
    \min \left\{ K \sqrt{\frac{\pi}{2}} , \frac{1}{\sqrt{n_{*}}} + \sqrt{\frac{\log(2K) + \log(n_{*})}{2}} \right\}.
\end{align}
Similarly, under the randomized response model, Assumption~\ref{assumption:regularity-dist-delta} becomes
\begin{align*}
  \inf_{t \in (0,1)} \Delta_k(t) \geq - \alpha + c(n_k) 
  + \frac{2 \epsilon (1-\tilde{\rho}_k)}{(K \tilde{\rho}_k-\epsilon) \sqrt{n_{*}} } 
  \left( \frac{1}{\sqrt{n_{*}}} + 2 \sqrt{\frac{\log(2K) + \log(n_{*})}{2}} \right),
\end{align*}
which is always satisfied in the large-sample limit, $n_{*} \to \infty$, as long as
\begin{align} \label{eq:epsilon-bound-asymptotically-tight}
  \epsilon \leq \min \left\{K\tilde{\rho}_k,  \frac{\alpha K \tilde{\rho}_k}{1+\alpha - \tilde{\rho}_k}\right\},
\end{align}
regardless of whether the stochastic dominance condition in~\eqref{eq:assump_scores-cond} holds.

Combined with Theorem~\ref{thm:algorithm-correction-upper}, these expressions tell us that the prediction sets output by Algorithm~\ref{alg:correction} have asymptotically tight coverage if the noise parameter $\epsilon$ is not too large and the regularity conditions of Assumptions~\ref{assumption:regularity-dist}--\ref{assumption:consitency-scores} hold.
For example, if $\alpha=0.1$ and $\tilde{\rho}_k=1/K$ for all $k \in [K]$, the upper bound in~\eqref{eq:epsilon-bound-asymptotically-tight} becomes $\epsilon \leq 0.167$ if $K=2$, and $\epsilon \leq 0.111$ if $K=5$.

Further, this model makes it easy to bound from above the term $\varphi_k(n_k,n_{*})$ in Theorem~\ref{thm:algorithm-correction-upper}, with a bound that only increases with $K$ at rate $\sqrt{\log K}$. Intuitively, this means that the asymptotic tightness of the prediction sets output by Algorithm~\ref{alg:correction} also holds for classification problems with many possible classes.

\subsubsection{Adaptive coverage under a bounded label contamination model}

The randomized response model also allows simplifying our general method for constructing adaptive prediction sets under imperfect knowledge of the label contamination process.
Since the label frequencies $\tilde{\rho}_k$ are easy to estimate accurately from the available contaminated data for all $k \in [K]$, the expression for the matrix $V$ in~\eqref{eq:V-rr} involves only one possibly unknown quantity, the scalar noise parameter $\epsilon \in [0,1)$.
In fact, it is easy to verify that, for any $l,k \in [K]$, the matrix entry $V_{kl}$ can be equivalently rewritten as
\begin{align}  \label{eq:V-RR-simplified}
  V_{kl} 
  & = \frac{\tilde{\rho}_k}{\tilde{\rho}_k-\epsilon/K} \I{k=l} - \frac{\xi \tilde{\rho}_l}{K \tilde{\rho}_k + \xi (K \tilde{\rho}_k - 1)},
\end{align}
where $\xi := \epsilon/(1-\epsilon) > 0$ is a monotone increasing transformation of $\epsilon$.

Therefore, in this special case, implementing the methods from Section~\ref{sec:method-bounded-noise} only requires a valid confidence interval for $\xi$, in lieu of a joint confidence region for all off-diagonal elements of $V$.
In particular, if a confidence interval $[\hat{\xi}^{\mathrm{low}}, \hat{\xi}^{\mathrm{upp}}]$ at level $1-\alpha_V$ is available for $\xi$, then it follows immediately from~\eqref{eq:V-RR-simplified} and~\eqref{eq:tilde-rho-rr} that a valid joint confidence region $[\hat{V}^{\mathrm{low}}, \hat{V}^{\mathrm{upp}}]$ for the off-diagonal elements of $V$ is given by:
\begin{align}  \label{eq:V-RR-CI-simplified}
  & \hat{V}^{\mathrm{upp}}_{kl} = - \frac{\hat{\xi}^{\mathrm{low}} \tilde{\rho}_l}{K \tilde{\rho}_k + \hat{\xi}^{\mathrm{upp}} (K \tilde{\rho}_k - 1)},
  & \hat{V}^{\mathrm{low}}_{kl} = - \frac{\hat{\xi}^{\mathrm{upp}} \tilde{\rho}_l}{K \tilde{\rho}_k + \hat{\xi}^{\mathrm{low}} (K \tilde{\rho}_k - 1)},
  && \forall l \neq k.
\end{align}

To simplify the following notation as much as possible, but without much loss of generality, let us make the additional assumption that the noisy label frequencies are uniform: $\tilde{\rho}_k = 1/K$ for all $k \in [K]$. Note that~\eqref{eq:tilde-rho-rr} tells us this is always the case if the true labels are uniform.
Then, Equation \eqref{eq:V-RR-CI-simplified} implies that, for all $l \neq k$,
\begin{align*} 
  & \hat{V}^{\mathrm{upp}}_{kl} = - \frac{\hat{\xi}^{\mathrm{low}}}{K},
  & \hat{V}^{\mathrm{low}}_{kl} = - \frac{\hat{\xi}^{\mathrm{upp}}}{K},
  &&   \hat{\delta}^{(V)}_{kl} := \hat{V}^{\mathrm{upp}}_{kl} - \hat{V}^{\mathrm{low}}_{kl}  = \frac{\hat{\xi}^{\mathrm{upp}}-\hat{\xi}^{\mathrm{low}}}{K},
\end{align*}
and therefore $\hat{\delta}^{(V)}_{k*} := \max_{l \neq k} \hat{\delta}^{(V)}_{kl} = \delta_{\hat{\xi}}/K$, where $\delta_{\hat{\xi}} := \hat{\xi}^{\mathrm{upp}}-\hat{\xi}^{\mathrm{low}}$.
Further, $\hat{\zeta}_k$ from~\eqref{eq:def-zeta-k} is equal to zero because $\hat{V}^{\mathrm{upp}}_{kl} = \hat{V}^{\mathrm{upp}}_{kl'}$ and $V_{kl} = V_{kl'}$ for all $l,l' \neq k$.
Thus, in this special case, our method from Section~\ref{sec:method-bounded-noise} can be implemented using the following simplified versions of the plug-in estimator $\hat{\Delta}_k^{\mathrm{ci}}(t)$ defined in~\eqref{eq:delta-hat-k-epsilon-bound}:
\begin{align} \label{eq:delta-hat-k-epsilon-bound-RR}
  \hat{\Delta}_k^{\mathrm{ci}}(t)
  & = \hat{\xi}^{\mathrm{low}} \left(1 - \frac{1}{K}\right) \left[ \hat{F}_k^{k}(t) - \frac{\sum_{l \neq k} \hat{F}_l^{k}(t)}{K-1}  \right] - \delta_{\hat{\xi}} \left(1-\frac{1}{K} \right) \left| \hat{F}_k^{k}(t) - \frac{\sum_{l \neq k}\hat{F}_l^{k}(t)}{K-1}  \right|,
\end{align}
and $\delta^{\mathrm{ci}}(n_k, n_*)$ in \eqref{eq:delta-constant-epsilon-ci},
\begin{align*}
  \delta^{\mathrm{ci}}(n_k, n_*)
  & = c(n_k) + \frac{2 \hat{\xi}^{\mathrm{upp}}(1-1/K)}{\sqrt{n_{*}}} \min \left\{ K \sqrt{\frac{\pi}{2}} , \frac{1}{\sqrt{n_{*}}} + \sqrt{\frac{\log(2K) + \log(n_{*})}{2}} \right\} \\
  & \qquad + 2 \alpha_V \bar{\xi}^{\mathrm{upp}} (1-1/K),
\end{align*}
where $\bar{\xi}^{\mathrm{upp}}$ is a (possibly very conservative) deterministic upper bound on $\xi$.

These simplified expressions highlight that, in principle, Algorithm~\ref{alg:correction-ci} only requires a one-sided confidence interval $[0,\hat{\xi}^{\mathrm{upp}}]$ for  $\xi$ in order to achieve valid coverage, because one could always evaluate $\hat{\Delta}_k^{\mathrm{ci}}(t)$ in \eqref{eq:delta-hat-k-epsilon-bound-RR} using $\smash{\hat{\xi}^{\mathrm{low}}=0}$ and $\smash{\delta_{\hat{\xi}}=\hat{\xi}^{\mathrm{upp}}}$. 
However, we already know from Section~\ref{sec:method-bounded-noise} that the prediction sets computed by Algorithm~\ref{alg:correction-ci} tend to be more informative if the confidence bounds are tighter.
In particular, the $1-\alpha + \varphi_k^{\mathrm{ci}}(n_k,n_{*})$ coverage upper bound given by Theorem~\ref{thm:algorithm-correction-upper-ci} can be interpreted even more intuitively in this special case, since
  \begin{align*}
    \varphi_k^{\mathrm{ci}}(n_k,n_{*})
    & \leq \frac{1}{n_{*}} + 2c(n_k) + \frac{2}{n_k} \cdot \frac{1}{1-\epsilon} \left[ 1 + \epsilon \cdot \frac{f_{\max}}{f_{\min}} \cdot \sum_{j=1}^{n_k+1} \frac{1}{j} \right] + \left( 1 + 4 \bar{\xi}^{\mathrm{upp}} \right) \alpha_V + 2 \EV{\delta_{\hat{\xi}}}  \\
    & \qquad\qquad + \frac{4 \EV{\hat{\xi}^{\mathrm{upp}}} }{\sqrt{n_{*}}} \min \left\{ K \sqrt{\frac{\pi}{2}} , \frac{1}{\sqrt{n_{*}}} + \sqrt{\frac{\log(2K) + \log(n_{*})}{2}} \right\}.
  \end{align*}
This highlights that the prediction sets output by Algorithm~\ref{alg:correction-ci} can be (approximately) asymptotically tight, under Assumption~\ref{assumption:regularity-dist-delta-ci},  if $\bar{\xi}^{\mathrm{upp}}$ is finite, $\alpha_V$ is small, and the expected $\mathbb{E}[\delta_{\hat{\xi}}]$ length of the confidence interval for $\xi$ is also small.
Further, under this model, Assumption~\ref{assumption:regularity-dist-delta-ci} becomes approximately equivalent, in the limit of $n_{*} \to \infty$, to 
\begin{align}
  \xi + 2 \delta_{\hat{\xi}}
  \leq \alpha - 2 \alpha_V \bar{\xi}^{\mathrm{upp}}
\end{align}
This means that Assumption~\ref{assumption:regularity-dist-delta-ci} is often realistic, similarly to Assumption~\ref{assumption:regularity-dist-delta}, as long as $\alpha_V \leq \alpha$, the noise parameter $\epsilon$ is not too large, and the confidence interval $[\hat{\xi}^{\mathrm{low}}, \hat{\xi}^{\mathrm{upp}}]$ is not too wide.
For example, suppose $K=2$, $\alpha=0.1$, $\bar{\xi}^{\mathrm{upp}}=0.125$, and the confidence interval $[\hat{\xi}^{\mathrm{low}}, \hat{\xi}^{\mathrm{upp}}]=[0.075,0.125]$ has level $\alpha_V=0.01$ and width $\delta_{\hat{\xi}} = 0.05$. 
Then, the upper bound for the true noise parameter $\xi$ in~\eqref{eq:epsilon-bound-asymptotically-tight-ci} is $\xi \leq 0.0975$, which corresponds to $\epsilon \leq 0.089$.
Additionally, Assumption~\ref{assumption:regularity-dist-delta-ci} is also guaranteed to be satisfied in the large-sample limit if the stochastic dominance condition in~\eqref{eq:assump_scores-cond} holds (i.e., $\inf_{t \in (0,1)} \Delta_k(t) \geq 0$) and a sufficiently tight confidence interval $[\hat{\xi}^{\mathrm{low}}, \hat{\xi}^{\mathrm{upp}}]$ for a small enough $\alpha_V$ is available. 

\subsubsection{Estimating the parameter $\epsilon$} \label{app:RR-model-estim}

The scalar noise parameter $\epsilon \in [0,1)$ in the randomized response model described in Section~\ref{app:simplified-model} can be empirically estimated, from a small amount of clean data, by applying a simplified version of the general procedure described in Section~\ref{sec:method-noise-estim}.

Recall that the matrices $\tilde{Q}$ and $Q$, defined in (\ref{eq:Q-def}), are such that, for any $l,k \in [K]$,
\begin{align*}
    & \tilde{Q}_{kl} := \P{ \hat{f}(X) = l \mid \tilde{Y}=k, \hat{f} }, 
    & Q_{kl} := \P{ \hat{f}(X) = l \mid Y=k, \hat{f} }.
\end{align*}
We prove in Section~\ref{app:proofs-model-fitting} that, under the randomized response model, the estimating equation $\tilde{Q} = M Q$ stated in (\ref{eq:Q-M}) implies
\begin{align} \label{eq:epsilon-estim-RR}
  \epsilon  = \frac{\psi - \tilde{\psi}}{\psi - 1/K},
\end{align}
where $\tilde{\psi}$ and $\psi$ are the probabilities that the classifier guesses correctly the corrupted and true label, respectively, of a new independent data point,
\begin{align} \label{eq:psi-psi-tilde}
    & \tilde{\psi} := \P{ \hat{f}(X) = \tilde{Y} \mid \hat{f} },
    & \psi := \P{ \hat{f}(X) = Y \mid \hat{f} }.      
\end{align}

This result suggests the following simplified method for estimating $\epsilon$.
The clean data in $\mathcal{D}^{0}$ can be used to compute an intuitive empirical estimate of $\psi$:
\begin{align} \label{eq:psi-estim}
    \psi
    & \approx \frac{1}{|\mathcal{D}^{0}|} \sum_{i \in \mathcal{D}^{0}} \I{Y_i = \hat{f}(X_i)}.
\end{align}
Similarly, $\tilde{\psi}$ can be estimated using the held-out contaminated data in $\mathcal{D}^{1}_b$:
\begin{align} \label{eq:psi-tilde-estim}
    \tilde{\psi}
    & \approx \frac{1}{|\mathcal{D}^{1}_b|} \sum_{i \in \mathcal{D}^{1}_b} \I{\tilde{Y}_i = \hat{f}(X_i)}.
\end{align}
If the contaminated observations are relatively abundant (i.e., $|\mathcal{D}^{1}| \gg |\mathcal{D}^{0}|$), one should expect~\eqref{eq:psi-tilde-estim} to provide an empirical estimate of $\tilde{\psi}$ with low variance compared to that of $\psi$ in~\eqref{eq:psi-estim}.
Therefore, \eqref{eq:epsilon-estim-RR} tells us that the leading source of uncertainty in $\epsilon$ is due to $\psi$, which depends on the unknown joint distribution of $(\hat{f}(X), Y)$ conditional on the trained classifier $\hat{f}$.
The latter is a multinomial distribution with $K^2$ categories and event probabilities equal to 
$$
\lambda_{lk} = \P{ \hat{f}(X) = k, Y=l  \mid \hat{f} }, \qquad \forall l,k \in [K].
$$
Then, since $\psi = \sum_{l=1}^{K} \lambda_{ll}$, it is easy to see that $\epsilon$ in~\eqref{eq:epsilon-estim-RR} can be written as a function of the multinomial parameter vector $\lambda = (\lambda_{lk})_{l,k \in [K]}$, as well as of other quantities ($\tilde{\psi}$ and $\tilde{\rho}_k$) that are already known with relatively high accuracy:
\begin{align} \label{eq:epsilon-estim-multin-RR}
  \epsilon = \frac{\sum_{l=1}^{K} \lambda_{ll} - \tilde{\psi}}{\sum_{l=1}^{K} \lambda_{ll} - 1/K}.
\end{align}
Thus, a confidence interval for $\epsilon$ can be directly obtained by applying standard parametric bootstrap techniques for multinomial parameters \citep{sison1995simultaneous}, similarly to the more general approach presented in Section~\ref{sec:method-noise-estim}.
In turn, this immediately translates into a confidence interval $[\hat{\xi}^{\mathrm{low}},\hat{\xi}^{\mathrm{upp}}  ]$ for $\xi$, which is a monotone increasing function of $\epsilon$, at any desired significance level $\alpha_V \in (0,1)$.
Alternatively, one could consider seeking only a point estimate $\hat{\epsilon}$ of $\epsilon$, by replacing the multinomial parameters in~\eqref{eq:epsilon-estim-multin-RR} with their standard maximum-likelihood estimates.

\subsection{The two-level randomized response model} \label{app:BRR-model}

\subsubsection{Model description} \label{app:simplified-model-BRR}

This section delves into a different special case of the general label contamination model introduced in Section~\ref{subsec:linear-contam-model}, extending the streamlined methods presented in Section~\ref{app:RR-model} for the randomized response model \citep{warner1965randomized} to a two-level hierarchical setting.
The model considered here describes a natural label contamination process involving two clearly defined groups of labels.
To illustrate, consider an animal image recognition task: one label group could represent various dog breeds, while the other could denote different cat breeds. The objective is to accurately identify not only the species but also the specific breed for each new image. 
In this context, one would often expect that a realistic data annotation process may lead to more frequent mislabeling of dog (or cat) breeds rather than mistakenly identifying a dog as a cat or vice-versa.

The two-level label contamination scenario mentioned above can be formalized using an intuitive model with two scalar parameters: $\epsilon \in [0,1)$ and $\nu \in [0,1]$. The parameter $\epsilon$ influences the likelihood that the observed label $\tilde{Y}$ deviates from the true label $Y$. This mirrors the role of $\epsilon$ in the randomized response model outlined in Section~\ref{app:RR-model}. Conversely, the parameter $\nu$ governs the interaction between the two distinct label groups.
In the special case of $\nu=0$, this model will reduce to the randomized response model \citep{warner1965randomized}, indicating an equal probability of mislabeling across breeds or species. 
On the other hand, in the special case of $\nu=1$, this model will describe a scenario in which two separate species-specific randomized response models operate independently of one another, and accurate species labeling is always ensured. The most interesting cases will of course be those in between of these two extremes, for values of $\nu \in (0,1)$.
The specifics of this model are presented next.

For simplicity, let us assume that the total number of possible labels, $K$, is even.
We describe the relation between $X,Y$ and $\tilde{Y}$ with
\begin{align}  \label{eq:def-BRR}
  \P{\tilde{Y} = k \mid X, Y=l} = T_{kl},
\end{align}
for all $k,l \in [K]$, where $T \in [0,1]^{K \times K}$ is a ($2 \times 2$) block matrix defined as
\begin{align*}
  T
  & =  \left(\begin{array}{cc}
    D^{\mathrm{BRR}} + B^{\mathrm{BRR}} & C^{\mathrm{BRR}} \\
    C^{\mathrm{BRR}} & D^{\mathrm{BRR}} + B^{\mathrm{BRR}}
  \end{array}\right).
\end{align*}
Above, the matrix $D^{\mathrm{BRR}}$ is diagonal and such that, for any $k,l \in [K/2]$,
\begin{align*}
  D^{\mathrm{BRR}}_{kl} = (1-\epsilon) \I{k=l},
\end{align*}
while the matrices $B^{\mathrm{BRR}}$ and $C^{\mathrm{BRR}}$ are constant and such that, for any $k,l \in [K/2]$,
\begin{align*}
  & B^{\mathrm{BRR}}_{kl} = \frac{\epsilon}{K} \left( 1 + \nu \right),
  & C^{\mathrm{BRR}}_{kl} = \frac{\epsilon}{K} \left( 1 - \nu \right).
\end{align*}
For example, in the special case of $K=4$, the matrix $T$ would look like
\begin{align*}
  T^{(4)}  =
\left(
\begin{array}{cccc}
  1 -\epsilon + \frac{(1 + \nu) \epsilon}{4}  & \frac{(1 + \nu) \epsilon}{4}  & \frac{(1-\nu ) \epsilon}{4}  & \frac{(1-\nu ) \epsilon}{4}  \\ 
  \frac{(1 + \nu) \epsilon}{4}  & 1 -\epsilon + \frac{(1 + \nu) \epsilon}{4}  & \frac{(1-\nu ) \epsilon}{4}  & \frac{(1-\nu ) \epsilon}{4}  \\ 
  \frac{(1-\nu ) \epsilon}{4}  & \frac{(1-\nu ) \epsilon}{4}  & 1 -\epsilon + \frac{(1 + \nu) \epsilon}{4}  & \frac{(1 + \nu) \epsilon}{4}  \\ 
  \frac{(1-\nu ) \epsilon}{4}  & \frac{(1-\nu ) \epsilon}{4}  & \frac{(1 + \nu) \epsilon}{4}  & 1 -\epsilon + \frac{(1 + \nu) \epsilon}{4}  \\
\end{array}
\right).
\end{align*}

It is easy to verify that this two-level randomized response model generally leads to the following contaminated label frequencies:
\begin{align*}
\tilde{\rho}_k
  & = (1-\epsilon) \rho_k + \frac{\epsilon}{K} + \nu \frac{\epsilon}{K} \left( 2 \sum_{l=1}^{K/2} \rho_l - 1 \right), \qquad \forall k \in [K].
\end{align*}
However, we will assume henceforth that $\rho_k = 1/K$ for all $k \in [K]$, which implies $\tilde{\rho}_k = 1/K$ for all $k \in [K]$. 
This simplification is not crucial in principle, but it is useful to make the following computations less tedious.
In particular, under the assumption of uniform label frequencies, the matrix $M$ in (\ref{prop:indep-linear}) is simply equal to $T$:
\begin{align*}
  M = T
  & =  \left(\begin{array}{cc}
    D^{\mathrm{BRR}} + B^{\mathrm{BRR}} & C^{\mathrm{BRR}} \\
    C^{\mathrm{BRR}} & D^{\mathrm{BRR}} + B^{\mathrm{BRR}}
  \end{array}\right).
\end{align*}

Having a relatively simple expression for $M$ is useful because it simplifies the implementation of our method, which relies directly on the inverse matrix $V=M^{-1}$.
The latter can now be obtained analytically by combining the Sherman-Morrison formula with standard techniques for block-matrix inversion. 
This leads to the following close-formula expressions.
For any $k \in [K]$,
\begin{align*}
  V_{kk}
  & = \frac{1}{1-\epsilon} \left( 1 - \frac{\epsilon}{K} \right) - \frac{\epsilon \nu}{K(1-\epsilon)\left[ 1-\epsilon(1-\nu)\right]}.
\end{align*}
For any $l \in \mathcal{B}_k \setminus \{k\}$, where $\mathcal{B}_k$ indicates the block to which label $k$ belongs---that is, $\mathcal{B}_k=\{1,\ldots,K/2\}$ if $k \leq K/2$ and $\mathcal{B}_k=\{K/2+1,\ldots,K\}$ otherwise---the term $V_{kl}$ is
\begin{align*}
  V_{kl } = - \frac{\epsilon}{K(1-\epsilon)} \cdot \left( 1 + \frac{\nu}{1 - \epsilon(1-\nu)} \right).
\end{align*}
Finally, for any $l \in \mathcal{B}_k^{\mathrm{c}}$, where $\mathcal{B}_k^{\mathrm{c}} := [K] \setminus \mathcal{B}_k$, 
\begin{align*}
  V_{kl } = - \frac{\epsilon}{K(1-\epsilon)} \cdot \left( 1 - \frac{\nu}{1 - \epsilon(1-\nu)} \right).
\end{align*}
For example, in the special case of $K=4$, the matrix $V$ would look like
\begin{align*}
V^{(4)} = \left(
\begin{array}{cccc}
 \frac{\epsilon ^2 - \nu  \epsilon ^2 + 3 \nu  \epsilon - 5 \epsilon +4}{4 (1-\epsilon) (1 -\epsilon + \nu  \epsilon )} & -\frac{\epsilon  (1 -\epsilon + \nu + \nu  \epsilon )}{4 (1-\epsilon) (1 -\epsilon + \nu  \epsilon )} & -\frac{(1-\nu) \epsilon }{4 (1 -\epsilon + \nu  \epsilon )} & -\frac{(1-\nu) \epsilon }{4 (1 -\epsilon + \nu  \epsilon )} \\
 - \frac{\epsilon  (1 -\epsilon + \nu + \nu  \epsilon )}{4 (1-\epsilon) (\nu \epsilon -\epsilon +1)} & \frac{\epsilon ^2 - \nu  \epsilon ^2 + 3 \nu  \epsilon - 5 \epsilon +4}{4 (1-\epsilon) (1 -\epsilon + \nu  \epsilon )} & -\frac{(1-\nu) \epsilon }{4 (1 -\epsilon + \nu  \epsilon )} & -\frac{(1-\nu) \epsilon }{4 (1 -\epsilon + \nu  \epsilon )} \\
 -\frac{(1-\nu) \epsilon }{4 (1 -\epsilon + \nu  \epsilon )} & -\frac{(1-\nu) \epsilon }{4 (1 -\epsilon + \nu  \epsilon )} & \frac{\epsilon ^2 - \nu  \epsilon ^2 + 3 \nu  \epsilon - 5 \epsilon +4}{4 (1-\epsilon) (1 -\epsilon + \nu  \epsilon )} & -\frac{\epsilon  (1 -\epsilon + \nu + \nu  \epsilon )}{4 (1-\epsilon) (1 -\epsilon + \nu  \epsilon )} \\
 -\frac{(1-\nu) \epsilon }{4 (1 -\epsilon + \nu  \epsilon )} & -\frac{(1-\nu) \epsilon }{4 (1 -\epsilon + \nu  \epsilon )} & - \frac{\epsilon  (\nu +\nu \epsilon -\epsilon +1)}{4 (1-\epsilon) (1 -\epsilon + \nu  \epsilon )} & \frac{\epsilon ^2 - \nu  \epsilon ^2 + 3 \nu  \epsilon - 5 \epsilon +4}{4 (1-\epsilon) (1 -\epsilon + \nu  \epsilon )} \\
\end{array}
\right).
\end{align*}

As detailed below, the relatively tractable structure of the matrix $V$ under this two-level randomized response model leads to an easier-to-interpret version of the general methodology presented in Section~\ref{sec:methods}.
Further, the scalar nature of the unknown parameters $\epsilon \in [0,1)$ $\nu \in [0,1]$ simplifies the task of empirically fitting the model given a limited amount of clean data, thus overcoming the main practical limitation of the general estimation approach described in Section~\ref{sec:method-noise-estim}.

\subsubsection{Adaptive coverage under a known label contamination model}

 We begin by showing how the adaptive conformal prediction methodology from Section~\ref{sec:method-known-noise} simplifies under a two-level randomized response model for the label contamination process.
First, note that the plug-in estimate $\hat{\Delta}_k(t)$ of the coverage inflation factor $\Delta_k(t)$ utilized by Algorithm~\ref{alg:correction}, originally defined in~\eqref{eq:delta-hat}, simplifies to:
\begin{align} \label{eq:Delta-hat-BRR}
\begin{split}
  \hat{\Delta}_k(t)
    & = \frac{\epsilon}{1-\epsilon}  \left( 1 - \frac{1}{K} \right) \left[ \hat{F}_k^{k}(t) - \frac{1}{K-1} \sum_{l \neq k} \hat{F}_l^{k}(t) \right] \\
    & \qquad - \frac{\epsilon}{1-\epsilon} \cdot \frac{\nu}{1 - \epsilon(1-\nu)} \left[ \frac{1}{K} \sum_{l \in \mathcal{B}_k} \hat{F}_l^{k}(t) - \frac{1}{K} \sum_{l \in \mathcal{B}_k^C} \hat{F}_l^{k}(t) \right].
  \end{split}
\end{align}
In the special case of $\nu=0$, this recovers to same expression obtained under the standard randomized response model in~\eqref{eq:Delta-hat-RR}.
In the other extreme case, if $\nu=1$, the expression for $\hat{\Delta}_k(t)$ in~\eqref{eq:Delta-hat-BRR} reduces to
\begin{align}
\begin{split}
  \hat{\Delta}_k(t)
    & = \frac{\epsilon}{1-\epsilon}  \left( 1 - \frac{2}{K} \right) \left[ \hat{F}_k^{k}(t) - \frac{1}{K/2-1} \sum_{l \in \mathcal{B}_l \setminus \{k\}} \hat{F}_l^{k}(t) \right],
  \end{split}
\end{align}
consistently with two block-specific randomized response models operating independently of one another.

In the general case, Equation~\eqref{eq:Delta-hat-BRR} implies that
\begin{align}
  \sum_{l\neq k}|V_{kl}|
  & = \frac{\epsilon}{1-\epsilon}\left(1-\frac{1}{K} - \frac{\nu}{K[1-\epsilon(1-\nu)]} \right).
\end{align}
Therefore, under a two-level randomized response model, the finite-sample correction term $\delta(n_k,n_{*})$ defined in~\eqref{eq:delta-constant} simplifies to
\begin{align*}
\begin{split}
  \delta(n_k,n_{*})
  & = c(n_k)
    + \frac{2\epsilon}{(1-\epsilon)\sqrt{n_{*}}}\left(1-\frac{1}{K} - \frac{\nu}{K[1-\epsilon(1-\nu)]} \right)  \\
    & \qquad \cdot \min \left\{ K \sqrt{\frac{\pi}{2}} , \frac{1}{\sqrt{n_{*}}} + \sqrt{\frac{\log(2K) + \log(n_{*})}{2}} \right\}.
  \end{split}
\end{align*}

Similarly, under this two-level randomized response model, Assumption~\ref{assumption:regularity-dist-delta} becomes
\begin{align*}
  \inf_{t \in (0,1)} \Delta_k(t) \geq - \alpha + c(n_k) 
  + \frac{2 \epsilon \left(1-\frac{1}{K} - \frac{\nu}{K[1-\epsilon(1-\nu)]} \right)}{(1-\epsilon)\sqrt{n_*}} 
  \left( \frac{1}{\sqrt{n_{*}}} + 2 \sqrt{\frac{\log(2K) + \log(n_{*})}{2}} \right).
\end{align*}
This is always satisfied in the limit of $n_{*} \to \infty$, as long as the following inequality holds:
\begin{align*}
  \epsilon \leq \frac{\alpha}{1+\alpha - \frac{1}{K} + \frac{1}{2} \frac{\nu}{1-\epsilon(1-\nu)}}.
\end{align*}
Instead of solving this quadratic inequality, it suffices here to note that a stricter condition, for any value of $\nu \in [0,1]$, is
\begin{align} \label{eq:epsilon-bound-asymptotically-tight-BRR}
  \epsilon \leq \frac{\alpha}{2+\alpha - \frac{1}{K}}.
\end{align}
Combined with Theorem~\ref{thm:algorithm-correction-upper}, these expressions tell us that the prediction sets output by Algorithm~\ref{alg:correction} have asymptotically tight coverage if the noise parameter $\epsilon$ is not too large and the regularity conditions of Assumptions~\ref{assumption:regularity-dist}--\ref{assumption:consitency-scores} hold, consistently with the simpler special case of the standard randomized response model (Section~\ref{app:RR-model}).
For example, if $\alpha=0.1$, the upper bound in~\eqref{eq:epsilon-bound-asymptotically-tight-BRR} becomes $\epsilon \leq 0.0625$ if $K=2$, and $\epsilon \leq 0.0526$ if $K=5$.
Further, also consistently with Section~\ref{app:RR-model}, this model makes it easy to bound from above the term $\varphi_k(n_k,n_{*})$ in Theorem~\ref{thm:algorithm-correction-upper}, with a bound that only increases with $K$ at rate $\sqrt{\log K}$.

\subsubsection{Adaptive coverage under a bounded label contamination model}

We now turn our attention to the problem in which the parameters $\epsilon$ and $\nu$ of the two-level randomized response model are not known exactly.
Under this model, applying the adaptive conformal prediction methods presented in Section~\ref{sec:method-bounded-noise} only requires a simultaneously valid pair of confidence intervals for $\epsilon$ and $\nu$, in lieu of a joint confidence region for all off-diagonal elements of $V$.

In particular, if two simultaneously valid confidence intervals $[\hat{\xi}^{\mathrm{low}}, \hat{\xi}^{\mathrm{upp}}]$ and $[\hat{\nu}^{\mathrm{low}}, \hat{\nu}^{\mathrm{upp}}]$ at level $1-\alpha_V$ are available for $\xi := \epsilon/(1-\epsilon)$ and $\nu$, then it follows immediately that a valid joint confidence region $[\hat{V}^{\mathrm{low}}, \hat{V}^{\mathrm{upp}}]$ for the off-diagonal elements of $V$ is given as follows.

For any $k \in [K]$ and $l \in \mathcal{B}_k \setminus \{k\}$,
\begin{align*}
  V_{kl} = - \frac{\xi}{K} \cdot \frac{ 1 + \nu(1 + 2 \xi)}{1 + \nu \xi }.
\end{align*}
It is easy to verify by taking partial derivatives that this is a monotone decreasing function of $\nu$ for any fixed $\xi \geq 0$, and a monotone decreasing function of $\xi$ for any fixed $\nu \in [0,1]$.
Therefore, for any $k \in [K]$ and $l \in \mathcal{B}_k \setminus \{k\}$,
\begin{align*}
  & \hat{V}^{\mathrm{upp}}_{kl} = - \frac{\hat{\xi}^{\mathrm{low}}}{K} \cdot \frac{ 1 + \hat{\nu}^{\mathrm{low}} ( 1 + 2 \hat{\xi}^{\mathrm{low}})}{1 + \hat{\nu}^{\mathrm{low}} \hat{\xi}^{\mathrm{low}}}, 
  & \hat{V}^{\mathrm{low}}_{kl} = - \frac{\hat{\xi}^{\mathrm{upp}}}{K} \cdot \frac{ 1 + \hat{\nu}^{\mathrm{upp}} (1+ 2 \hat{\xi}^{\mathrm{upp}})}{1 + \hat{\nu}^{\mathrm{upp}} \hat{\xi}^{\mathrm{upp}}}.
\end{align*}

Similarly, for any $k \in [K]$ and $l \in \mathcal{B}_k^{\mathrm{c}}$, 
\begin{align*}
  V_{kl} = - \frac{\xi}{K} \cdot \frac{ 1 - \nu}{1 + \nu \xi }.
\end{align*}
It is easy to verify by taking partial derivatives that this is a monotone increasing function of $\nu$ for any fixed $\xi \geq 0$, and a monotone decreasing function of $\xi$ for any fixed $\nu \in [0,1]$.
Therefore, for any $k \in [K]$ and $l \in \mathcal{B}_k^{\mathrm{c}}$, 
\begin{align*}
  & \hat{V}^{\mathrm{upp}}_{kl} = - \frac{\hat{\xi}^{\mathrm{low}}}{K} \cdot \frac{ 1 - \hat{\nu}^{\mathrm{upp}}}{1 + \hat{\nu}^{\mathrm{upp}} \hat{\xi}^{\mathrm{low}}}, 
  & \hat{V}^{\mathrm{low}}_{kl} = - \frac{\hat{\xi}^{\mathrm{upp}}}{K} \cdot \frac{ 1 - \hat{\nu}^{\mathrm{low}}}{1 + \hat{\nu}^{\mathrm{low}} \hat{\xi}^{\mathrm{upp}}}.
\end{align*}

Next, we require an upper confidence bound for the parameter $\hat{\zeta}_k$ defined in~\eqref{eq:def-zeta-k}.
Leveraging the confidence bounds for $\xi$ and $\nu$ as well as amenable structure of our $V$ matrix, we prove in Section~\ref{app:proofs-model-fitting} that  a valid upper bound for $\hat{\zeta}_k$ is given by
\begin{align} \label{eq:zeta-upp-BRR}
\begin{split}
  \hat{\zeta}^{\mathrm{upp}}_k
  & = \frac{2}{K} \cdot \frac{\left( \hat{\nu}^{\mathrm{upp}} - \hat{\nu}^{\mathrm{low}} \right) + \left[ \hat{\nu}^{\mathrm{upp}} \left( \hat{\xi}^{\mathrm{upp}}\right)^2 - \hat{\nu}^{\mathrm{low}} \left( \hat{\xi}^{\mathrm{low}}\right)^2 \right] + \hat{\nu}^{\mathrm{low}} \hat{\nu}^{\mathrm{upp}} \hat{\xi}^{\mathrm{low}} \hat{\xi}^{\mathrm{upp}} \left( \hat{\xi}^{\mathrm{upp}} - \hat{\xi}^{\mathrm{low}} \right) }{\left( 1 + \hat{\nu}^{\mathrm{upp}} \hat{\xi}^{\mathrm{upp}} \right) \left( 1 + \hat{\nu}^{\mathrm{low}} \hat{\xi}^{\mathrm{low}} \right) }  \\
  & \qquad  + \frac{\hat{\xi}^{\mathrm{low}}}{K} \cdot \frac{\left( \hat{\nu}^{\mathrm{upp}} - \hat{\nu}^{\mathrm{low}} \right) \left( 1 + \hat{\xi}^{\mathrm{low}} \right) }{\left( 1 + \hat{\nu}^{\mathrm{low}} \hat{\xi}^{\mathrm{low}} \right) \left( 1 + \hat{\nu}^{\mathrm{upp}} \hat{\xi}^{\mathrm{low}} \right)}.
\end{split}
\end{align}
Note that, in the special case of $\hat{\nu}^{\mathrm{upp}}=0$ and $\hat{\nu}^{\mathrm{low}}=0$, which corresponds to a standard randomized response model with $\nu=0$, we recover that $\hat{\zeta}^{\mathrm{upp}}_k=0$, consistently with the results presented earlier in Section~\ref{app:RR-model}.

Leveraging the explicit expressions for $\hat{V}^{\mathrm{low}}$, $\hat{V}^{\mathrm{upp}}$ and $\hat{\zeta}^{\mathrm{upp}}_k$ derived above, one can then directly apply the general conformal prediction methodology detailed in Section~\ref{sec:method-bounded-noise}.
Note that further simplifications of our method's remaining components within this two-tier randomized response model, though possible in principle, are not explicitly shown here due to the increasingly tedious nature of such analytical calculations beyond this point.
Instead, we will later investigate the performance of our method as a function of the confidence intervals for the parameters $\xi$ and $\nu$ through numerical experiments.

\subsubsection{Estimating the parameters $\epsilon$ and $\nu$} \label{sec:app:simplified-model-BRR-estim}

Both parameters $\epsilon \in [0,1)$ and $\nu \in [0,1)$ in the two-level randomized response model described in Section~\ref{app:simplified-model-BRR} can be empirically estimated, from a small amount of clean data, by applying a simplified version of the general procedure described in Section~\ref{sec:method-noise-estim}.

Recall that the matrices $\tilde{Q}$ and $Q$, defined in (\ref{eq:Q-def}), are such that, for any $l,k \in [K]$,
\begin{align*}
    & \tilde{Q}_{kl} := \P{ \hat{f}(X) = l \mid \tilde{Y}=k, \hat{f} }, 
    & Q_{kl} := \P{ \hat{f}(X) = l \mid Y=k, \hat{f} }.
\end{align*}
Recall also that $\tilde{\psi}$ and $\psi$ are the probabilities that the classifier guesses correctly the corrupted and true label, respectively, of a new independent data point, as defined in~\eqref{eq:psi-psi-tilde}.
Then, let us define two new quantities, $\tilde{\phi}$ and $\phi$. These are the probabilities that the classifier guesses correctly the group to which the corrupted and true label, respectively, of a new independent data point belong. That is,
\begin{align*}
  & \tilde{\phi} := \sum_{k=1}^{K} \sum_{l \in \mathcal{B}_k} \P{ \hat{f}(X) = l, \tilde{Y}=k \mid  \hat{f} },
 & \phi := \sum_{k=1}^{K} \sum_{l \in \mathcal{B}_k} \P{ \hat{f}(X) = l, Y=k \mid  \hat{f} }.
\end{align*}

We prove in Section~\ref{app:proofs-model-fitting} that, under the two-level randomized response model, the estimating equation $\tilde{Q} = M Q$ stated in (\ref{eq:Q-M}) implies the following system of two equations:
\begin{align} \label{eq:BRR-estimating-system}
  \begin{split}
\begin{cases}
    \tilde{\psi} & = (1-\epsilon) \psi + \frac{\epsilon}{K}  + \frac{\epsilon \nu}{K} \left( 2 \phi - 1 \right), \\
    \tilde{\phi} & = \phi - \epsilon (1-\nu) \left( \phi - \frac{1}{2} \right).
  \end{cases}
\end{split}
\end{align}
It is easy to verify that the solution of this system is:
\begin{align} \label{eq:epsilon-estim-BRR}
  & \epsilon = \frac{\frac{K}{2} \left( \psi - \tilde{\psi} \right) - \left(\phi - \tilde{\phi} \right)}{\frac{K}{2} \psi - \phi},
  & \nu 
  = 1 - \frac{\left( \phi - \tilde{\phi} \right) \left( \frac{K}{2} \psi - \phi \right) }{\left( \phi-\frac{1}{2} \right) \left[ \frac{K}{2} \left( \psi - \tilde{\psi} \right) - \left( \phi - \tilde{\phi} \right)  \right] }.
\end{align}
This result suggests the following simplified method for estimating $\epsilon$ and $\nu$.

The clean data in $\mathcal{D}^{0}$ can be used to compute an intuitive empirical estimate of $\psi$ using~\eqref{eq:psi-estim}, and an estimate of $\phi$ as:
\begin{align} \label{eq:phi-estim}
    \phi
    & \approx \frac{1}{|\mathcal{D}^{0}|} \sum_{i \in \mathcal{D}^{0}} \sum_{k=1}^{K} \sum_{l \in \mathcal{B}_k} \I{ \hat{f}(X_i) = l, Y_i=k }.
\end{align}
Similarly, $\tilde{\psi}$ can be estimated using the held-out contaminated data in $\mathcal{D}^{1}_b$ using~\eqref{eq:psi-tilde-estim}, and $\tilde{\phi}$ can be estimated using
\begin{align} \label{eq:phi-tilde-estim}
    \tilde{\phi}
    & \approx \frac{1}{|\mathcal{D}^{1}_b|} \sum_{i \in \mathcal{D}^{1}_b} \sum_{k=1}^{K} \sum_{l \in \mathcal{B}_k} \I{\hat{f}(X_i) = l, \tilde{Y}_i=k}.
\end{align}
If the contaminated observations are relatively abundant (i.e., $|\mathcal{D}^{1}| \gg |\mathcal{D}^{0}|$), one should expect~\eqref{eq:psi-tilde-estim} and~\eqref{eq:phi-tilde-estim} to provide empirical estimates of $\tilde{\psi}$ and $\tilde{\phi}$, respectively, with low variance compared to those of $\psi$ and $\phi$.
Therefore, the leading source of uncertainty in $\epsilon$ and $\nu$ is due to $\psi$ and $\phi$, which depend on the unknown joint distribution of $(\hat{f}(X), Y)$ conditional on the trained classifier $\hat{f}$.
We already know that the latter is a multinomial distribution with $K^2$ categories and event probabilities equal to 
$$
\lambda_{lk} = \P{ \hat{f}(X) = k, Y=l  \mid \hat{f} }, \qquad \forall l,k \in [K].
$$
Further,
\begin{align}
  & \psi = \sum_{l=1}^{K} \lambda_{ll},
  &  \phi = \sum_{k=1}^{K} \sum_{l \in \mathcal{B}_k}  \lambda_{kl}.
\end{align}
Therefore, the parameters $\epsilon$ and $\nu$ in~\eqref{eq:epsilon-estim-BRR} can be written as a function of the multinomial parameter vector $\lambda = (\lambda_{lk})_{l,k \in [K]}$, as well as of other quantities ($\tilde{\psi}$ and $\tilde{\phi}$) that are already known with relatively high accuracy.

Thus, confidence intervals for $\epsilon$ (or, equivalently, $\xi$) and $\nu$ can be directly obtained by applying standard parametric bootstrap techniques for multinomial parameters \citep{sison1995simultaneous}, similarly to the case of the standard randomized response model discussed in Section~\ref{app:RR-model}.
Alternatively, one could consider seeking only a point estimate $\hat{\epsilon}$ and $\hat{\nu}$, by replacing the multinomial parameters in~\eqref{eq:epsilon-estim-BRR} with their standard maximum-likelihood estimates.

\section{Comparison to worst-case coverage bounds} \label{app:worst-case}

\subsection{Theoretical bounds under a general linear contamination model}

As previously mentioned in Section~\ref{subsec:linear-contam-model}, it is notable that Theorem~\ref{thm:coverage-lab-cond} could also be utilized to derive worst-case coverage bounds for standard conformal prediction sets calibrated by Algorithm~\ref{alg:standard-lab-cond} using contaminated data.
This short digression highlights an interesting link with the sophisticated theoretical results of \cite{barber2022conformal}, although it does not directly add to the methodological advancements presented in this paper.

\begin{corollary}\label{thm:coverage-lab-cond-worst-case}
Suppose $(X_i,Y_i,\tilde{Y}_i)$ are i.i.d.~for all $i \in [n+1]$, and assume also that Assumption~\ref{assumption:linear-contam} holds.
Fix any prediction function $\mathcal{C}$ satisfying Definition~\ref{def:pred-function}, and let $\hat{C}(X_{n+1})$ indicate the prediction set output by Algorithm~\ref{alg:standard-lab-cond} applied using the corrupted labels $\tilde{Y}_i$ instead of the clean labels $Y_i$, for all $i \in [n]$.
Then,
\begin{align} \label{eq:prop-label-cond-coverage-lower-worst-case}
  \P{Y_{n+1} \in \hat{C}(X_{n+1}) \mid Y_{n+1} = k} \geq 1 - \alpha - \sum_{l \neq k} |V_{kl}|.
\end{align}
Further, if the conformity scores $\hat{s}(X_i,\tilde{Y}_i)$ used by Algorithm~\ref{alg:standard-lab-cond} are almost-surely distinct,
\begin{align} \label{eq:prop-label-cond-coverage-upper-worst-case}
  \P{Y_{n+1} \in \hat{C}(X_{n+1}) \mid Y_{n+1} = k} \leq 1 - \alpha + \frac{1}{n_k+1} + \sum_{l \neq k} |V_{kl}|.
\end{align}
\end{corollary}

The upper and lower coverage bounds established by Corollary~\ref{thm:coverage-lab-cond-worst-case} are independent of the conformity scores employed by Algorithm~\ref{alg:standard-lab-cond}, consistently with a worst-case type of analysis.
However, a practical drawback of this generality is the presence of a constant term in the gap between the upper \eqref{eq:prop-label-cond-coverage-upper-worst-case} and lower \eqref{eq:prop-label-cond-coverage-lower-worst-case} bounds. Therefore, this theoretical worst-case gap cannot vanish asymptotically as the sample size increases:
\begin{align}
\begin{split}
\text{Worst-Case (WC) Gap} 
  & = \text{WC-UB} - \text{WC-LB} = \frac{1}{n_k+1} + 2\sum_{l \neq k} |V_{kl}|.
\end{split}
\end{align}
This is why worst-case coverage bounds such as those provided by Corollary~\ref{thm:coverage-lab-cond-worst-case} are not as practically relevant for our objectives as the novel methodologies developed in Section~\ref{sec:methods}.
Indeed, our methods are carefully designed to adaptively learn the unknown distributions of the conformity scores from the existing data.
This allows us to obtain more informative prediction sets with asymptotically tight coverage guarantees, under the assumption that the label contamination process is random as opposed to worst-case.
For example, Theorems~\ref{thm:algorithm-correction} and~\ref{thm:algorithm-correction-upper} provide coverage lower and upper bounds equal to $1-\alpha$ and $1 - \alpha + \varphi_k(n_k,n_{*})$, respectively, for the prediction sets output by our Algorithm~\ref{alg:correction}, where $\varphi_k(n_k,n_{*})$ is a known sequence converging to zero  as $n_k \to \infty$ and $n_{*} \to \infty$.

It is important to highlight that the presence of a constant theoretical gap impacts not only the practical utility of our Corollary~\ref{thm:coverage-lab-cond-worst-case} but also that of the more nuanced theoretical worst-case bounds previously derived by \cite{barber2022conformal}.
This observation sets the stage for a detailed comparison of these worst-case bounds, which we present in the following subsection along with more detailed evidence of their practical limitations in our context.
For clarity but without much loss of generality, the comparison will be undertaken under the simpler label contamination model previously outlined in Section~\ref{app:RR-model}.

\subsection{Theoretical bounds under a randomized-response model} \label{app:worst-case-RR}

Assume for simplicity that the label contamination process follows the randomized response model defined in Section~\ref{app:simplified-model}, with noise parameter $\epsilon \in [0,1)$, and uniform label frequencies $\rho_k = 1/K = \tilde{\rho}_k$ for all $k \in [K]$.
Then, it is not difficult to see that the worst-case coverage lower and upper bounds from Corollary~\ref{thm:coverage-lab-cond-worst-case} become:
\begin{align} \label{eq:worst-case-ours}
\begin{split}
  \P{Y_{n+1} \in \hat{C}(X_{n+1}) \mid Y_{n+1} = k} & \geq 1 - \alpha - \frac{\epsilon}{1-\epsilon} \left( 1 - \frac{1}{K} \right), \\
  \P{Y_{n+1} \in \hat{C}(X_{n+1}) \mid Y_{n+1} = k} & \leq 1 - \alpha + \frac{1}{n_k+1} + \frac{\epsilon}{1-\epsilon} \left( 1 - \frac{1}{K} \right).
\end{split}
\end{align}

Let us now compare these theoretical bounds with those found by \cite{barber2022conformal}. 
We refer in particular to Appendix C of \cite{barber2022conformal}, which studies the behavior of standard conformal predictions under a {\em Huber contamination} model.
Recall that a $\epsilon'$-Huber contamination model is a standard model for describing situations in which a fraction $\epsilon' \in [0,1)$ of data points are contaminated by outliers from an unknown distribution. Thus, our randomized response model with noise parameter $\epsilon$ intuitively corresponds to a special case of a Huber contamination model with parameter $\epsilon' = \epsilon (1-1/K)$, because the latter is the expected proportion of contaminated labels $\tilde{Y}$ that differ from the true $Y$.
Then, Theorems 2 and 3 in \cite{barber2022conformal} provide the following worst-case coverage bounds:
\begin{align} \label{eq:worst-case-barber-add}
\begin{split}
  \P{Y_{n+1} \in \hat{C}(X_{n+1}) \mid Y_{n+1} = k} & \geq 1 - \alpha - \epsilon \left( 1 - \frac{1}{K} \right), \\
  \P{Y_{n+1} \in \hat{C}(X_{n+1}) \mid Y_{n+1} = k} & \leq 1 - \alpha + \frac{1}{n_k+1} + \epsilon \left( 1 - \frac{1}{K} \right).
\end{split}
\end{align}
These theoretical bounds are similar to ours but tend to be relatively tighter, especially if the label noise parameter is large, because $\epsilon/(1-\epsilon) > \epsilon$ for all $\epsilon>0$.
Further, Theorem 6 in \cite{barber2022conformal} also provides an alternative multiplicative lower bound for the worst-case coverage that can be even tighter than that in \eqref{eq:worst-case-barber-add} if the significance level $\alpha$ is small:
\begin{align} \label{eq:worst-case-barber-mult}
  \P{Y_{n+1} \in \hat{C}(X_{n+1}) \mid Y_{n+1} = k} & \geq 1 - \frac{\alpha}{1- \epsilon \left( 1 - \frac{1}{K} \right)} 
                                                      \approx 1 - \alpha - \alpha \epsilon \left( 1 - \frac{1}{K} \right),
\end{align}
where the approximation above holds for small values of $\epsilon$.

In summary, our comparison of \eqref{eq:worst-case-ours} with \eqref{eq:worst-case-barber-add} and \eqref{eq:worst-case-barber-mult} indicates that a byproduct of Theorem~\ref{thm:coverage-lab-cond} is a worst-case coverage analysis of standard conformal prediction methods in scenarios with label contamination. 
This analysis aligns with the more general findings of \cite{barber2022conformal}, though certainly with less elegance and precision. 
It is however worth repeating that the intended function of Theorem~\ref{thm:coverage-lab-cond} is not to conduct a worst-case theoretical analysis of standard conformal prediction methods.
On the contrary, our objective is to develop a new method for adaptively correcting the potential over-coverage or under-coverage of standard conformal predictions, thereby generating more insightful prediction sets. 
Theorem~\ref{thm:coverage-lab-cond}, as shown in this paper, paves the path towards effectively accomplishing this goal. 

\begin{figure}[!htb]
\centering 
\includegraphics[width=0.9\linewidth]{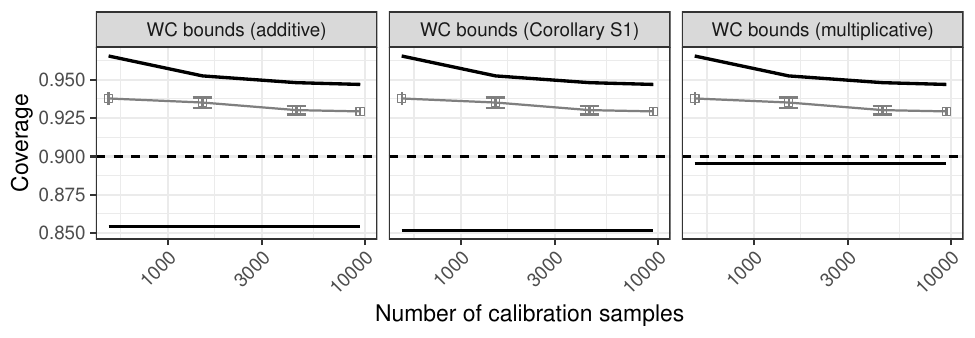}
\caption{Visualization of theoretical worst-case bounds for the coverage achieved by standard conformal prediction sets, calibrated for label-conditional coverage, on the CIFAR-10H data with noisy human-assigned labels. The thick black lines correspond to the theoretical bounds, whereas the light gray line denotes the average empirical coverage achieved by the standard conformal prediction approach in our experiments. The dashed line indicates the nominal 90\% coverage level.}
\label{fig:exp-cifar10-lab-cond-theory_bands}
\end{figure}

In contrast, theoretical worst-case analyses such as those of \cite{barber2022conformal} are not useful for our aims because they offer no actionable guidance on how to adjust the standard conformal prediction sets. 
The fixed and considerable gap between their lower and upper bounds prevents us from understanding how to correct the standard conformal prediction sets in the presence of  random label contamination. 
As an example, Figure~\ref{fig:exp-cifar10-lab-cond-theory_bands} plots different types of theoretical worst-case upper and lower bounds for the coverage achieved by standard conformal prediction sets, calibrated for label-conditional coverage, on the CIFAR-10H data set studied in Section~\ref{sec:empirical-cifar}.
See also Figure~\ref{fig:exp-cifar10-marginal-theory_bands} for similar results in the context of calibration for marginal instead of label-conditional coverage (see also Section~\ref{app:extensions-theory}).
Overall, these results demonstrate that the theoretical bounds may be technically valid but do not provide actionable insight, particularly because they do not even tell us whether the standard conformal prediction sets are too wide or too narrow in the presence of label contamination!

\begin{figure}[!htb]
\centering 
\includegraphics[width=0.9\linewidth]{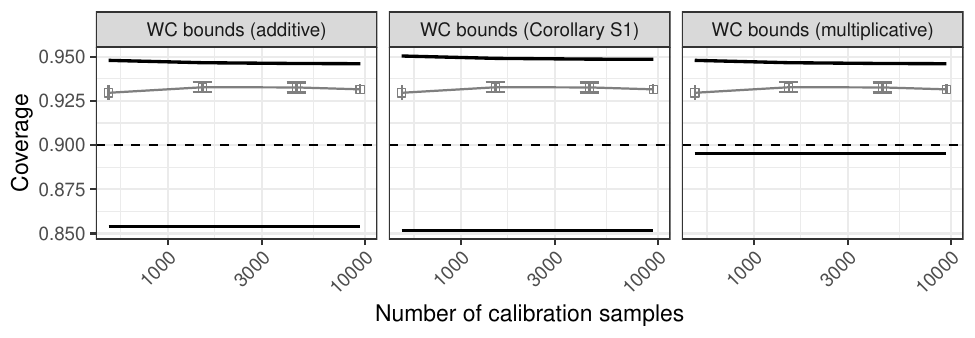}
\caption{Performance of theoretical worst-case bounds for the coverage achieved by standard conformal prediction sets, calibrated for marginal coverage, on the CIFAR-10H data. Other details are as in Figure~\ref{fig:exp-cifar10-lab-cond-theory_bands}.}
\label{fig:exp-cifar10-marginal-theory_bands}
\end{figure}

While the current worst-case analysis suggests that the sole principled approach in the context of label contamination may be to enlarge the standard conformal prediction sets, given that the theoretical lower bounds fall below the nominal $1-\alpha$ level, that is not a satisfactory solution.
It lacks adaptability to specific characteristics of the data and classifier, always resulting in prediction sets that are more conservative than the standard ones---which this paper shows already tend to be overly pessimistic.
This issue is exemplified in Figures~\ref{fig:exp-cifar10-lab-cond-theory} and~\ref{fig:exp-cifar10-marginal-theory}, which compare the performance of our adaptive method to that of a naive alternative approach inspired by the worst-case theoretical analysis.
The latter simply consists of applying standard conformal prediction methods at a modified significance level $\alpha'$ defined in such a way that the tightest available (multiplicative) coverage lower bound is exactly equal to the nominal $1-\alpha$ coverage level. 

As expected, the naive theoretical benchmark in a certain sense worsens rather than mitigates the problem considered in this paper. In fact, it leads to even less informative prediction sets compared to the standard approach ignoring the presence of label contamination.
By contrast, our adaptive method needs no worst-case analysis; instead, leveraging some knowledge of the random label contamination process, it is able to learn from the available data how to correct the standard conformal prediction sets. Thus, our method produces more informative prediction sets that achieve the desired 90\% coverage tightly.

\begin{figure}[!htb]
\centering 
\includegraphics[width=0.95\linewidth]{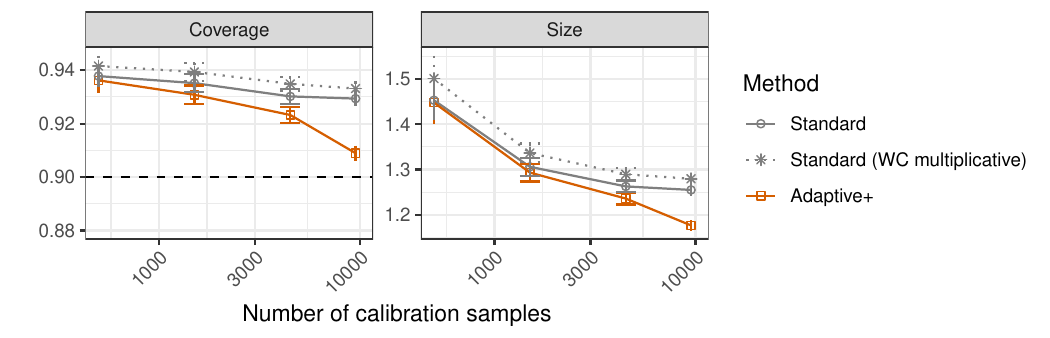}
\caption{Performance of the proposed adaptive conformal method on CIFAR-10H image data with noisy human-assigned labels, compared to two benchmarks. 
All methods are calibrated to seek 90\% label-conditional coverage. The first benchmark is the standard conformal prediction method ignoring the presence of label contamination. The second benchmark is the standard method applied at a more conservative level, chosen such that the best available worst-case theoretical lower coverage bound matches the desired coverage level $1-\alpha$.
Other details are as in Figure~\ref{fig:exp-cifar10-lab-cond}.}
\label{fig:exp-cifar10-lab-cond-theory}
\end{figure}

\begin{figure}[!htb]
\centering 
\includegraphics[width=0.95\linewidth]{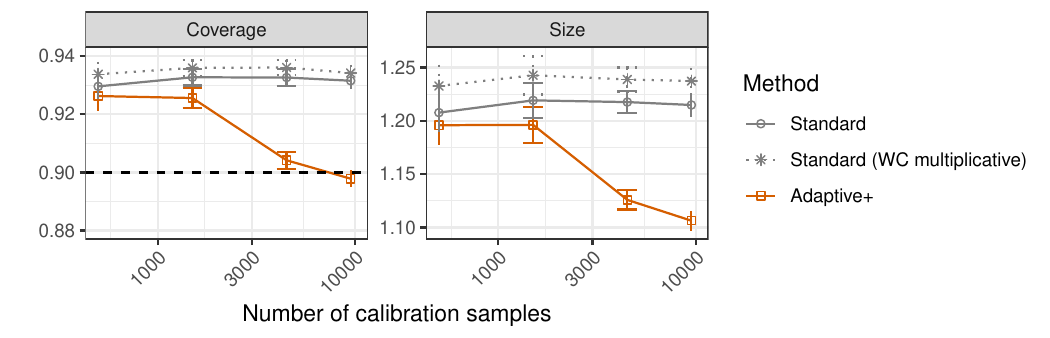}
\caption{Performance of the proposed adaptive conformal method on CIFAR-10H image data with noisy human-assigned labels, compared to two benchmarks. 
All methods are calibrated to seek 90\% marginal coverage.
Other details are as in Figure~\ref{fig:exp-cifar10-marginal} and Figure~\ref{fig:exp-cifar10-lab-cond-theory}.}
\label{fig:exp-cifar10-marginal-theory}
\end{figure}

\clearpage

\section{Extensions of preliminary theoretical results} \label{app:extensions-theory}

\subsection{General marginal coverage bounds under label contamination}

A result similar to Theorem~\ref{thm:coverage-lab-cond}, from Section~\ref{sec:general-coverage-bounds}, can be reached about the marginal coverage of the prediction sets produced by Algorithm~\ref{alg:standard-marg}.
In this case, a useful quantity to define, for any $t \in [0,1]$, is the {\em marginal coverage inflation factor},
 \begin{align} \label{eq:delta-marg}
    \Delta(t) & := \sum_{k=1}^{K} \left[ \rho_k F_k^{k}(t) - \tilde{\rho}_k \tF_k^{k}(t)\right] = F(t) - \tF(t),
  \end{align}
where $\rho_k := \mathbb{P}[Y = k]$ and $\tilde{\rho}_k = \mathbb{P}[\tilde{Y} = k]$ are the expected proportions of clean and corrupted labels equal to $k$, respectively. Above, $F(t)$ and $\tilde{F}(t)$ are the marginal cumulative distribution functions of $\hat{s}(X,Y)$ and $\hat{s}(X,\tilde{Y})$, respectively, given $\mathcal{D}$; i.e.,
 \begin{align*}
   & F(t) := \P{\hat{s}(X,Y) \leq t \mid \mathcal{D}^{\text{train}}},
   & \tilde{F}(t) := \P{\hat{s}(X,\tilde{Y}) \leq t \mid \mathcal{D}^{\text{train}}}.
 \end{align*}

\begin{theorem}\label{thm:coverage-marginal}
Suppose $(X_i,Y_i,\tilde{Y}_i)$ are i.i.d.~for all $i \in [n+1]$.
Fix any prediction function $\mathcal{C}$ satisfying Definition~\ref{def:pred-function}, and let $\hat{C}(X_{n+1})$ indicate the prediction set output by Algorithm~\ref{alg:standard-marg} applied using the corrupted labels $\tilde{Y}_i$ instead of the clean labels $Y_i$, for all $i \in [n]$.
Then,
\begin{align} \label{eq:prop-marg-coverage-lower}
  \P{Y_{n+1} \in \hat{C}(X_{n+1})} \geq 1 - \alpha + \EV{\Delta(\htau)}.
\end{align}
Further, if the scores $\hat{s}(X_i,\tilde{Y}_i)$ used by Algorithm~\ref{alg:standard-marg} are almost-surely distinct,
\begin{align} \label{eq:prop-marg-coverage-upper}
  \P{Y_{n+1} \in \hat{C}(X_{n+1})} \leq 1 - \alpha + \frac{1}{n_{\mathrm{cal}}+1} + \EV{\Delta(\htau)}.
\end{align}
\end{theorem}

\subsection{Coverage lower bounds under a linear contamination model}

A similar stochastic dominance condition as in Corollary~\ref{cor:coverage-cond} also implies that the prediction sets output by Algorithm~\ref{alg:standard-marg} are conservative in the marginal coverage sense of~\eqref{eq:def-marg-coverage}.

\begin{corollary} \label{cor:coverage-marg}
Consider the same setting of Theorem~\ref{thm:coverage-marginal} and assume Assumption~\ref{assumption:linear-contam} holds.
Suppose also that the cumulative distribution functions of the scores \eqref{eq:cdf-scores} satisfy
    \begin{equation} \label{eq:assump_scores-cond-marg}
        \max_{l \neq k} F_k^{l}(t)  \leq F_k^{k}(t),
      \end{equation}
      for all $t\in\mathbb{R}$ and $k\in[K]$.
Then, $\Delta(\hat{\tau}) \geq 0$ almost-surely, and hence the predictions $\hat{C}(X_{n+1})$ of Algorithm~\ref{alg:standard-marg} satisfy~\eqref{eq:def-marg-coverage}.
\end{corollary}

Equation~\eqref{eq:assump_scores-cond-marg} is similar to~\eqref{eq:assump_scores-cond}, although the two conditions are not exactly equivalent. Intuitively, \eqref{eq:assump_scores-cond-marg} states that the scores $\hat{s}(X,k)$ assigned by the machine learning model tend to be smaller than any other scores $\hat{s}(X,l)$ for $l \neq k$ among data points with true label $Y=k$. In other words, this could be interpreted as saying that the correct label is the most likely point prediction of the machine learning model, for each possible class $k \in [K]$.

\clearpage

\section{Methodology extensions} \label{app:extensions}

\subsection{Adaptive prediction sets with marginal coverage} \label{app:methods-marginal}

While this paper has so far focused on achieving tight label-conditional coverage~\eqref{eq:def-lab-cond-coverage}, the proposed methods can be adapted to alternatively control the weaker notion of marginal coverage~\eqref{eq:def-marg-coverage}.
Concretely, we present here Algorithm~\ref{alg:correction-marg}, which extends for that purpose Algorithm~\ref{alg:correction} from Section~\ref{sec:method-known-noise}.
It easy to see that the analogous extensions of the methods described in Sections~\ref{sec:method-bounded-noise}--\ref{sec:method-noise-estim} would also follow similarly.

Algorithm~\ref{alg:correction-marg} differs from Algorithm~\ref{alg:correction} in that it calculates a single threshold
\begin{align} \label{eq:Ik-set-tauk-marg}
  \hat{\tau}^{\mathrm{marg}}
  & = \begin{cases}
    S_{(\hat{i}^{\mathrm{marg}})} \text{ where } \hat{i}^{\mathrm{marg}} = \min\{i \in \hat{\mathcal{I}}^{\mathrm{marg}}\}, & \text{if } \hat{\mathcal{I}}^{\mathrm{marg}} \neq \emptyset, \\
    1, & \text{if } \hat{\mathcal{I}}^{\mathrm{marg}} = \emptyset,
  \end{cases}
\end{align}
where the set $\hat{\mathcal{I}}^{\mathrm{marg}}$ is defined as
\begin{align} \label{eq:Ik-set-marg}
  \hat{\mathcal{I}}^{\mathrm{marg}} := \left\{i \in [n_{\mathrm{cal}}] : \frac{i}{n_{\mathrm{cal}}} \geq 1 - \alpha - \hDelta(S_{(i)}) + \delta^{\mathrm{marg}}(n_{\mathrm{cal}},n_{*})  \right\},
\end{align}
for an empirical estimate $\hat{\Delta}(t)$ of the factor $\Delta(t)$ in~\eqref{eq:delta-marg} given by
\begin{align} \label{eq:delta-hat-marg}
  \hat{\Delta}(t)
  & := \sum_{k=1}^{K} \left[ (\rho_kV_{kk}-\tilde{\rho}_k) \hat{F}_k^{k}(t) + \rho_k \sum_{l \neq k} V_{kl} \hat{F}_l^{k}(t) \right],
\end{align}
where, for any $k \in [K]$,
\begin{align}
  \rho_k = \sum_{l=1}^{K} M_{lk} \tilde{\rho}_l,
\end{align}
and a finite-sample correction term taking the form
\begin{align} \label{eq:delta-constant-marg}
\begin{split}
  \delta^{\mathrm{marg}}(n_{\mathrm{cal}},n_{*})
  & := c(n_{\mathrm{cal}}) + \frac{2 \max_{k \in [K]} \sum_{l\neq k} |V_{kl}| + \sum_{k=1}^{K} | \rho_k - \tilde{\rho}_k  |}{\sqrt{n_{*}}} \\
  & \qquad \cdot \min \left\{ K^2 \sqrt{\frac{\pi}{2}} , \frac{1}{\sqrt{n_{*}}} + \sqrt{\frac{\log(2K^2) + \log(n_{*})}{2}} \right\}.
\end{split}
\end{align}

It is worth pointing out that, unlike our adaptive methods for label-conditional coverage presented in Section~\ref{sec:methods}, Algorithm~\ref{alg:correction-marg} requires knowledge of the expected class frequencies for the contaminated labels; i.e., $\tilde{\rho}_k$ for all $k \in [K]$.
This additional assumption could be relaxed by replacing $\tilde{\rho}_k$ with an empirical estimate obtained from the available contaminated data, but for simplicity we choose to treat $\tilde{\rho}_k$ as known. Fortunately, this does not prevent our solution from being practical because $\tilde{\rho}_k$ is easy to estimate.

\begin{algorithm}[!htb]
\DontPrintSemicolon

\KwIn{Data set $\{(X_i, \tilde{Y}_i)\}_{i=1}^{n}$ with corrupted labels $\tilde{Y}_i \in [K]$.}
\myinput{The inverse $V$ of the matrix $M$ in \eqref{eq:contam_model}.}
\myinput{The expected contaminated label frequencies $\tilde{\rho}_k$, for all $k \in [K]$.}
\myinput{Unlabeled test point with features $X_{n+1}$.}
\myinput{Machine learning algorithm $\mathcal{A}$ for training a $K$-class classifier.}
\myinput{Prediction function $\mathcal{C}$ satisfying Definition~\ref{def:pred-function}; e.g., \eqref{eq:pred-function-hps}.}
\myinput{Desired coverage level $1-\alpha \in (0,1)$.}

Randomly split $[n]$ into two disjoint subsets, $\mathcal{D}^{\text{train}}$ and $\mathcal{D}^{\mathrm{cal}}$.\;
Train the classifier $\mathcal{A}$ on the data in $\mathcal{D}^{\text{train}}$. \;
Compute conformity scores $\hat{s}(X_i,k)$ using \eqref{eq:conf-scores} for all $i \in \mathcal{D}^{\mathrm{cal}}$ and all $k \in [K]$.\;
Define the empirical CDF $\hat{F}_l^{k}$ of $\{\hat{s}(X_i,k) : i \in \mathcal{D}_l^{\mathrm{cal}} \}$ for all $k,l \in [K]$,  as in \eqref{eq:F-hat}. \;

  Sort $\{\hat{s}(X_i,Y_i) : i \in \mathcal{D}^{\mathrm{cal}} \}$ into $(S_{(1)}, S_{(2)}, \dots, S_{(n_{\mathrm{cal})}})$, in ascending order. \;
  Compute $\hDelta(S_{(i)})$ for all $i \in [n_{\mathrm{cal}}]$, as in \eqref{eq:delta-hat-marg}. \;
  Compute $\delta^{\mathrm{marg}}(n_{\mathrm{cal}},n_{*})$ using \eqref{eq:delta-constant-marg}, based on a Monte Carlo estimate of $c(n_{\mathrm{cal}})$.\;
  Construct the set $\hat{\mathcal{I}}^{\mathrm{marg}} \subseteq [K]$ as in \eqref{eq:Ik-set-marg}.\;
  Evaluate $\hat{\tau}^{\mathrm{marg}}$ based on $\hat{\mathcal{I}}^{\mathrm{marg}}$ as in \eqref{eq:Ik-set-tauk-marg}.\;
Evaluate $\hat{C}^{\mathrm{marg}}(X_{n+1}) = \mathcal{C}(X_{n+1}, \hat{\tau}^{\mathrm{marg}}; \hat{\pi})$.

\nonl
\textbf{Output: } Conformal prediction set $\hat{C}^{\mathrm{marg}}(X_{n+1})$ for $Y_{n+1}$.

\caption{Contamination-adaptive classification with marginal coverage}
\label{alg:correction-marg}
\end{algorithm}

Below, Theorem~\ref{thm:algorithm-correction-marg} establishes that the marginal coverage of the prediction sets output by Algorithm~\ref{alg:correction-marg} is bounded from below by $1-\alpha$, as long as our method is applied based on the true model matrix $M$ in \eqref{eq:contam_model} and the correct contaminated label frequencies $\tilde{\rho}_k$ for all $k \in [K]$.

\begin{theorem} \label{thm:algorithm-correction-marg}
Suppose $(X_i,Y_i,\tilde{Y}_i)$ are i.i.d.~for all $i \in [n+1]$.
Assume the label contamination model in Assumption~\ref{assumption:linear-contam} holds.
Fix any prediction function $\mathcal{C}$ satisfying Definition~\ref{def:pred-function}, and let $\hat{C}^{\mathrm{marg}}(X_{n+1})$ indicate the prediction set output by Algorithm~\ref{alg:correction-marg} based on the inverse $V$ of the model matrix $M$ in the label contamination model~\eqref{eq:contam_model} and the true values of the contaminated label frequencies $\tilde{\rho}_k$ for all $k \in [K]$.
Then,
  \begin{align*}
    \P{Y_{n+1} \in \hat{C}^{\mathrm{marg}}(X_{n+1})} \geq 1 - \alpha.
  \end{align*}
\end{theorem}

Next, we prove that Algorithm~\ref{alg:correction-marg} is not overly conservative, following an approach similar to that of Theorem~\ref{thm:algorithm-correction-upper} for Algorithm~\ref{alg:correction}.
This requires a technical lower bound for the marginal coverage inflation factor $\Delta(t)$, whose interpretation is similar to that of Assumption~\ref{assumption:regularity-dist-delta}.
\begin{assumption} \label{assumption:regularity-dist-delta-marg}
The factor $\Delta(t)$ in~\eqref{eq:delta-hat-marg} is bounded from below by:
\begin{align*}
  & \inf_{t \in (0,1)} \Delta(t) \geq - \alpha + c(n_{\mathrm{cal}}) \\
  & \qquad \qquad \qquad + \left( 2 \max_{k \in [K]} \sum_{l\neq k} |V_{kl}| + \sum_{k=1}^{K} | \rho_k - \tilde{\rho}_k  | \right) \frac{1+2\sqrt{\log(2K^2) + \log(n_{*})}}{\sqrt{n_{*}}}.
\end{align*}
\end{assumption}

Under this setup, an upper bound for the marginal coverage of the conformal prediction sets output by Algorithm~\ref{alg:correction-marg} is established below.

\begin{theorem} \label{thm:algorithm-correction-marg-upper}
Under the setup of Theorem~\ref{thm:algorithm-correction-marg}, let $\hat{C}^{\mathrm{marg}}(X_{n+1})$ indicate the prediction set output by Algorithm~\ref{alg:correction-marg} based on the inverse $V$ of the model matrix $M$ in~\eqref{eq:contam_model} and the true values of the label frequencies $\rho_k,\tilde{\rho}_k$ for all $k \in [K]$.
Suppose also that Assumptions~\ref{assumption:regularity-dist}, \ref{assumption:consitency-scores}, and~\ref{assumption:regularity-dist-delta-marg} hold, and that $\tilde{\rho}_k > 0$ for all $k \in [K]$. Then,
  \begin{align*}
    \P{Y_{n+1} \in \hat{C}^{\mathrm{marg}}(X_{n+1})} \leq  1 - \alpha + \varphi^{\mathrm{marg}}_k(n_k,n_{*}),
  \end{align*}
where
\begin{align}
\begin{split}
    & \varphi^{\mathrm{marg}}(n_{\mathrm{cal}}, n_{*}) = 2 \delta^{\mathrm{marg}}(n_{\mathrm{cal}},n_{*}) + \frac{1}{n_{\mathrm{cal}}} + \frac{1}{n_{*}} + \frac{1}{n_{\mathrm{cal}}} \cdot \max_{k \in [K]} \left( \frac{\rho_k}{\tilde{\rho}_k} \sum_{l=1}^{K} |V_{kl}| \right)  \\
  & \qquad \qquad + \frac{\sum_{j=1}^{n_{\mathrm{cal}}+1} 1/j}{n_{\mathrm{cal}}} \cdot \left[ \max_{k \in [K]} \frac{\rho_k V_{kk} - \tilde{\rho}_k}{\tilde{\rho}_k}  +  \frac{f_{\max}}{f_{\min}} \cdot \max_{k \in [K]} \left( \frac{\rho_k}{\tilde{\rho}_k} \sum_{l \neq k} V_{kl} \right) \right].
\end{split}
\end{align}
\end{theorem}

The interpretation of Theorem~\ref{thm:algorithm-correction-marg-upper} is analogous to that of Theorem~\ref{thm:algorithm-correction-upper}: the marginal coverage of the prediction sets output by Algorithm~\ref{alg:correction-marg} is guaranteed to be asymptotically tight because $\varphi^{\mathrm{marg}}(n_k,n_{*}) \to 0$ as $n_{*} \to \infty$.

We conclude this section by noting that the power of Algorithm~\ref{alg:correction-marg} can be further boosted without losing marginal coverage, as long as a relatively mild ``optimistic'' assumption on the marginal coverage inflation factor defined in~\eqref{eq:delta-marg} holds.
Concretely, we propose to apply Algorithm~\ref{alg:correction-marg} with the set $\hat{\mathcal{I}}^{\mathrm{marg}}$ in~\eqref{eq:Ik-set-marg} replaced by:
\begin{align} \label{eq:Ik-set-marg-optimist}
  \hat{\mathcal{I}}^{\mathrm{marg}} := \left\{i \in [n_{\mathrm{cal}}] : \frac{i}{n_{\mathrm{cal}}} \geq 1 - \alpha - \max\left\{ \hDelta(S_{(i)}) - \delta^{\mathrm{marg}}(n_{\mathrm{cal}},n_{*}) , - \frac{(1-\alpha)}{n_{\mathrm{cal}}} \right\} \right\}.
\end{align}
This optimistic variation of Algorithm~\ref{alg:correction-marg} is analogous to the extension of Algorithm~\ref{alg:correction} presented earlier in Section~\ref{sec:boosting-optimist}, and it enjoys a similar coverage guarantee.

\begin{proposition} \label{thm:algorithm-correction-marg-optimist}
Under the same setup as  Theorem~\ref{thm:algorithm-correction-marg}, assume also that  $\inf_{t\in \R}\Delta(t) \geq \delta^{\mathrm{marg}}(n_{\mathrm{cal}}, n_{*})- (1-\alpha)/n_{\mathrm{cal}}$.
If $\hat{C}^{\mathrm{marg}}(X_{n+1})$ is the prediction set output by Algorithm~\ref{alg:correction-marg} applied with $\hat{\mathcal{I}}$ in~\eqref{eq:Ik-set-marg-optimist} instead of~\eqref{eq:Ik-set-marg}, then $\P{Y_{n+1} \in \hat{C}^{\mathrm{marg}}(X_{n+1})} \geq 1 - \alpha$.
\end{proposition}

\subsection{Adaptive prediction sets with calibration-conditional coverage} \label{app:methods-cal-cond}

The methods presented in this paper can also be extended to construct prediction sets $\hat{C}^{\mathrm{cc}}(X_{n+1})$ guaranteeing the following notion of calibration-conditional coverage \citep{vovk2012conditional}:
  \begin{align} \label{eq:cal-cond-coverage}
    \P{ \P{Y_{n+1} \in \hat{C}^{\mathrm{cc}}(X_{n+1}) \mid Y = k, \mathcal{D}} \geq 1 - \alpha} \geq 1-\gamma,
  \end{align}
for any given $\alpha, \gamma  \in (0,1)$.
Concretely, we present here Algorithm~\ref{alg:correction-cc}, which extends for the aforementioned purpose Algorithm~\ref{alg:correction} from Section~\ref{alg:correction}.
Analogous extensions of the methods described in Sections~\ref{sec:method-bounded-noise}--\ref{sec:method-noise-estim} and Section~\ref{app:methods-marginal} would also follow similarly.
In a nutshell, Algorithm~\ref{alg:correction-cc} differs from Algorithm~\ref{alg:correction} only in that it uses a finite-sample correction factor
\begin{align} \label{eq:delta-constant-cc}
  \delta^{\mathrm{cc}}(n_k,n_{*},\gamma)
  & :=  \sqrt{\frac{\log(1/\gamma_1)}{2 n_k}} + 2\sum_{\l\neq k}|V_{kl}| \sqrt{\frac{\log(2K) + \log(1/\gamma_2)}{2 n_{*}}},
\end{align}
where
\begin{align*}
& \gamma_1=  \gamma \left( 1- \frac{1}{2} \cdot \frac{\sum_{\l\neq k}|V_{kl}|}{\sum_{l=1}^{K}|V_{kl}|}\right),
& \gamma_2=  \frac{\gamma}{2} \cdot \frac{\sum_{\l\neq k}|V_{kl}|}{\sum_{l=1}^{K}|V_{kl}|},
\end{align*}
instead of the factor $\delta(n_k,n_{*})$ defined in \eqref{eq:delta-constant}.
Note that the second term on the right-hand-side of \eqref{eq:delta-constant-cc} simply vanishes in the special case where $\sum_{l=1}^{K}|V_{kl}|=0$, which corresponds to the absence of label contamination.

\begin{algorithm}[!ht]
\DontPrintSemicolon

\KwIn{Data set $\{(X_i, \tilde{Y}_i)\}_{i=1}^{n}$ with corrupted labels $\tilde{Y}_i \in [K]$.}
\myinput{
The inverse $V$ of the matrix $M$ in \eqref{eq:contam_model}.}
\myinput{Unlabeled test point with features $X_{n+1}$.}
\myinput{Machine learning algorithm $\mathcal{A}$ for training a $K$-class classifier.}
\myinput{Prediction function $\mathcal{C}$ satisfying Definition~\ref{def:pred-function}; e.g., \eqref{eq:pred-function-hps}.}
\myinput{Desired coverage level $1-\alpha \in (0,1)$}
\myinput{Desired calibration-conditional significance level $\gamma \in (0,1)$.}

Randomly split $[n]$ into two disjoint subsets, $\mathcal{D}^{\text{train}}$ and $\mathcal{D}^{\mathrm{cal}}$.\;
Train the classifier $\mathcal{A}$ on the data in $\mathcal{D}^{\text{train}}$. \;
Compute conformity scores $\hat{s}(X_i,k)$ using \eqref{eq:conf-scores} for all $i \in \mathcal{D}^{\mathrm{cal}}$ and all $k \in [K]$.\;
Define the empirical CDF $\hat{F}_l^{k}$ of $\{\hat{s}(X_i,k) : i \in \mathcal{D}_l^{\mathrm{cal}} \}$ for all $l \in [K]$,  as in \eqref{eq:F-hat}. \;
\For{$k=1, \dots, K$}{
  Define $\mathcal{D}_k^{\mathrm{cal}} = \{ i \in \mathcal{D}^{\mathrm{cal}} : \tilde{Y}_i = k\}$ and $n_k = |\mathcal{D}_k^{\mathrm{cal}}|$.\;
  Sort $\{\hat{s}(X_i,k) : i \in \mathcal{D}_k^{\mathrm{cal}} \}$ into $(S^k_{(1)}, S^k_{(2)}, \dots, S^k_{(n_k)})$, in ascending order. \;
  Compute $\hat{F}_l^{k}(S^k_{(i)})$ for all $i \in [n_k]$ and $l \in [K]$. \;
  Compute $\hDelta_k(S^k_{(i)})$ for all $i \in [n_k]$, as in \eqref{eq:delta-hat}. \;
  Compute $\delta^{\mathrm{cc}}(n_k,n_{*},\gamma)$ using \eqref{eq:delta-constant-cc}. \;
  Construct the set $\hat{\mathcal{I}}_k \subseteq [K]$ as in \eqref{eq:Ik-set}, using $\delta^{\mathrm{cc}}(n_k,n_{*},\gamma)$ instead of $\delta(n_k,n_{*})$.\;
  Evaluate $\hat{\tau}_k$ based on $\hat{\mathcal{I}}_k$ as in \eqref{eq:Ik-set-tauk}.

}
Evaluate $\hat{C}^{\mathrm{cc}}(X_{n+1}) = \mathcal{C}(X_{n+1}, \hat{\tau}; \hat{\pi})$, where $\hat{\tau} = (\hat{\tau}_1, \ldots, \hat{\tau}_K)$.

\nonl
\textbf{Output: } Conformal prediction set $\hat{C}^{\mathrm{cc}}(X_{n+1})$ for $Y_{n+1}$.

\caption{Extension of Algorithm~\ref{alg:correction} with calibration-conditional coverage}
\label{alg:correction-cc}
\end{algorithm}

\begin{theorem} \label{thm:algorithm-correction-cc}
Suppose $(X_i,Y_i,\tilde{Y}_i)$ are i.i.d.~for all $i \in [n+1]$.
Assume the label contamination model in Assumption~\ref{assumption:linear-contam} holds.
Fix any prediction function $\mathcal{C}$ satisfying Definition~\ref{def:pred-function}, and let $\hat{C}^{\mathrm{cc}}(X_{n+1})$ indicate the prediction set output by Algorithm~\ref{alg:correction-cc} based on the inverse $V$ of the model matrix $M$ in the label contamination model~\eqref{eq:contam_model}.
Then, for any $k \in [K]$,
  \begin{align*}
    \P{ \P{Y_{n+1} \in \hat{C}^{\mathrm{cc}}(X_{n+1}) \mid Y = k, \mathcal{D}} \geq 1 - \alpha} \geq 1-\gamma.
  \end{align*}
\end{theorem}

It is interesting to compare the finite-sample correction factor in \eqref{eq:delta-constant-cc} to the standard approach for constructing conformal prediction sets with calibration-conditional coverage.
 In fact, Proposition 2a in \citet{vovk2012conditional} implies that calibration-conditional coverage~\eqref{eq:cal-cond-coverage} can be achieved in the absence of label contamination by simply applying Algorithm \ref{alg:standard-lab-cond}, in Section~\ref{app:standard-lab-cond}, with the nominal level $\alpha$ replaced by
\begin{align*}
  \alpha' = \alpha - \sqrt{\frac{\log(1/\gamma)}{2 n_{k}}}.
\end{align*}
Intuitively, this correction is similar to the first term on the right-hand-side of \eqref{eq:delta-constant-cc}, and it becomes equivalent in the case that the matrix $V$ is diagonal, which corresponds to the absence of label contamination.

Following an approach similar to that of Theorem~\ref{thm:algorithm-correction-upper}, it can be proved that Algorithm~\ref{alg:correction-cc} is not overly conservative.
This requires a technical lower bound for the coverage inflation factor $\Delta_k(t)$, whose interpretation is similar to that of Assumption~\ref{assumption:regularity-dist-delta}.
\begin{assumption} \label{assumption:regularity-dist-delta-cc}
The coverage inflation factor $\Delta_k(t)$ is bounded from below by:
\begin{align*}
  \inf_{t \in (0,1)} \Delta_k(t) \geq - \alpha + \sqrt{\frac{\log(1/\gamma_1)}{2 n_k}} + 4\sum_{\l\neq k}|V_{kl}| \sqrt{\frac{\log(2K) + \log(1/\gamma_2)}{2 n_{*}}}.
\end{align*}
\end{assumption}

Under this setup, a finite-sample upper bound for the marginal coverage of the conformal prediction sets output by Algorithm~\ref{alg:correction-cc} is established below.

\begin{theorem} \label{thm:algorithm-correction-cc-upper}
Suppose $(X_i,Y_i,\tilde{Y}_i)$ are i.i.d.~for all $i \in [n+1]$.
Assume the label contamination model in Assumption~\ref{assumption:linear-contam} holds.
Fix any prediction function $\mathcal{C}$ satisfying Definition~\ref{def:pred-function}, and let $\hat{C}^{\mathrm{cc}}(X_{n+1})$ indicate the prediction set output by Algorithm~\ref{alg:correction-cc} based on the inverse $V$ of the model matrix $M$ in the label contamination model~\eqref{eq:contam_model}.
Suppose also that Assumptions~\ref{assumption:regularity-dist}, \ref{assumption:consitency-scores}, and~\ref{assumption:regularity-dist-delta-cc} hold.
Then, for any $k \in [K]$,
  \begin{align*}
    \P{ \P{Y_{n+1} \in \hat{C}^{\mathrm{cc}}(X_{n+1}) \mid Y = k, \mathcal{D}} \leq 1 - \alpha + \varphi_k^{\mathrm{cc}}(n_k,n_{*},\gamma)} \geq 1-\gamma.
  \end{align*}
where
\begin{align*}
  \varphi_k^{\mathrm{cc}}(n_k,n_{*},\gamma)
  & =  \delta^{\mathrm{cc}}(n_k,n_{*},\gamma) + \frac{1}{n_k} \left( 1 + \frac{V_{kk} + \sum_{l \neq k} |V_{kl}|}{\gamma \bar{v}}\right) + \sqrt{\frac{\log[3/(\gamma \bar{v})]}{2 n_k}} \\
  & \qquad + 2\sum_{\l\neq k}|V_{kl}| \sqrt{\frac{\log(2K) + \log[3/(\gamma \bar{v})]}{2 n_{*}}} \\
  & \qquad + 2\sum_{l\neq k}|V_{kl}| \cdot \frac{f_{\max}}{f_{\min}} \cdot \frac{1}{n_k} \cdot \left[ \log(n_k+1) + \frac{3}{\gamma \bar{v}}  \sum_{j=1}^{n_k+1} \frac{1}{j}   \right],
\end{align*}
and
\begin{align*}
  \bar{v}
    = \frac{1}{2} \left( 1 - \frac{1}{2} \cdot \frac{\sum_{\l\neq k}|V_{kl}|}{\sum_{l=1}^{K}|V_{kl}|} \right).
\end{align*}
\end{theorem}

The interpretation of Theorem~\ref{thm:algorithm-correction-cc-upper} is similar to that of Theorem~\ref{thm:algorithm-correction-upper}, because $\varphi_k^{\mathrm{cc}}(n_k,n_{*},\gamma) \to 0$ as $n_{*} \to \infty$, for any fixed $\gamma > 0$.

We conclude this section by noting that the power of Algorithm~\ref{alg:correction-cc} can also be further boosted without losing the calibration-conditional coverage guarantee, as long as a relatively mild ``optimistic'' assumption on the coverage inflation factor defined in~\eqref{eq:delta} holds.
Concretely, we propose to apply Algorithm~\ref{alg:correction-cc} with the set $\hat{\mathcal{I}}_k$ in~\eqref{eq:Ik-set} replaced by:
\begin{align} \label{eq:Ik-set-cc-optimist}
  \hat{\mathcal{I}}_k := \left\{i \in [n_{k}] : \frac{i}{n_{k}} \geq 1 - \alpha - \max\left\{ \hDelta_k(S_{(i)}) - \delta^{\mathrm{cc}}(n_{k},n_{*},\gamma) , - \sqrt{\frac{\log(1/\gamma)}{2 n_{k}}} \right\} \right\}.
\end{align}
The motivation behind this approach is that, in the absence of label contamination, valid calibration-conditional coverage~\eqref{eq:cal-cond-coverage} can be achieved by simply applying Algorithm \ref{alg:standard-lab-cond} with the nominal level $\alpha$ replaced by $\alpha' = \alpha - \sqrt{\log(1/\gamma)/(2 n_{k})}$ \citep{vovk2012conditional}.
Therefore, in analogy with the optimistic variation of Algorithm~\ref{alg:correction-cc} presented earlier in Section~\ref{sec:boosting-optimist}, it is now intuitive to propose an optimistic extension of Algorithm~\ref{alg:correction-cc} that is never more conservative than the standard benchmark for clean data.
The following result establishes that this optimistic approach still guarantees the desired calibration-conditional coverage~\eqref{eq:cal-cond-coverage}, as long as some relatively mild assumption on the coverage inflation factor holds.

\begin{proposition} \label{thm:algorithm-correction-cc-optimist}
Under the setup of Theorem~\ref{thm:algorithm-correction-cc}, assume also that  $\inf_{t\in \R}\Delta_k(t) \geq \delta^{\mathrm{cc}}(n_k, n_{*}) - \sqrt{ \log(1/\gamma) / (2 n_{k})}$.
If $\hat{C}^{\mathrm{cc}}(X_{n+1})$ is the set output by Algorithm~\ref{alg:correction-cc} applied with $\hat{\mathcal{I}}_k$ in~\eqref{eq:Ik-set-cc-optimist}, then, $\P{ \P{Y_{n+1} \in \hat{C}^{\mathrm{cc}}(X_{n+1}) \mid Y = k, \mathcal{D}} \geq 1 - \alpha} \geq 1-\gamma$ for any $k \in [K]$.
\end{proposition}


\clearpage

\section{Mathematical proofs} \label{app:proofs}

\subsection{Preliminaries}

\begin{proof}[Proof of Theorem~\ref{thm:coverage-lab-cond}]
  By definition of the conformity score function in~\eqref{eq:conf-scores}, $Y_{n+1} \in \hat{C}(X_{n+1})$ if and only if $\hat{s}(X_{n+1},Y_{n+1}) \leq \hat{\tau}_{Y_{n+1}}$.
  Therefore,
  \begin{align*}
    & \P{Y_{n+1} \in \hat{C}(X_{n+1}) \mid Y_{n+1} = k} \\
    & \qquad = \P{\tilde{Y}_{n+1} \in \hat{C}(X_{n+1}) \mid \tilde{Y}_{n+1} = k} + \\
    & \qquad \qquad \P{Y_{n+1} \in \hat{C}(X_{n+1}) \mid Y_{n+1} = k} - \P{\tilde{Y}_{n+1} \in \hat{C}(X_{n+1}) \mid \tilde{Y}_{n+1} = k} , \\
    & \qquad = \P{\tilde{Y}_{n+1} \in \hat{C}(X_{n+1}) \mid \tilde{Y}_{n+1} = k} \\
    & \qquad \qquad + \EV{ \P{Y_{n+1} \in \hat{C}(X_{n+1}) \mid Y_{n+1} = k, \mathcal{D}} - \P{\tilde{Y}_{n+1} \in \hat{C}(X_{n+1}) \mid \tilde{Y}_{n+1} = k, \mathcal{D}} }\\
    & \qquad = \P{\tilde{Y}_{n+1} \in \hat{C}(X_{n+1}) \mid \tilde{Y}_{n+1} = k} \\
    & \qquad \qquad + \EV{ \P{\hat{s}(X_{n+1}, k) \leq \hat{\tau}_{k} \mid Y_{n+1} = k, \mathcal{D}} - \P{\hat{s}(X_{n+1}, k) \leq \hat{\tau}_{k}  \mid \tilde{Y}_{n+1} = k, \mathcal{D}}  } \\
   & \qquad = \P{\tilde{Y}_{n+1} \in \hat{C}(X_{n+1}) \mid \tilde{Y}_{n+1} = k}  + \EV{ F_{k}^{k}(\hat{\tau}_{k}) - \tF_{k}^{k}(\hat{\tau}_{k})  } \\
   & \qquad = \P{\tilde{Y}_{n+1} \in \hat{C}(X_{n+1}) \mid \tilde{Y}_{n+1} = k}  + \EV{\Delta_k(\htau_k)}.
  \end{align*}
  Above, the notation $\EV{\cdot \mid \mathcal{D}}$ indicates the expected value conditional on the training and calibration data sets.
  From this, Equations~\eqref{eq:prop-label-cond-coverage-lower} and~\eqref{eq:prop-label-cond-coverage-upper} follow directly because
   \begin{align*}
     1-\alpha \leq
     \P{\tilde{Y}_{n+1} \in \hat{C}(X_{n+1}) \mid \tilde{Y}_{n+1} = k}
     \leq 1 - \alpha + \frac{1}{n_k+1},
   \end{align*}
  by Proposition~\ref{prop:standard-coverage-label}, which applies here since the data pairs $(X_i,\tilde{Y}_i)$ are i.i.d.~random samples.
 \end{proof}

\begin{proof}[Proof of Proposition~\ref{prop:indep-linear}]
By definition of $\tilde{P}_k$ and $P_l$,
\begin{align*}
  \tilde{P}_k
  & = \P{X \mid \tilde{Y}=k}
    = \sum_{l=1}^{K} \P{X, Y=l \mid \tilde{Y}=k} \\
  & = \sum_{l=1}^{K} \frac{\P{X, \tilde{Y}=k \mid Y=l} \cdot \P{Y=l}}{\P{\tilde{Y}=k}}  \\
  & = \sum_{l=1}^{K} \frac{\P{X \mid Y=l} \cdot \P{\tilde{Y}=k \mid Y=l} \cdot \P{Y=l}}{\P{\tilde{Y}=k}} 
  = \sum_{l=1}^{K} P_l \cdot M_{kl}.
\end{align*}
Above, the fourth equality follows directly from Assumption~\ref{assumption:linear-contam}.
\end{proof}

\begin{proof}[Proof of Corollary~\ref{cor:coverage-cond}]
  By Theorem~\ref{thm:coverage-lab-cond}, it suffices to prove that $\EV{\Delta_k(t)} \geq 0$ for all $t \in \mathbb{R}$ and $k \in [K]$.
  To establish that, note that combining~\eqref{eq:delta} with Proposition~\ref{prop:indep-linear} gives:
  \begin{align*}
    \Delta_k(t)
    & = F_k^{k}(t) - \tF_k^{k}(t)
      = (1-M_{kk}) F_k^{k} (t) - \sum_{j\neq k} M_{kj} F_j^{k} (t) \\
    & \geq (1-M_{kk}) F_k^{k} (t) - \bigg(\sum_{j\neq k} M_{kj}\bigg) \cdot \max_{j\neq k} F_j^{k} (t).
  \end{align*}
  Further, by~\eqref{eq:assump_scores-cond},
  \begin{align*}
    \Delta_k(t)
    & \geq (1-M_{kk}) F_k^{k} (t) - \bigg(\sum_{j\neq k} M_{kj}\bigg) \cdot F_k^{k} (t)
     = \left( 1 - \sum_{j=1}^{K} M_{kj} \right) \cdot F_k^{k} (t)
      = 0,
  \end{align*}
where the last equality follows from the fact that $\sum_{j=1}^{K} M_{kj}=1$ for all $k \in [K]$.
This implies that $\EV{\Delta_k(\htau_k)} \geq 0$, completing the proof.
\end{proof}

\subsection{Adaptive coverage under a known label contamination model}

\begin{proof}[Proof of Theorem~\ref{thm:algorithm-correction}]
  Suppose $Y_{n+1} = k$, for some $k \in [K]$.
  By definition of the conformity score function in~\eqref{eq:conf-scores}, the event $k \notin \hat{C}(X_{n+1})$ occurs if and only if $\hat{s}(X_{n+1},k) > \hat{\tau}_k$.
  We will assume without loss of generality that $\hat{\mathcal{I}_k} \neq \emptyset$ and $\hat{i}_k = \min\{i \in \hat{\mathcal{I}_k}\}$; otherwise, $\hat{\tau}_k=1$ and the result trivially holds.
As in the proof of Theorem~\ref{thm:coverage-lab-cond}, the probability of miscoverage conditional on $Y_{n+1}=k$ and on the labeled data in $\mathcal{D}$ can be decomposed as:
  \begin{align*}
    & \P{Y \notin C(X, S^k_{(\hat{i}_k)}) \mid Y = k, \mathcal{D}} \\
    & \qquad = \P{ \tilde{Y} \notin C(X, S^k_{(\hat{i}_k)}) \mid \tilde{Y} = k, \mathcal{D}} \\
    & \qquad \quad + \P{Y \notin C(X, S^k_{(\hat{i}_k)}) \mid Y = k, \mathcal{D}} - \P{\tilde{Y} \notin C(X, S^k_{(\hat{i}_k)}) \mid \tilde{Y} = k, \mathcal{D}}\\
    & \qquad = 1 - \tilde{F}^{k}_k(S^k_{(\hat{i}_k)}) \\
    & \qquad \quad + \P{ \hat{s}(X_{n+1},k) > S^k_{(\hat{i}_k)} \mid Y=k, \mathcal{D} } - \P{ \hat{s}(X_{n+1},k) > S^k_{(\hat{i}_k)} \mid \tilde{Y}=k, \mathcal{D} } \\
    & \qquad = 1 - \tilde{F}^{k}_k(S^k_{(\hat{i}_k)}) + \left[ 1-F_k^k(S^k_{(\hat{i}_k)}) \right] - \left[ 1-\tilde{F}_k^k(S^k_{(\hat{i}_k)}) \right]\\
     & \qquad = 1 - \tilde{F}^{k}_k(S^k_{(\hat{i}_k)}) - \Delta_k(S^k_{(\hat{i}_k)}) \\
     & \qquad = 1 - \hat{F}^{k}_k(S^k_{(\hat{i}_k)}) + \hat{F}^{k}_k(S^k_{(\hat{i}_k)}) - \tilde{F}^{k}_k(S^k_{(\hat{i}_k)}) - \Delta_k(S^k_{(\hat{i}_k)}) \\
    & \qquad = \left[ 1 - \frac{\hat{i}_k}{n_k} - \Delta_k(S^k_{(\hat{i}_k)}) \right] + \hat{F}^{k}_k(S^k_{(\hat{i}_k)}) - \tilde{F}^{k}_k(S^k_{(\hat{i}_k)}) \\
    & \qquad = \left[ 1 - \frac{\hat{i}_k}{n_k} - \hat{\Delta}_k(S^k_{(\hat{i}_k)}) + \delta(n_k,n_{*}) \right] - \delta(n_k,n_{*})  + \\
    & \qquad \qquad + \hat{\Delta}_k(S^k_{(\hat{i}_k)}) - \Delta_k(S^k_{(\hat{i}_k)}) + \hat{F}^{k}_k(S^k_{(\hat{i}_k)}) - \tilde{F}^{k}_k(S^k_{(\hat{i}_k)}) \\
    & \qquad \leq \sup_{i \in \hat{\mathcal{I}}_k} \left\{ 1 - \frac{i}{n_k} - \hat{\Delta}_k(S^k_{(i)}) + \delta(n_k,n_{*}) \right\} - \delta(n_k,n_{*}) \\
    & \qquad \qquad + \sup_{t \in \mathbb{R}} [\hat{\Delta}_k(t) - \Delta_k(t) ] + \hat{F}^{k}_k(S^k_{(\hat{i}_k)}) - \tilde{F}^{k}_k(S^k_{(\hat{i}_k)}).
  \end{align*}
By definition of $\hat{\mathcal{I}}_k$, for all $i \in \hat{\mathcal{I}}_k$,
  \begin{align*}
    1 - \frac{i}{n_k} - \hat{\Delta}_k(S^k_{(i)}) + \delta(n_k,n_{*}) \leq \alpha,
  \end{align*}
which implies a.s.
  \begin{align*}
   \sup_{i \in \hat{\mathcal{I}}_k} \left\{ 1 - \frac{i}{n_k} - \hat{\Delta}_k(S^k_{(i)}) + \delta(n_k,n_{*}) \right\} \leq \alpha.
  \end{align*}
Therefore,
  \begin{align*}
    & \P{Y \notin C(X, S^k_{(\hat{i}_k)}) \mid Y = k} \\
    & \qquad \leq \alpha  + \EV{\sup_{t \in \mathbb{R}} [\hat{\Delta}_k(t) - \Delta_k(t)] } + \EV{ \hat{F}^{k}_k(S^k_{(\hat{i}_k)}) - \tilde{F}^{k}_k(S^k_{(\hat{i}_k)})  } - \delta(n_k,n_{*}).
  \end{align*}
The second term on the right-hand-side above can be bound using the DKW inequality, as made precise by Lemma~\ref{lemma:dkw-delta-expected}.
This leads to:
  \begin{align*}
    & \P{Y \notin C(X, S^k_{(\hat{i}_k)}) \mid Y = k} \\
    & \qquad \leq \alpha + \frac{2 \sum_{l\neq k}|V_{kl}|}{\sqrt{n_{*}}} \min \left\{ K \sqrt{\frac{\pi}{2}} , \frac{1}{\sqrt{n_{*}}} + \sqrt{\frac{\log(2K) + \log(n_{*})}{2}} \right\} \\
    & \quad \qquad + \EV{ \hat{F}^{k}_k(S^k_{(\hat{i}_k)}) - \tilde{F}^{k}_k(S^k_{(\hat{i}_k)})  } - \delta(n_k,n_{*}).
  \end{align*}
In order to bound the last expected value, let $U_1, \ldots, U_{n_k}$ be i.i.d.~uniform random variables on $[0,1]$, and denote their order statistics as $U_{(1)}, \ldots, U_{(n_k)}$. This implies that
\begin{align*}
  \hat{F}^{k}_k(S^k_{(\hat{i}_k)}) - \tilde{F}^{k}_k(S^k_{(\hat{i}_k)})
  \overset{d}{=} \frac{\hat{i}_k}{n_k} - U_{(\hat{i}_k)}.
\end{align*}
Therefore, it follows from~\eqref{eq:define-cn} that
\begin{align}
  \EV{ \hat{F}^{k}_k(S^k_{(\hat{i}_k)}) - \tilde{F}^{k}_k(S^k_{(\hat{i}_k)})  }
  \leq \EV{ \sup_{i \in [n_k]} \left\{ \frac{i}{n_k} - U_{(i)} \right\} } = c(n_k).
  \label{eq:tighter_cdf_bound}
\end{align}
We can thus conclude that:
  \begin{align*}
    & \P{Y \notin C(X, S^k_{(\hat{i}_k)}) \mid Y = k} \\
    & \qquad \leq \alpha + \frac{2 \sum_{l\neq k}|V_{kl}|}{\sqrt{n_{*}}} \min \left\{ K \sqrt{\frac{\pi}{2}} , \frac{1}{\sqrt{n_{*}}} + \sqrt{\frac{\log(2K) + \log(n_{*})}{2}} \right\} + c(n_k) - \delta(n_k,n_{*}) \\
    & \qquad = \alpha.
  \end{align*}

\end{proof}

\begin{lemma}\label{lemma:dkw-delta-expected}
Under the assumptions of Theorem~\ref{thm:algorithm-correction}, for any $k \in [K]$,
\begin{align*}
  & \EV{ \sup_{t \in \mathbb{R}} | \hat{\Delta}_k(t) - \Delta_k(t)  |}
    \leq \frac{2\sum_{\l\neq k}|V_{kl}|}{\sqrt{n_{*}}} \min \left\{ K \sqrt{\frac{\pi}{2}}, \frac{1}{\sqrt{n_{*}}} + \sqrt{\frac{\log(2K) + \log(n_{*})}{2}} \right\}.
\end{align*}
\end{lemma}

\begin{proof}[Proof of Lemma~\ref{lemma:dkw-delta-expected}]
Note that, for any $k \in [K]$ and $t \in \mathbb{R}$,
\begin{align*}
	\hat{\Delta}_k(t) - \Delta_k(t)
	& = (V_{kk}-1) \left[ \hat{F}_k^{k}(t) - \tF^{k}_k(t)  \right] + \sum_{l\neq k} V_{kl} \left[ \hat{F}_l^{k}(t) - \tF^{k}_l(t)  \right]    \\
	& = \sum_{l\neq k} V_{kl} \left\{ \left[ \hat{F}_l^{k}(t) - \tF^{k}_l(t)  \right] - \left[ \hat{F}_k^{k}(t) - \tF^{k}_k(t)  \right] \right\}	\\
	& \leq 2 \left( \sum_{\l\neq k}|V_{kl}| \right) \max_{l \in [K]} | \hat{F}_l^{k}(t)   - \tF^{k}_{l}(t) |,
\end{align*}
where in the second-to-last line we used the fact that $V$ has row sums equal to 1 because it is the inverse of $M$, which has row sums equal to 1.
Therefore,
\begin{align*}
  \EV{ \sup_{t \in \mathbb{R}} | \hat{\Delta}_k(t) - \Delta_k(t)  | }
  \leq 2\sum_{\l\neq k}|V_{kl}| \cdot \EV{ \max_{l \in [K]} \sup_{t \in \mathbb{R}} | \hat{F}_l^{k}(t) - \tF^{k}_{l}(t)  |}.
\end{align*}

Define $l_{\min} := \arg\min_{l \in [K]} n_j$ and $n_{*} := \min_{l \in [K]} n_l$.
Then, combining the DKW inequality with a union bound gives that, for any $\eta > 0$,
\begin{align*}
  \P{ \max_{l \in [K]} \sup_{t \in \mathbb{R}} | \hat{F}_l^{k}(t) - \tF^{k}_{l}(t)    | > \eta }
  & \leq K \cdot \P{ \sup_{t \in \mathbb{R}} | \hat{F}_{l_{\min}}^{k}(t) - \tF^{k}_{l_{\min}}(t)    | > \eta } \\
  & \leq 2 K e^{- 2 n_{*} \eta^2}.
\end{align*}
This implies that
\begin{align*}
  \EV{ \max_{l \in [K]} \sup_{t \in \mathbb{R}} | \hat{F}_l^{k}(t) - \tF^{k}_{l}(t)   | }
  & \leq \int_{0}^{\infty} \P{ \max_{l \in [K]} \sup_{t \in \mathbb{R}} | \hat{F}_l^{k}(t) - \tF^{k}_{l}(t)   | > \eta } d\eta \\
  & \leq 2 K \int_{0}^{\infty} e^{- 2 n_{*} \eta^2} d\eta 
    = K\sqrt{\frac{\pi}{2 n_{*}}}.
\end{align*}

Similarly, the DKW inequality also implies that, for any $\eta > 0$,
\begin{align*}
  & \P{ \sup_{t \in \mathbb{R}} |\hat{\Delta}_k(t) - \Delta_k(t) | > 2\sum_{\l\neq k}|V_{kl}| \sqrt{\frac{\log(2K) + \log(1/\eta)}{2 n_{*}}} }
   \leq \eta;
\end{align*}
see Lemma~\ref{lemma:dkw-delta-prob}.
Therefore, setting $\eta = 1/n_{*}$, we obtain
\begin{align*}
  \EV{ \sup_{t \in \mathbb{R}} |\hat{\Delta}_k(t) - \Delta_k(t) | }
  & \leq 2\sum_{\l\neq k}|V_{kl}| \left[ \sqrt{\frac{\log(2K) + \log(n_{*})}{2 n_{*}}} + \frac{1}{n_{*}} \right].
\end{align*}
\end{proof}

\begin{lemma}\label{lemma:dkw-delta-prob}
Under the assumptions of Theorem~\ref{thm:algorithm-correction}, for any $k \in [K]$ and $\eta > 0$,
\begin{align*}
  \P{ \sup_{t \in \mathbb{R}} |\hat{\Delta}_k(t) - \Delta_k(t) | > 2\sum_{\l\neq k}|V_{kl}| \sqrt{\frac{\log(2K) + \log(1/\eta)}{2 n_{*}}} }
  & \leq \eta.
\end{align*}
\end{lemma}

\begin{proof}[Proof of Lemma~\ref{lemma:dkw-delta-prob}]
It follows from the definitions of $\hat{\Delta}_k(t)$ and $\Delta_k(t)$, and from the DKW inequality, that, for any $\eta > 0$,
\begin{align*}
  & \P{ \sup_{t \in \mathbb{R}} |\hat{\Delta}_k(t) - \Delta_k(t) | > 2\sum_{\l\neq k}|V_{kl}|  \sqrt{\frac{\log(2K) + \log(1/\eta)}{2 n_{*}}} } \\
  & \qquad \leq \P{ \max_{j \in [K]} | \hat{F}_j^{k}(t) - \tF^{k}_{j}(t)| > \sqrt{\frac{\log(2K) + \log(1/\eta)}{2 n_{*}}} } \\
  & \qquad \leq K \cdot \P{ | \hat{F}_{l_{\min}}^{k}(t) - \tF^{k}_{l_{\min}}(t)  | > \sqrt{\frac{\log(2K) + \log(1/\eta)}{2 n_{*}}} } \\
  & \qquad \leq 2K \cdot \exp\left[ - 2 n_{*} \frac{\log(2K) + \log(1/\eta)}{2 n_{*}} \right] \\
  & \qquad = 2K \cdot \exp\left[ - \log(2K) - \log(1/\eta) \right] 
    = \eta.
\end{align*}

\end{proof}

\begin{proof}[Proof of Theorem~\ref{thm:algorithm-correction-upper}]

Define the events $\cA_1=\{\hat{\cI}_k=\emptyset \}$ and $\cA_2=\{\hat{i}_k=1\}$.
Then,
\begin{align}
  \begin{split} 	\label{eq:upper_bound_events}
	& \P{Y_{n+1} \in \hat{C}(X_{n+1}) \mid Y = k}  \\
	& \qquad \leq \P{\cA_1} + \EV{\P{Y_{n+1} \in \hat{C}(X_{n+1}) \mid Y = k, \cD}\I{\cA_2}} \\
	& \qquad \qquad + \EV{\P{Y_{n+1} \in \hat{C}(X_{n+1}) \mid Y = k, \cD}\I{\cA_1^c\cap \cA_2^c}}.
      \end{split}
\end{align}

We will now separately bound the three terms on the right-hand-side of~\eqref{eq:upper_bound_events}.
The following notation will be useful for this purpose.
For all $i \in [n_k]$, let $U_i \sim \text{Uniform}(0,1)$ be independent and identically distributed uniform random variables, and
denote their order statistics as $U_{(1)} < U_{(2)} < \ldots < U_{(n_k)}$.

\begin{itemize}
\item The probability of the event $\cA_1$ can be bound from above as:
  \begin{align} \label{eq:upper_bound_event1}
    \P{\hat{\mathcal{I}}_k = \emptyset} \leq \frac{1}{n_{*}}.
  \end{align}

  To simplify the notation in the proof of~\eqref{eq:upper_bound_event1}, define
\begin{align*}
  d_k := c(n_k) + \frac{2 \sum_{l\neq k}|V_{kl}|}{\sqrt{n_{*}}} \left( \frac{1}{\sqrt{n_{*}}} + 2 \sqrt{\frac{\log(2K) + \log(n_{*})}{2}} \right).
\end{align*}
Then, under Assumption~\ref{assumption:regularity-dist-delta},
  \begin{align*}
    & 1 - \left( \alpha + \hDelta_k(S^k_{(i)}) - \delta(n_k,n_{*}) \right) \\
    & \qquad = 1 - \alpha - \Delta_k(S^k_{(i)}) + \delta(n_k,n_{*}) + \Delta_k(S^k_{(i)}) - \hDelta_k(S^k_{(i)}) \\
    & \qquad \leq 1 - d_k + \delta(n_k,n_{*}) + \sup_{t \in [0,1]} | \Delta_k(t) - \hDelta_k(t)|.
  \end{align*}
Further, it follows from Lemma~\ref{lemma:dkw-delta-prob} and~\eqref{eq:delta-constant} that, with probability at least $1-1/n_{*}$,
  \begin{align*}
    & 1 - \left( \alpha + \hDelta_k(S^k_{(i)}) - \delta(n_k,n_{*}) \right) \\
    & \qquad \leq 1 - d_k + \delta(n_k,n_{*}) + 2 \sum_{l\neq k}|V_{kl}| \sqrt{\frac{\log(2K) + \log(n_{*})}{2 n_{*}}} \\
    & \qquad \leq 1 - d_k + c(n_k) + \frac{2 \sum_{l\neq k}|V_{kl}|}{\sqrt{n_{*}}} \left( \frac{1}{\sqrt{n_{*}}} + 2 \sqrt{\frac{\log(2K) + \log(n_{*})}{2}} \right)
      = 1.
  \end{align*}

This implies that the set $\hat{\mathcal{I}}_k$ defined in~\eqref{eq:Ik-set} is non-empty with probability at least $1-1/n_{*}$, because $n_k \in \hat{\mathcal{I}}_k$ if $1 - [ \alpha + \hDelta_k(S^k_{(i)}) - \delta(n_k,n_{*}) ] \leq 1$.

\item Under $\cA_2$, the second term on the right-hand-side of~\eqref{eq:upper_bound_events} can be written as:
  \begin{align*}
    \P{Y_{n+1} \in \hat{C}(X_{n+1}) \mid Y = k, \cD} & = \P{\hat{s}(X_{n+1}, k) \leq S^k_{(1)} \mid Y = k, \cD} = F^{k}_k(S^k_{(1)}).
\end{align*}
Therefore,
\begin{align*}
  \begin{split}
    & \EV{\P{Y_{n+1} \in \hat{C}(X_{n+1}) \mid Y = k, \cD}\I{\cA_2}} \\
    & \qquad \leq   \EV{F^{k}_k(S^k_{(1)})}  \\
    & \qquad = V_{kk} \EV{\tilde{F}^{k}_k(S^k_{(1)})} + \sum_{l \neq k} V_{kl} \EV{\tilde{F}^{k}_l(S^k_{(1)})}\\
    & \qquad = V_{kk} \EV{U_{(1)}} + \sum_{l \neq k} V_{kl} \EV{\tilde{F}^{k}_l(S^k_{(1)})}\\
    & \qquad = \frac{V_{kk}}{n_k+1} + \sum_{l \neq k} V_{kl} \EV{\tilde{F}^{k}_l(S^k_{(1)})}.
  \end{split}
\end{align*}
Next, using Assumption~\ref{assumption:consitency-scores}, we can write:
\begin{align}     \label{eq:upper_bound_event2}
  \begin{split}
    \EV{\P{Y_{n+1} \in \hat{C}(X_{n+1}) \mid Y = k, \cD}\I{\cA_2}} 
    & \leq \frac{V_{kk}}{n_k+1} + \EV{\tilde{F}^{k}_k(S^k_{(1)})} \sum_{l \neq k} |V_{kl}|  \\
    & = \frac{V_{kk}}{n_k+1} + \EV{U_{(1)} \sum_{l \neq k} |V_{kl}|}  \\
      & = \frac{V_{kk} + \sum_{l \neq k} |V_{kl}|}{n_k+1}.
  \end{split}
\end{align}

\item Under $\cA_1^c\cap \cA_2^c$, by definition of $\hat{\mathcal{I}}_k$, for any $i \leq \hat{i}_k - 1$,
  \begin{align*}
    \frac{i}{n_k} < 1 - \left( \alpha + \hDelta_k(S^k_{(i)}) - \delta(n_k,n_{*}) \right).
  \end{align*}
  Therefore, choosing $i = \hat{i}_k-1$, we get:
  \begin{align*}
    \frac{\hat{i}_k}{n_k} < 1 - \left( \alpha + \hDelta_k(S^k_{(\hat{i}_k-1)}) - \delta(n_k,n_{*}) \right) + \frac{1}{n_k}.
  \end{align*}
As in the proof of Theorem~\ref{thm:algorithm-correction}, the probability of coverage conditional on $Y_{n+1}=k$ and on the labeled data in $\mathcal{D}$ can be decomposed as:
  \begin{align*}
    & \P{Y \in C(X, S^k_{(\hat{i}_k)}) \mid Y = k, \mathcal{D}} \\
    & \quad = \P{ \tilde{Y} \in C(X, S^k_{(\hat{i}_k)}) \mid \tilde{Y} = k, \mathcal{D}} \\
    & \quad \quad + \P{Y \in C(X, S^k_{(\hat{i}_k)}) \mid Y = k, \mathcal{D}} - \P{\tilde{Y} \in C(X, S^k_{(\hat{i}_k)}) \mid \tilde{Y} = k, \mathcal{D}}\\
    & \quad = \tilde{F}^{k}_k(S^k_{(\hat{i}_k)}) + \P{ \hat{s}(X_{n+1},k) \leq S^k_{(\hat{i}_k)} \mid Y=k, \mathcal{D} } - \P{ \hat{s}(X_{n+1},k) \leq S^k_{(\hat{i}_k)} \mid \tilde{Y}=k, \mathcal{D} } \\
    & \quad = \tilde{F}^{k}_k(S^k_{(\hat{i}_k)}) + F_k^k(S^k_{(\hat{i}_k)})  - \tilde{F}_k^k(S^k_{(\hat{i}_k)}) 
      = \tilde{F}^{k}_k(S^k_{(\hat{i}_k)}) + \Delta_k(S^k_{(\hat{i}_k)}) \\
    & \quad = \hat{F}^{k}_k(S^k_{(\hat{i}_k)}) + \tilde{F}^{k}_k(S^k_{(\hat{i}_k)}) - \hat{F}^{k}_k(S^k_{(\hat{i}_k)})+ \Delta_k(S^k_{(\hat{i}_k)}) \\
    & \quad = \frac{\hat{i}_k}{n_k} + \tilde{F}^{k}_k(S^k_{(\hat{i}_k)}) - \hat{F}^{k}_k(S^k_{(\hat{i}_k)})+ \Delta_k(S^k_{(\hat{i}_k)}) \\
    & \quad < 1 - \left( \alpha + \hDelta_k(S^k_{(\hat{i}_k-1)}) - \delta(n_k,n_{*}) \right) + \frac{1}{n_k} + \tilde{F}^{k}_k(S^k_{(\hat{i}_k)}) - \hat{F}^{k}_k(S^k_{(\hat{i}_k)})+ \Delta_k(S^k_{(\hat{i}_k)}) \\
    & \quad = 1 - \alpha + \delta(n_k,n_{*}) + \frac{1}{n_k} + \tilde{F}^{k}_k(S^k_{(\hat{i}_k)}) - \hat{F}^{k}_k(S^k_{(\hat{i}_k)})+ \Delta_k(S^k_{(\hat{i}_k)}) - \hDelta_k(S^k_{(\hat{i}_k-1)}).
  \end{align*}

  To bound the last term above, note that
  \begin{align*}
     \Delta_k(S^k_{(\hat{i}_k)}) - \hDelta_k(S^k_{(\hat{i}_k-1)})
    & = (\Delta_k(S^k_{(\hat{i}_k)}) - \Delta_k(S^k_{(\hat{i}_k-1)})) + (\Delta_k(S^k_{(\hat{i}_k-1)}) - \hDelta_k(S^k_{(\hat{i}_k-1)})) \\
    & \leq \sup_{t \in \mathbb{R}} |\Delta_k(t) - \hDelta_k(t)| + (\Delta_k(S^k_{(\hat{i}_k)}) - \Delta_k(S^k_{(\hat{i}_k-1)})),
  \end{align*}
  where the expected value of the first term above can be bounded using Lemma~\ref{lemma:dkw-delta-expected}, and the second term is given by
  \begin{align*}
    & \Delta_k(S^k_{(\hat{i}_k)}) - \Delta_k(S^k_{(\hat{i}_k-1)}) \\
    & \qquad = (V_{kk}-1)  \left[ \tF^{k}_k(S^k_{(\hat{i}_k)}) -  \tF^{k}_k(S^k_{(\hat{i}_k-1)}) \right] + \sum_{l\neq k}  V_{kl} \left[ \tF^{k}_{l}(S^k_{(\hat{i}_k)}) - \tF^{k}_{l}(S^k_{(\hat{i}_k-1)}) \right]	\\
    & \qquad = \sum_{l\neq k} V_{kl} \left[ \left( \tF^{k}_{l}(S^k_{(\hat{i}_k)}) - \tF^{k}_{l}(S^k_{(\hat{i}_k-1)}) \right) - \left(\tF^{k}_k(S^k_{(\hat{i}_k)}) - \tF^{k}_k(S^k_{(\hat{i}_k-1)}) \right) \right] \\
    & \qquad \leq  2\sum_{l\neq k}|V_{kl}| \cdot  \max_{2\leq i\leq n_k}\max_{l\in[K]}\left| \tF^{k}_{l}(S^k_{(i)}) - \tF^{k}_{l}(S^k_{(i-1)}) \right|.
  \end{align*}
  Combining the above, we find that, under $\cA_1^c\cap \cA_2^c$,
  \begin{align*}
    & \P{Y \in C(X, S^k_{(\hat{i}_k)}) \mid Y = k, \mathcal{D}}	\\
    & \qquad \leq 1 - \alpha + \delta(n_k,n_{*}) + \frac{1}{n_k} + \tilde{F}^{k}_k(S^k_{(\hat{i}_k)}) - \hat{F}^{k}_k(S^k_{(\hat{i}_k)}) +  \sup_{t \in \mathbb{R}} |\Delta_k(t) - \hDelta_k(t)| \\
    & \qquad \qquad + 2\sum_{l\neq k}|V_{kl}|  \cdot \max_{2\leq i\leq n_k}\max_{l\in[K]}\left| \tF^{k}_{l}(S^k_{(i)}) - \tF^{k}_{l}(S^k_{(i-1)}) \right|.
  \end{align*}

  It remains to bound the expectation of the last term. By Assumption \ref{assumption:regularity-dist},
  \begin{align*}
    \max_{2\leq i\leq n_k}\max_{l\in[K]}\left| \tF^{k}_{l}(S^k_{(i)}) - \tF^{k}_{l}(S^k_{(i-1)}) \right| & \leq f_{\max} \max_{2\leq i\leq n_k}| S^k_{(i)} - S^k_{(i-1)}|	\\
                                                                                                         & \leq \frac{f_{\max}}{f_{\min}} \max_{2\leq i\leq n_k} |\tF^{k}_{k}(S^k_{(i)}) - \tF^{k}_{k}(S^k_{(i-1)})|	\\
    \qquad \qquad & \stackrel{d}{=} \frac{f_{\max}}{f_{\min}} \max_{2\leq i\leq n_k} (U_{(i)} - U_{(i-1)}) \\
                                                                                                         &\leq \frac{f_{\max}}{f_{\min}}  \max_{1 \leq i \leq n_k +1 } D_i,
  \end{align*}
  where $D_1=U_{(1)}$, $D_i=U_{(i)} - U_{(i-1)}$ for $i=2, \dots, n_k$, and $D_{n_k+1} = 1-U_{(n_k)}$.
  By a standard result on maximum uniform spacing,
  \begin{align*}
    \EV{\max_{1\leq i \leq n_k+1} D_i} = \frac{1}{n_k+1}\sum_{j=1}^{n_k+1} \frac{1}{j}.
  \end{align*}

  Therefore, under $\cA_1^c\cap \cA_2^c$,
\begin{align*}
	& \EV{\P{Y_{n+1} \in \hat{C}(X_{n+1}) \mid Y = k, \cD}\I{\cA_1^c\cap \cA_2^c}} \\
	& \qquad \leq  1 - \alpha + \delta(n_k,n_{*}) + \frac{1}{n_k} + \EV{ \tilde{F}^{k}_k(S^k_{(\hat{i}_k)}) - \hat{F}^{k}_k(S^k_{(\hat{i}_k)}) } + \EV{\sup_{t \in \mathbb{R}} |\Delta_k(t) - \hDelta_k(t)| } \\
 & \qquad \qquad + 2\sum_{l\neq k}|V_{kl}|  \cdot \EV{\max_{2\leq i\leq n_k}\max_{l\in[K]}\left| \tF^{k}_{l}(S^k_{(i)}) - \tF^{k}_{l}(S^k_{(i-1)}) \right|}.
\end{align*}
Thus, applying Lemma~\ref{lemma:dkw-delta-expected} and the bound in~\eqref{eq:tighter_cdf_bound}, it follows that, under $\cA_1^c\cap \cA_2^c$,
\begin{align} \label{eq:upper_bound_event3}
\begin{split}
  & \EV{\P{Y_{n+1} \in \hat{C}(X_{n+1}) \mid Y = k, \cD}\I{\cA_1^c\cap \cA_2^c}} \\
  & \qquad \leq  1 - \alpha + \delta(n_k,n_{*}) + \frac{1}{n_k} + \EV{ \max_{i \in [n_k]} \left( U_{i} - \frac{i}{n_k} \right) } + \EV{\sup_{t \in \mathbb{R}} |\Delta_k(t) - \hDelta_k(t)| } \\
	& \qquad \qquad + 2\sum_{l\neq k}|V_{kl}|  \cdot \frac{f_{\max}}{f_{\min}} \cdot \frac{1}{n_k+1}\sum_{j=1}^{n_k+1} \frac{1}{j}	\\
	& \qquad =  1 - \alpha + \delta(n_k,n_{*}) + \frac{1}{n_k} + c(n_k) + \EV{\sup_{t \in \mathbb{R}} |\Delta_k(t) - \hDelta_k(t)| } \\
	& \qquad \qquad + 2\sum_{l\neq k}|V_{kl}|  \cdot \frac{f_{\max}}{f_{\min}} \cdot \frac{1}{n_k+1}\sum_{j=1}^{n_k+1} \frac{1}{j}	\\
	& \qquad   \leq  1 - \alpha + \delta(n_k,n_{*}) + \frac{1}{n_k} + c(n_k) \\
	& \qquad \qquad +  \frac{2\sum_{l\neq k}|V_{kl}| }{\sqrt{n_{*}}} \min \left\{ K \sqrt{\frac{\pi}{2}} , \frac{1}{\sqrt{n_{*}}} + \sqrt{\frac{\log(2K) + \log(n_{*})}{2}} \right\} \\
	& \qquad \qquad \qquad + 2\sum_{l\neq k}|V_{kl}| \cdot \frac{f_{\max}}{f_{\min}} \cdot \frac{1}{n_k+1}\sum_{j=1}^{n_k+1} \frac{1}{j} \\
	& \qquad   =  1 - \alpha + 2 \delta(n_k,n_{*}) + \frac{1}{n_k} + 2\sum_{l\neq k}|V_{kl}| \cdot \frac{f_{\max}}{f_{\min}} \cdot \frac{1}{n_k+1}\sum_{j=1}^{n_k+1} \frac{1}{j}.
      \end{split}
\end{align}

\end{itemize}

Finally, combining~\eqref{eq:upper_bound_events} with \eqref{eq:upper_bound_event1}, \eqref{eq:upper_bound_event2}, and \eqref{eq:upper_bound_event3} leads to the desired result:
\begin{align*}
  & \P{Y_{n+1} \in \hat{C}(X_{n+1}) \mid Y = k}  \\
  & \qquad \leq \P{\cA_1} + \EV{\P{Y_{n+1} \in \hat{C}(X_{n+1}) \mid Y = k, \cD}\I{\cA_2}} \\
  & \qquad \qquad + \EV{\P{Y_{n+1} \in \hat{C}(X_{n+1}) \mid Y = k, \cD}\I{\cA_1^c\cap \cA_2^c}} \\
   & \qquad \leq 1 - \alpha + \frac{1}{n_{*}} + \frac{V_{kk} + \sum_{l \neq k} |V_{kl}|}{n_k+1} \\
  & \qquad \qquad + 2 \delta(n_k,n_{*}) + \frac{1}{n_k} + 2\sum_{l\neq k} |V_{kl}| \cdot \frac{f_{\max}}{f_{\min}} \cdot \frac{1}{n_k+1}\sum_{j=1}^{n_k+1} \frac{1}{j} \\
   & \qquad \leq 1 - \alpha + \frac{1}{n_{*}} + 2 \delta(n_k,n_{*}) + \frac{1}{n_k} \left[1+ 2 \sum_{l\neq k} |V_{kl}| \cdot \frac{f_{\max}}{f_{\min}} \sum_{j=1}^{n_k+1} \frac{1}{j} \right] \\
  & \qquad \qquad + \frac{V_{kk} + \sum_{l \neq k} |V_{kl}|}{n_k+1}.
\end{align*}

\end{proof}

\begin{proof}[Proof of Proposition~\ref{thm:algorithm-correction-optimistic}]
  Suppose $Y_{n+1} = k$, for some $k \in [K]$. Proceeding exactly as in the proof of Theorem~\ref{thm:algorithm-correction}, we obtain:
\begin{align*}
	& \P{Y \notin C(X, S^k_{(\hat{i}_k)}) \mid Y = k, \mathcal{D}}  \\
	& \qquad = 1 - \tilde{F}^{k}_k(S^k_{(\hat{i}_k)}) - \Delta_k(S^k_{(\hat{i}_k)}) 	\\
	& \qquad = 1- \hat{F}^{k}_k(S^k_{(\hat{i}_k)}) - \max\{\hDelta_k(S^k_{(\hat{i}_k)}) - \delta(n_k,n_{*}), -(1-\alpha)/n_k \}	\\
	& \qquad\qquad + \hat{F}^{k}_k(S^k_{(\hat{i}_k)}) - \tilde{F}^{k}_k(S^k_{(\hat{i}_k)}) \\
	& \qquad\qquad +  \max\{\hDelta_k(S^k_{(\hat{i}_k)}) - \delta(n_k,n_{*}), -(1-\alpha)/n_k \} - (\Delta_k(S^k_{(\hat{i}_k)}) - \delta(n_k,n_{*})) - \delta(n_k,n_{*}) 	\\
	& \qquad = \left[ 1- \frac{\hat{i}_k}{n_k} - \max\{\hDelta_k(S^k_{(i)}) - \delta(n_k,n_{*}), -(1-\alpha)/n_k \} \right]	\\
	& \qquad\qquad + \hat{F}^{k}_k(S^k_{(\hat{i}_k)}) - \tilde{F}^{k}_k(S^k_{(\hat{i}_k)}) \\
	& \qquad\qquad +  \max\{\hDelta_k(S^k_{(\hat{i}_k)}) - \delta(n_k,n_{*}) + (1-\alpha)/n_k, 0 \} \\
        & \qquad \qquad - \left( \Delta_k(S^k_{(\hat{i}_k)}) - \delta(n_k,n_{*}) + (1-\alpha)/n_k \right) - \delta(n_k,n_{*}) 	\\
	& \qquad \leq \left[ 1- \frac{\hat{i}_k}{n_k} - \max\{\hDelta_k(S^k_{(i)}) - \delta(n_k,n_{*}), -(1-\alpha)/n_k \} \right]	\\
	& \qquad\qquad + \hat{F}^{k}_k(S^k_{(\hat{i}_k)}) - \tilde{F}^{k}_k(S^k_{(\hat{i}_k)}) + \sup_{t \in \mathbb{R}} |\hat{\Delta}_k(t) - \Delta_k(t)| - \delta(n_k,n_{*}),
\end{align*}
using the fact that $\inf_{t\in \R}\Delta_k(t) \geq \delta(n_k, n_{*}) - (1-\alpha)/n_k$ implies
\begin{align*}
 & |\max\{\hDelta_k(t) - \delta(n_k,n_{*}) + (1-\alpha)/n_k, 0 \} - \left(\Delta_k(t) - \delta(n_k,n_{*}) + (1-\alpha)/n_k \right)| \\
   & \qquad \leq |\hat{\Delta}_k(t) - \Delta_k(t)|,
\end{align*}
for all $t \in \mathbb{R}$.
The proof is then completed by proceeding as in the proof of Theorem~\ref{thm:algorithm-correction}.
\end{proof}

\subsection{Adaptive coverage under a bounded label contamination model}

\begin{proof}[Proof of Theorem~\ref{thm:algorithm-correction-ci}]
  Suppose $Y_{n+1} = k$, for some $k \in [K]$.
  By definition of the conformity score function in~\eqref{eq:conf-scores}, the event $k \notin \hat{C}^{\mathrm{ci}}(X_{n+1})$ occurs if and only if $\hat{s}(X_{n+1},k) > \hat{\tau}_k$.
  We will assume without loss of generality that $\hat{\mathcal{I}}^{\mathrm{ci}}_k \neq \emptyset$ and $\hat{i}_k = \min\{i \in \hat{\mathcal{I}}^{\mathrm{ci}}_k\}$; otherwise, $\hat{\tau}_k=1$ and the result trivially holds.
By proceeding exactly as in the proof of Theorem~\ref{thm:algorithm-correction}, we have
  \begin{align*}
    & \P{Y \notin C(X, S^k_{(\hat{i}_k)}) \mid Y = k, \hat{V}^{\mathrm{low}}, \hat{V}^{\mathrm{upp}}} \\
    &  \qquad \leq \alpha  + \EV{\sup_{t \in \mathbb{R}} [\hat{\Delta}_k^{\mathrm{ci}}(t) - \Delta_k(t)] \mid \hat{V}^{\mathrm{low}}, \hat{V}^{\mathrm{upp}} } + c(n_k) - \delta^{\mathrm{ci}}(n_k,n_{*}).
  \end{align*}
The second term on the right-hand-side above can be bound using the DKW inequality, as made precise by Lemma~\ref{lemma:dkw-delta-expected-ci}.
This leads to:
  \begin{align*}
    & \P{Y \notin C(X, S^k_{(\hat{i}_k)}) \mid Y = k, \hat{V}^{\mathrm{low}}, \hat{V}^{\mathrm{upp}}} \\
    & \qquad \leq \alpha +
      \frac{2 \sum_{l \neq k} \left( |\hat{V}^{\mathrm{upp}}_{kl}| + \hat{\delta}^{(V)}_{kl} \right) }{\sqrt{n_{*}}} \min \left\{ K \sqrt{\frac{\pi}{2}} , \frac{1}{\sqrt{n_{*}}} + \sqrt{\frac{\log(2K) + \log(n_{*})}{2}} \right\} \\
  & \qquad \qquad + 2 \cdot \I{V \notin [\hat{V}^{\mathrm{low}}, \hat{V}^{\mathrm{upp}}] } \cdot \sum_{l \neq k} |\bar{V}^{\mathrm{upp}}_{kl}| 
    + c(n_k)- \delta^{\mathrm{ci}}(n_k,n_{*}) \\
    & \qquad = \alpha + 2 \I{V \notin [\hat{V}^{\mathrm{low}}, \hat{V}^{\mathrm{upp}}] } \sum_{l \neq k} |\bar{V}^{\mathrm{upp}}_{kl}| 
      - 2 \alpha_V \sum_{l \neq k} |\bar{V}_{kl}^{\mathrm{upp}}|  \\
    & \qquad \leq \alpha + 2 \left\{ \I{V \notin [\hat{V}^{\mathrm{low}}, \hat{V}^{\mathrm{upp}}] }  - \alpha_V \right\} \sum_{l \neq k} |\bar{V}_{kl}^{\mathrm{upp}}|.
  \end{align*}
Finally, taking an expectation with respect to the data used to estimate the confidence region for $V$ cancels the last term on the right-hand-side above, leading to:
  \begin{align*}
    \P{Y \notin C(X, S^k_{(\hat{i}_k)}) \mid Y = k}
    & \leq \alpha.
  \end{align*}
\end{proof}

\begin{lemma}\label{lemma:dkw-delta-expected-ci}
Under the assumptions of Theorem~\ref{thm:algorithm-correction-ci}, for any $k \in [K]$,
\begin{align*}
  & \EV{ \sup_{t \in \mathbb{R}} [ \hat{\Delta}_k^{\mathrm{ci}}(t) - \Delta_k(t) ] \mid \hat{V}^{\mathrm{low}}, \hat{V}^{\mathrm{upp}}} \\
  & \qquad \leq \EV{ \sup_{t \in \mathbb{R}} \max\left\{ \hat{\Delta}_k^{\mathrm{ci}}(t) - \Delta_k(t) , 0 \right\} \mid \hat{V}^{\mathrm{low}}, \hat{V}^{\mathrm{upp}}} \\
  & \qquad \leq \frac{2 \sum_{l \neq k} \left(|\hat{V}^{\mathrm{upp}}_{kl}| + \hat{\delta}^{(V)}_{kl} \right) }{\sqrt{n_{*}}} \min \left\{ K \sqrt{\frac{\pi}{2}} , \frac{1}{\sqrt{n_{*}}} + \sqrt{\frac{\log(2K) + \log(n_{*})}{2}} \right\} \\
  & \qquad \qquad + 2 \cdot \I{V \notin [\hat{V}^{\mathrm{low}}, \hat{V}^{\mathrm{upp}}] } \cdot \sum_{l \neq k} |\bar{V}_{kl}^{\mathrm{upp}}|.
\end{align*}
Further,
\begin{align*}
  & \EV{ \sup_{t \in \mathbb{R}} | \hat{\Delta}_k^{\mathrm{ci}}(t) - \Delta_k(t) | \mid \hat{V}^{\mathrm{low}}, \hat{V}^{\mathrm{upp}}} \\
  & \qquad \leq \frac{2 \sum_{l \neq k} |\hat{V}^{\mathrm{upp}}_{kl}| }{\sqrt{n_{*}}} \min \left\{ K \sqrt{\frac{\pi}{2}} , \frac{1}{\sqrt{n_{*}}} + \sqrt{\frac{\log(2K) + \log(n_{*})}{2}} \right\} \\
  & \qquad \qquad + 2 \I{V \notin [\hat{V}^{\mathrm{low}}, \hat{V}^{\mathrm{upp}}] } \cdot \sum_{l \neq k} |\bar{V}_{kl}^{\mathrm{upp}}| + (K-1) \left( \hat{\delta}^{(V)}_{k*}  + |\hat{\zeta}^{\mathrm{upp}}_k| \right).
\end{align*}

\end{lemma}

\begin{proof}[Proof of Lemma~\ref{lemma:dkw-delta-expected-ci}]
To simplify the notation in the following, define
\begin{align*}
  \hat{\psi}_{k} := \min_{l \neq k} \left( \hat{V}^{\mathrm{upp}}_{kl} - V_{kl} \right).
\end{align*}

Note that $\hat{\Delta}_k^{\mathrm{ci}}(t)$ can be equivalently written as:
\begin{align*}
  \hat{\Delta}_k^{\mathrm{ci}}(t)
  & := \sum_{l \neq k} \hat{V}^{\mathrm{upp}}_{kl} \left( \hat{F}_l^{k}(t) - \hat{F}_k^{k}(t) \right) -  \hat{\delta}^{(V)}_{k*} \left| \sum_{l \neq k}  \left( \hat{F}_l^{k}(t) - \hat{F}_k^{k}(t) \right) \right| \\
  & \qquad - |\hat{\zeta}^{\mathrm{upp}}_k| \sum_{l \neq k} \left| \hat{F}_l^{k}(t) - \hat{F}_k^{k}(t) \right|.
\end{align*}
Therefore, for any $k \in [K]$ and $t \in \mathbb{R}$,
\begin{align*}
  \hat{\Delta}_k^{\mathrm{ci}}(t) - \Delta_k^{\mathrm{ci}}(t)
  & = \sum_{l \neq k} \hat{V}^{\mathrm{upp}}_{kl} \left( \hat{F}_l^{k}(t) - \hat{F}_k^{k}(t) \right) - \sum_{l \neq k} V_{kl} \left( \tF_l^{k}(t) - \tF_k^{k}(t)\right) \\
  & \qquad -  \hat{\delta}^{(V)}_{k*} \left| \sum_{l \neq k}  \left( \hat{F}_l^{k}(t) - \hat{F}_k^{k}(t) \right) \right| - |\hat{\zeta}^{\mathrm{upp}}_k| \sum_{l \neq k} \left| \hat{F}_l^{k}(t) - \hat{F}_k^{k}(t) \right| \\
  & = \sum_{l \neq k} \hat{V}^{\mathrm{upp}}_{kl} \left( \hat{F}_l^{k}(t) - \tilde{F}_l^{k}(t) + \tilde{F}_k^{k}(t) - \hat{F}_k^{k}(t)   \right) \\
  & \qquad + \sum_{l \neq k} \left( \hat{V}^{\mathrm{upp}}_{kl} - V_{kl} \right) \left( \tF_l^{k}(t) - \hat{F}_l^{k}(t) + \hat{F}_k^{k}(t) - \tF_k^{k}(t)\right) \\
  & \qquad + \sum_{l \neq k} \left( \hat{V}^{\mathrm{upp}}_{kl} - V_{kl} \right) \left( \hat{F}_l^{k}(t) - \hat{F}_k^{k}(t) \right) \\
  & \qquad -  \hat{\delta}^{(V)}_{k*} \left| \sum_{l \neq k}  \left( \hat{F}_l^{k}(t) - \hat{F}_k^{k}(t) \right) \right| - |\hat{\zeta}^{\mathrm{upp}}_k| \sum_{l \neq k} \left| \hat{F}_l^{k}(t) - \hat{F}_k^{k}(t) \right| \\
  & = \sum_{l \neq k} \hat{V}^{\mathrm{upp}}_{kl} \left( \hat{F}_l^{k}(t) - \tilde{F}_l^{k}(t) + \tilde{F}_k^{k}(t) - \hat{F}_k^{k}(t)   \right) \\
  & \qquad + \sum_{l \neq k} \left( \hat{V}^{\mathrm{upp}}_{kl} - V_{kl} \right) \left( \tF_l^{k}(t) - \hat{F}_l^{k}(t) + \hat{F}_k^{k}(t) - \tF_k^{k}(t)\right) \\
  & \qquad + \sum_{l \neq k} \left( \hat{V}^{\mathrm{upp}}_{kl} - V_{kl} - \hat{\psi}_k \right) \left( \hat{F}_l^{k}(t) - \hat{F}_k^{k}(t) \right) 
   + \hat{\psi}_k \sum_{l \neq k} \left( \hat{F}_l^{k}(t) - \hat{F}_k^{k}(t) \right) \\
  & \qquad -  \hat{\delta}^{(V)}_{k*} \left| \sum_{l \neq k}  \left( \hat{F}_l^{k}(t) - \hat{F}_k^{k}(t) \right) \right| - |\hat{\zeta}^{\mathrm{upp}}_k| \sum_{l \neq k} \left| \hat{F}_l^{k}(t) - \hat{F}_k^{k}(t) \right|.
\end{align*}
Next, note that,
\begin{align*}
  \sum_{l \neq k} \left( \hat{V}^{\mathrm{upp}}_{kl} - V_{kl} - \hat{\psi}_k \right) \left( \hat{F}_l^{k}(t) - \hat{F}_k^{k}(t) \right) 
  &  \leq \sum_{l \neq k} \left( \hat{V}^{\mathrm{upp}}_{kl} - V_{kl} - \hat{\psi}_k \right) \left| \hat{F}_l^{k}(t) - \hat{F}_k^{k}(t) \right| \\
  &  \leq \left[ \max_{l \neq k} \left( \hat{V}^{\mathrm{upp}}_{kl} - V_{kl} \right) - \hat{\psi}_k \right] \cdot \sum_{l \neq k} \left| \hat{F}_l^{k}(t) - \hat{F}_k^{k}(t) \right| \\
  &  = \hat{\zeta}_k \sum_{l \neq k} \left| \hat{F}_l^{k}(t) - \hat{F}_k^{k}(t) \right|,
\end{align*}
because $\hat{V}^{\mathrm{upp}}_{kl} - V_{kl} \geq \hat{\psi}_k$ for all $l \neq k$.
Therefore,
\begin{align*}
  \hat{\Delta}_k^{\mathrm{ci}}(t) - \Delta_k^{\mathrm{ci}}(t)
  & \leq 2 \left( \max_{l \in [K]} \left| \hat{F}_l^{k}(t) - \tilde{F}_l^{k}(t) \right| \right) \sum_{l \neq k} \left( |\hat{V}^{\mathrm{upp}}_{kl}| + \left| \hat{V}^{\mathrm{upp}}_{kl} - V_{kl} \right| \right)\\
  & \qquad + \hat{\zeta}_k \sum_{l \neq k} \left| \hat{F}_l^{k}(t) - \hat{F}_k^{k}(t) \right| \\
  & \qquad + \hat{\psi}_k \left| \sum_{l \neq k} \left( \hat{F}_l^{k}(t) - \hat{F}_k^{k}(t) \right) \right| \\
  & \qquad -  \hat{\delta}^{(V)}_{k*} \left| \sum_{l \neq k}  \left( \hat{F}_l^{k}(t) - \hat{F}_k^{k}(t) \right) \right| - |\hat{\zeta}^{\mathrm{upp}}_k| \sum_{l \neq k} \left| \hat{F}_l^{k}(t) - \hat{F}_k^{k}(t) \right|.
\end{align*}

In the event that $V_{kl} \in [\hat{V}^{\mathrm{low}}_{kl} , \hat{V}^{\mathrm{upp}}_{kl}]$ for all $l \neq k$, which occurs with probability $1-\alpha_V$,
\begin{align*}
  & 0 \leq \hat{V}^{\mathrm{upp}}_{kl} - V_{kl} \leq \hat{\delta}^{(V)}_{kl},
  & 0 \leq \hat{\psi}_k \leq \hat{\delta}^{(V)}_{kl},
\end{align*}
which implies $\hat{V}^{\mathrm{upp}}_{kl} - V_{kl} - \hat{\psi}_k \leq \hat{\delta}^{(V)}_{kl} \leq \hat{\delta}^{(V)}_{k*}$.
Further, in this event, $\hat{\zeta}_k \leq \hat{\zeta}^{\mathrm{upp}}_k$.
Therefore, if $V_{kl} \in [\hat{V}^{\mathrm{low}}_{kl} , \hat{V}^{\mathrm{upp}}_{kl}]$ for all $l \neq k$,
\begin{align*}
  \hat{\Delta}_k^{\mathrm{ci}}(t) - \Delta_k^{\mathrm{ci}}(t)
  & \leq 2 \left( \max_{l \in [K]} \left| \hat{F}_l^{k}(t) - \tilde{F}_l^{k}(t) \right| \right) \sum_{l \neq k} \left( |\hat{V}^{\mathrm{upp}}_{kl}| + \hat{\delta}^{(V)}_{kl} \right).
\end{align*}

Imagine now that it is not the case that $V_{kl} \in [\hat{V}^{\mathrm{low}}_{kl} , \hat{V}^{\mathrm{upp}}_{kl}]$ for all $l \neq k$.
Then,
\begin{align*}
  \hat{\Delta}_k^{\mathrm{ci}}(t) - \Delta_k^{\mathrm{ci}}(t)
  & = \sum_{l \neq k} \hat{V}^{\mathrm{upp}}_{kl} \left( \hat{F}_l^{k}(t) - \hat{F}_k^{k}(t) \right) - \sum_{l \neq k} V_{kl} \left( \tF_l^{k}(t) - \tF_k^{k}(t)\right) \\
  & \qquad -  \hat{\delta}^{(V)}_{k*} \left| \sum_{l \neq k}  \left( \hat{F}_l^{k}(t) - \hat{F}_k^{k}(t) \right) \right| - |\hat{\zeta}^{\mathrm{upp}}_k| \sum_{l \neq k} \left| \hat{F}_l^{k}(t) - \hat{F}_k^{k}(t) \right| \\
  & = \sum_{l \neq k} \hat{V}^{\mathrm{upp}}_{kl} \left( \hat{F}_l^{k}(t) - \tilde{F}_l^{k}(t) + \tilde{F}_k^{k}(t) - \hat{F}_k^{k}(t)   \right) \\
  & \qquad + \sum_{l \neq k} \left( \hat{V}^{\mathrm{upp}}_{kl} - V_{kl} \right) \left( \tF_l^{k}(t) - \tF_k^{k}(t)\right) \\
  & \qquad -  \hat{\delta}^{(V)}_{k*} \left| \sum_{l \neq k}  \left( \hat{F}_l^{k}(t) - \hat{F}_k^{k}(t) \right) \right| - |\hat{\zeta}^{\mathrm{upp}}_k| \sum_{l \neq k} \left| \hat{F}_l^{k}(t) - \hat{F}_k^{k}(t) \right| \\
  & \leq 2 \left( \max_{l \in [K]} \left| \hat{F}_l^{k}(t) - \tilde{F}_l^{k}(t) \right| \right) \sum_{l \neq k} |\hat{V}^{\mathrm{upp}}_{kl}|
    + \sum_{l \neq k} \left| \hat{V}^{\mathrm{upp}}_{kl} - V_{kl} \right|.
\end{align*}

In conclusion,
\begin{align*}
  \hat{\Delta}_k^{\mathrm{ci}}(t) - \Delta_k^{\mathrm{ci}}(t)
  & \leq 2 \left( \max_{l \in [K]} \left| \hat{F}_l^{k}(t) - \tilde{F}_l^{k}(t) \right| \right) \sum_{l \neq k} \left( |\hat{V}^{\mathrm{upp}}_{kl}| + \hat{\delta}^{(V)}_{kl} \right)\\
  & \qquad + \I{V \notin [\hat{V}^{\mathrm{low}}, \hat{V}^{\mathrm{upp}}] } \cdot \sum_{l \neq k} \left| \hat{V}^{\mathrm{upp}}_{kl} - V_{kl} \right|.
\end{align*}
From here, proceeding as in the proof of Lemma~\ref{lemma:dkw-delta-expected}, one easily arrives at:
\begin{align*}
  & \EV{ \sup_{t \in \mathbb{R}} \max\left\{ \hat{\Delta}_k^{\mathrm{ci}}(t) - \Delta_k(t) , 0 \right\} \mid \hat{V}^{\mathrm{low}}, \hat{V}^{\mathrm{upp}}} \\
  & \qquad \leq \frac{2 \sum_{l \neq k} \left(|\hat{V}^{\mathrm{upp}}_{kl}| + \hat{\delta}^{(V)}_{kl} \right) }{\sqrt{n_{*}}} \min \left\{ K \sqrt{\frac{\pi}{2}} , \frac{1}{\sqrt{n_{*}}} + \sqrt{\frac{\log(2K) + \log(n_{*})}{2}} \right\} \\
  & \qquad \qquad + \I{V \notin [\hat{V}^{\mathrm{low}}, \hat{V}^{\mathrm{upp}}] } \cdot \sum_{l \neq k} \left| \hat{V}^{\mathrm{upp}}_{kl} - V_{kl} \right| \\
  & \qquad \leq \frac{2 \sum_{l \neq k} \left(|\hat{V}^{\mathrm{upp}}_{kl}| + \hat{\delta}^{(V)}_{kl} \right) }{\sqrt{n_{*}}} \min \left\{ K \sqrt{\frac{\pi}{2}} , \frac{1}{\sqrt{n_{*}}} + \sqrt{\frac{\log(2K) + \log(n_{*})}{2}} \right\} \\
  & \qquad \qquad + 2 \cdot \I{V \notin [\hat{V}^{\mathrm{low}}, \hat{V}^{\mathrm{upp}}] } \cdot \sum_{l \neq k} |\bar{V}_{kl}^{\mathrm{upp}}|.
\end{align*}
This completes the first part of the proof.

To prove the second part, recall that
\begin{align*}
  \hat{\Delta}_k^{\mathrm{ci}}(t) - \Delta_k^{\mathrm{ci}}(t)
  & = \sum_{l \neq k} \hat{V}^{\mathrm{upp}}_{kl} \left( \hat{F}_l^{k}(t) - \tilde{F}_l^{k}(t) + \tilde{F}_k^{k}(t) - \hat{F}_k^{k}(t)   \right) \\
  & \qquad + \sum_{l \neq k} \left( \hat{V}^{\mathrm{upp}}_{kl} - V_{kl} \right) \left( \tF_l^{k}(t) - \tF_k^{k}(t)\right) \\
  & \qquad -  \hat{\delta}^{(V)}_{k*} \left| \sum_{l \neq k}  \left( \hat{F}_l^{k}(t) - \hat{F}_k^{k}(t) \right) \right| 
    - |\hat{\zeta}^{\mathrm{upp}}_k| \sum_{l \neq k} \left| \hat{F}_l^{k}(t) - \hat{F}_k^{k}(t) \right|.
\end{align*}
Therefore,
\begin{align*}
  \left| \hat{\Delta}_k^{\mathrm{ci}}(t) - \Delta_k^{\mathrm{ci}}(t) \right|
  & \leq \sum_{l \neq k} |\hat{V}^{\mathrm{upp}}_{kl}| \left| \hat{F}_l^{k}(t) - \tilde{F}_l^{k}(t) + \tilde{F}_k^{k}(t) - \hat{F}_k^{k}(t)   \right| \\
  & \qquad + \sum_{l \neq k} \left| \hat{V}^{\mathrm{upp}}_{kl} - V_{kl} \right| + (K-1) \left( \hat{\delta}^{(V)}_{k*}  + |\hat{\zeta}^{\mathrm{upp}}_k| \right) \\
  & \leq 2 \left( \max_{l \in [K]} \left| \hat{F}_l^{k}(t) - \tilde{F}_l^{k}(t)\right| \right) \sum_{l \neq k} |\hat{V}^{\mathrm{upp}}_{kl}| \\
  & \qquad + \sum_{l \neq k} \left| \hat{V}^{\mathrm{upp}}_{kl} - V_{kl} \right| + (K-1) \left( \hat{\delta}^{(V)}_{k*}  + |\hat{\zeta}^{\mathrm{upp}}_k| \right).
\end{align*}

Then, in the event that $V_{kl} \in [\hat{V}^{\mathrm{low}}_{kl} , \hat{V}^{\mathrm{upp}}_{kl}]$ for all $l \neq k$,
\begin{align*}
  \left| \hat{\Delta}_k^{\mathrm{ci}}(t) - \Delta_k^{\mathrm{ci}}(t) \right|
  & \leq 2 \left( \max_{l \in [K]} \left| \hat{F}_l^{k}(t) - \tilde{F}_l^{k}(t)\right| \right) \sum_{l \neq k} |\hat{V}^{\mathrm{upp}}_{kl}|
    +  (K-1) \left( 2 \hat{\delta}^{(V)}_{k*}  + |\hat{\zeta}^{\mathrm{upp}}_k| \right).
\end{align*}
By contrast, if it is not the case that $V_{kl} \in [\hat{V}^{\mathrm{low}}_{kl} , \hat{V}^{\mathrm{upp}}_{kl}]$ for all $l \neq k$,
\begin{align*}
  \left| \hat{\Delta}_k^{\mathrm{ci}}(t) - \Delta_k^{\mathrm{ci}}(t) \right|
  & \leq 2 \left( \max_{l \in [K]} \left| \hat{F}_l^{k}(t) - \tilde{F}_l^{k}(t)\right| \right) \sum_{l \neq k} |\hat{V}^{\mathrm{upp}}_{kl}| \\
  & \qquad + \sum_{l \neq k} \left| \hat{V}^{\mathrm{upp}}_{kl} - V_{kl} \right| + (K-1) \left( \hat{\delta}^{(V)}_{k*}  + |\hat{\zeta}^{\mathrm{upp}}_k| \right) \\
  & \leq 2 \left( \max_{l \in [K]} \left| \hat{F}_l^{k}(t) - \tilde{F}_l^{k}(t)\right| \right) \sum_{l \neq k} |\hat{V}^{\mathrm{upp}}_{kl}| \\
  & \qquad + \sum_{l \neq k} \left( |\hat{V}^{\mathrm{upp}}_{kl}| + |V_{kl}| \right) + (K-1) \left( \hat{\delta}^{(V)}_{k*}  + |\hat{\zeta}^{\mathrm{upp}}_k| \right) \\
  & \leq 2 \left( \max_{l \in [K]} \left| \hat{F}_l^{k}(t) - \tilde{F}_l^{k}(t)\right| \right) \sum_{l \neq k} |\hat{V}^{\mathrm{upp}}_{kl}| \\
  & \qquad + 2 \sum_{l \neq k} |\bar{V}_{kl}^{\mathrm{upp}}| + (K-1) \left( \hat{\delta}^{(V)}_{k*}  + |\hat{\zeta}^{\mathrm{upp}}_k| \right).
\end{align*}
From here, proceeding as in the proof of Lemma~\ref{lemma:dkw-delta-expected}, one easily arrives at:
\begin{align*}
  & \EV{ \sup_{t \in \mathbb{R}} | \hat{\Delta}_k^{\mathrm{ci}}(t) - \Delta_k(t) | \mid \hat{V}^{\mathrm{low}}, \hat{V}^{\mathrm{upp}}} \\
  & \qquad \leq \frac{2 \sum_{l \neq k} |\hat{V}^{\mathrm{upp}}_{kl}|  }{\sqrt{n_{*}}} \min \left\{ K \sqrt{\frac{\pi}{2}} , \frac{1}{\sqrt{n_{*}}} + \sqrt{\frac{\log(2K) + \log(n_{*})}{2}} \right\} \\
  & \qquad \qquad + \I{V \notin [\hat{V}^{\mathrm{low}}, \hat{V}^{\mathrm{upp}}] } \cdot 2 \sum_{l \neq k} |\bar{V}_{kl}^{\mathrm{upp}}| + (K-1) \left( 2 \hat{\delta}^{(V)}_{k*}  + |\hat{\zeta}^{\mathrm{upp}}_k| \right).
\end{align*}
This completes the second part of the proof.
\end{proof}

\begin{proof}[Proof of Theorem~\ref{thm:algorithm-correction-upper-ci}]
We follow the same strategy as in the proof of Theorem~\ref{thm:algorithm-correction-upper}.
Define the events $\cA_1=\{\hat{\cI}^{\mathrm{ci}}_k=\emptyset \}$ and $\cA_2=\{\hat{i}_k=1\}$.
Then,
\begin{align}
  \begin{split} 	\label{eq:upper_bound_events-ci}
	& \P{Y_{n+1} \in \hat{C}^{\mathrm{ci}}(X_{n+1}) \mid Y = k, \hat{V}^{\mathrm{low}}, \hat{V}^{\mathrm{upp}}}  \\
	& \qquad \leq \P{\cA_1 \mid \hat{V}^{\mathrm{low}}, \hat{V}^{\mathrm{upp}}} \\
        & \qquad \qquad + \EV{\P{Y_{n+1} \in \hat{C}^{\mathrm{ci}}(X_{n+1}) \mid Y = k, \cD, \hat{V}^{\mathrm{low}}, \hat{V}^{\mathrm{upp}}}\I{\cA_2} \mid \hat{V}^{\mathrm{low}}, \hat{V}^{\mathrm{upp}}} \\
	& \qquad \qquad + \EV{\P{Y_{n+1} \in \hat{C}^{\mathrm{ci}}(X_{n+1}) \mid Y = k, \cD, \hat{V}^{\mathrm{low}}, \hat{V}^{\mathrm{upp}}}\I{\cA_1^c\cap \cA_2^c} \mid \hat{V}^{\mathrm{low}}, \hat{V}^{\mathrm{upp}}}.
      \end{split}
\end{align}

We will now separately bound the three terms on the right-hand-side of~\eqref{eq:upper_bound_events-ci}.
The following notation will be useful for this purpose.
For all $i \in [n_k]$, let $U_i \sim \text{Uniform}(0,1)$ be independent and identically distributed uniform random variables, and
denote their order statistics as $U_{(1)} < U_{(2)} < \ldots < U_{(n_k)}$.

\begin{itemize}
\item The probability of $\cA_1$ conditional on $\hat{V}^{\mathrm{low}}$ and $\hat{V}^{\mathrm{upp}}$  can be bound from above as:
  \begin{align} \label{eq:upper_bound_event1-ci}
    \P{\hat{\mathcal{I}}^{\mathrm{ci}}_k = \emptyset \mid \hat{V}^{\mathrm{low}}, \hat{V}^{\mathrm{upp}}}
    \leq \frac{1}{n_*} + \I{V \notin [\hat{V}^{\mathrm{low}}, \hat{V}^{\mathrm{upp}}]}.
  \end{align}

  To simplify the notation in the proof of~\eqref{eq:upper_bound_event1-ci}, define
\begin{align*}
  d_k^{\mathrm{ci}}
  & := c(n_k) + \frac{2 \sum_{l \neq k} \left( |\hat{V}^{\mathrm{upp}}_{kl}| + \hat{\delta}^{(V)}_{kl} \right) }{\sqrt{n_{*}}} \left( \frac{1}{\sqrt{n_{*}}} + 2 \sqrt{\frac{\log(2K) + \log(n_{*})}{2}} \right) \\
  & \qquad + 2 \alpha_V \sum_{l \neq k} |\bar{V}_{kl}^{\mathrm{upp}}| + (K-1) \left( 2 \hat{\delta}^{(V)}_{k*}  + |\hat{\zeta}^{\mathrm{upp}}_k| \right).
\end{align*}

Then, under Assumption~\ref{assumption:regularity-dist-delta-ci},
  \begin{align*}
    & 1 - \left( \alpha + \hDelta_k^{\mathrm{ci}}(S^k_{(i)}) - \delta^{\mathrm{ci}}(n_k, n_*) \right) \\
    & \qquad = 1 - \alpha - \Delta_k(S^k_{(i)}) + \delta^{\mathrm{ci}}(n_k, n_*) + \Delta_k(S^k_{(i)}) - \hDelta_k^{\mathrm{ci}}(S^k_{(i)}) \\
    & \qquad \leq 1 - d_k^{\mathrm{ci}} + \delta^{\mathrm{ci}}(n_k, n_*) + \sup_{t \in [0,1]} | \Delta_k(t) - \hDelta_k^{\mathrm{ci}}(t)|.
  \end{align*}
Further, it follows from the definition of $\delta^{\mathrm{ci}}(n_k, n_*)$ in~\eqref{eq:delta-constant-epsilon-ci} and Lemma~\ref{lemma:dkw-delta-prob-ci} applied with $\eta = 1/n_*$ that, with probability at least
$$1 - \frac{1}{n_*} - \I{V \notin [\hat{V}^{\mathrm{low}}, \hat{V}^{\mathrm{upp}}]},$$
we have
  \begin{align*}
    & 1 - \left( \alpha + \hDelta_k^{\mathrm{ci}}(S^k_{(i)}) - \delta^{\mathrm{ci}}(n_k, n_*) \right) \\
    & \quad \leq 1 - d_k^{\mathrm{ci}} + \delta^{\mathrm{ci}}(n_k, n_*) +  (K-1) \left( 2 \hat{\delta}^{(V)}_{k*}  + |\hat{\zeta}^{\mathrm{upp}}_k| \right) + 2 \sum_{l \neq k} |\hat{V}^{\mathrm{upp}}_{kl}| \sqrt{\frac{\log(2K) + \log(n_*)}{2 n_{*}}}\\
    & \quad \leq 1 - d_k^{\mathrm{ci}} + c(n_k) + \frac{2 \sum_{l \neq k} \left( |\hat{V}^{\mathrm{upp}}_{kl}| + \hat{\delta}^{(V)}_{kl} \right) }{\sqrt{n_{*}}} \left( \frac{1}{\sqrt{n_{*}}} + 2 \sqrt{\frac{\log(2K) + \log(n_{*})}{2}} \right) \\
    & \quad \quad + 2 \alpha_V \sum_{l \neq k} |\bar{V}_{kl}^{\mathrm{upp}}| + (K-1) \left( 2 \hat{\delta}^{(V)}_{k*}  + |\hat{\zeta}^{\mathrm{upp}}_k| \right)
     = 1.
  \end{align*}

This implies that the set $\hat{\mathcal{I}}^{\mathrm{ci}}_k$ defined in~\eqref{eq:Ik-set}, with $\hat{\Delta}_k(t)$ and $\delta^{\mathrm{ci}}(n_k, n_*)$ instead of $\hat{\Delta}_k^{\mathrm{ci}}(t)$ and $\delta(n_k,n_{*})$, is non-empty with probability at least $1 - 1/n_* - \I{V \notin [\hat{V}^{\mathrm{low}}, \hat{V}^{\mathrm{upp}}]}$, because $n_k \in \hat{\mathcal{I}}^{\mathrm{ci}}_k$ if $1 - [ \alpha + \hDelta_k^{\mathrm{ci}}(S^k_{(i)}) - \delta^{\mathrm{ci}}(n_k, n_*) ] \leq 1$.

\item Under $\cA_2$, the second term on the right-hand-side of~\eqref{eq:upper_bound_events} can be written as:
  \begin{align*}
    \P{Y_{n+1} \in \hat{C}^{\mathrm{ci}}(X_{n+1}) \mid Y = k, \cD} & = \P{\hat{s}(X_{n+1}, k) \leq S^k_{(1)} \mid Y = k, \cD} = F^{k}_k(S^k_{(1)}).
\end{align*}
Therefore, by the independence of $\hat{V}^{\mathrm{low}}$ and $\hat{V}^{\mathrm{upp}}$ from the other data,
\begin{align*}
  \begin{split}
    & \EV{\P{Y_{n+1} \in \hat{C}^{\mathrm{ci}}(X_{n+1}) \mid Y = k, \cD, \hat{V}^{\mathrm{low}}, \hat{V}^{\mathrm{upp}}}\I{\cA_2} \mid \hat{V}^{\mathrm{low}}, \hat{V}^{\mathrm{upp}}} \\
    & \qquad \leq \EV{\tF^{k}_k(S^k_{(1)})} +  \EV{F^{k}_k(S^k_{(1)})-\tilde{F}^{k}_k(S^k_{(1)})}  \\
    & \qquad = \EV{U_{(1)}} + (V_{kk}-1) \EV{\tilde{F}^{k}_k(S^k_{(1)})} + \sum_{l \neq k} V_{kl} \EV{\tilde{F}^{k}_l(S^k_{(1)})}\\
    & \qquad = V_{kk} \EV{U_{(1)}} + \sum_{l \neq k} V_{kl} \EV{\tilde{F}^{k}_l(S^k_{(1)})}\\
    & \qquad = \frac{V_{kk}}{n_k+1} + \sum_{l \neq k} V_{kl} \EV{\tilde{F}^{k}_l(S^k_{(1)})}.
  \end{split}
\end{align*}
Next, using Assumption~\ref{assumption:consitency-scores} as in the proof of Theorem~\ref{thm:algorithm-correction-upper}, we arrive to:
\begin{align}     \label{eq:upper_bound_event2-ci}
  \begin{split}
    & \EV{\P{Y_{n+1} \in \hat{C}^{\mathrm{ci}}(X_{n+1}) \mid Y = k, \cD, \hat{V}^{\mathrm{low}}, \hat{V}^{\mathrm{upp}}}\I{\cA_2} \mid \hat{V}^{\mathrm{low}}, \hat{V}^{\mathrm{upp}}} \\
      & \qquad \leq \frac{V_{kk} + \sum_{l \neq k} |V_{kl}|}{n_k+1} 
       = \frac{1 + \sum_{l \neq k} \left( |V_{kl}| - V_{kl} \right) }{n_k+1} \\
      & \qquad \leq \frac{1 + 2 \sum_{l \neq k} |V_{kl}|}{n_k+1}.
  \end{split}
\end{align}

\item Under $\cA_1^c\cap \cA_2^c$, by definition of $\hat{\mathcal{I}}^{\mathrm{ci}}_k$, for any $i \leq \hat{i}_k - 1$,
  \begin{align*}
    \frac{i}{n_k} < 1 - \left( \alpha + \hDelta_k^{\mathrm{ci}}(S^k_{(i)}) - \delta^{\mathrm{ci}}(n_k, n_*) \right).
  \end{align*}
  Therefore, choosing $i = \hat{i}_k-1$, we get:
  \begin{align*}
    \frac{\hat{i}_k}{n_k} < 1 - \left( \alpha + \hDelta_k^{\mathrm{ci}}(S^k_{(\hat{i}_k-1)}) - \delta^{\mathrm{ci}}(n_k, n_*) \right) + \frac{1}{n_k}.
  \end{align*}
As in the proof of Theorem~\ref{thm:algorithm-correction-upper}, the probability of coverage conditional on $Y_{n+1}=k$, on $\mathcal{D}$, and on the confidence interval $[\hat{V}^{\mathrm{low}}, \hat{V}^{\mathrm{upp}}]$ can be decomposed as:
  \begin{align*}
    & \P{Y \in C(X, S^k_{(\hat{i}_k)}) \mid Y = k, \mathcal{D}, \hat{V}^{\mathrm{low}}, \hat{V}^{\mathrm{upp}}} \\
    & \qquad = \frac{\hat{i}_k}{n_k} + \tilde{F}^{k}_k(S^k_{(\hat{i}_k)}) - \hat{F}^{k}_k(S^k_{(\hat{i}_k)})+ \Delta_k(S^k_{(\hat{i}_k)}) \\
    & \qquad < 1 - \left( \alpha + \hDelta_k^{\mathrm{ci}}(S^k_{(\hat{i}_k-1)}) - \delta^{\mathrm{ci}}(n_k, n_*) \right) + \frac{1}{n_k} + \tilde{F}^{k}_k(S^k_{(\hat{i}_k)}) - \hat{F}^{k}_k(S^k_{(\hat{i}_k)})+ \Delta_k(S^k_{(\hat{i}_k)}) \\
    & \qquad = 1 - \alpha + \delta^{\mathrm{ci}}(n_k, n_*) + \frac{1}{n_k} + \tilde{F}^{k}_k(S^k_{(\hat{i}_k)}) - \hat{F}^{k}_k(S^k_{(\hat{i}_k)})- \hDelta_k^{\mathrm{ci}}(S^k_{(\hat{i}_k)}) + \Delta_k(S^k_{(\hat{i}_k-1)}).
  \end{align*}

  To bound the last term above, note that
  \begin{align*}
     \Delta_k(S^k_{(\hat{i}_k)}) - \hDelta_k^{\mathrm{ci}}(S^k_{(\hat{i}_k-1)})
    &  = (\Delta_k(S^k_{(\hat{i}_k)}) - \Delta_k(S^k_{(\hat{i}_k-1)})) + (\Delta_k(S^k_{(\hat{i}_k-1)}) - \hDelta_k^{\mathrm{ci}}(S^k_{(\hat{i}_k-1)})) \\
    &  \leq \sup_{t \in \mathbb{R}} |\Delta_k(t) - \hDelta_k^{\mathrm{ci}}(t)| + (\Delta_k(S^k_{(\hat{i}_k)}) - \Delta_k(S^k_{(\hat{i}_k-1)})),
  \end{align*}
  where the expected value of the first term above can be bounded using Lemma~\ref{lemma:dkw-delta-expected-ci}, and the second term is given by
  \begin{align*}
    & \Delta_k(S^k_{(\hat{i}_k)}) - \Delta_k(S^k_{(\hat{i}_k-1)}) \\
    & \qquad = (V_{kk}-1)  \left[ \tF^{k}_k(S^k_{(\hat{i}_k)}) -  \tF^{k}_k(S^k_{(\hat{i}_k-1)}) \right] + \sum_{l\neq k}  V_{kl} \left[ \tF^{k}_{l}(S^k_{(\hat{i}_k)}) - \tF^{k}_{l}(S^k_{(\hat{i}_k-1)}) \right]	\\
    & \qquad = \sum_{l\neq k} V_{kl} \left[ \left( \tF^{k}_{l}(S^k_{(\hat{i}_k)}) - \tF^{k}_{l}(S^k_{(\hat{i}_k-1)}) \right) - \left(\tF^{k}_k(S^k_{(\hat{i}_k)}) - \tF^{k}_k(S^k_{(\hat{i}_k-1)}) \right) \right] \\
    & \qquad \leq  2\sum_{l\neq k}|V_{kl}| \cdot  \max_{2\leq i\leq n_k}\max_{l\in[K]}\left| \tF^{k}_{l}(S^k_{(i)}) - \tF^{k}_{l}(S^k_{(i-1)}) \right|.
  \end{align*}
  Combining the above, we find that, under $\cA_1^c\cap \cA_2^c$,
  \begin{align*}
    & \P{Y \in C(X, S^k_{(\hat{i}_k)}) \mid Y = k, \mathcal{D}, \hat{V}^{\mathrm{low}}, \hat{V}^{\mathrm{upp}}}	\\
    & \qquad \leq 1 - \alpha + \delta^{\mathrm{ci}}(n_k, n_*) + \frac{1}{n_k} + \tilde{F}^{k}_k(S^k_{(\hat{i}_k)}) - \hat{F}^{k}_k(S^k_{(\hat{i}_k)})+  \sup_{t \in \mathbb{R}} |\Delta_k(t) - \hDelta_k^{\mathrm{ci}}(t)| \\
    & \qquad \qquad + 2\sum_{l\neq k}|V_{kl}| \cdot \max_{2\leq i\leq n_k}\max_{l\in[K]}\left| \tF^{k}_{l}(S^k_{(i)}) - \tF^{k}_{l}(S^k_{(i-1)}) \right|.
  \end{align*}

  It remains to bound the expectation of the last term. By Assumption \ref{assumption:regularity-dist},
  \begin{align*}
    \max_{2\leq i\leq n_k}\max_{l\in[K]}\left| \tF^{k}_{l}(S^k_{(i)}) - \tF^{k}_{l}(S^k_{(i-1)}) \right| & \leq f_{\max} \max_{2\leq i\leq n_k}| S^k_{(i)} - S^k_{(i-1)}|	\\
                                                                                                         & \leq \frac{f_{\max}}{f_{\min}} \max_{2\leq i\leq n_k} |\tF^{k}_{k}(S^k_{(i)}) - \tF^{k}_{k}(S^k_{(i-1)})|	\\
    \qquad \qquad & \stackrel{d}{=} \frac{f_{\max}}{f_{\min}} \max_{2\leq i\leq n_k} (U_{(i)} - U_{(i-1)}) \\
                                                                                                         &\leq \frac{f_{\max}}{f_{\min}}  \max_{1 \leq i \leq n_k +1 } D_i,
  \end{align*}
  where $D_1=U_{(1)}$, $D_i=U_{(i)} - U_{(i-1)}$ for $i=2, \dots, n_k$, and $D_{n_k+1} = 1-U_{(n_k)}$.
  By a standard result on maximum uniform spacing,
  \begin{align*}
    \EV{\max_{1\leq i \leq n_k+1} D_i} = \frac{1}{n_k+1}\sum_{j=1}^{n_k+1} \frac{1}{j}.
  \end{align*}

  Therefore, under $\cA_1^c\cap \cA_2^c$,
\begin{align} \label{eq:upper_bound_event3-ci}
\begin{split}
	& \EV{\P{Y_{n+1} \in \hat{C}^{\mathrm{ci}}(X_{n+1}) \mid Y = k, \cD, \hat{V}^{\mathrm{low}}, \hat{V}^{\mathrm{upp}}}\I{\cA_1^c\cap \cA_2^c} \mid \hat{V}^{\mathrm{low}}, \hat{V}^{\mathrm{upp}}} \\
	& \qquad \leq  1 - \alpha + \delta^{\mathrm{ci}}(n_k, n_*) + \frac{1}{n_k} + c(n_k) + \EV{\sup_{t \in \mathbb{R}} |\Delta_k(t) - \hDelta_k^{\mathrm{ci}}(t)| \mid \hat{V}^{\mathrm{low}}, \hat{V}^{\mathrm{upp}}} \\
	& \qquad \qquad + 2\sum_{l\neq k}|V_{kl}|  \cdot \frac{f_{\max}}{f_{\min}} \cdot \frac{1}{n_k+1}\sum_{j=1}^{n_k+1} \frac{1}{j}	\\
	& \qquad   \leq  1 - \alpha + \delta^{\mathrm{ci}}(n_k, n_*) + \frac{1}{n_k} + c(n_k) \\
        & \qquad \qquad + \frac{2 \sum_{l \neq k} |\hat{V}^{\mathrm{upp}}_{kl}| }{\sqrt{n_{*}}} \min \left\{ K \sqrt{\frac{\pi}{2}} , \frac{1}{\sqrt{n_{*}}} + \sqrt{\frac{\log(2K) + \log(n_{*})}{2}} \right\} \\
  & \qquad \qquad + \I{V \notin [\hat{V}^{\mathrm{low}}, \hat{V}^{\mathrm{upp}}] } \cdot \sum_{l \neq k} \left( |\hat{V}^{\mathrm{upp}}_{kl}|  + |V_{kl}| \right) + (K-1) \left( 2 \hat{\delta}^{(V)}_{k*}  + |\hat{\zeta}^{\mathrm{upp}}_k| \right)  \\
	& \qquad \qquad + 2\sum_{l\neq k}|V_{kl}|  \cdot \frac{f_{\max}}{f_{\min}} \cdot \frac{1}{n_k+1}\sum_{j=1}^{n_k+1} \frac{1}{j}	\\
	& \qquad  =  1 - \alpha + 2 \delta^{\mathrm{ci}}(n_k, n_*) + \frac{1}{n_k} \\
        & \qquad \qquad + \I{V \notin [\hat{V}^{\mathrm{low}}, \hat{V}^{\mathrm{upp}}] } \cdot \sum_{l \neq k} \left( |\hat{V}^{\mathrm{upp}}_{kl}|  + |V_{kl}| \right) - 2 \alpha_V \sum_{l \neq k} |\bar{V}_{kl}^{\mathrm{upp}}| \\
	& \qquad \qquad + (K-1) \left( 2 \hat{\delta}^{(V)}_{k*}  + |\hat{\zeta}^{\mathrm{upp}}_k| \right) + 2\sum_{l\neq k}|V_{kl}| \cdot \frac{f_{\max}}{f_{\min}} \cdot \frac{1}{n_k+1}\sum_{j=1}^{n_k+1} \frac{1}{j} \\
	& \qquad  \leq  1 - \alpha + 2 \delta^{\mathrm{ci}}(n_k, n_*) + \frac{1}{n_k} \\
        & \qquad \qquad + 2   \sum_{l \neq k} |\bar{V}_{kl}^{\mathrm{upp}}| \cdot \left( \I{V \notin [\hat{V}^{\mathrm{low}}, \hat{V}^{\mathrm{upp}}] } - \alpha_V \right) \\
	& \qquad \qquad + (K-1) \left( 2 \hat{\delta}^{(V)}_{k*}  + |\hat{\zeta}^{\mathrm{upp}}_k| \right) + 2\sum_{l\neq k}|V_{kl}| \cdot \frac{f_{\max}}{f_{\min}} \cdot \frac{1}{n_k+1}\sum_{j=1}^{n_k+1} \frac{1}{j}.
      \end{split}
\end{align}
Above, the second inequality is obtained from Lemma~\ref{lemma:dkw-delta-expected-ci} and the bound in~\eqref{eq:tighter_cdf_bound}.

\end{itemize}

Combining~\eqref{eq:upper_bound_events-ci} with \eqref{eq:upper_bound_event1-ci}, \eqref{eq:upper_bound_event2-ci}, and \eqref{eq:upper_bound_event3-ci} leads to:
\begin{align*}
  & \P{Y_{n+1} \in \hat{C}^{\mathrm{ci}}(X_{n+1}) \mid Y = k, \hat{V}^{\mathrm{low}}, \hat{V}^{\mathrm{upp}}}  \\
  & \qquad \leq \P{\cA_1 \mid \hat{V}^{\mathrm{low}}, \hat{V}^{\mathrm{upp}}} \\
  & \qquad \qquad + \EV{\P{Y_{n+1} \in \hat{C}^{\mathrm{ci}}(X_{n+1}) \mid Y = k, \cD, \hat{V}^{\mathrm{low}}, \hat{V}^{\mathrm{upp}}}\I{\cA_2} \mid \hat{V}^{\mathrm{low}}, \hat{V}^{\mathrm{upp}}} \\
  & \qquad \qquad + \EV{\P{Y_{n+1} \in \hat{C}^{\mathrm{ci}}(X_{n+1}) \mid Y = k, \cD, \hat{V}^{\mathrm{low}}, \hat{V}^{\mathrm{upp}}}\I{\cA_1^c\cap \cA_2^c} \mid \hat{V}^{\mathrm{low}}, \hat{V}^{\mathrm{upp}}} \\
  & \qquad \leq 1 - \alpha + \frac{1}{n_{*}} + \I{V \notin [\hat{V}^{\mathrm{low}}, \hat{V}^{\mathrm{upp}}]} + \frac{1 + 2 \sum_{l \neq k} |V_{kl}|}{n_k+1} \\
  & \qquad \qquad + 2 \delta^{\mathrm{ci}}(n_k, n_*) + \frac{1}{n_k} + 2   \sum_{l \neq k} |\bar{V}_{kl}^{\mathrm{upp}}| \cdot \left( \I{V \notin [\hat{V}^{\mathrm{low}}, \hat{V}^{\mathrm{upp}}] } - \alpha_V \right) \\
	& \qquad \qquad + (K-1) \left( 2 \hat{\delta}^{(V)}_{k*}  + |\hat{\zeta}^{\mathrm{upp}}_k| \right) + 2\sum_{l\neq k}|V_{kl}| \cdot \frac{f_{\max}}{f_{\min}} \cdot \frac{1}{n_k+1}\sum_{j=1}^{n_k+1} \frac{1}{j} \\
  & \qquad \leq 1 - \alpha + \frac{1}{n_{*}} + \I{V \notin [\hat{V}^{\mathrm{low}}, \hat{V}^{\mathrm{upp}}]}  + (K-1) \left( 2 \hat{\delta}^{(V)}_{k*}  + |\hat{\zeta}^{\mathrm{upp}}_k| \right) \\
  & \qquad \qquad + 2 \delta^{\mathrm{ci}}(n_k, n_*) + 2   \sum_{l \neq k} |\bar{V}_{kl}^{\mathrm{upp}}| \cdot \left( \I{V \notin [\hat{V}^{\mathrm{low}}, \hat{V}^{\mathrm{upp}}] } - \alpha_V \right) \\ 
	& \qquad \qquad + \frac{2}{n_k} \left[ 1 + \sum_{l \neq k} |V_{kl}| + \sum_{l\neq k}|V_{kl}| \cdot \frac{f_{\max}}{f_{\min}} \cdot \sum_{j=1}^{n_k+1} \frac{1}{j} \right].
\end{align*}

Finally, taking an expectation with respect to $\hat{V}^{\mathrm{low}}$ and $\hat{V}^{\mathrm{upp}}$, we arrive at:
\begin{align*}
  & \P{Y_{n+1} \in \hat{C}^{\mathrm{ci}}(X_{n+1}) \mid Y = k}  \\ 
  & \qquad \leq 1 - \alpha + \frac{1}{n_{*}} + \alpha_V + 2 \EV{ \delta^{\mathrm{ci}}(n_k, n_*) } + (K-1) \left( 2 \EV{\hat{\delta}^{(V)}_{k*}}  + \EV{|\hat{\zeta}^{\mathrm{upp}}_k|} \right) \\
	& \qquad \qquad + \frac{2}{n_k} \left[ 1 + \sum_{l \neq k} |V_{kl}| + \sum_{l\neq k}|V_{kl}| \cdot \frac{f_{\max}}{f_{\min}} \cdot \sum_{j=1}^{n_k+1} \frac{1}{j} \right] \\
  & \qquad = 1 - \alpha + \frac{1}{n_{*}} + \alpha_V + 2c(n_k) + 4 \alpha_V \sum_{l \neq k} |\bar{V}_{kl}^{\mathrm{upp}}| \\
    & \qquad\qquad + \frac{4 \sum_{l \neq k} \EV{ |\hat{V}^{\mathrm{upp}}_{kl}| + \hat{\delta}^{(V)}_{kl} } }{\sqrt{n_{*}}} \min \left\{ K \sqrt{\frac{\pi}{2}} , \frac{1}{\sqrt{n_{*}}} + \sqrt{\frac{\log(2K) + \log(n_{*})}{2}} \right\}\\
  & \qquad \qquad  + (K-1) \left( 2 \EV{\hat{\delta}^{(V)}_{k*}}  + \EV{|\hat{\zeta}^{\mathrm{upp}}_k|} \right)
    + \frac{2}{n_k} \left[ 1 + \sum_{l \neq k} |V_{kl}| + \sum_{l\neq k}|V_{kl}| \cdot \frac{f_{\max}}{f_{\min}} \cdot \sum_{j=1}^{n_k+1} \frac{1}{j} \right] \\
  & \qquad = 1 - \alpha + \frac{1}{n_{*}} + \left( 1 + 4 \sum_{l \neq k} |\bar{V}_{kl}^{\mathrm{upp}}| \right) \alpha_V + 2c(n_k) \\
    & \qquad\qquad + \frac{4 \sum_{l \neq k} \EV{ |\hat{V}^{\mathrm{upp}}_{kl}| + \hat{\delta}^{(V)}_{kl} } }{\sqrt{n_{*}}} \min \left\{ K \sqrt{\frac{\pi}{2}} , \frac{1}{\sqrt{n_{*}}} + \sqrt{\frac{\log(2K) + \log(n_{*})}{2}} \right\}\\
  & \qquad \qquad + (K-1) \left( 2 \EV{\hat{\delta}^{(V)}_{k*}}  + \EV{|\hat{\zeta}^{\mathrm{upp}}_k|} \right) + \frac{2}{n_k} \left[ 1 + \sum_{l \neq k} |V_{kl}| + \sum_{l\neq k}|V_{kl}| \cdot \frac{f_{\max}}{f_{\min}} \cdot \sum_{j=1}^{n_k+1} \frac{1}{j} \right].
\end{align*}

\end{proof}

\begin{lemma}\label{lemma:dkw-delta-prob-ci}
Under the assumptions of Theorem~\ref{thm:algorithm-correction-ci}, for any $k \in [K]$ and $\eta > 0$,
\begin{align*}
  & \P{ \sup_{t \in \mathbb{R}} |\hat{\Delta}_k^{\mathrm{ci}}(t) - \Delta_k(t) | > (K-1) \left( 2 \hat{\delta}^{(V)}_{k*}  + |\hat{\zeta}^{\mathrm{upp}}_k| \right) + 2 \Upsilon(K,\eta,n_{*}) \sum_{l \neq k} |\hat{V}^{\mathrm{upp}}_{kl}|  \mid \hat{V}^{\mathrm{low}}, \hat{V}^{\mathrm{upp}} } \\
  & \qquad \leq \eta \I{V \in [\hat{V}^{\mathrm{low}}, \hat{V}^{\mathrm{upp}}]}   + \I{V \notin [\hat{V}^{\mathrm{low}}, \hat{V}^{\mathrm{upp}}]},
\end{align*}
where $\Upsilon(K,\eta,n_{*}) := \sqrt{[\log(2K) + \log(1/\eta)]/(2 n_{*})}$.
\end{lemma}

\begin{proof}[Proof of Lemma~\ref{lemma:dkw-delta-prob-ci}]
By proceeding as in the proof of Lemma~\ref{lemma:dkw-delta-expected-ci}, we obtain that, in the event that $V \in [\hat{V}^{\mathrm{low}}, \hat{V}^{\mathrm{upp}}]$, for any $k \in [K]$ and $t \in \mathbb{R}$,
\begin{align*}
  |\hat{\Delta}_k^{\mathrm{ci}}(t) - \Delta_k(t)|
  & \leq 2 \left( \max_{l \in [K]} \left| \hat{F}_l^{k}(t) - \tilde{F}_l^{k}(t)\right| \right) \sum_{l \neq k} |\hat{V}^{\mathrm{upp}}_{kl}| + (K-1) \left( 2 \hat{\delta}^{(V)}_{k*}  + |\hat{\zeta}^{\mathrm{upp}}_k| \right).
\end{align*}
Combined with the DKW inequality, this implies that, for any $\eta > 0$,
\begin{align*}
  & \P{ \sup_{t \in \mathbb{R}} |\hat{\Delta}_k^{\mathrm{ci}}(t) - \Delta_k(t) | > (K-1) \left( 2 \hat{\delta}^{(V)}_{k*}  + |\hat{\zeta}^{\mathrm{upp}}_k| \right) + 2 \Upsilon(K,\eta,n_{*}) \sum_{l \neq k} |\hat{V}^{\mathrm{upp}}_{kl}| \mid \hat{V}^{\mathrm{low}}, \hat{V}^{\mathrm{upp}} } \\
  & \qquad \leq \P{ \max_{j \in [K]} | \hat{F}_j^{k}(t) - \tF^{k}_{j}(t)| > \sqrt{\frac{\log(2K) + \log(1/\eta)}{2 n_{*}}} } \I{V \in [\hat{V}^{\mathrm{low}}, \hat{V}^{\mathrm{upp}}]} \\
  & \qquad \qquad + \I{V \notin [\hat{V}^{\mathrm{low}}, \hat{V}^{\mathrm{upp}}]} \\
  & \qquad \leq \eta \I{V \in [\hat{V}^{\mathrm{low}}, \hat{V}^{\mathrm{upp}}]}   + \I{V \notin [\hat{V}^{\mathrm{low}}, \hat{V}^{\mathrm{upp}}]}.
\end{align*}
Note that the last inequality above follows exactly as in the proof of Lemma~\ref{lemma:dkw-delta-prob}.
\end{proof}

\begin{proof}[Proof of Proposition~\ref{thm:algorithm-correction-ci-optimist}]
This proof combines elements of the proofs of Proposition~\ref{thm:algorithm-correction-optimistic}, Theorem~\ref{thm:algorithm-correction}, and Theorem~\ref{thm:algorithm-correction-ci}.
  Suppose $Y_{n+1} = k$, for some $k \in [K]$. Proceeding exactly as in the proof of Theorem~\ref{thm:algorithm-correction}, we obtain:
\begin{align*}
	& \P{Y \notin C(X, S^k_{(\hat{i}_k)}) \mid Y = k, \mathcal{D}, \hat{V}^{\mathrm{low}}, \hat{V}^{\mathrm{upp}}}  \\
	& \qquad = 1 - \tilde{F}^{k}_k(S^k_{(\hat{i}_k)}) - \Delta_k(S^k_{(\hat{i}_k)}) 	\\
	& \qquad = 1- \hat{F}^{k}_k(S^k_{(\hat{i}_k)}) - \max\left\{\hDelta^{\mathrm{ci}}_k(S^k_{(\hat{i}_k)}) - \delta^{\mathrm{ci}}(n_k,n_{*}), -(1-\alpha)/n_k \right\}	\\
	& \qquad\qquad + \hat{F}^{k}_k(S^k_{(\hat{i}_k)}) - \tilde{F}^{k}_k(S^k_{(\hat{i}_k)}) - \delta^{\mathrm{ci}}(n_k,n_{*}) 	\\
	& \qquad\qquad +  \max\left\{\hDelta^{\mathrm{ci}}_k(S^k_{(\hat{i}_k)}) - \delta^{\mathrm{ci}}(n_k,n_{*}), -(1-\alpha)/n_k \right\} - (\Delta_k(S^k_{(\hat{i}_k)}) - \delta^{\mathrm{ci}}(n_k,n_{*})) \\
	& \qquad = 1- \hat{F}^{k}_k(S^k_{(\hat{i}_k)}) - \max\left\{\hDelta^{\mathrm{ci}}_k(S^k_{(\hat{i}_k)}) - \delta^{\mathrm{ci}}(n_k,n_{*}), -(1-\alpha)/n_k \right\}	\\
	& \qquad\qquad + \hat{F}^{k}_k(S^k_{(\hat{i}_k)}) - \tilde{F}^{k}_k(S^k_{(\hat{i}_k)}) - \delta^{\mathrm{ci}}(n_k,n_{*})  \\
	& \qquad\qquad +  \max\left\{\hDelta^{\mathrm{ci}}_k(S^k_{(\hat{i}_k)}) - \delta^{\mathrm{ci}}(n_k,n_{*}), -(1-\alpha)/n_k \right\} \\
        & \qquad \qquad - \max \left\{ \Delta_k(S^k_{(\hat{i}_k)}) - \delta^{\mathrm{ci}}(n_k,n_{*}), -(1-\alpha)/n_k \right\} 	\\
	& \qquad = 1- \hat{F}^{k}_k(S^k_{(\hat{i}_k)}) - \max\left\{\hDelta^{\mathrm{ci}}_k(S^k_{(\hat{i}_k)}) - \delta^{\mathrm{ci}}(n_k,n_{*}), -(1-\alpha)/n_k \right\}	\\
	& \qquad\qquad + \hat{F}^{k}_k(S^k_{(\hat{i}_k)}) - \tilde{F}^{k}_k(S^k_{(\hat{i}_k)}) - \delta^{\mathrm{ci}}(n_k,n_{*})  \\
	& \qquad\qquad +  \max\left\{\hDelta^{\mathrm{ci}}_k(S^k_{(\hat{i}_k)}), \delta^{\mathrm{ci}}(n_k,n_{*}) -(1-\alpha)/n_k \right\} \\
        & \qquad \qquad - \max \left\{ \Delta_k(S^k_{(\hat{i}_k)}),  \delta^{\mathrm{ci}}(n_k,n_{*}) -(1-\alpha)/n_k\right\} 	\\
	& \qquad \leq 1- \hat{F}^{k}_k(S^k_{(\hat{i}_k)}) - \max\left\{\hDelta^{\mathrm{ci}}_k(S^k_{(\hat{i}_k)}) - \delta^{\mathrm{ci}}(n_k,n_{*}), -(1-\alpha)/n_k \right\}	\\
	& \qquad\qquad + \hat{F}^{k}_k(S^k_{(\hat{i}_k)}) - \tilde{F}^{k}_k(S^k_{(\hat{i}_k)}) - \delta^{\mathrm{ci}}(n_k,n_{*})
 +  \max\left\{\hDelta^{\mathrm{ci}}_k(S^k_{(\hat{i}_k)}) - \Delta_k(S^k_{(\hat{i}_k)}), 0 \right\} \\
	 & \qquad \leq \left[ 1- \frac{\hat{i}_k}{n_k} - \max\left\{\hDelta^{\mathrm{ci}}_k(S^k_{(i)}) - \delta^{\mathrm{ci}}(n_k,n_{*}), -(1-\alpha)/n_k \right\} \right]	\\
	 & \qquad\qquad + \hat{F}^{k}_k(S^k_{(\hat{i}_k)}) - \tilde{F}^{k}_k(S^k_{(\hat{i}_k)}) - \delta^{\mathrm{ci}}(n_k,n_{*}) + \sup_{t \in \mathbb{R}} \max\left\{ \hat{\Delta}^{\mathrm{ci}}_k(t) - \Delta_k(t), 0 \right\},
\end{align*}
using the assumption that $\inf_{t\in \R}\Delta_k(t) \geq \delta^{\mathrm{ci}}(n_k, n_{*}) -(1-\alpha)/n_k$ in the third equality above.
The proof is then completed by proceeding as in the proof of Theorem~\ref{thm:algorithm-correction-ci}, using Lemma~\ref{lemma:dkw-delta-expected-ci}.

\end{proof}

\subsection{Fitting the label contamination model} \label{app:proofs-model-fitting}

\begin{proof}[Proof of Equation~\eqref{eq:Q-M}]

By definition of $Q$ and $\tilde{Q}$ in~\eqref{eq:Q-def}, for any $k,l \in [K]$,
\begin{align*}
  \tilde{Q}_{lk}
  & = \P{ \hat{f}(X) = k \mid \tilde{Y}=l, \hat{f} } \\  
  & = \sum_{l'=1}^{K} \P{ \hat{f}(X) = k, Y=l' \mid \tilde{Y}=l, \hat{f} } \\  
  & = \sum_{l'=1}^{K} \P{ Y=l' \mid \tilde{Y}=l, \hat{f} } \P{\hat{f}(X) = k \mid Y=l', \tilde{Y}=l, \hat{f} } \\
  & = \sum_{l'=1}^{K} M_{ll'} \P{\hat{f}(X) = k \mid Y=l', \hat{f} } \\
  & = \sum_{l'=1}^{K} M_{ll'} Q_{l'k} \\
  & = (MQ)_{lk}.
\end{align*}

\end{proof}

\begin{proof}[Proof of Equation~\eqref{eq:epsilon-estim-RR}]
The proof of~\eqref{eq:epsilon-estim-RR} follows quite directly from Equation (\ref{eq:Q-M}) and the definitions of $\tilde{\psi}$ and $\psi$ in Equation~\eqref{eq:psi-psi-tilde}:
\begin{align*}
  \tilde{\psi} 
  & := \sum_{k=1}^{K} \P{ \hat{f}(X) = k, \tilde{Y}=k \mid  \hat{f} } \\
  & = \sum_{k=1}^{K} \tilde{\rho}_k \tilde{Q}_{kk}
    = \sum_{k=1}^{K} \tilde{\rho}_k \sum_{s=1}^{K} M_{ks} Q_{sk} \\
  & = \sum_{k=1}^{K} \tilde{\rho}_k \left( M_{kk} Q_{kk} + \sum_{s \neq k} M_{ks} Q_{sk} \right) \\
  & = \sum_{k=1}^{K} \tilde{\rho}_k \left( \frac{(1-\epsilon) \rho_k}{(1-\epsilon)\rho_k + \epsilon/K}  Q_{kk} + \sum_{s=1}^{K} \frac{\rho_s \cdot \epsilon/K }{(1-\epsilon)\rho_k + \epsilon/K} Q_{sk} \right) \\
  & = (1-\epsilon) \sum_{k=1}^{K} \rho_k  Q_{kk} + \frac{\epsilon}{K} \sum_{k=1}^{K} \sum_{s=1}^{K} \rho_s Q_{sk}  \\
  & = (1-\epsilon) \sum_{k=1}^{K} \P{ \hat{f}(X) = k, Y=k \mid  \hat{f} } + \frac{\epsilon}{K} \sum_{k=1}^{K} \sum_{s=1}^{K} \P{ \hat{f}(X) = k, Y=s \mid  \hat{f} }  \\
  & = (1-\epsilon) \psi + \frac{\epsilon}{K}.
\end{align*}
\end{proof}

\begin{proof}[Proof of Equation~\eqref{eq:zeta-upp-BRR}]
By the definition of $\hat{\zeta}_k$ in~\eqref{eq:def-zeta-k},
\begin{align*}
  \hat{\zeta}_k
  & := \max_{l \neq k} \left( \hat{V}^{\mathrm{upp}}_{kl} - V_{kl} \right) - \min_{l \neq k} \left( \hat{V}^{\mathrm{upp}}_{kl} - V_{kl} \right) \\
  & = \Bigg| \left( \frac{\xi}{K} \cdot \frac{ 1 + \nu(1 + 2 \xi)}{1 + \nu \xi } - \frac{\hat{\xi}^{\mathrm{low}}}{K} \cdot \frac{ 1 + \hat{\nu}^{\mathrm{low}} ( 1 + 2 \hat{\xi}^{\mathrm{low}})}{1 + \hat{\nu}^{\mathrm{low}} \hat{\xi}^{\mathrm{low}}} \right) \\
  & \qquad - \left( \frac{\xi}{K} \cdot \frac{ 1 - \nu}{1 + \nu \xi } - \frac{\hat{\xi}^{\mathrm{low}}}{K} \cdot \frac{ 1 - \hat{\nu}^{\mathrm{upp}}}{1 + \hat{\nu}^{\mathrm{upp}} \hat{\xi}^{\mathrm{low}}} \right) \Bigg|  \\
  & = \left| \frac{\xi}{K}  \cdot \frac{ 2 \nu(1 + \xi) }{1 + \nu \xi } - \frac{\hat{\xi}^{\mathrm{low}}}{K}  \left( \frac{ 1 + \hat{\nu}^{\mathrm{low}} ( 1 + 2 \hat{\xi}^{\mathrm{low}})}{1 + \hat{\nu}^{\mathrm{low}} \hat{\xi}^{\mathrm{low}}} - \frac{ 1 - \hat{\nu}^{\mathrm{upp}}}{1 + \hat{\nu}^{\mathrm{upp}} \hat{\xi}^{\mathrm{low}}}  \right) \right|  \\
  & = \left| \frac{\xi}{K}  \cdot \frac{ 2 \nu(1 + \xi) }{1 + \nu \xi } - \frac{\hat{\xi}^{\mathrm{low}}}{K} \cdot \frac{ \left( 1 + \hat{\xi}^{\mathrm{low}} \right) \left( \hat{\nu}^{\mathrm{low}} + \hat{\nu}^{\mathrm{upp}} + 2 \hat{\nu}^{\mathrm{low}} \hat{\nu}^{\mathrm{upp}} \hat{\xi}^{\mathrm{low}} \right) }{\left( 1 + \hat{\nu}^{\mathrm{low}} \hat{\xi}^{\mathrm{low}} \right) \left( 1 + \hat{\nu}^{\mathrm{upp}} \hat{\xi}^{\mathrm{low}} \right)} \right|.
\end{align*}
Then, applying the triangle inequality gives us:
\begin{align*}
  \hat{\zeta}_k
& \leq \left| \frac{\xi}{K}  \cdot \frac{ 2 \nu(1 + \xi) }{1 + \nu \xi } - \frac{\hat{\xi}^{\mathrm{low}}}{K}  \cdot \frac{ 2 \hat{\nu}^{\mathrm{low}}(1 + \hat{\xi}^{\mathrm{low}}) }{1 + \hat{\nu}^{\mathrm{low}} \hat{\xi}^{\mathrm{low}} } \right| \\
  & \qquad + \frac{\hat{\xi}^{\mathrm{low}}}{K} \left| \frac{ 2 \hat{\nu}^{\mathrm{low}}(1 + \hat{\xi}^{\mathrm{low}}) }{1 + \hat{\nu}^{\mathrm{low}} \hat{\xi}^{\mathrm{low}} } - \frac{ \left( 1 + \hat{\xi}^{\mathrm{low}} \right) \left( \hat{\nu}^{\mathrm{low}} + \hat{\nu}^{\mathrm{upp}} + 2 \hat{\nu}^{\mathrm{low}} \hat{\nu}^{\mathrm{upp}} \hat{\xi}^{\mathrm{low}} \right) }{\left( 1 + \hat{\nu}^{\mathrm{low}} \hat{\xi}^{\mathrm{low}} \right) \left( 1 + \hat{\nu}^{\mathrm{upp}} \hat{\xi}^{\mathrm{low}} \right)} \right|  \\
  & = \left| \frac{\xi}{K}  \cdot \frac{ 2 \nu(1 + \xi) }{1 + \nu \xi } - \frac{\hat{\xi}^{\mathrm{low}}}{K}  \cdot \frac{ 2 \hat{\nu}^{\mathrm{low}}(1 + \hat{\xi}^{\mathrm{low}}) }{1 + \hat{\nu}^{\mathrm{low}} \hat{\xi}^{\mathrm{low}} } \right| 
    + \frac{\hat{\xi}^{\mathrm{low}}}{K} \cdot \frac{\left( \hat{\nu}^{\mathrm{upp}} - \hat{\nu}^{\mathrm{low}} \right) \left( 1 + \hat{\xi}^{\mathrm{low}} \right) }{\left( 1 + \hat{\nu}^{\mathrm{low}} \hat{\xi}^{\mathrm{low}} \right) \left( 1 + \hat{\nu}^{\mathrm{upp}} \hat{\xi}^{\mathrm{low}} \right)}.
\end{align*}
Next, the monotonicity of the function $(\nu, \epsilon) \mapsto (\nu(1+\epsilon))/((1+\nu\epsilon))$ in both inputs allows us to conclude that:
\begin{align*}
  \hat{\zeta}_k
& \leq \left| \frac{\hat{\xi}^{\mathrm{upp}}}{K}  \cdot \frac{ 2 \hat{\nu}^{\mathrm{upp}}(1 + \hat{\xi}^{\mathrm{upp}}) }{1 + \hat{\nu}^{\mathrm{upp}} \hat{\xi}^{\mathrm{upp}} } - \frac{\hat{\xi}^{\mathrm{low}}}{K}  \cdot \frac{ 2 \hat{\nu}^{\mathrm{low}}(1 + \hat{\xi}^{\mathrm{low}}) }{1 + \hat{\nu}^{\mathrm{low}} \hat{\xi}^{\mathrm{low}} } \right| \\
  & \qquad  + \frac{\hat{\xi}^{\mathrm{low}}}{K} \cdot \frac{\left( \hat{\nu}^{\mathrm{upp}} - \hat{\nu}^{\mathrm{low}} \right) \left( 1 + \hat{\xi}^{\mathrm{low}} \right) }{\left( 1 + \hat{\nu}^{\mathrm{low}} \hat{\xi}^{\mathrm{low}} \right) \left( 1 + \hat{\nu}^{\mathrm{upp}} \hat{\xi}^{\mathrm{low}} \right)}  \\
  & = \frac{2}{K} \cdot \frac{\left( \hat{\nu}^{\mathrm{upp}} - \hat{\nu}^{\mathrm{low}} \right) + \left[ \hat{\nu}^{\mathrm{upp}} \left( \hat{\xi}^{\mathrm{upp}}\right)^2 - \hat{\nu}^{\mathrm{low}} \left( \hat{\xi}^{\mathrm{low}}\right)^2 \right] + \hat{\nu}^{\mathrm{low}} \hat{\nu}^{\mathrm{upp}} \hat{\xi}^{\mathrm{low}} \hat{\xi}^{\mathrm{upp}} \left( \hat{\xi}^{\mathrm{upp}} - \hat{\xi}^{\mathrm{low}} \right) }{\left( 1 + \hat{\nu}^{\mathrm{upp}} \hat{\xi}^{\mathrm{upp}} \right) \left( 1 + \hat{\nu}^{\mathrm{low}} \hat{\xi}^{\mathrm{low}} \right) }  \\
  & \qquad  + \frac{\hat{\xi}^{\mathrm{low}}}{K} \cdot \frac{\left( \hat{\nu}^{\mathrm{upp}} - \hat{\nu}^{\mathrm{low}} \right) \left( 1 + \hat{\xi}^{\mathrm{low}} \right) }{\left( 1 + \hat{\nu}^{\mathrm{low}} \hat{\xi}^{\mathrm{low}} \right) \left( 1 + \hat{\nu}^{\mathrm{upp}} \hat{\xi}^{\mathrm{low}} \right)}.
\end{align*}
\end{proof}

\begin{proof}[Proof of Equation~\eqref{eq:BRR-estimating-system}]
We begin by proving the first equation, namely
\begin{align*}
    \tilde{\psi} & = (1-\epsilon) \psi + \frac{\epsilon}{K}  + \frac{\epsilon \nu}{K} \left( 2 \phi - 1 \right).
\end{align*}

It follows directly from the general estimating equation (\ref{eq:Q-M}) that 
\begin{align*}
  \tilde{\psi} 
  & := \sum_{k=1}^{K} \P{ \hat{f}(X) = k, \tilde{Y}=k \mid  \hat{f} } \\
  & = \frac{1}{K} \sum_{k=1}^{K}  \tilde{Q}_{kk} \\
  & = \frac{1}{K} \sum_{k=1}^{K}  \sum_{s=1}^{K} M_{ks} Q_{sk} \\
  & = \frac{1}{K} \sum_{k=1}^{K}  \left( M_{kk} Q_{kk} + \sum_{s \neq k} M_{ks} Q_{sk} \right) \\
  & = \frac{1}{K} \sum_{k=1}^{K}  \left( (1-\epsilon) Q_{kk} + \sum_{s \in \mathcal{B}_k} \frac{\epsilon}{K} \left( 1 + \nu \right) Q_{sk} + \sum_{s \in \mathcal{B}^{\mathrm{c}}_k} \frac{\epsilon}{K} \left( 1 - \nu \right) Q_{sk} \right) \\
  & = (1-\epsilon) \psi + \frac{\epsilon(1+\nu)}{K^2} \sum_{k=1}^{K} \sum_{s \in \mathcal{B}_k} Q_{sk} + \frac{\epsilon(1-\nu)}{K^2} \sum_{k=1}^{K} \sum_{s \in \mathcal{B}^{\mathrm{c}}_k} Q_{sk} \\
  & = (1-\epsilon) \psi + \frac{\epsilon(1+\nu)}{K} \sum_{k=1}^{K} \sum_{s \in \mathcal{B}_k} \P{ \hat{f}(X) = k, Y=s \mid  \hat{f} } \\
  & \qquad + \frac{\epsilon(1-\nu)}{K} \sum_{k=1}^{K} \sum_{s \in \mathcal{B}^{\mathrm{c}}_k} \P{ \hat{f}(X) = k, Y=s \mid  \hat{f} } \\
  & = (1-\epsilon) \psi + \frac{\epsilon}{K} \sum_{k=1}^{K} \sum_{s=1}^{K} \P{ \hat{f}(X) = k, Y=s \mid  \hat{f} } \\
  & \qquad + \frac{\epsilon \nu}{K} \sum_{k=1}^{K} \left( \sum_{s \in \mathcal{B}_k} \P{ \hat{f}(X) = k, Y=s \mid  \hat{f} } - \sum_{s \in \mathcal{B}^{\mathrm{c}}_k} \P{ \hat{f}(X) = k, Y=s \mid  \hat{f} } \right)\\
  & = (1-\epsilon) \psi + \frac{\epsilon}{K}  + \frac{\epsilon \nu}{K} \left( 2 \sum_{k=1}^{K} \sum_{s \in \mathcal{B}_k} \P{ \hat{f}(X) = k, Y=s \mid  \hat{f} } - 1 \right)\\
  & = (1-\epsilon) \psi + \frac{\epsilon}{K}  + \frac{\epsilon \nu}{K} \left( 2 \phi - 1 \right),
\end{align*}
where
\begin{align*}
  \phi 
  & := \sum_{k=1}^{K} \sum_{l \in \mathcal{B}_k} \P{ \hat{f}(X) = l, Y=k \mid  \hat{f} }.
\end{align*}
This completes the first part of the proof. 

Let us now prove the second equation, namely
\begin{align*}
  \tilde{\phi} 
  & = \phi - \epsilon (1-\nu) \left( \phi - \frac{1}{2} \right).
\end{align*}

It follows directly from the general estimating equation (\ref{eq:Q-M}) that 
\begin{align*}
  \tilde{\phi} 
  & := \sum_{k=1}^{K} \sum_{l \in \mathcal{B}_k} \P{ \hat{f}(X) = l, \tilde{Y}=k \mid  \hat{f} } \\
  & = \frac{1}{K} \sum_{k=1}^{K} \sum_{l \in \mathcal{B}_k} \tilde{Q}_{kl} \\
  & = \frac{1}{K} \sum_{k=1}^{K} \sum_{l \in \mathcal{B}_k} \sum_{s=1}^{K} M_{ks} Q_{sl} \\
  & = \frac{1}{K} \sum_{k=1}^{K} \sum_{l \in \mathcal{B}_k} \left( M_{kk} Q_{kl} + \sum_{s \neq k} M_{ks} Q_{sl} \right) \\
  & = \frac{1}{K} \sum_{k=1}^{K} \sum_{l \in \mathcal{B}_k} \left( (1-\epsilon) Q_{kl} + \sum_{s \in \mathcal{B}_k} \frac{\epsilon}{K} \left( 1 + \nu \right) Q_{sl} + \sum_{s \in \mathcal{B}^{\mathrm{c}}_k} \frac{\epsilon}{K} \left( 1 - \nu \right) Q_{sk}\right) \\
  & = (1-\epsilon)  \sum_{k=1}^{K} \sum_{l \in \mathcal{B}_k} \frac{Q_{kl}}{K} + \frac{\epsilon \left( 1 + \nu \right)}{K} \sum_{k=1}^{K} \sum_{l \in \mathcal{B}_k} \sum_{s \in \mathcal{B}_k} \frac{Q_{sl}}{K} + \frac{\epsilon \left( 1 - \nu \right)}{K} \sum_{k=1}^{K} \sum_{l \in \mathcal{B}_k} \sum_{s \in \mathcal{B}^{\mathrm{c}}_k} \frac{Q_{sl}}{K} \\
   & = (1-\epsilon) \phi + \frac{\epsilon \left( 1 + \nu \right)}{K} \frac{K}{2} \sum_{k=1}^{K} \sum_{s \in \mathcal{B}_k} \frac{Q_{sk}}{K} + \frac{\epsilon \left( 1 - \nu \right)}{K} \frac{K}{2} \sum_{k=1}^{K} \sum_{s \in \mathcal{B}^{\mathrm{c}}_k} \frac{Q_{sk}}{K} \\
   & = (1-\epsilon) \phi + \frac{\epsilon \left( 1 + \nu \right)}{2} \phi + \frac{\epsilon \left( 1 - \nu \right)}{2} (1-\phi).
\end{align*}
This concludes the second part of the proof.

\end{proof}

\subsection{Extensions of preliminary theoretical results}

\begin{proof}[Proof of Theorem~\ref{thm:coverage-marginal}]
The proof strategy is similar to that of Theorem~\ref{thm:coverage-lab-cond}.
  By definition of the conformity score function in~\eqref{eq:conf-scores}, $Y_{n+1} \in \hat{C}(X_{n+1})$ if and only if $\hat{s}(X_{n+1},Y_{n+1}) \leq \hat{\tau}_{Y_{n+1}}$.
  Therefore,
  \begin{align*}
    & \P{Y_{n+1} \in \hat{C}(X_{n+1})} \\
    & \qquad = \P{\tilde{Y}_{n+1} \in \hat{C}(X_{n+1})} + \left( \P{Y_{n+1} \in \hat{C}(X_{n+1})} - \P{\tilde{Y}_{n+1} \in \hat{C}(X_{n+1})} \right).
  \end{align*}
The proof is then completed by noting that the second term on the right-hand-side above can be written as:
  \begin{align*}
    & \P{Y_{n+1} \in \hat{C}(X_{n+1})} - \P{\tilde{Y}_{n+1} \in \hat{C}(X_{n+1})} \\
    & \qquad = \sum_{k=1}^{K} \left( \rho_k \cdot \P{Y_{n+1} \in \hat{C}(X_{n+1}) \mid Y_{n+1} = k} - \tilde{\rho}_k \cdot  \P{\tilde{Y}_{n+1} \in \hat{C}(X_{n+1}) \mid \tilde{Y}_{n+1} = k} \right) \\
    & \qquad = \sum_{k=1}^{K} \left( \rho_k \cdot \P{\hat{s}(X_{n+1}, k) \leq \hat{\tau} \mid Y_{n+1} = k} - \tilde{\rho}_k \cdot  \P{ \hat{s}(X_{n+1}, k) \leq \hat{\tau}  \mid \tilde{Y}_{n+1} = k} \right) \\
    & \qquad = \EV{ \sum_{k=1}^{K} \left[ \rho_k F_k^{k}(\hat{\tau}) - \tilde{\rho}_k \tF_k^{k}(\hat{\tau})\right] }
     = \EV{ \Delta(\hat{\tau})}.
  \end{align*}
\end{proof}

\begin{proof}[Proof of Corollary~\ref{cor:coverage-marg}]
  By Theorem~\ref{thm:coverage-marginal}, it suffices to prove $\EV{\Delta(t)} \geq 0$ for all $t \in \mathbb{R}$.
  To establish that, note that combining~\eqref{eq:delta-marg} with Proposition~\ref{prop:indep-linear} gives:
  \begin{align*}
    \Delta(t)
    & = \sum_{k=1}^{K} \left[ \rho_k F_k^{k}(t) - \tilde{\rho}_k \tF_k^{k}(t)\right]
     = \sum_{k=1}^{K} \left[ \rho_k  F_k^{k}(t) - \sum_{l=1}^{K} M_{kl}  \tilde{\rho}_k  F_l^{k}(t)  \right].
  \end{align*}
  To simplify the notation in the following, define
  \begin{align*}
    W_{kl} & := \mathbb{P}[\tilde{Y}=k \mid  X, Y=l]
     = M_{kl} \frac{\tilde{\rho}_k}{\rho_l}.
  \end{align*}
  Then,
  \begin{align*}
    \Delta(t)
    & = \sum_{l=1}^{K} \rho_l  F_l^{l}(t) -  \sum_{l=1}^{K} \rho_l  \sum_{k=1}^{K}	W_{kl} F_l^{k}(t) \\
    & = \sum_{l=1}^{K} \rho_l  F_l^{l}(t) -  \sum_{l=1}^{K} \rho_l  \left[ W_{ll} F_l^{l}(t) + \sum_{k \neq l} W_{kl} F_l^{k}(t) \right] \\
    & \geq \sum_{l=1}^{K} \rho_l  F_l^{l}(t) -  \sum_{l=1}^{K} \rho_l  \left[ W_{ll} F_l^{l}(t) + \left(\sum_{k \neq l} W_{kl}\right) \max_{k \neq l} F_l^{k}(t) \right] \\
    & = \sum_{l=1}^{K} \rho_l  F_l^{l}(t) -  \sum_{l=1}^{K} \rho_l  \left[ W_{ll} F_l^{l}(t) + \left(\sum_{k=1}^{K} W_{kl} \right) \max_{k \neq l} F_l^{k}(t) - W_{ll} \max_{k \neq l} F_l^{k}(t) \right] \\
    & = \sum_{l=1}^{K} \rho_l  F_l^{l}(t) -  \sum_{l=1}^{K} \rho_l  \left[ W_{ll} F_l^{l}(t) + \max_{k \neq l} F_l^{k}(t) - W_{ll} \max_{k \neq l} F_l^{k}(t) \right] \\
    & = \sum_{l=1}^{K} \rho_l (1-W_{ll}) \left[ F_l^{l}(t) - \max_{k\neq l} F_l^{k}(t) \right] \geq 0,
  \end{align*}
  where the last inequality is given by~\eqref{eq:assump_scores-cond-marg}.
This implies that $\Delta(\htau) \geq 0$ almost-surely.
\end{proof}

\subsection{Prediction sets with marginal coverage}

\begin{proof}[Proof of Theorem~\ref{thm:algorithm-correction-marg}]
The proof strategy is similar to those of Theorem~\ref{thm:coverage-marginal} and Theorem~\ref{thm:algorithm-correction}.
  By definition of the conformity score function in~\eqref{eq:conf-scores}, the event $Y_{n+1} \in \hat{C}^{\mathrm{marg}}(X_{n+1})$ occurs if and only if $\hat{s}(X_{n+1},Y_{n+1}) \leq \hat{\tau}^{\mathrm{marg}}$.
Further, note that
  \begin{align*}
    \tilde{F}^{k}(t) = \sum_{k=1}^{K} \tilde{\rho}_{k} \tilde{F}^{k}_k(t).
  \end{align*}
  Therefore, the probability of a miscoverage event conditional on the data in $\mathcal{D}$ can be decomposed as:
  \begin{align*}
    & \P{Y_{n+1} \notin \hat{C}^{\mathrm{marg}}(X_{n+1}) \mid \mathcal{D}} \\
    & \qquad = \P{\tilde{Y}_{n+1} \notin \hat{C}^{\mathrm{marg}}(X_{n+1}) \mid \mathcal{D}} \\
    & \qquad \qquad + \left( \P{Y_{n+1} \notin \hat{C}^{\mathrm{marg}}(X_{n+1}) \mid \mathcal{D}} - \P{\tilde{Y}_{n+1} \notin \hat{C}^{\mathrm{marg}}(X_{n+1}) \mid \mathcal{D}} \right) \\
    & \qquad = \sum_{k=1}^{K} \tilde{\rho}_{k} \left[ 1 - \tilde{F}^{k}_k(S_{(\hat{i}^{\mathrm{marg}})}) \right] \\
    & \qquad \qquad + \sum_{k=1}^{K}  \rho_{k} \P{ \hat{s}(X_{n+1},k) > S^k_{(\hat{i}^{\mathrm{marg}})} \mid \mathcal{D}, Y_{n+1}=k } \\
    & \qquad \qquad - \sum_{k=1}^{K} \tilde{\rho}_{k} \P{ \hat{s}(X_{n+1},k) > S^k_{(\hat{i}^{\mathrm{marg}})} \mid \mathcal{D}, \tilde{Y}_{n+1}=k } \\
    & \qquad = \sum_{k=1}^{K} \tilde{\rho}_{k} \left[ 1 - \tilde{F}^{k}_k(S_{(\hat{i}^{\mathrm{marg}})}) \right]
      + \sum_{k=1}^{K} \left[ \rho_{k} - \rho_{k} F^{k}_k(S_{(\hat{i}^{\mathrm{marg}})}) -  \tilde{\rho}_{k} + \tilde{\rho}_{k}  \tilde{F}^{k}_k(S_{(\hat{i}^{\mathrm{marg}})}) \right]\\
    & \qquad = 1 - \sum_{k=1}^{K} \tilde{\rho}_{k} \tilde{F}^{k}_k(S_{(\hat{i}^{\mathrm{marg}})}) - \Delta(S_{(\hat{i}^{\mathrm{marg}})}) \\
    & \qquad = 1 - \tilde{F}(S_{(\hat{i}^{\mathrm{marg}})}) - \Delta(S_{(\hat{i}^{\mathrm{marg}})}) \\
     & \qquad = 1 - \hat{F}(S_{(\hat{i}^{\mathrm{marg}})}) - \hat{\Delta}(S_{(\hat{i}^{\mathrm{marg}})}) \\
     & \qquad \qquad + \left[ \hat{F}(S_{(\hat{i}^{\mathrm{marg}})}) - \tilde{F}(S_{(\hat{i}^{\mathrm{marg}})}) \right] + \left[ \hat{\Delta}(S_{(\hat{i}^{\mathrm{marg}})}) - \Delta(S_{(\hat{i}^{\mathrm{marg}})}) \right] \\
     & \qquad = \left\{ 1 - \frac{\hat{i}^{\mathrm{marg}}}{n_{\mathrm{cal}}} - \hat{\Delta}(S_{(\hat{i}^{\mathrm{marg}})}) + \delta^{\mathrm{marg}}(n_{\mathrm{cal}},n_{*}) \right\} - \delta^{\mathrm{marg}}(n_{\mathrm{cal}},n_{*})\\
     & \qquad \qquad + \left[ \hat{F}(S_{(\hat{i}^{\mathrm{marg}})}) - \tilde{F}(S_{(\hat{i}^{\mathrm{marg}})}) \right] + \left[ \hat{\Delta}(S_{(\hat{i}^{\mathrm{marg}})}) - \Delta(S_{(\hat{i}^{\mathrm{marg}})}) \right] \\
     & \qquad \leq \sup_{i \in \hat{\mathcal{I}}^{\mathrm{marg}}} \left\{ 1 - \frac{i}{n_{\mathrm{cal}}} - \hat{\Delta}(S_{(i)}) + \delta^{\mathrm{marg}}(n_{\mathrm{cal}},n_{*}) \right\} - \delta^{\mathrm{marg}}(n_{\mathrm{cal}},n_{*})\\
     & \qquad \qquad + \left[ \hat{F}(S_{(\hat{i}^{\mathrm{marg}})}) - \tilde{F}(S_{(\hat{i}^{\mathrm{marg}})}) \right] + \sup_{t \in \mathbb{R}} \left[ \hat{\Delta}(t) - \Delta(t) \right].
  \end{align*}
By definition of $\hat{\mathcal{I}}^{\mathrm{marg}}$, for all $i \in \hat{\mathcal{I}}^{\mathrm{marg}}$,
  \begin{align*}
    1 - \frac{i}{n_{\mathrm{cal}}} - \hat{\Delta}(S_{(i)}) + \delta^{\mathrm{marg}}(n_{\mathrm{cal}},n_{*}) \leq \alpha,
  \end{align*}
which implies a.s.
  \begin{align*}
    \sup_{i \in \hat{\mathcal{I}}^{\mathrm{marg}}} \left\{ 1 - \frac{i}{n_{\mathrm{cal}}} - \hat{\Delta}(S_{(i)}) + \delta^{\mathrm{marg}}(n_{\mathrm{cal}},n_{*}) \right\} \leq \alpha.
  \end{align*}
Therefore,
  \begin{align*}
    & \P{Y_{n+1} \in \hat{C}^{\mathrm{marg}}(X_{n+1})} \\
    & \qquad \leq \alpha + \EV{ \sup_{t \in \mathbb{R}} \left[ \hat{\Delta}(t) - \Delta(t) \right] }
      + \EV{ \hat{F}(S_{(\hat{i}^{\mathrm{marg}})}) - \tilde{F}(S_{(\hat{i}^{\mathrm{marg}})}) }  - \delta^{\mathrm{marg}}(n_{\mathrm{cal}},n_{*}).
  \end{align*}
The first expected value on the right-hand-side above can be bound using the DKW inequality, as made precise by Lemma~\ref{lemma:dkw-delta-expected-marg}.
This leads to:
  \begin{align*}
    \P{Y_{n+1} \in \hat{C}^{\mathrm{marg}}(X_{n+1})}
    & \leq \alpha  + \EV{ \hat{F}(S_{(\hat{i}^{\mathrm{marg}})}) - \tilde{F}(S_{(\hat{i}^{\mathrm{marg}})}) }  - \delta^{\mathrm{marg}}(n_{\mathrm{cal}},n_{*}) \\
    & \qquad + \frac{2 \max_{k \in [K]} \sum_{l\neq k} |V_{kl}| + \sum_{k=1}^{K} | \rho_k - \tilde{\rho}_k  |}{\sqrt{n_{*}}} \\
    & \qquad \cdot \min \left\{ K^2 \sqrt{\frac{\pi}{2}} , \frac{1}{\sqrt{n_{*}}} + \sqrt{\frac{\log(2K^2) + \log(n_{*})}{2}} \right\}.
  \end{align*}
Finally, the remaining expected value can be bound as the corresponding term in the proof of Theorem~\ref{thm:algorithm-correction}.
Let $U_1, \ldots, U_{n_{\mathrm{cal}}}$ be i.i.d.~uniform random variables on $[0,1]$, and denote their order statistics as $U_{(1)}, \ldots, U_{(n_{\mathrm{cal}})}$.
Then,
\begin{align*}
  \hat{F}(S_{(\hat{i}^{\mathrm{marg}})}) - \tilde{F}(S_{(\hat{i}^{\mathrm{marg}})})
  \overset{d}{=} \frac{\hat{i}^{\mathrm{marg}}}{n_{\mathrm{cal}}} - U_{(\hat{i}^{\mathrm{marg}})}.
\end{align*}
Therefore, it follows from~\eqref{eq:define-cn} that
\begin{align}
  \EV{ \hat{F}(S_{(\hat{i}^{\mathrm{marg}})}) - \tilde{F}(S_{(\hat{i}^{\mathrm{marg}})}) }
  \leq \EV{ \sup_{i \in [n_{\mathrm{cal}}]} \left\{ \frac{i}{n_{\mathrm{cal}}} - U_{(i)} \right\} } = c(n_{\mathrm{cal}}).
  \label{eq:tighter_cdf_bound-marg}
\end{align}
  \begin{align*}
    \P{Y_{n+1} \in \hat{C}^{\mathrm{marg}}(X_{n+1})}
    & \leq \alpha  + c(n_{\mathrm{cal}})  - \delta^{\mathrm{marg}}(n_{\mathrm{cal}},n_{*}) \\
    & \qquad + \frac{2 \max_{k \in [K]} \sum_{l\neq k} |V_{kl}| + \sum_{k=1}^{K} | \rho_k - \tilde{\rho}_k  |}{\sqrt{n_{*}}} \\
    & \qquad \cdot \min \left\{ K^2 \sqrt{\frac{\pi}{2}} , \frac{1}{\sqrt{n_{*}}} + \sqrt{\frac{\log(2K^2) + \log(n_{*})}{2}} \right\} \\
    & = \alpha.
  \end{align*}

\end{proof}

\begin{lemma}\label{lemma:dkw-delta-expected-marg}
Under the assumptions of Theorem~\ref{thm:algorithm-correction-marg},
\begin{align*}
  \EV{ \sup_{t \in \mathbb{R}} | \hat{\Delta}(t) - \Delta(t)  | }
  & \leq \frac{2 \max_{k \in [K]} \sum_{l\neq k} |V_{kl}| + \sum_{k=1}^{K} | \rho_k - \tilde{\rho}_k  |}{\sqrt{n_{*}}} \\
  & \qquad \cdot \min \left\{ K^2 \sqrt{\frac{\pi}{2}} , \frac{1}{\sqrt{n_{*}}} + \sqrt{\frac{\log(2K^2) + \log(n_{*})}{2}} \right\}.
\end{align*}

\end{lemma}

\begin{proof}[Proof of Lemma~\ref{lemma:dkw-delta-expected-marg}]
Note that, for any $k \in [K]$ and $t \in \mathbb{R}$,
\begin{align*}
  \hat{\Delta}(t) - \Delta(t)
  & = \sum_{k=1}^{K} \left( \rho_k V_{kk}- \tilde{\rho}_k \right) \left[ \hat{F}_k^{k}(t) - \tF^{k}_k(t)  \right]
    + \sum_{k=1}^{K} \rho_k \sum_{l\neq k} V_{kl} \left[ \hat{F}_l^{k}(t) - \tF^{k}_l(t)  \right]    \\
  & = \sum_{k=1}^{K} \rho_k \left(  V_{kk}- 1 \right) \left[ \hat{F}_k^{k}(t) - \tF^{k}_k(t)  \right] + \sum_{k=1}^{K} \rho_k \sum_{l\neq k} V_{kl} \left[ \hat{F}_l^{k}(t) - \tF^{k}_l(t)  \right] \\
  & \qquad \qquad + \sum_{k=1}^{K} \left( \rho_k - \tilde{\rho}_k  \right) \left[ \hat{F}_k^{k}(t) - \tF^{k}_k(t)  \right] \\
  & = \sum_{k=1}^{K} \rho_k \sum_{l\neq k} V_{kl} \left\{ \left[ \hat{F}_l^{k}(t) - \tF^{k}_l(t)  \right] - \left[ \hat{F}_k^{k}(t) - \tF^{k}_k(t)  \right] \right\}\\
  & \qquad \qquad + \sum_{k=1}^{K} \left( \rho_k - \tilde{\rho}_k  \right) \left[ \hat{F}_k^{k}(t) - \tF^{k}_k(t)  \right],
\end{align*}
where in the last equality we used the fact that $V$ has row sums equal to 1 because it is the inverse of $M$, which has row sums equal to 1.
Therefore,
\begin{align*}
  |\hat{\Delta}(t) - \Delta(t)|
  & \leq 2 \sum_{k=1}^{K} \rho_k \left( \sum_{l\neq k} |V_{kl}| \right) \left( \max_{l \neq k} | \hat{F}_l^{k}(t) - \tF^{k}_l(t) | \right) \\
  & \qquad \qquad + \sum_{k=1}^{K} | \rho_k - \tilde{\rho}_k  | \left( \max_{k \in [K]} | \hat{F}_k^{k}(t) - \tF^{k}_k(t) | \right) \\
  & \leq 2 \left( \max_{k \in [K]} \sum_{l \neq k} |V_{kl}| \right) \left( \max_{k,l \in [K] \,:\, l \neq k} | \hat{F}_l^{k}(t) - \tF^{k}_l(t) | \right) \sum_{k=1}^{K} \rho_k  \\
  & \qquad \qquad + \left( \max_{k \in [K]} | \hat{F}_k^{k}(t) - \tF^{k}_k(t) | \right) \sum_{k=1}^{K} | \rho_k - \tilde{\rho}_k  | \\
  & = 2 \left( \max_{k \in [K]} \sum_{l \neq k} |V_{kl}| \right) \left( \max_{k,l \in [K] \,:\, l \neq k} | \hat{F}_l^{k}(t) - \tF^{k}_l(t) | \right)  \\
  & \qquad \qquad + \left( \sum_{k=1}^{K} | \rho_k - \tilde{\rho}_k  | \right) \left( \max_{k \in [K]} | \hat{F}_k^{k}(t) - \tF^{k}_k(t) | \right) \\
  & \leq  \left( 2 \max_{k \in [K]} \sum_{l \neq k} |V_{kl}| + \sum_{k=1}^{K} | \rho_k - \tilde{\rho}_k  | \right) \left( \max_{k,l \in [K] } | \hat{F}_l^{k}(t) - \tF^{k}_l(t) | \right),
\end{align*}
and thus
\begin{align*}
  \EV{ |\hat{\Delta}(t) - \Delta(t)| }
  & \leq \left( 2 \max_{k \in [K]} \sum_{l\neq k} |V_{kl}| + \sum_{k=1}^{K} | \rho_k - \tilde{\rho}_k  | \right) \EV{ \max_{k,l \in [K] } | \hat{F}_l^{k}(t) - \tF^{k}_l(t) | }.
\end{align*}

As in the proof of Lemma~\ref{lemma:dkw-delta-expected}, define
\begin{align*}
  & l_{\min} := \arg\min_{l \in [K]} n_j.
  & n_{*} := \min_{l \in [K]} n_l.
\end{align*}
Then, combining the DKW inequality with a union bound leads to:
\begin{align*}
  \EV{ \max_{k,l \in [K]} \sup_{t \in \mathbb{R}} | \hat{F}_l^{k}(t) - \tF^{k}_{l}(t)   | }
  & \leq K^2 \sqrt{\frac{\pi}{2n_{*}}}.
\end{align*}

Similarly, the DKW inequality also implies that, for any $\eta > 0$,
\begin{align*}
  & \P{ \sup_{t \in \mathbb{R}} |\hat{\Delta}(t) - \Delta(t) | > \left( 2 \max_{k \in [K]} \sum_{l\neq k} |V_{kl}| + \sum_{k=1}^{K} | \rho_k - \tilde{\rho}_k  | \right)  \sqrt{\frac{\log(2K^2) + \log(1/\eta)}{2 n_{*}}} } \leq \eta.
\end{align*}
see Lemma~\ref{lemma:dkw-delta-prob-marg}.
Therefore, setting $\eta = 1/n_{*}$, we obtain
\begin{align*}
  & \EV{ \sup_{t \in \mathbb{R}} |\hat{\Delta}(t) - \Delta(t) | }
    \leq \left[ 2 \max_{k \in [K]} \sum_{l\neq k} |V_{kl}| + \sum_{k=1}^{K} | \rho_k - \tilde{\rho}_k  | \right] \left[ \sqrt{\frac{\log(2K^2) + \log(n_{*})}{2 n_{*}}} + \frac{1}{n_{*}} \right].
\end{align*}
This concludes the proof that
\begin{align*}
  \EV{ \sup_{t \in \mathbb{R}} | \hat{\Delta}(t) - \Delta(t)  | }
  & \leq \frac{2 \max_{k \in [K]} \sum_{l\neq k} |V_{kl}| + \sum_{k=1}^{K} | \rho_k - \tilde{\rho}_k  |}{\sqrt{n_{*}}} \\
  & \qquad \cdot \min \left\{ K^2 \sqrt{\frac{\pi}{2}} , \frac{1}{\sqrt{n_{*}}} + \sqrt{\frac{\log(2K^2) + \log(n_{*})}{2}} \right\}.
\end{align*}

\end{proof}

\begin{lemma}\label{lemma:dkw-delta-prob-marg}
Under the assumptions of Theorem~\ref{thm:algorithm-correction-marg}, for any $\eta > 0$,
\begin{align*}
  & \P{ \sup_{t \in \mathbb{R}} |\hat{\Delta}(t) - \Delta(t) | > \left( 2 \max_{k \in [K]} \sum_{l\neq k} |V_{kl}| + \sum_{k=1}^{K} | \rho_k - \tilde{\rho}_k  | \right)  \sqrt{\frac{\log(2K^2) + \log(1/\eta)}{2 n_{*}}} } \leq \eta.
\end{align*}
\end{lemma}

\begin{proof}[Proof of Lemma~\ref{lemma:dkw-delta-prob-marg}]
It follows from the definitions of $\hat{\Delta}(t)$ and $\Delta(t)$, and from the DKW inequality, that, for any $\eta > 0$ and $k_0 \in [K]$,
\begin{align*}
  & \P{ \sup_{t \in \mathbb{R}} |\hat{\Delta}(t) - \Delta(t) | > \left( 2 \max_{k \in [K]} \sum_{l\neq k} |V_{kl}| + \sum_{k=1}^{K} | \rho_k - \tilde{\rho}_k  | \right)  \sqrt{\frac{\log(2K^2) + \log(1/\eta)}{2 n_{*}}} } \\
  & \qquad \leq \P{ \max_{k,l \in [K]} | \hat{F}_l^{k}(t) - \tF^{k}_{l}(t)| > \sqrt{\frac{\log(2K^2) + \log(1/\eta)}{2 n_{*}}} } \\
  & \qquad \leq K^2 \cdot \P{ | \hat{F}_{l_{\min}}^{k_0}(t) - \tF^{k_0}_{l_{\min}}(t)  | > \sqrt{\frac{\log(2K^2) + \log(1/\eta)}{2 n_{*}}} } \\
  & \qquad \leq 2K^2 \cdot \exp\left[ - 2 n_{*} \frac{\log(2K^2) + \log(1/\eta)}{2 n_{*}} \right] \\
  & \qquad = 2K^2 \cdot \exp\left[ - \log(2K^2) - \log(1/\eta) \right] \\
  & \qquad = \eta.
\end{align*}

\end{proof}

\begin{proof}[Proof of Theorem~\ref{thm:algorithm-correction-marg-upper}]

The proof combines the strategies from the proof of Theorem~\ref{thm:algorithm-correction-upper} and the proof of Theorem~\ref{thm:algorithm-correction-marg}.
Define the events $\cA_1=\{\hat{\cI}^{\mathrm{marg}}=\emptyset \}$ and $\cA_2=\{\hat{i}^{\mathrm{marg}}=1\}$.
Then,
\begin{align}
  \begin{split} 	\label{eq:upper_bound_events-marg}
	\P{Y_{n+1} \in \hat{C}^{\mathrm{marg}}(X_{n+1}) }
	& \leq \P{\cA_1} + \EV{\P{Y_{n+1} \in \hat{C}^{\mathrm{marg}}(X_{n+1}) \mid \cD}\I{\cA_2}} \\
	& \qquad + \EV{\P{Y_{n+1} \in \hat{C}^{\mathrm{marg}}(X_{n+1}) \mid \cD}\I{\cA_1^c\cap \cA_2^c}}.
      \end{split}
\end{align}

We will now separately bound the three terms on the right-hand-side of~\eqref{eq:upper_bound_events-marg}.
The following notation will be useful for this purpose.
For all $i \in [n_{\mathrm{cal}}]$, let $U_i \sim \text{Uniform}(0,1)$ be independent and identically distributed uniform random variables, and
denote their order statistics as $U_{(1)} < U_{(2)} < \ldots < U_{(n_{\mathrm{cal}})}$.

\begin{itemize}
\item The probability of the event $\cA_1$ can be bound from above as:
  \begin{align} \label{eq:upper_bound_event1-marg}
    \P{\hat{\mathcal{I}}_k = \emptyset} \leq \frac{1}{n_{*}}.
  \end{align}

  To simplify the notation in the proof of~\eqref{eq:upper_bound_event1-marg}, define
\begin{align*}
  & d^{\mathrm{marg}}  \\
  & \quad := c(n_{\mathrm{cal}})  + \left( 2 \max_{k \in [K]} \sum_{l\neq k} |V_{kl}| + \sum_{k=1}^{K} | \rho_k - \tilde{\rho}_k  | \right) \frac{1+2\sqrt{\log(2K^2) + \log(n_{*})}}{\sqrt{n_{*}}}.
\end{align*}
Then, under Assumption~\ref{assumption:regularity-dist-delta-marg},
  \begin{align*}
    & 1 - \left( \alpha + \hDelta(S_{(i)}) - \delta^{\mathrm{marg}}(n_{\mathrm{cal}},n_{*}) \right) \\
    & \qquad = 1 - \alpha - \delta^{\mathrm{marg}}(S_{(i)}) + \delta^{\mathrm{marg}}(n_{\mathrm{cal}},n_{*}) + \delta^{\mathrm{marg}}(S_{(i)}) - \hDelta(S_{(i)}) \\
    & \qquad \leq 1 - d^{\mathrm{marg}} + \delta^{\mathrm{marg}}(n_{\mathrm{cal}},n_{*}) + \sup_{t \in [0,1]} | \delta^{\mathrm{marg}}(t) - \hDelta(t)|.
  \end{align*}
Further, it follows from Lemma~\ref{lemma:dkw-delta-prob-marg} and~\eqref{eq:delta-constant-marg} that, with probability at least $1-1/n_{*}$,
  \begin{align*}
    & 1 - \left( \alpha + \hDelta(S_{(i)}) - \delta^{\mathrm{marg}}(n_{\mathrm{cal}},n_{*}) \right) \\
    & \qquad \leq 1 - d^{\mathrm{marg}} + \delta^{\mathrm{marg}}(n_{\mathrm{cal}},n_{*}) \\
    & \qquad \qquad + \left( 2 \max_{k \in [K]} \sum_{l\neq k} |V_{kl}| + \sum_{k=1}^{K} | \rho_k - \tilde{\rho}_k  | \right)  \sqrt{\frac{\log(2K^2) + \log(n_{*})}{2 n_{*}}} \\
    & \qquad \leq 1 - d^{\mathrm{marg}} + c(n_{\mathrm{cal}}) \\
    & \qquad \qquad + \left( 2 \max_{k \in [K]} \sum_{l\neq k} |V_{kl}| + \sum_{k=1}^{K} | \rho_k - \tilde{\rho}_k  | \right) \frac{1+2\sqrt{\log(2K^2) + \log(n_{*})}}{\sqrt{n_{*}}} \\
    & \qquad = 1.
  \end{align*}

This implies that the set $\hat{\mathcal{I}}^{\mathrm{marg}}$ defined in~\eqref{eq:Ik-set-marg} is non-empty with probability at least $1-1/n_{*}$, because $n_{\mathrm{cal}} \in \hat{\mathcal{I}}^{\mathrm{marg}}$ if $1 - [ \alpha + \hDelta(S_{(i)}) - \delta^{\mathrm{marg}}(n_{\mathrm{cal}},n_{*}) ] \leq 1$.

\item Under $\cA_2$, the second term on the right-hand-side of~\eqref{eq:upper_bound_events-marg} can be written as:
  \begin{align*}
    \P{Y_{n+1} \in \hat{C}^{\mathrm{marg}}(X_{n+1}) \mid \cD}
    & = \P{\hat{s}(X_{n+1}, Y_{n+1}) \leq S_{(1)} \mid \cD} \\
    & = F(S_{(1)}) \\
    & = \sum_{k=1}^{K} \rho_k F^k_k(S_{(1)}) \\
    & = \sum_{k=1}^{K} \rho_k \sum_{l=1}^{K} V_{kl} \tilde{F}^k_l(S_{(1)}).
\end{align*}
Then, using Assumption~\ref{assumption:consitency-scores}, we obtain:
  \begin{align*}
    & \P{Y_{n+1} \in \hat{C}^{\mathrm{marg}}(X_{n+1}) \mid \cD}  \\
    & \qquad \leq \sum_{k=1}^{K} \rho_k \left( V_{kk} + \sum_{l \neq k} |V_{kl}| \right)  \tilde{F}^k_k(S_{(1)}) \\
    & \qquad = \sum_{k=1}^{K} \frac{\rho_k}{\tilde{\rho}_k} \left( V_{kk} + \sum_{l \neq k} |V_{kl}| \right) \tilde{\rho}_k \tilde{F}^k_k(S_{(1)}) \\
    & \qquad \leq \left[ \max_{k \in [K]} \left( \frac{\rho_k}{\tilde{\rho}_k} \sum_{l=1}^{K} |V_{kl}| \right)   \right] \sum_{k=1}^{K} \tilde{\rho}_k  \tilde{F}^k_k(S_{(1)}) \\
    & \qquad = \left[ \max_{k \in [K]} \left( \frac{\rho_k}{\tilde{\rho}_k} \sum_{l=1}^{K} |V_{kl}| \right)   \right] \tilde{F}(S_{(1)}).
\end{align*}
Therefore,
\begin{align}
  \begin{split}
     \EV{\P{Y_{n+1} \in \hat{C}^{\mathrm{marg}}(X_{n+1}) \mid \cD}\I{\cA_2}}
    & \leq \left[ \max_{k \in [K]} \left( \frac{\rho_k}{\tilde{\rho}_k} \sum_{l=1}^{K} |V_{kl}| \right)   \right] \EV{\tilde{F}(S_{(1)})} \\
    & = \left[ \max_{k \in [K]} \left( \frac{\rho_k}{\tilde{\rho}_k} \sum_{l=1}^{K} |V_{kl}| \right)   \right] \EV{U_{(1)}} \\
    & = \left[ \max_{k \in [K]} \left( \frac{\rho_k}{\tilde{\rho}_k} \sum_{l=1}^{K} |V_{kl}| \right)   \right] \frac{1}{n_{\mathrm{cal}}+1}.
  \end{split}\label{eq:upper_bound_event2-marg}
\end{align}

\item Under $\cA_1^c\cap \cA_2^c$, by definition of $\hat{\mathcal{I}}^{\mathrm{marg}}$, for any $i \leq \hat{i}^{\mathrm{marg}} - 1$,
  \begin{align*}
    \frac{i}{n_{\mathrm{cal}}} < 1 - \left( \alpha + \hDelta(S_{(i)}) - \delta^{\mathrm{marg}}(n_{\mathrm{cal}},n_{*}) \right).
  \end{align*}
  Therefore, choosing $i = \hat{i}^{\mathrm{marg}}-1$, we get:
 \begin{align*}
   \frac{\hat{i}^{\mathrm{marg}}}{n_{\mathrm{cal}}} < 1 - \left( \alpha + \hDelta(S_{(\hat{i}^{\mathrm{marg}}-1)}) - \delta^{\mathrm{marg}}(n_{\mathrm{cal}},n_{*}) \right) + \frac{1}{n_{\mathrm{cal}}}.
 \end{align*}

As in the proofs of Theorems~\ref{thm:algorithm-correction-upper} and~\ref{thm:algorithm-correction-marg}, the probability of coverage conditional on the labeled data in $\mathcal{D}$ can be written as:
\begin{align*}
    & \P{Y_{n+1} \in \hat{C}^{\mathrm{marg}}(X_{n+1}) \mid \mathcal{D}} \\
    & \qquad = \tilde{F}(S_{(\hat{i}^{\mathrm{marg}})}) + \Delta(S_{(\hat{i}^{\mathrm{marg}})}) \\
    & \qquad = \hat{F}(S_{(\hat{i}^{\mathrm{marg}})}) + \Delta(S_{(\hat{i}^{\mathrm{marg}})}) + \left[ \tilde{F}(S_{(\hat{i}^{\mathrm{marg}})}) - \hat{F}(S_{(\hat{i}^{\mathrm{marg}})}) \right] \\
    & \qquad = \frac{\hat{i}^{\mathrm{marg}}}{n_{\mathrm{cal}}} + \Delta(S_{(\hat{i}^{\mathrm{marg}})}) + \left[ \tilde{F}(S_{(\hat{i}^{\mathrm{marg}})}) - \hat{F}(S_{(\hat{i}^{\mathrm{marg}})}) \right] \\
    & \qquad < 1 - \left( \alpha + \hDelta(S_{(\hat{i}^{\mathrm{marg}}-1)}) - \delta^{\mathrm{marg}}(n_{\mathrm{cal}},n_{*}) \right) + \frac{1}{n_{\mathrm{cal}}} \\
    & \qquad \qquad + \Delta(S_{(\hat{i}^{\mathrm{marg}})}) + \left[ \tilde{F}(S_{(\hat{i}^{\mathrm{marg}})}) - \hat{F}(S_{(\hat{i}^{\mathrm{marg}})}) \right] \\
    & \qquad = 1 - \alpha + \delta^{\mathrm{marg}}(n_{\mathrm{cal}},n_{*}) + \frac{1}{n_{\mathrm{cal}}} \\
    & \qquad \qquad + \left[ \tilde{F}(S_{(\hat{i}^{\mathrm{marg}})}) - \hat{F}(S_{(\hat{i}^{\mathrm{marg}})}) \right] + \left[ \Delta(S_{(\hat{i}^{\mathrm{marg}})}) - \hDelta(S_{(\hat{i}^{\mathrm{marg}}-1)}) \right].
  \end{align*}
  The last term above can be bound by proceeding as in the proof of Theorem~\ref{thm:algorithm-correction-upper}:
  \begin{align*}
    & \Delta(S_{(\hat{i}^{\mathrm{marg}})}) - \hDelta(S_{(\hat{i}^{\mathrm{marg}}-1)})\\
    & \qquad \leq \sup_{t \in \mathbb{R}} |\Delta(t) - \hDelta(t)| + (\Delta(S_{(\hat{i}^{\mathrm{marg}})}) - \Delta(S_{(\hat{i}^{\mathrm{marg}}-1)})),
  \end{align*}
  where the expected value of the first term above can be bounded using Lemma~\ref{lemma:dkw-delta-expected-marg}, and the second term is given by
  \begin{align*}
    & \Delta(S_{(\hat{i}^{\mathrm{marg}})}) - \Delta(S_{(\hat{i}^{\mathrm{marg}}-1)}) \\
    & \qquad = \sum_{k=1}^{K} \left( \rho_k V_{kk} - \tilde{\rho}_k \right) \left[ \tF^{k}_k(S_{(\hat{i}^{\mathrm{marg}})}) -  \tF^{k}_k(S_{(\hat{i}^{\mathrm{marg}}-1)}) \right] \\
    & \qquad \qquad + \sum_{k=1}^{K} \rho_k \sum_{l \neq k} V_{kl} \left[ \tF^{k}_{l}(S^k_{(\hat{i}^{\mathrm{marg}})}) - \tF^{k}_{l}(S^k_{(\hat{i}^{\mathrm{marg}}-1)}) \right] \\
    & \qquad = \sum_{k=1}^{K} \tilde{\rho}_k \cdot \frac{\rho_k V_{kk} - \tilde{\rho}_k}{\tilde{\rho}_k} \left[ \tF^{k}_k(S_{(\hat{i}^{\mathrm{marg}})}) -  \tF^{k}_k(S_{(\hat{i}^{\mathrm{marg}}-1)}) \right] \\
    & \qquad \qquad + \sum_{k=1}^{K} \rho_k \sum_{l \neq k} V_{kl} \left[ \tF^{k}_{l}(S_{(\hat{i}^{\mathrm{marg}})}) - \tF^{k}_{l}(S_{(\hat{i}^{\mathrm{marg}}-1)}) \right] \\
    & \qquad \leq \left( \max_{k \in [K]} \frac{\rho_k V_{kk} - \tilde{\rho}_k}{\tilde{\rho}_k} \right) \sum_{k=1}^{K} \tilde{\rho}_k  \left[ \tF^{k}_k(S_{(\hat{i}^{\mathrm{marg}})}) -  \tF^{k}_k(S_{(\hat{i}^{\mathrm{marg}}-1)}) \right] \\
    & \qquad \qquad + \sum_{k=1}^{K} \rho_k \left( \sum_{l \neq k} V_{kl} \right) \max_{l \neq k} \left[ \tF^{k}_{l}(S_{(\hat{i}^{\mathrm{marg}})}) - \tF^{k}_{l}(S_{(\hat{i}^{\mathrm{marg}}-1)}) \right] \\
    & \qquad = \left( \max_{k \in [K]} \frac{\rho_k V_{kk} - \tilde{\rho}_k}{\tilde{\rho}_k} \right) \left[ \tF(S_{(\hat{i}^{\mathrm{marg}})}) -  \tF(S_{(\hat{i}^{\mathrm{marg}}-1)}) \right] \\
    & \qquad \qquad + \sum_{k=1}^{K} \rho_k \left( \sum_{l \neq k} V_{kl} \right) \max_{l \neq k} \left[ \tF^{k}_{l}(S_{(\hat{i}^{\mathrm{marg}})}) - \tF^{k}_{l}(S_{(\hat{i}^{\mathrm{marg}}-1)}) \right] \\
    & \qquad \leq \left( \max_{k \in [K]} \frac{\rho_k V_{kk} - \tilde{\rho}_k}{\tilde{\rho}_k} \right) \max_{2 \leq i \leq n_{\mathrm{cal}}} \left[ \tF(S_{(i)}) -  \tF(S_{(i-1)}) \right] \\
    & \qquad \qquad + \sum_{k=1}^{K} \rho_k \left( \sum_{l \neq k} V_{kl} \right) \max_{l \neq k} \left[ \tF^{k}_{l}(S_{(\hat{i}^{\mathrm{marg}})}) - \tF^{k}_{l}(S_{(\hat{i}^{\mathrm{marg}}-1)}) \right],
  \end{align*}
where the first inequality follows from the fact that, by definition of the order statistics, $\tF^{k}_l(S_{(\hat{i}^{\mathrm{marg}})}) \geq \tF^{k}_l(S_{(\hat{i}^{\mathrm{marg}}-1)})$ for all $l \in [K]$.
  By a standard result on maximum uniform spacing,
  \begin{align*}
    \EV{ \max_{2 \leq i \leq n_{\mathrm{cal}}} \left[ \tF(S_{(i)}) -  \tF(S_{(i-1)}) \right] }
    & = \EV{\max_{1\leq i \leq n_{\mathrm{cal}}+1} D_i} \\
    & = \frac{1}{n_{\mathrm{cal}}+1}\sum_{j=1}^{n_{\mathrm{cal}}+1} \frac{1}{j},
  \end{align*}
  where $D_1=U_{(1)}$, $D_i=U_{(i)} - U_{(i-1)}$ for $i=2, \dots, n_{\mathrm{cal}}$, and $D_{n_{\mathrm{cal}}+1} = 1-U_{(n_{\mathrm{cal}})}$.
Therefore,
  \begin{align*}
    & \EV{ \Delta(S_{(\hat{i}^{\mathrm{marg}})}) - \Delta(S_{(\hat{i}^{\mathrm{marg}}-1)}) } \\
    & \qquad \leq \left( \max_{k \in [K]} \frac{\rho_k V_{kk} - \tilde{\rho}_k}{\tilde{\rho}_k} \right) \frac{1}{n_{\mathrm{cal}}+1}\sum_{j=1}^{n_{\mathrm{cal}}+1} \frac{1}{j}\\
    & \qquad \qquad + \sum_{k=1}^{K} \rho_k \left( \sum_{l \neq k} V_{kl} \right) \EV{ \max_{l \neq k} \left[ \tF^{k}_{l}(S_{(\hat{i}^{\mathrm{marg}})}) - \tF^{k}_{l}(S_{(\hat{i}^{\mathrm{marg}}-1)}) \right] }.
  \end{align*}
  By Assumption \ref{assumption:regularity-dist}, for any $k \in [K]$,
  \begin{align*}
    & \sum_{k=1}^{K} \rho_k \left( \sum_{l \neq k} V_{kl} \right) \EV{ \max_{l \neq k} \left[ \tF^{k}_{l}(S_{(\hat{i}^{\mathrm{marg}})}) - \tF^{k}_{l}(S_{(\hat{i}^{\mathrm{marg}}-1)}) \right] } \\
    & \qquad \leq \sum_{k=1}^{K} \rho_k \left( \sum_{l \neq k} V_{kl} \right) \EV{ f_{\max} \cdot \left[ S_{(\hat{i}^{\mathrm{marg}})} - S_{(\hat{i}^{\mathrm{marg}}-1)} \right] } \\
    & \qquad \leq \sum_{k=1}^{K} \rho_k \left( \sum_{l \neq k} V_{kl} \right) \EV{ \frac{f_{\max}}{f_{\min}} \cdot \left[ \tF^{k}_{k}(S_{(\hat{i}^{\mathrm{marg}})}) - \tF^{k}_{k}(S_{(\hat{i}^{\mathrm{marg}}-1)}) \right] } \\
    & \qquad = \frac{f_{\max}}{f_{\min}} \sum_{k=1}^{K} \tilde{\rho}_k \cdot \frac{\rho_k}{\tilde{\rho}_k} \left( \sum_{l \neq k} V_{kl} \right) \EV{ \tF^{k}_{k}(S_{(\hat{i}^{\mathrm{marg}})}) - \tF^{k}_{k}(S_{(\hat{i}^{\mathrm{marg}}-1)}) } \\
    & \qquad \leq \frac{f_{\max}}{f_{\min}} \left[ \max_{k \in [K]} \left( \frac{\rho_k}{\tilde{\rho}_k} \sum_{l \neq k} V_{kl} \right) \right] \sum_{k=1}^{K} \tilde{\rho}_k  \EV{ \tF^{k}_{k}(S_{(\hat{i}^{\mathrm{marg}})}) - \tF^{k}_{k}(S_{(\hat{i}^{\mathrm{marg}}-1)}) } \\
    & \qquad = \frac{f_{\max}}{f_{\min}} \left[ \max_{k \in [K]} \left( \frac{\rho_k}{\tilde{\rho}_k} \sum_{l \neq k} V_{kl} \right) \right] \EV{ \tF(S_{(\hat{i}^{\mathrm{marg}})}) - \tF(S_{(\hat{i}^{\mathrm{marg}}-1)}) } \\
    & \qquad \leq \frac{f_{\max}}{f_{\min}} \left[ \max_{k \in [K]} \left( \frac{\rho_k}{\tilde{\rho}_k} \sum_{l \neq k} V_{kl} \right) \right] \EV{ \max_{2\leq i\leq n_{\mathrm{cal}}} | \tF(S_{(i)}) - \tF(S_{(i-1)})|  } \\
    & \qquad = \frac{f_{\max}}{f_{\min}} \left[ \max_{k \in [K]} \left( \frac{\rho_k}{\tilde{\rho}_k} \sum_{l \neq k} V_{kl} \right) \right] \EV{ \max_{2\leq i\leq n_{\mathrm{cal}}} D_i  } \\
    & \qquad = \frac{f_{\max}}{f_{\min}} \left[ \max_{k \in [K]} \left( \frac{\rho_k}{\tilde{\rho}_k} \sum_{l \neq k} V_{kl} \right) \right] \frac{1}{n_{\mathrm{cal}}+1}\sum_{j=1}^{n_{\mathrm{cal}}+1} \frac{1}{j}.
  \end{align*}
Therefore,  under $\cA_1^c\cap \cA_2^c$,
  \begin{align*}
    & \EV{ \Delta(S_{(\hat{i}^{\mathrm{marg}})}) - \Delta(S_{(\hat{i}^{\mathrm{marg}}-1)}) } \\
    & \qquad \leq \left( \max_{k \in [K]} \frac{\rho_k V_{kk} - \tilde{\rho}_k}{\tilde{\rho}_k} \right) \frac{1}{n_{\mathrm{cal}}+1}\sum_{j=1}^{n_{\mathrm{cal}}+1} \frac{1}{j}\\
    & \qquad \qquad + \frac{f_{\max}}{f_{\min}} \left[ \max_{k \in [K]} \left( \frac{\rho_k}{\tilde{\rho}_k} \sum_{l \neq k} V_{kl} \right) \right] \frac{1}{n_{\mathrm{cal}}+1}\sum_{j=1}^{n_{\mathrm{cal}}+1} \frac{1}{j} \\
    & \qquad \leq \frac{\sum_{j=1}^{n_{\mathrm{cal}}+1} 1/j}{n_{\mathrm{cal}}+1} \cdot \left[ \max_{k \in [K]} \frac{\rho_k V_{kk} - \tilde{\rho}_k}{\tilde{\rho}_k}  + \frac{f_{\max}}{f_{\min}} \cdot \max_{k \in [K]} \left( \frac{\rho_k}{\tilde{\rho}_k} \sum_{l \neq k} V_{kl} \right) \right].
\end{align*}

At this point, we have proved that
\begin{align}
\begin{split}
    & \EV{ \P{Y_{n+1} \in \hat{C}^{\mathrm{marg}}(X_{n+1}) \mid \cD}\I{\cA_1^c\cap \cA_2^c}} \\
    & \qquad \leq 1 - \alpha + \delta^{\mathrm{marg}}(n_{\mathrm{cal}},n_{*}) + \frac{1}{n_{\mathrm{cal}}} + \EV{ \max_{i \in [n_{\mathrm{cal}}]} | \tilde{F}(S_{(i)}) - \hat{F}(S_{(i}) } \\
  & \qquad \qquad + \frac{2 \max_{k \in [K]} \sum_{l\neq k} |V_{kl}| + \sum_{k=1}^{K} | \rho_k - \tilde{\rho}_k  |}{\sqrt{n_{*}}} \\
  & \qquad \qquad \cdot \min \left\{ K^2 \sqrt{\frac{\pi}{2}} , \frac{1}{\sqrt{n_{*}}} + \sqrt{\frac{\log(2K^2) + \log(n_{*})}{2}} \right\} \\
  & \qquad \qquad + \frac{\sum_{j=1}^{n_{\mathrm{cal}}+1} 1/j}{n_{\mathrm{cal}}+1} \cdot \left[ \max_{k \in [K]} \frac{\rho_k V_{kk} - \tilde{\rho}_k}{\tilde{\rho}_k}  + \frac{f_{\max}}{f_{\min}} \cdot \max_{k \in [K]} \left( \frac{\rho_k}{\tilde{\rho}_k} \sum_{l \neq k} V_{kl} \right) \right] \\
    & \qquad = 1 - \alpha + \delta^{\mathrm{marg}}(n_{\mathrm{cal}},n_{*}) + \frac{1}{n_{\mathrm{cal}}} + c(n_{\mathrm{cal}})\\
  & \qquad \qquad + \frac{2 \max_{k \in [K]} \sum_{l\neq k} |V_{kl}| + \sum_{k=1}^{K} | \rho_k - \tilde{\rho}_k  |}{\sqrt{n_{*}}} \\
  & \qquad \qquad \cdot \min \left\{ K^2 \sqrt{\frac{\pi}{2}} , \frac{1}{\sqrt{n_{*}}} + \sqrt{\frac{\log(2K^2) + \log(n_{*})}{2}} \right\} \\
  & \qquad \qquad + \frac{\sum_{j=1}^{n_{\mathrm{cal}}+1} 1/j}{n_{\mathrm{cal}}+1} \cdot \left[ \max_{k \in [K]} \frac{\rho_k V_{kk} - \tilde{\rho}_k}{\tilde{\rho}_k}  + \frac{f_{\max}}{f_{\min}} \cdot \max_{k \in [K]} \left( \frac{\rho_k}{\tilde{\rho}_k} \sum_{l \neq k} V_{kl} \right) \right] \\
    & \qquad = 1 - \alpha + 2 \delta^{\mathrm{marg}}(n_{\mathrm{cal}},n_{*}) + \frac{1}{n_{\mathrm{cal}}} + \\
  & \qquad \qquad + \frac{\sum_{j=1}^{n_{\mathrm{cal}}+1} 1/j}{n_{\mathrm{cal}}+1} \cdot \left[ \max_{k \in [K]} \frac{\rho_k V_{kk} - \tilde{\rho}_k}{\tilde{\rho}_k}  + \frac{f_{\max}}{f_{\min}} \cdot \max_{k \in [K]} \left( \frac{\rho_k}{\tilde{\rho}_k} \sum_{l \neq k} V_{kl} \right) \right],
\end{split} \label{eq:upper_bound_event3-marg}
\end{align}
where the inequality above follows from Lemma~\ref{lemma:dkw-delta-expected-marg} and the bound in~\eqref{eq:tighter_cdf_bound}.

\end{itemize}

 Finally, combining~\eqref{eq:upper_bound_events-marg} with \eqref{eq:upper_bound_event1-marg}, \eqref{eq:upper_bound_event2-marg}, and \eqref{eq:upper_bound_event3-marg} leads to the desired result:
 \begin{align*}
   & \P{Y_{n+1} \in \hat{C}^{\mathrm{marg}}(X_{n+1})}  \\
   & \qquad \leq \P{\cA_1} + \EV{\P{Y_{n+1} \in \hat{C}^{\mathrm{marg}}(X_{n+1}) \mid \cD}\I{\cA_2}} \\
   & \qquad \qquad + \EV{\P{Y_{n+1} \in \hat{C}^{\mathrm{marg}}(X_{n+1}) \mid \cD}\I{\cA_1^c\cap \cA_2^c}} \\
    & \qquad \leq 1 - \alpha + 2 \delta^{\mathrm{marg}}(n_{\mathrm{cal}},n_{*}) + \frac{1}{n_{\mathrm{cal}}} + \frac{1}{n_{*}} + \left[ \max_{k \in [K]} \left( \frac{\rho_k}{\tilde{\rho}_k} \sum_{l=1}^{K} |V_{kl}| \right)   \right] \frac{1}{n_{\mathrm{cal}}} \\
  & \qquad \qquad + \frac{\sum_{j=1}^{n_{\mathrm{cal}}+1} 1/j}{n_{\mathrm{cal}}} \cdot \left[ \max_{k \in [K]} \frac{\rho_k V_{kk} - \tilde{\rho}_k}{\tilde{\rho}_k}  + \frac{f_{\max}}{f_{\min}} \cdot \max_{k \in [K]} \left( \frac{\rho_k}{\tilde{\rho}_k} \sum_{l \neq k} V_{kl} \right) \right].
 \end{align*}

\end{proof}

\begin{proof}[Proof of Theorem~\ref{thm:algorithm-correction-marg-optimist}]
This proof combines elements of the proofs of Proposition~\ref{thm:algorithm-correction-optimistic} and Theorem~\ref{thm:algorithm-correction-marg}.
Proceeding exactly as in the proof of Theorem~\ref{thm:algorithm-correction-marg}, we obtain:
  \begin{align*}
    & \P{Y_{n+1} \notin \hat{C}^{\mathrm{marg}}(X_{n+1}) \mid \mathcal{D}} \\
    & \qquad = 1 - \tilde{F}(S_{(\hat{i}^{\mathrm{marg}})}) - \Delta(S_{(\hat{i}^{\mathrm{marg}})}) \\
     & \qquad = 1 - \hat{F}(S_{(\hat{i}^{\mathrm{marg}})}) - \max\left\{\hat{\Delta}(S_{(\hat{i}^{\mathrm{marg}})}) - \delta^{\mathrm{marg}}(n_{\mathrm{cal}},n_{*}), 0 \right\} - \delta^{\mathrm{marg}}(n_{\mathrm{cal}},n_{*}) \\
     & \qquad \qquad + \left[ \hat{F}(S_{(\hat{i}^{\mathrm{marg}})}) - \tilde{F}(S_{(\hat{i}^{\mathrm{marg}})}) \right] \\
    & \qquad \qquad + \max\{\hat{\Delta}(S_{(\hat{i}^{\mathrm{marg}})}) - \delta^{\mathrm{marg}}(n_{\mathrm{cal}},n_{*}), 0 \} - \left[ \Delta(S_{(\hat{i}^{\mathrm{marg}})}) - \delta^{\mathrm{marg}}(n_{\mathrm{cal}},n_{*}) \right] \\
     & \qquad \leq \left[ 1 - \frac{\hat{i}^{\mathrm{marg}}}{n_{\mathrm{cal}}} - \max\left\{\hat{\Delta}(S_{(\hat{i}^{\mathrm{marg}})}) - \delta^{\mathrm{marg}}(n_{\mathrm{cal}},n_{*}), 0 \right\} \right] - \delta^{\mathrm{marg}}(n_{\mathrm{cal}},n_{*}) \\
     & \qquad \qquad + \left[ \hat{F}(S_{(\hat{i}^{\mathrm{marg}})}) - \tilde{F}(S_{(\hat{i}^{\mathrm{marg}})}) \right] + \sup_{t \in \mathbb{R}} | \hat{\Delta}(t) - \Delta(t) |,
\end{align*}
using the fact that $\inf_{t\in \R}\Delta(t) \geq \delta(n_{\mathrm{cal}},n_{*})$ implies
$$
 |\max\{\hDelta(t) - \delta(n_{\mathrm{cal}},n_{*}), 0 \} - (\Delta(t) - \delta(n_{\mathrm{cal}},n_{*}))| \leq |\hat{\Delta}(t) - \Delta(t)|
$$
for all $t \in \mathbb{R}$.
The proof is then completed by proceeding as in the proof of Theorem~\ref{thm:algorithm-correction-marg}.
\end{proof}

\subsection{Prediction sets with calibration-conditional coverage}

\begin{proof}[Proof of Theorem~\ref{thm:algorithm-correction-cc}]
  The proof is similar to that of Theorem~\ref{thm:algorithm-correction}.
  Suppose $Y_{n+1} = k$, for some $k \in [K]$.
  We assume without loss of generality that $\sum_{l=1}^{K}|V_{kl}| \neq 0$; otherwise, there is no label contamination and the result is trivially true.
  By definition of the conformity score function in~\eqref{eq:conf-scores}, the event $k \notin \hat{C}(X_{n+1})$ occurs if and only if $\hat{s}(X_{n+1},k) > \hat{\tau}_k$.
  We will assume without loss of generality that $\hat{\mathcal{I}}_k \neq \emptyset$ and $\hat{i}_k = \min\{i \in \hat{\mathcal{I}}_k\}$; otherwise, $\hat{\tau}_k=1$ and the result trivially holds.
As in the proof of Theorem~\ref{thm:algorithm-correction}, the probability of miscoverage conditional on $Y_{n+1}=k$ and on the labeled data in $\mathcal{D}$ can be bounded from above as:
  \begin{align*}
    & \P{Y \notin C(X, S^k_{(\hat{i}_k)}) \mid Y = k, \mathcal{D}} \\
    & \qquad \leq \alpha - \delta^{\mathrm{cc}}(n_k,n_{*},\gamma) + \sup_{t \in \mathbb{R}} [\hat{\Delta}_k(t) - \Delta_k(t) ] + \sup_{t \in \mathbb{R}} \left[\hat{F}^{k}_k(t) - \tilde{F}^{k}_k(t) \right].
  \end{align*}
We know from Lemma~\ref{lemma:dkw-delta-prob} that, for any $\gamma_1 > 0$,
\begin{align*}
  & \P{ \sup_{t \in \mathbb{R}} |\hat{\Delta}_k(t) - \Delta_k(t) | > 2\sum_{\l\neq k}|V_{kl}| \sqrt{\frac{\log(2K) + \log(1/\gamma_1)}{2 n_{*}}} }
   \leq \gamma_1.
\end{align*}
Similarly, it follows directly from the DKW inequality that, for any $\gamma_2 > 0$,
\begin{align*}
  \P{ \sup_{t \in \mathbb{R}} \left[ \hat{F}_k^{k}(t) - \tF^{k}_{k}(t)  \right] > \sqrt{\frac{\log(1/\gamma_2)}{2 n_k}} }
  & \leq \gamma_2.
\end{align*}
Therefore, for any  $\gamma, \gamma_1, \gamma_2 > 0$ such that $\gamma = \gamma_1 + \gamma_2$, with probability at least $1-\gamma$,
\begin{align*}
  & \P{Y \notin C(X, S^k_{(\hat{i}_k)}) \mid Y = k, \mathcal{D}} \\
  & \qquad \leq \alpha - \delta^{\mathrm{cc}}(n_k,n_{*},\gamma)
    + 2\sum_{\l\neq k}|V_{kl}| \sqrt{\frac{\log(2K) + \log(1/\gamma_1)}{2 n_{*}}}
    + \sqrt{\frac{\log(1/\gamma_2)}{2 n_k}}.
\end{align*}
Finally, setting
\begin{align*}
& \gamma_1=  \frac{\gamma}{2} \cdot \frac{\sum_{\l\neq k}|V_{kl}|}{\sum_{l=1}^{K}|V_{kl}|},
& \gamma_2=  \gamma \left( 1- \frac{1}{2} \cdot \frac{\sum_{\l\neq k}|V_{kl}|}{\sum_{l=1}^{K}|V_{kl}|}\right),
\end{align*}
gives the desired result; that is, with probability at least $1-\gamma$,
\begin{align*}
  & \P{Y \notin C(X, S^k_{(\hat{i}_k)}) \mid Y = k, \mathcal{D}} \leq \alpha.
\end{align*}

\end{proof}

\begin{proof}[Proof of Theorem~\ref{thm:algorithm-correction-cc-upper}]

The proof follows an approach similar to that of the proof of Theorem~\ref{thm:algorithm-correction-upper}.
Define the events $\cA_1=\{\hat{\cI}_k=\emptyset \}$ and $\cA_2=\{\hat{i}_k=1\}$.
Then,
\begin{align}
  \begin{split} 	\label{eq:upper_bound_events-cc}
	& \P{Y_{n+1} \in \hat{C}^{\mathrm{cc}}(X_{n+1}) \mid \mathcal{D}, Y = k}  \\
	& \qquad \leq \I{\cA_1} + \P{Y_{n+1} \in \hat{C}^{\mathrm{cc}}(X_{n+1}) \mid Y = k, \cD}\I{\cA_2} \\
	& \qquad \qquad + \P{Y_{n+1} \in \hat{C}^{\mathrm{cc}}(X_{n+1}) \mid Y = k, \cD}\I{\cA_1^c\cap \cA_2^c}.
      \end{split}
\end{align}

We will now separately bound the three terms on the right-hand-side of~\eqref{eq:upper_bound_events-cc}.
The following notation will be useful for this purpose.
For all $i \in [n_k]$, let $U_i \sim \text{Uniform}(0,1)$ be independent and identically distributed uniform random variables, and
denote their order statistics as $U_{(1)} < U_{(2)} < \ldots < U_{(n_k)}$.

\begin{itemize}

\item The probability of the event $\cA_1$ can be bound from above as:
  \begin{align} \label{eq:upper_bound_event1-cc}
    \P{\hat{\mathcal{I}}_k = \emptyset} \leq \gamma_2.
  \end{align}

  To simplify the notation in the proof of~\eqref{eq:upper_bound_event1-cc}, define
\begin{align*}
  d_k := \sqrt{\frac{\log(1/\gamma_1)}{2 n_k}} + 4\sum_{\l\neq k}|V_{kl}| \sqrt{\frac{\log(2K) + \log(1/\gamma_2)}{2 n_{*}}}.
\end{align*}
Then, under Assumption~\ref{assumption:regularity-dist-delta-cc},
  \begin{align*}
    & 1 - \left( \alpha + \hDelta_k(S^k_{(i)}) - \delta^{\mathrm{cc}}(n_k,n_{*},\gamma) \right) \\
    & \qquad = 1 - \alpha - \Delta_k(S^k_{(i)}) + \delta^{\mathrm{cc}}(n_k,n_{*},\gamma)  + \Delta_k(S^k_{(i)}) - \hDelta_k(S^k_{(i)}) \\
    & \qquad \leq 1 - d_k + \delta^{\mathrm{cc}}(n_k,n_{*},\gamma) + \sup_{t \in [0,1]} | \Delta_k(t) - \hDelta_k(t)|.
  \end{align*}
Further, it follows from Lemma~\ref{lemma:dkw-delta-prob} and~\eqref{eq:delta-constant-cc} that, with probability at least $1-\gamma_2$,
  \begin{align*}
    & 1 - \left( \alpha + \hDelta_k(S^k_{(i)}) - \delta^{\mathrm{cc}}(n_k,n_{*},\gamma) \right) \\
    & \qquad \leq 1 - d_k + \delta^{\mathrm{cc}}(n_k,n_{*},\gamma) + 2 \sum_{l\neq k}|V_{kl}| \sqrt{\frac{\log(2K) + \log(1/\gamma_2)}{2 n_{*}}} \\
    & \qquad \leq 1 - d_k + \sqrt{\frac{\log(1/\gamma_1)}{2 n_k}} + 4\sum_{\l\neq k}|V_{kl}| \sqrt{\frac{\log(2K) + \log(1/\gamma_2)}{2 n_{*}}}  \\
    & \qquad = 1.
  \end{align*}

This implies that the set $\hat{\mathcal{I}}_k$ defined in~\eqref{eq:Ik-set} is non-empty with probability at least $1-\gamma_2$, because $n_k \in \hat{\mathcal{I}}_k$ if $1 - [ \alpha + \hDelta_k(S^k_{(i)}) - \delta^{\mathrm{cc}}(n_k,n_{*},\gamma) ] \leq 1$.

\item Under $\cA_2$, the second term on the right-hand-side of~\eqref{eq:upper_bound_events-cc} can be written as:
  \begin{align*}
    \P{Y_{n+1} \in \hat{C}^{\mathrm{cc}}(X_{n+1}) \mid Y = k, \cD}\I{\cA_2}
    & = \P{\hat{s}(X_{n+1}, k) \leq S^k_{(1)} \mid Y = k, \cD} \\
    & = F^{k}_k(S^k_{(1)}).
\end{align*}
Therefore, using Assumption~\ref{assumption:consitency-scores}, we can write that, for any $a > 0$,
  \begin{align*}
    & \P{ \P{Y_{n+1} \in \hat{C}^{\mathrm{cc}}(X_{n+1}) \mid Y = k, \cD}\I{\cA_2} \geq a } \\
    & \qquad \leq \P{ F^{k}_k(S^k_{(1)}) \geq a } \\
    & \qquad \leq \P{ \left( V_{kk} + \sum_{l \neq k} |V_{kl}| \right) \tF^{k}_k(S^k_{(1)}) \geq a } \\
    & \qquad = \P{ U_{(1)} \geq \frac{a}{V_{kk} + \sum_{l \neq k} |V_{kl}|} } \\
    & \qquad = \left( 1 - \frac{a}{V_{kk} + \sum_{l \neq k} |V_{kl}|} \right)^{n_k} \\
    & \qquad \leq \frac{1}{n_k} \cdot \frac{V_{kk} + \sum_{l \neq k} |V_{kl}|}{a}.
  \end{align*}
  In particular, choosing $a = (V_{kk} + \sum_{l \neq k} |V_{kl}|)/(n_k \bar{\gamma})$, for any $\bar{\gamma} \in (0,1)$, gives us:
  \begin{align} \label{eq:upper_bound_event2-cc}
     \P{ \P{Y_{n+1} \in \hat{C}^{\mathrm{cc}}(X_{n+1}) \mid Y = k, \cD}\I{\cA_2} \leq \frac{V_{kk} + \sum_{l \neq k} |V_{kl}|}{n_k \bar{\gamma}} }
    \geq 1 - \bar{\gamma}.
  \end{align}

\item By proceeding  as in the proof of Theorem~\ref{thm:algorithm-correction-upper}, using Assumption~\ref{assumption:regularity-dist} we obtain that, under $\cA_1^c\cap \cA_2^c$,
  \begin{align} \label{eq:upper_bound_events-cc-upper}
    \begin{split}
    & \P{Y \in C(X, S^k_{(\hat{i}_k)}) \mid Y = k, \mathcal{D}} \\
    & \qquad < 1 - \alpha + \delta^{\mathrm{cc}}(n_k,n_{*},\gamma) + \frac{1}{n_k} + \sup_{t \in \mathbb{R}} \left[ \tilde{F}^{k}_k(t) - \hat{F}^{k}_k(t)  \right] + \sup_{t \in \mathbb{R}} |\Delta_k(t) - \hDelta_k(t)| \\
    & \qquad \qquad + 2\sum_{l\neq k}|V_{kl}| \cdot \frac{f_{\max}}{f_{\min}} \cdot  \max_{2\leq i\leq n_k} \left| \tF^{k}_{k}(S^k_{(i)}) - \tF^{k}_{k}(S^k_{(i-1)}) \right|.
  \end{split}
  \end{align}
  Now, recall that for any $\bar{\gamma} \in (0,1)$, by the DKW inequality,
\begin{align*}
  \P{ \sup_{t \in \mathbb{R}} \left[ \hat{F}_k^{k}(t) - \tF^{k}_{k}(t)  \right] > \sqrt{\frac{\log(1/\bar{\gamma})}{2 n_k}} }
  & \leq \bar{\gamma}.
\end{align*}
Similarly, we know from Lemma~\ref{lemma:dkw-delta-prob} that, for any $\bar{\gamma} \in (0,1)$,
\begin{align*}
  & \P{ \sup_{t \in \mathbb{R}} |\hat{\Delta}_k(t) - \Delta_k(t) | > 2\sum_{\l\neq k}|V_{kl}| \sqrt{\frac{\log(2K) + \log(1/\bar{\gamma})}{2 n_{*}}} }
   \leq \bar{\gamma}.
\end{align*}
It remains to bound the last term on the right-hand-side of~\eqref{eq:upper_bound_events-cc-upper}.
Note that, for any $\bar{\gamma} \in (0,1)$,
\begin{align*}
  & \P{\max_{2\leq i\leq n_k} \left| \tF^{k}_{k}(S^k_{(i)}) - \tF^{k}_{k}(S^k_{(i-1)}) \right| > \bar{\gamma} }
    = \P{\max_{1 \leq i \leq n_k +1 } D_i  > \bar{\gamma} },
\end{align*}
where $U_i \sim \text{Uniform}(0,1)$ are independent and identically distributed uniform random variables for all $i \in [n_k]$, and $D_1=U_{(1)}$, $D_i=U_{(i)} - U_{(i-1)}$ for $i=2, \dots, n_k$, and $D_{n_k+1} = 1-U_{(n_k)}$.

Further, by combining the Markov inequality with a standard result on the asymptotic behavior of the maximum uniform spacing, we obtain that, for any $a > 0$,
\begin{align*}
  \P{\max_{1 \leq i \leq n_k +1 } D_i  > a }
  & \leq \frac{1}{a} \cdot \EV{\max_{1 \leq i \leq n_k +1 } D_i}
    = \frac{1}{a} \cdot \frac{1}{n_k} \sum_{j=1}^{n_k+1} \frac{1}{j}.
\end{align*}
Therefore, setting $a = (\sum_{j=1}^{n_k+1} 1/j) / (n_k \bar{\gamma})$ for any $\bar{\gamma} > 0$, we obtain that
\begin{align*}
  \P{\max_{1 \leq i \leq n_k +1 } D_i  > \frac{1}{\bar{\gamma}} \cdot \frac{1}{n_k} \sum_{j=1}^{n_k+1} \frac{1}{j} }
  \leq \bar{\gamma}.
\end{align*}

    With a union bound, this implies that, for any $\bar{\gamma} \in (0,1)$,
  \begin{align} \label{eq:upper_bound_event3-cc}
    \P{ \P{Y \in C(X, S^k_{(\hat{i}_k)}) \mid Y = k, \mathcal{D}} \I{\cA_1^c\cap \cA_2^c} > 1 - \alpha + \omega(n_k, n_*, \gamma, \bar{\gamma}) } \leq \bar{\gamma},
  \end{align}
  where
  \begin{align*}
    & \omega(n_k, n_*, \gamma, \bar{\gamma}) \\
    & \qquad =  \delta^{\mathrm{cc}}(n_k,n_{*},\gamma) + \frac{1}{n_k} + \sqrt{\frac{\log(3/\bar{\gamma})}{2 n_k}}
      + 2\sum_{\l\neq k}|V_{kl}| \sqrt{\frac{\log(2K) + \log(3/\bar{\gamma})}{2 n_{*}}} \\
    & \qquad\qquad + 2\sum_{l\neq k}|V_{kl}| \cdot \frac{f_{\max}}{f_{\min}} \cdot \frac{1}{n_k} \cdot \left[ \log(n_k+1) + \frac{3}{\bar{\gamma}}  \sum_{j=1}^{n_k+1} \frac{1}{j}   \right].
  \end{align*}

\end{itemize}

Finally, combining~\eqref{eq:upper_bound_events-cc} with \eqref{eq:upper_bound_event1-cc}, \eqref{eq:upper_bound_event2-cc}
, and \eqref{eq:upper_bound_event3-cc} leads to
\begin{align*}
	& \P{ \P{Y_{n+1} \in \hat{C}^{\mathrm{cc}}(X_{n+1}) \mid \mathcal{D}, Y = k} \geq 1 - \alpha + \omega(n_k, n_*, \gamma, \bar{\gamma}) + \frac{V_{kk} + \sum_{l \neq k} |V_{kl}|}{n_k \bar{\gamma}} } \\
	& \qquad \leq \gamma_2 + 2 \bar{\gamma}.
\end{align*}
Thus, the desired result is obtained by setting
\begin{align*}
  \bar{\gamma}
  & = \frac{\gamma - \gamma_2}{2}
    = \frac{\gamma}{2} \left( 1 - \frac{1}{2} \cdot \frac{\sum_{\l\neq k}|V_{kl}|}{\sum_{l=1}^{K}|V_{kl}|} \right).
\end{align*}

\end{proof}

\begin{proof}[Proof of Theorem~\ref{thm:algorithm-correction-cc-optimist}]
This proof combines elements of the proofs of Proposition~\ref{thm:algorithm-correction-optimistic} and Theorem~\ref{thm:algorithm-correction-cc}.
  Suppose $Y_{n+1} = k$, for some $k \in [K]$.
Proceeding exactly as in the proof of Theorem~\ref{thm:algorithm-correction-cc}, we obtain:
  \begin{align*}
    & \P{Y \notin C(X, S^k_{(\hat{i}_k)}) \mid Y = k, \mathcal{D}} \\
     & \qquad \leq \left[ 1 - \frac{\hat{i}_k}{n_k} - \max\left\{\hat{\Delta}_k(S_{(\hat{i}_k)}) - \delta^{\mathrm{cc}}(n_k,n_{*}), - \sqrt{\frac{\log(1/\gamma)}{2 n_{k}}} \right\} \right] - \delta^{\mathrm{cc}}(n_k,n_{*}) \\
     & \qquad \qquad + \left[ \hat{F}_k^k(S_{(\hat{i}_k)}) - \tilde{F}_k^k(S_{(\hat{i}_k)}) \right] \\
    & \qquad \qquad + \max\left\{\hat{\Delta}_k(S_{(\hat{i}_k)}) - \delta^{\mathrm{cc}}(n_{k},n_{*}), - \sqrt{\frac{\log(1/\gamma)}{2 n_{k}}} \right\} - \left[ \Delta_k(S_{(\hat{i}_k)}) - \delta^{\mathrm{cc}}(n_k,n_{*}) \right] \\
     & \qquad = \left[ 1 - \frac{\hat{i}_k}{n_k} - \max\left\{\hat{\Delta}_k(S_{(\hat{i}_k)}) - \delta^{\mathrm{cc}}(n_k,n_{*}), - \sqrt{\frac{\log(1/\gamma)}{2 n_{k}}} \right\} \right] - \delta^{\mathrm{cc}}(n_k,n_{*}) \\
     & \qquad \qquad + \left[ \hat{F}_k^k(S_{(\hat{i}_k)}) - \tilde{F}_k^k(S_{(\hat{i}_k)}) \right] \\
    & \qquad \qquad + \max\left\{\hat{\Delta}_k(S_{(\hat{i}_k)}) + \sqrt{\frac{\log(1/\gamma)}{2 n_{k}}} - \delta^{\mathrm{cc}}(n_{k},n_{*}), 0 \right\} \\
    & \qquad \qquad - \left[ \Delta_k(S_{(\hat{i}_k)}) + \sqrt{\frac{\log(1/\gamma)}{2 n_{k}}} - \delta^{\mathrm{cc}}(n_k,n_{*}) \right] \\
    & \qquad \leq \alpha - \delta^{\mathrm{cc}}(n_k,n_{*},\gamma) + \sup_{t \in \mathbb{R}} |\hat{\Delta}_k(t) - \Delta_k(t) | + \sup_{t \in \mathbb{R}} \left[\hat{F}^{k}_k(t) - \tilde{F}^{k}_k(t) \right].
  \end{align*}
using the fact that $\inf_{t\in \R}\Delta_k(t) \geq \delta^{\mathrm{cc}}(n_k, n_{*}) - \sqrt{ \log(1/\gamma) / (2 n_{k})}$ implies
  \begin{align*}
 & \left| \max\left\{\hat{\Delta}_k(S_{(\hat{i}_k)}) + \sqrt{\frac{\log(1/\gamma)}{2 n_{k}}} - \delta^{\mathrm{cc}}(n_{k},n_{*}), 0 \right\}  - \left[ \Delta_k(t) + \sqrt{\frac{\log(1/\gamma)}{2 n_{k}}} - \delta^{\mathrm{cc}}(n_k,n_{*}) \right] \right| \\
 & \qquad \leq |\hat{\Delta}_k(t) - \Delta_k(t)|
  \end{align*}
for all $t \in \mathbb{R}$.
\end{proof}

\section*{Comparison to worst-case coverage bounds}

\begin{proof}[Proof of Corollary~\ref{thm:coverage-lab-cond-worst-case}]

The expression for $\Delta_k(t)$ in (\ref{eq:delta-2}) can be equivalently rewritten as
\begin{align*}
  \Delta_k(t)
  & = \sum_{l \neq k} V_{kl} \left( \tF_l^{k}(t) - \tF_k^{k}(t) \right),
\end{align*}
which implies
\begin{align*}
  - \sum_{l \neq k} |V_{kl}| \leq
  \Delta_k(t)
  & \leq \sum_{l \neq k} |V_{kl}|, \qquad \forall t \in \mathbb{R}.
\end{align*}
The proof is then completed by applying Theorem~\ref{thm:coverage-lab-cond}.
\end{proof}

\clearpage

\section{Supplementary numerical results} \label{app:figures}

\subsection{Simulations under a known label contamination model} \label{app:figures-known}

\subsubsection{Additional views and performance metrics}

\begin{figure}[!htb]
\centering
\includegraphics[width=0.9\linewidth]{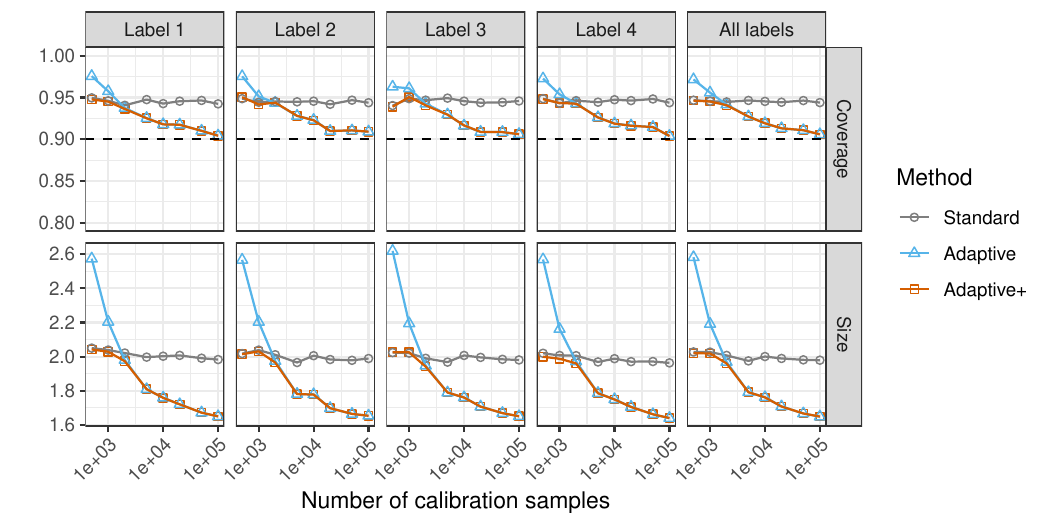}
\caption{Performances of different conformal methods on simulated data with randomly contaminated labels, as a function of the number of calibration samples.
The reported empirical coverage and average size of the prediction sets are stratified based on the true label of the test points. The strength parameter of the label contamination process is $\epsilon=0.1$. Other details are as in Figure~\ref{fig:exp-synthetic-1-lab-cond-K4-ncal}.}
\label{fig:exp-synthetic-1-lab-cond-K4-ncal_lc}
\end{figure}

\begin{figure}[!htb]
\centering
\includegraphics[width=0.9\linewidth]{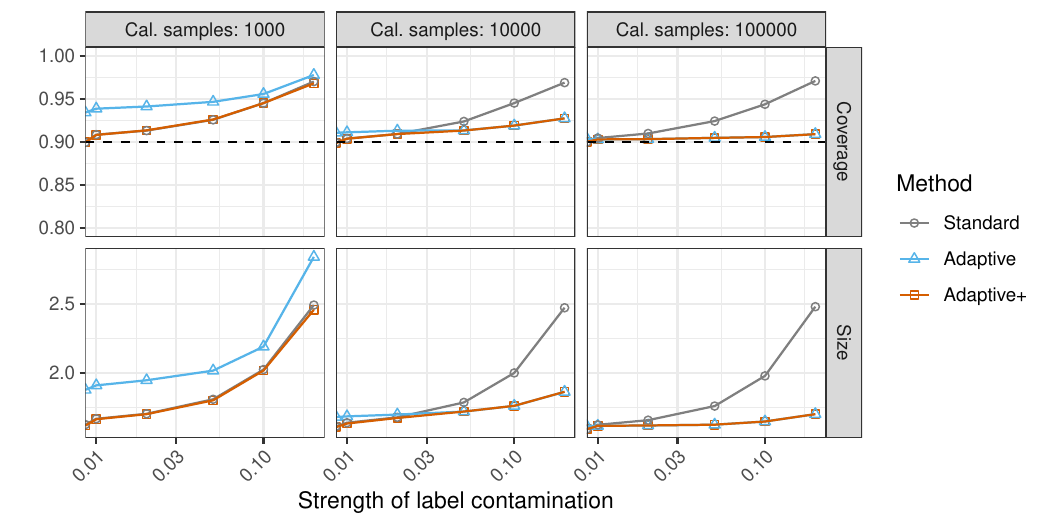}
\caption{Performances of different conformal methods on simulated data with varying numbers of calibration samples, as a function of the label contamination strength.
Other details are as in Figure~\ref{fig:exp-synthetic-1-lab-cond-K4-ncal}.}
\label{fig:exp-synthetic-1-lab-cond-K4-ncal-epsilon}
\end{figure}

\FloatBarrier
\subsubsection{The effect of the number of classes}

\begin{figure}[!htb]
\centering
\includegraphics[width=0.9\linewidth]{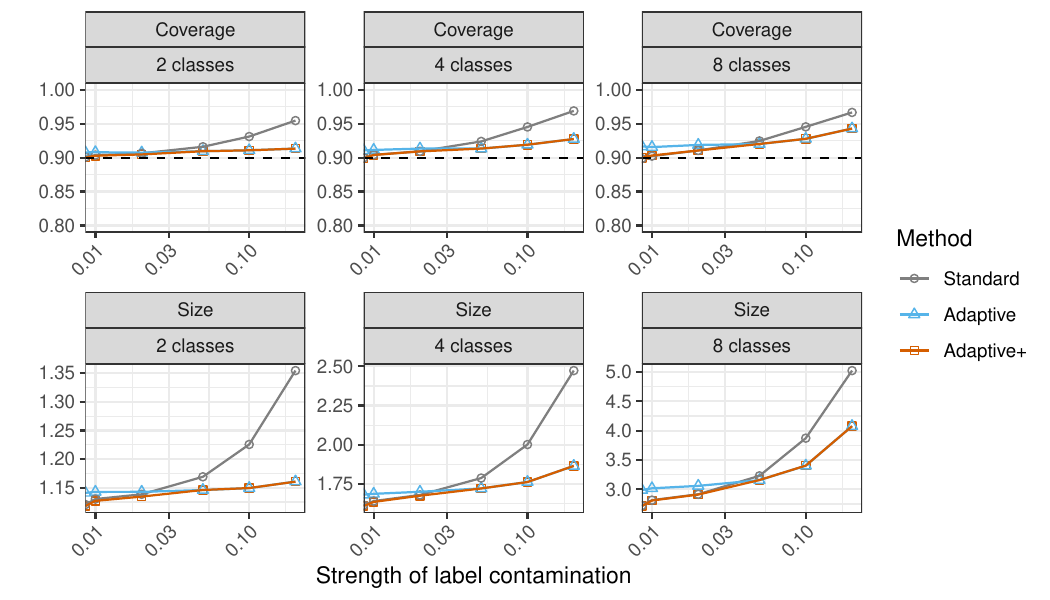}
\caption{Performances of different conformal methods on simulated classification data with different numbers of labels, as a function of the contamination strength.
The strength parameter of the label contamination process is $\epsilon=0.1$.
Other details are as in Figure~\ref{fig:exp-synthetic-1-lab-cond-K4-ncal-epsilon}.
}
\label{fig:exp-synthetic-1-lab-cond-ncal-epsilon}
\end{figure}

\clearpage
\subsubsection{The effect of the classifier}

\begin{figure}[!htb]
\centering
\includegraphics[width=0.9\linewidth]{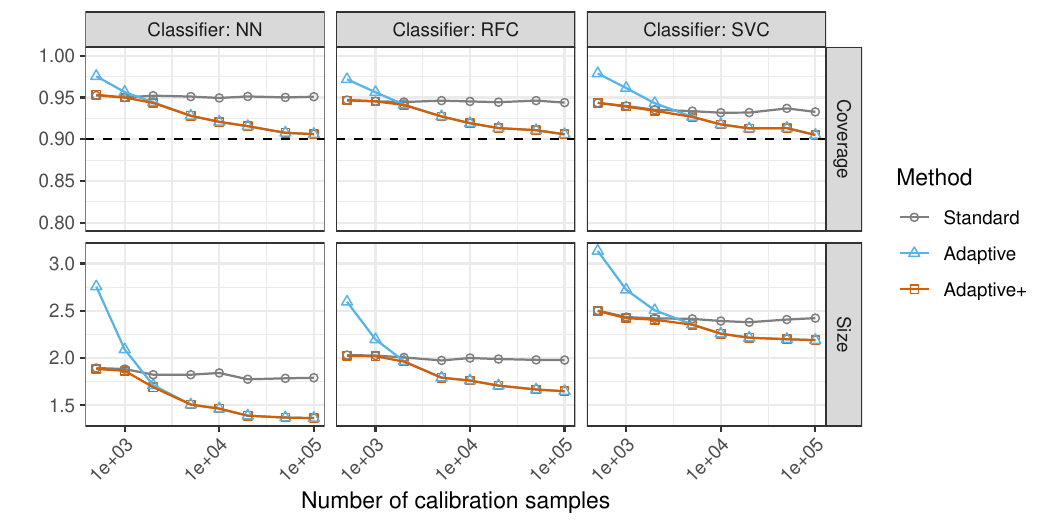}
\caption{Performances of different conformal methods on simulated data using different machine learning classifiers, as a function of the number of calibration samples. The strength parameter of the label contamination process is $\epsilon=0.1$. Other details are as in Figure~\ref{fig:exp-synthetic-1-lab-cond-K4-ncal}.}
\label{fig:exp-synthetic-1-lab-cond-K4-ncal-models}
\end{figure}

\FloatBarrier
\subsubsection{The effect of the data distribution}

\begin{figure}[!htb]
\centering
\includegraphics[width=0.9\linewidth]{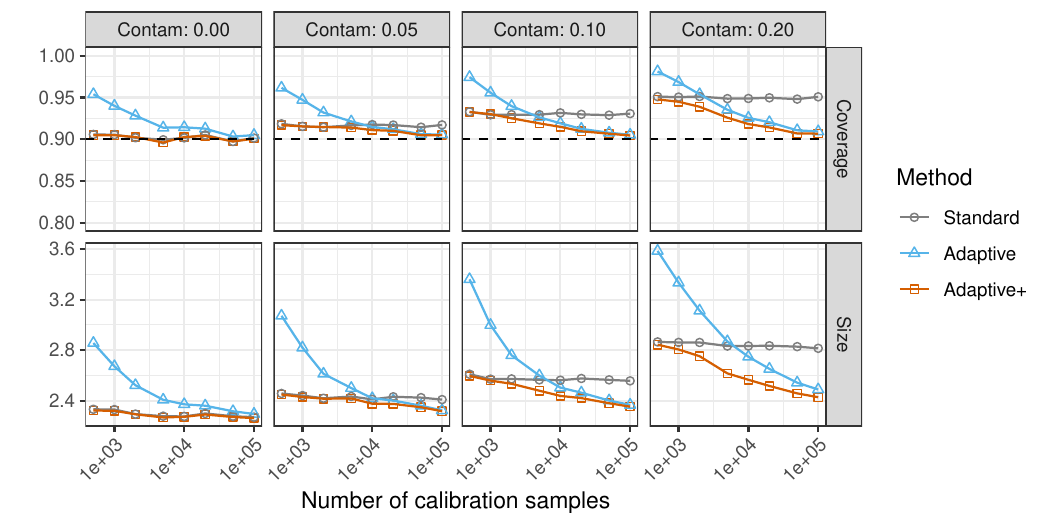}
\caption{Performances of different conformal methods on simulated data with random label contamination of varying strength, as a function of the number of calibration samples.
 All methods guarantee 90\% label-conditional coverage.
The data are simulated from a logistic model with random parameters.
Other details are as in Figure~\ref{fig:exp-synthetic-1-lab-cond-K4-ncal}.}
\label{fig:exp-synthetic-1-lab-cond-K4-ncal_synthetic2}
\end{figure}

\begin{figure}[!htb]
\centering
\includegraphics[width=0.9\linewidth]{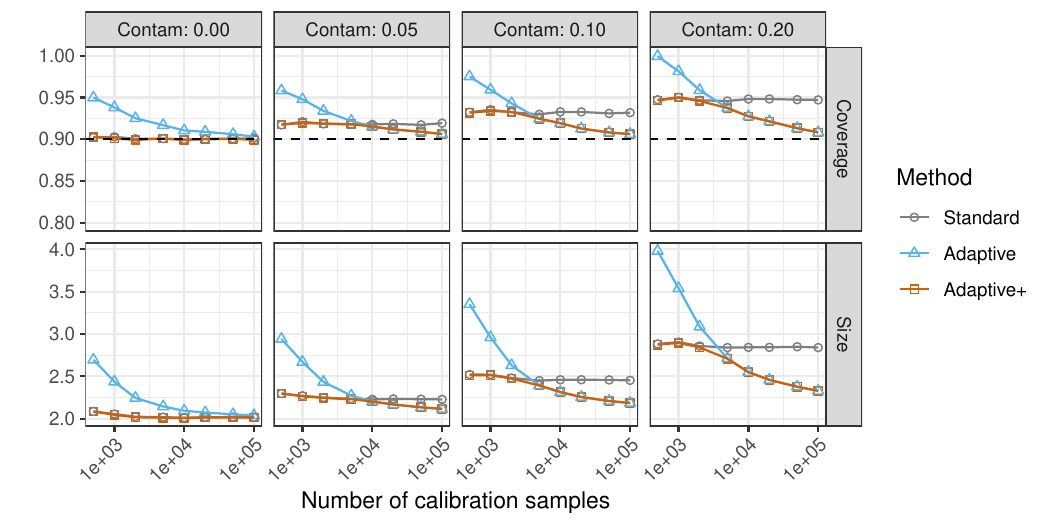}
\caption{Performances of different conformal methods on simulated data with random label contamination of varying strength, as a function of the number of calibration samples.
The data are simulated from a heteroscedastic decision-tree model.
Other details are as in Figure~\ref{fig:exp-synthetic-1-lab-cond-K4-ncal}.}
\label{fig:exp-synthetic-1-lab-cond-K4-ncal_synthetic3}
\end{figure}

\FloatBarrier
\subsubsection{The effect of the label contamination process}

\begin{figure}[!htb]
\centering
\includegraphics[width=0.9\linewidth]{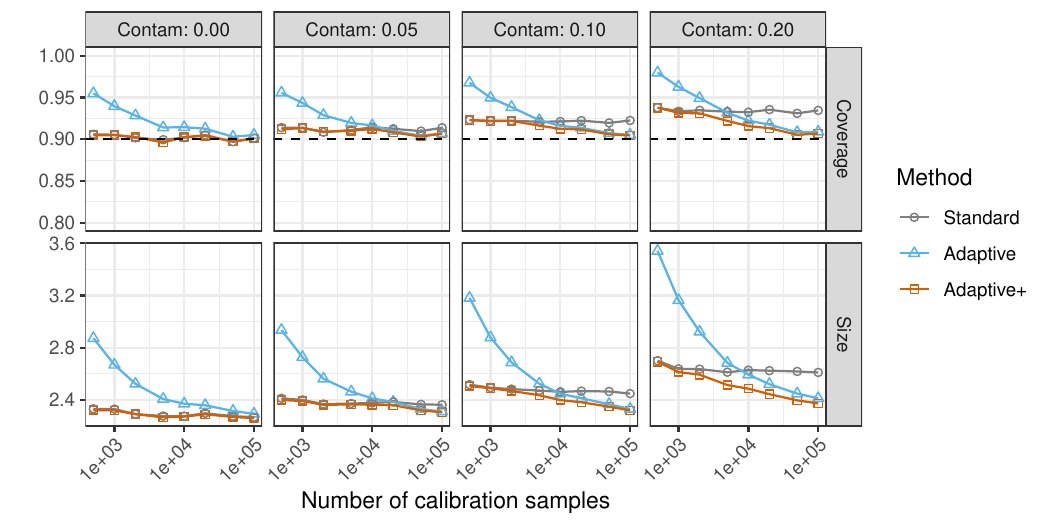}
\caption{Performances of different conformal methods on simulated data with random label contamination of varying strength, as a function of the number of calibration samples.
The contamination process has a block-like structure.
Other details are as in Figure~\ref{fig:exp-synthetic-1-lab-cond-K4-ncal_synthetic2}.}
\label{fig:exp-synthetic-1-lab-cond-K4-ncal_block}
\end{figure}

\begin{figure}[!htb]
\centering
\includegraphics[width=0.9\linewidth]{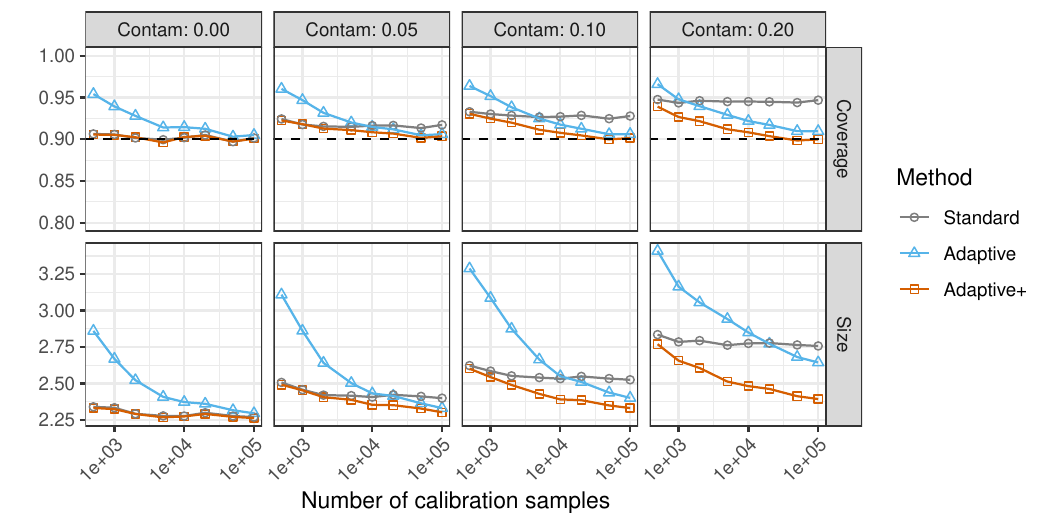}
\caption{Performances of different conformal methods on simulated data with random label contamination of varying strength, as a function of the number of calibration samples.
The label contamination process has a random heterogeneous structure.
Other details are as in Figure~\ref{fig:exp-synthetic-1-lab-cond-K4-ncal_synthetic2}.}
\label{fig:exp-synthetic-1-lab-cond-K4-ncal_random}
\end{figure}

\clearpage
\subsubsection{Prediction sets with marginal coverage}

\begin{figure}[!htb]
\centering
\includegraphics[width=0.9\linewidth]{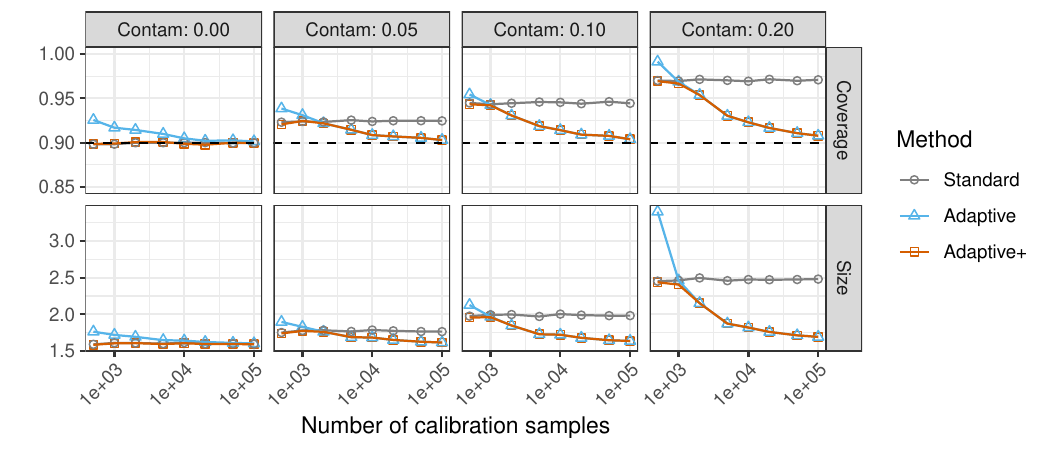}
\caption{Performances of different conformal prediction methods with marginal coverage on simulated data with random label contamination of varying strength, as a function of the number of calibration samples.
The dashed horizontal line indicates the 90\% nominal marginal coverage level.
Other details are as in Figure~\ref{fig:exp-synthetic-1-lab-cond-K4-ncal}.}
\label{fig:exp-synthetic-1-marginal-K4-ncal}
\end{figure}

\begin{figure}[!htb]
\centering
\includegraphics[width=0.9\linewidth]{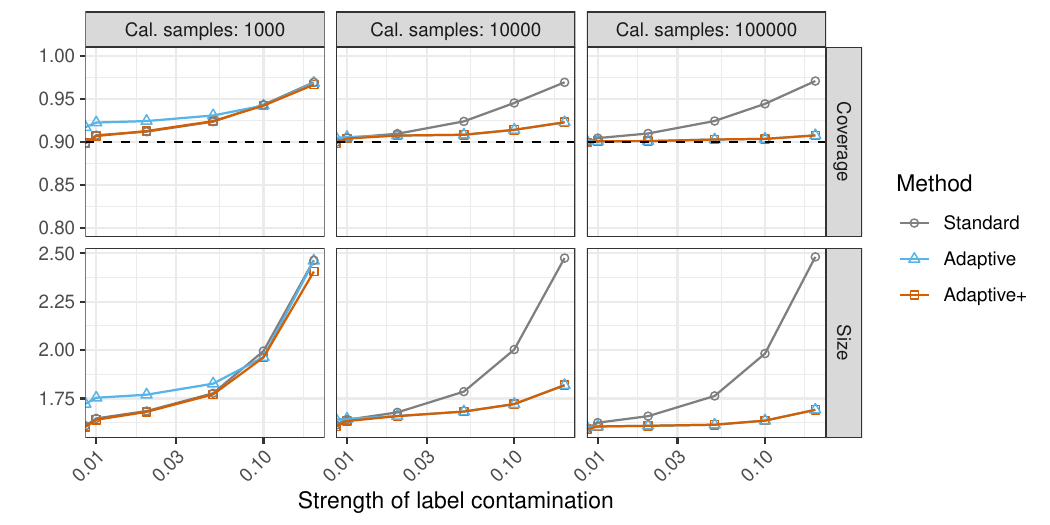}
\caption{Performances of different conformal prediction methods with marginal coverage on simulated data with varying numbers of calibration samples, as a function of the label contamination strength. The dashed horizontal line indicates the 90\% nominal marginal coverage level.
Other details are as in Figure~\ref{fig:exp-synthetic-1-lab-cond-K4-ncal-epsilon}.}
\label{fig:exp-synthetic-1-marginal-K4-ncal-epsilon}
\end{figure}

\begin{figure}[!htb]
\centering
\includegraphics[width=0.95\linewidth]{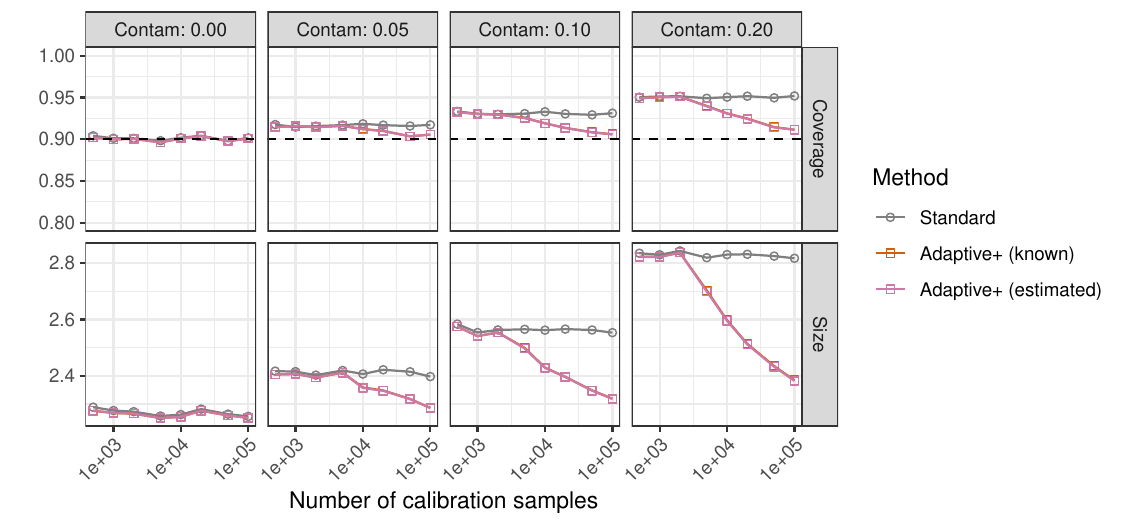}
\caption{Performances of different conformal methods on simulated data with label contamination.
The label contamination process follows a classical randomized response model.
The Adaptive+ method is applied with and without perfect knowledge of the contaminated label frequencies.
Note that there is no significant difference in performance here, because it is easy to estimate $\tilde{\rho}$ accurately from the available data.
The nominal marginal coverage level is 90\%.
Other details are as in Figure~\ref{fig:exp-synthetic-1-lab-cond-K4-ncal_synthetic2}.}
\label{fig:exp-synthetic-1-marginal-K4-ncal_synthetic2_uniform_estim-rho}
\end{figure}

\begin{figure}[!htb]
\centering
\includegraphics[width=0.95\linewidth]{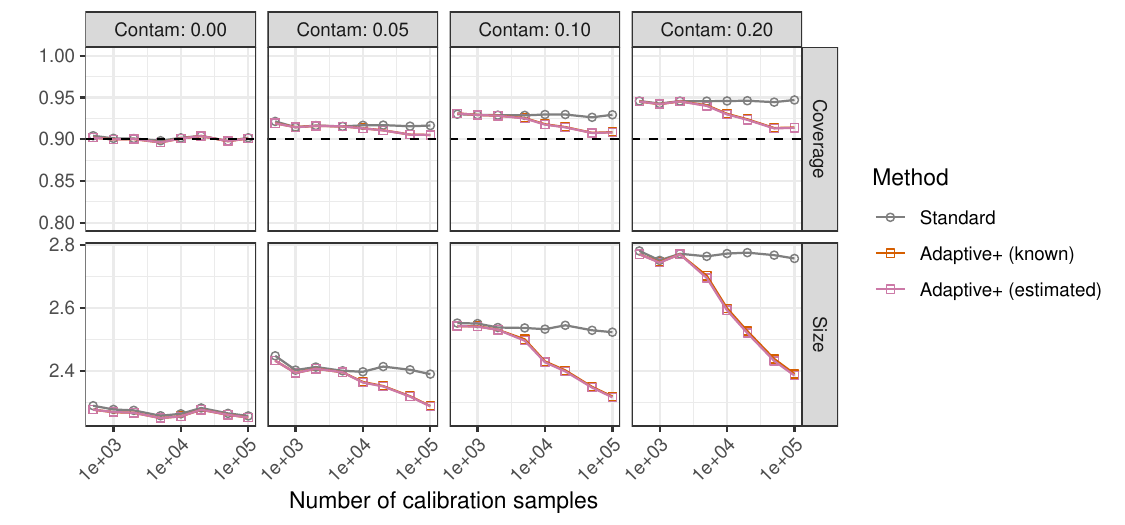}
\caption{Performances of different conformal methods on simulated data with label contamination.
 The label contamination process has a random heterogeneous structure.
The Adaptive+ method is applied with and without perfect knowledge of the contaminated label frequencies.
Note that there is no significant difference in performance here, because it is easy to estimate $\tilde{\rho}$ accurately from the available data.
The nominal marginal coverage level is 90\%.
Other details are as in Figure~\ref{fig:exp-synthetic-1-marginal-K4-ncal_synthetic2_uniform_estim-rho}.}
\label{fig:exp-synthetic-1-marginal-K4-ncal_synthetic2_random_estim-rho}
\end{figure}

\clearpage

\subsubsection{Prediction sets with calibration-conditional coverage}

\begin{figure}[!htb]
\centering
\includegraphics[width=0.9\linewidth]{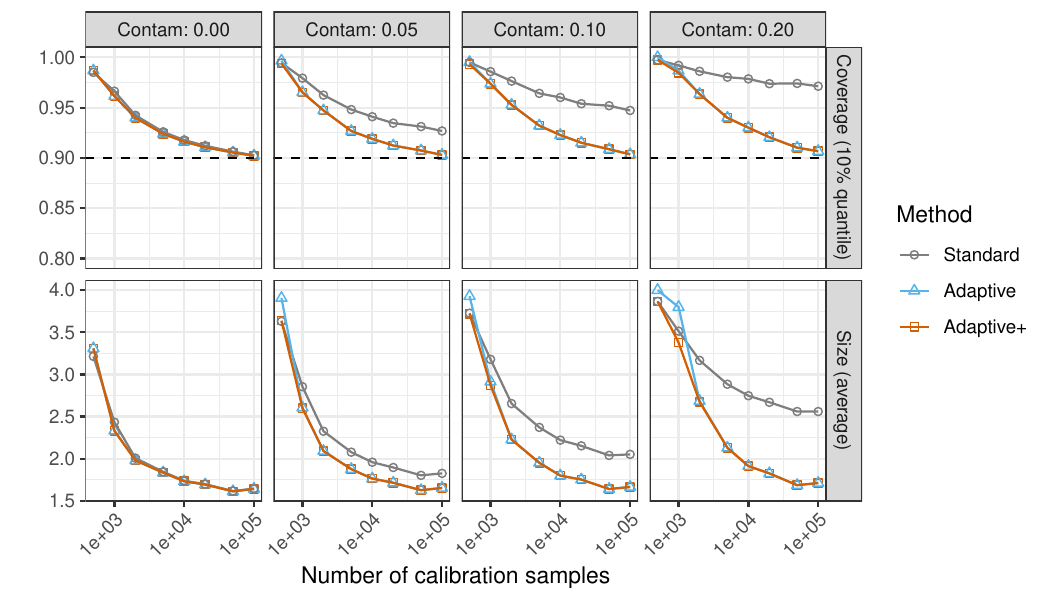}
\caption{Performances of different conformal methods on simulated data with random label contamination of varying strength, as a function of the number of calibration samples. All methods guarantee 90\% calibration and label-conditional coverage with probability at least 90\% over the calibration data. Other details are as in Figure~\ref{fig:exp-synthetic-1-marginal-K4-ncal}.}
\label{fig:exp-synthetic-1-lab-cond-K4-ncal-cc}
\end{figure}

\clearpage

\subsection{Simulations under a bounded label contamination model}  \label{app:figures-bounded}

\subsubsection{Randomized response model}  \label{app:figures-bounded-RR}

\begin{figure}[!htb]
\centering
\includegraphics[width=0.9\linewidth]{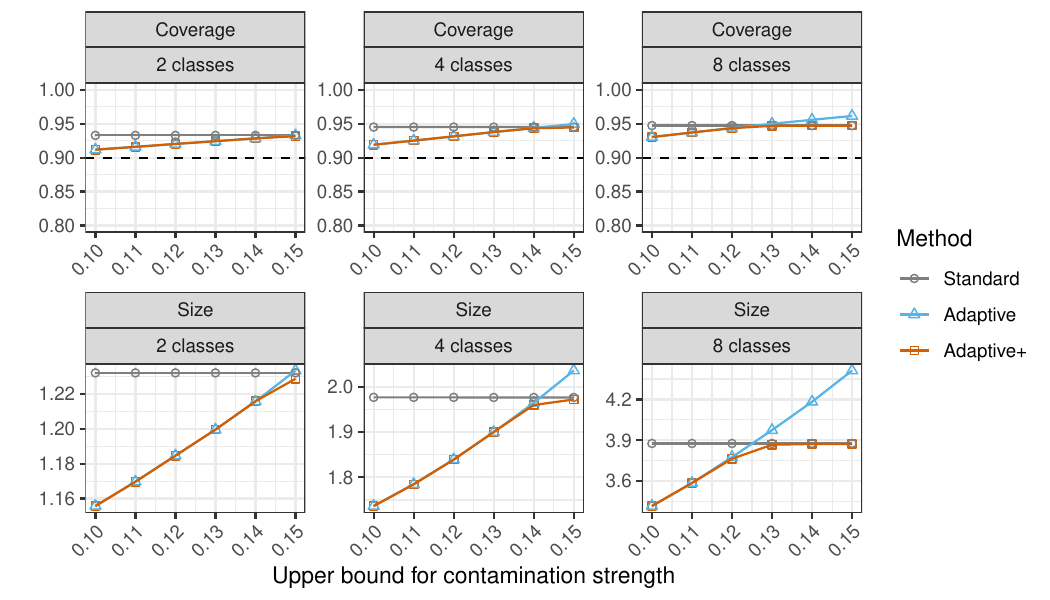}
\caption{Performances of different conformal methods on simulated data with random label contamination based on a randomized response model with unknown contamination strength
parameter $\epsilon$. 
The results are shown as a function of the known upper confidence bound for $\epsilon$, whose true value is $\epsilon=0.1$.
Other details are as in Figure~\ref{fig:exp-synthetic-1-bounded-ncal-eps0.2-lower}.}
\label{fig:exp-synthetic-1-lab-cond-K4-ncal-upper}
\end{figure}

\begin{figure}[!htb]
\centering
\includegraphics[width=0.9\linewidth]{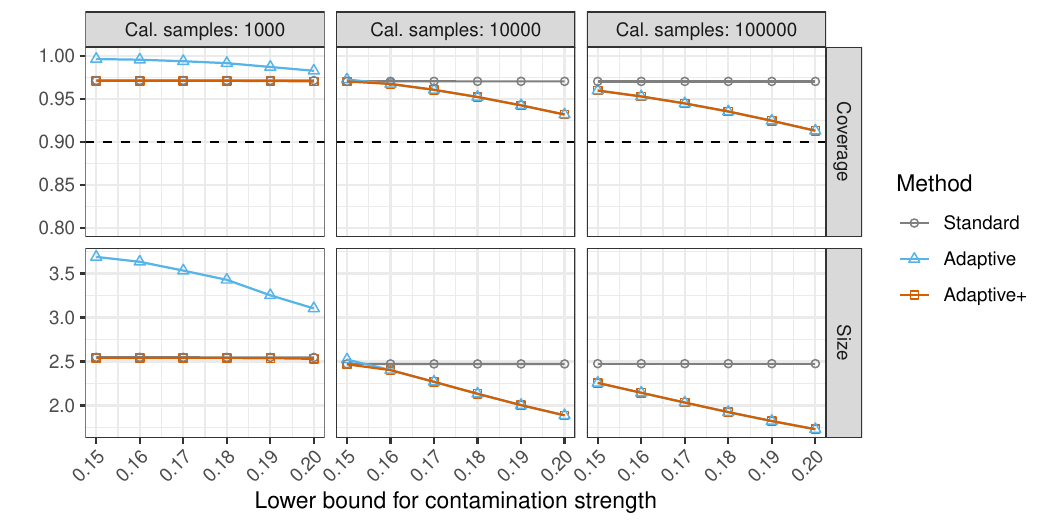}
\caption{Performances of different conformal methods on simulated data with random label contamination, as a function of the known lower confidence bound for the contamination strength parameter $\epsilon=0.2$.
The results are stratified based on the numbers of calibration samples.
Other details are as in Figure~\ref{fig:exp-synthetic-1-bounded-ncal-eps0.2-lower}.}
\label{fig:exp-synthetic-1-bounded-K4-ncal-eps0.2}
\end{figure}

\begin{figure}[!htb]
\centering
\includegraphics[width=0.9\linewidth]{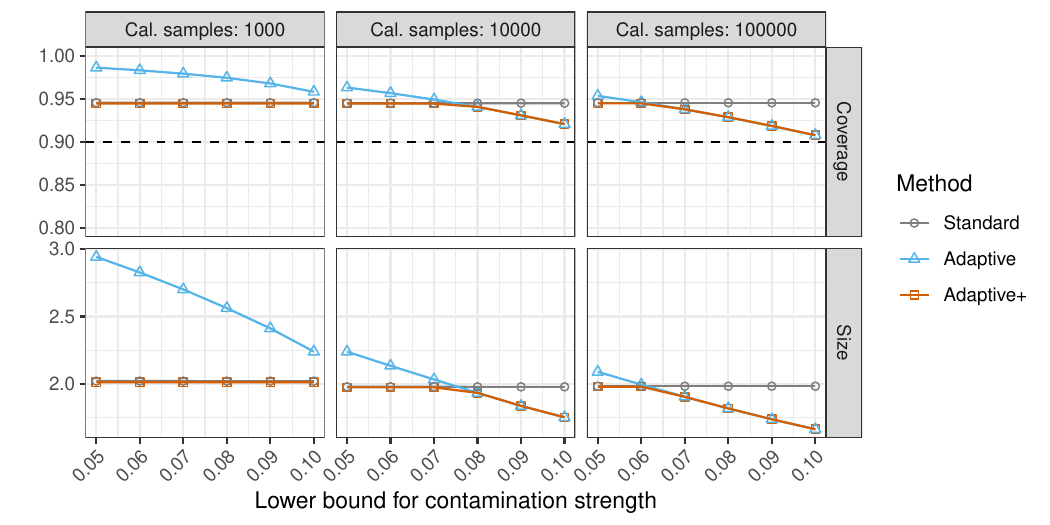}
\caption{Performances of different conformal methods on simulated data with random label contamination, as a function of the known lower confidence bound for the contamination strength parameter $\epsilon=0.1$.
The results are stratified based on the numbers of calibration samples.
Other details are as in Figure~\ref{fig:exp-synthetic-1-bounded-ncal-eps0.2-lower}.}
\label{fig:exp-synthetic-1-bounded-K4-ncal-eps0.1}
\end{figure}

\FloatBarrier
\clearpage

\subsubsection{Two-level randomized response model}  \label{app:figures-bounded-BRR}

\begin{figure}[!htb]
\centering
\includegraphics[width=\linewidth]{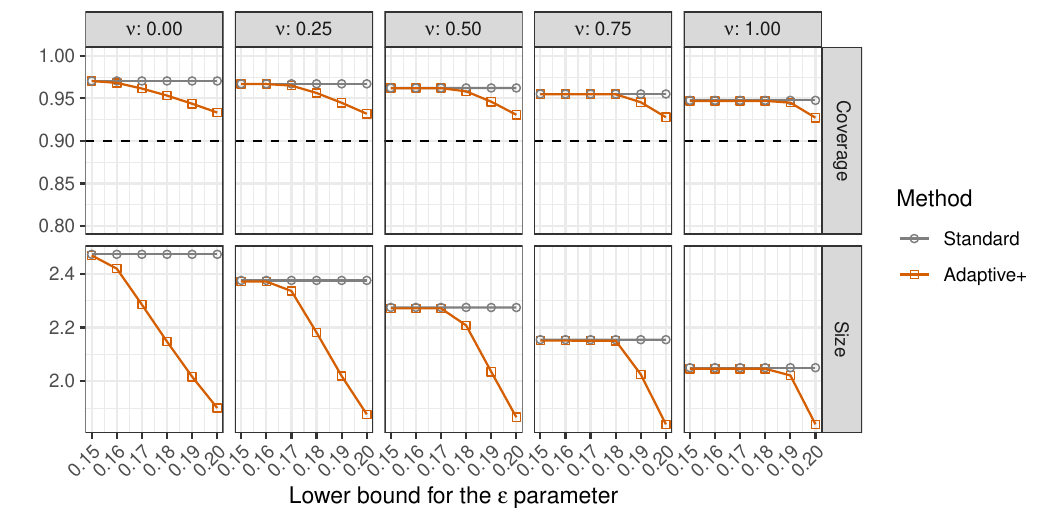}
\caption{Performances of different conformal methods on simulated data with random label contamination from a two-level randomized response model with partly unknown parameters.
The Adaptive+ method is implemented by applying the specialized version of Algorithm~\ref{alg:correction-ci} described in Section~\ref{app:BRR-model}, based on a 99\% confidence interval $[\hat{\epsilon}^{\mathrm{upp}}, \epsilon]$ for $\epsilon = 0.2$ and a degenerate interval $[\nu,\nu]$ for $\nu \in \{0,0.25,0.5,0.75,1\}$.
The results are shown as a function of $\hat{\epsilon}^{\mathrm{upp}}$. 
The calibration set size is $10,000$.
Other details are as in Figure~\ref{fig:exp-synthetic-1-bounded-ncal-eps0.2-lower}.}
\label{fig:exp-synthetic-1-bounded-BRR-K4-eps0.2-ncal10000-nuse0}
\end{figure}

\begin{figure}[!htb]
\centering
\includegraphics[width=\linewidth]{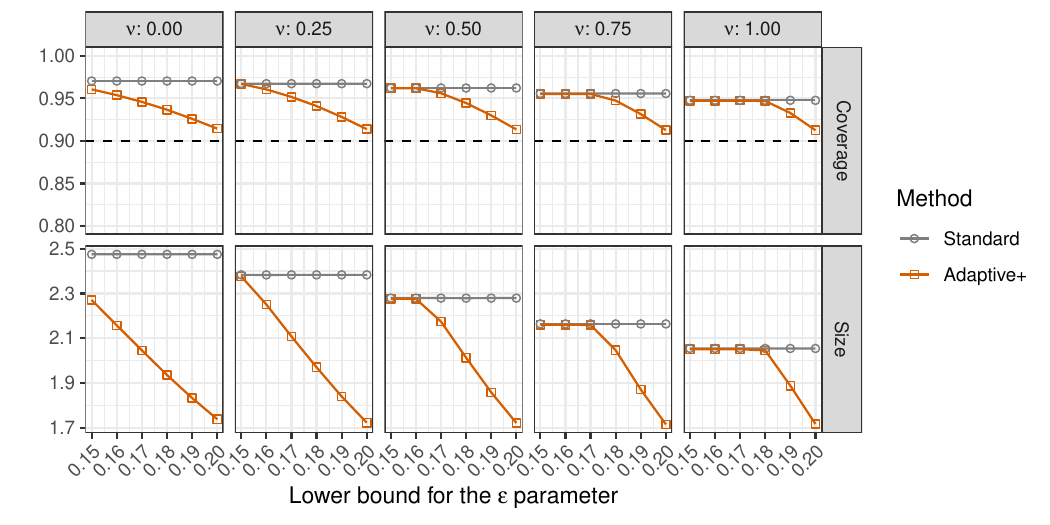}
\caption{Performances of different conformal methods on simulated data with random label contamination from a two-level randomized response model.
The Adaptive+ method is applied based on a 99\% confidence interval $[\hat{\epsilon}^{\mathrm{upp}}, \epsilon]$ for $\epsilon = 0.2$ and a degenerate interval $[\nu,\nu]$ for $\nu \in \{0,0.25,0.5,0.75,1\}$.
The results are shown as a function of $\hat{\epsilon}^{\mathrm{upp}}$. 
The calibration set size is $100,000$.
Other details are as in Figure~\ref{fig:exp-synthetic-1-bounded-BRR-K4-eps0.2-ncal10000-nuse0}.}
\label{fig:exp-synthetic-1-bounded-BRR-K4-eps0.2-ncal100000-nuse0}
\end{figure}

\begin{figure}[!htb]
\centering
\includegraphics[width=\linewidth]{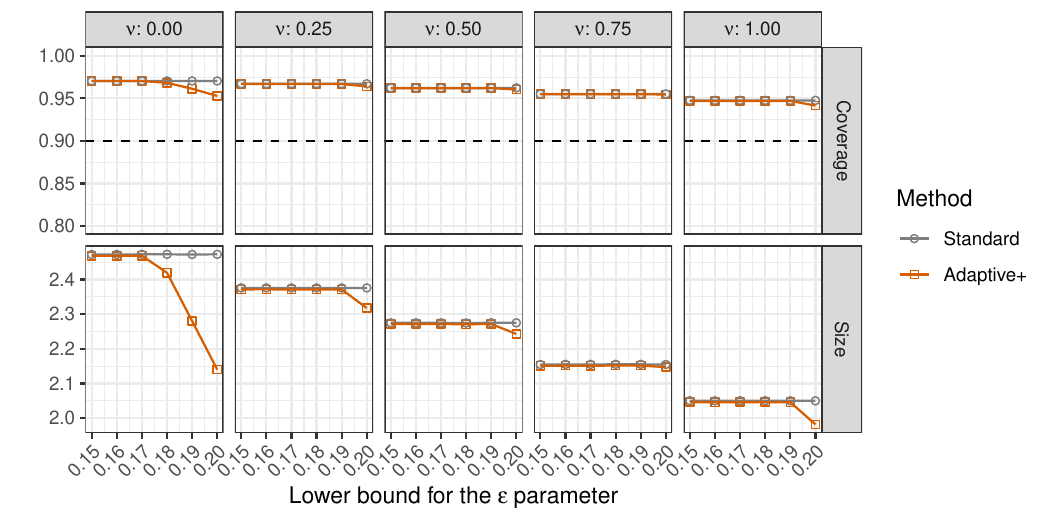}
\caption{Performances of different conformal methods on simulated data with random label contamination from a two-level randomized response model.
The Adaptive+ method is applied based on a 99\% confidence interval $[\hat{\epsilon}^{\mathrm{upp}}, \epsilon]$ for $\epsilon = 0.2$ and an interval $[\max\{0,\nu-0.02\},\min\{\nu+0.02,1\}]$ for $\nu \in \{0,0.25,0.5,0.75,1\}$.
The results are shown as a function of $\hat{\epsilon}^{\mathrm{upp}}$. 
The calibration set size is $10,000$.
Other details are as in Figure~\ref{fig:exp-synthetic-1-bounded-BRR-K4-eps0.2-ncal10000-nuse0}.}
\label{fig:exp-synthetic-1-bounded-BRR-K4-eps0.2-ncal10000-nuse0.02}
\end{figure}

\begin{figure}[!htb]
\centering
\includegraphics[width=\linewidth]{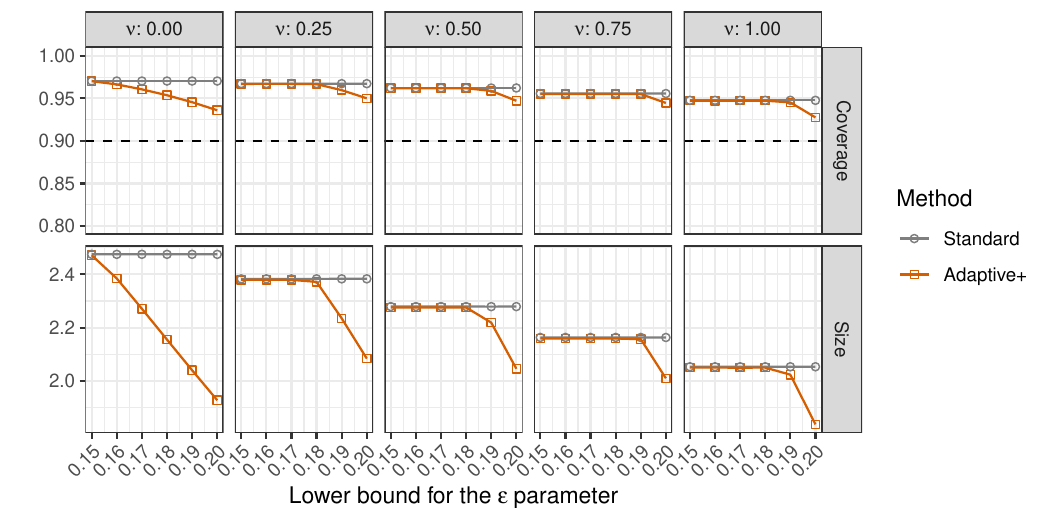}
\caption{Performances of different conformal methods on simulated data with random label contamination from a two-level randomized response model.
The Adaptive+ method is applied based on a 99\% confidence interval $[\hat{\epsilon}^{\mathrm{upp}}, \epsilon]$ for $\epsilon = 0.2$ and an interval $[\max\{0,\nu-0.02\},\min\{\nu+0.02,1\}]$ for $\nu \in \{0,0.25,0.5,0.75,1\}$.
The results are shown as a function of $\hat{\epsilon}^{\mathrm{upp}}$. 
The calibration set size is $100,000$.
Other details are as in Figure~\ref{fig:exp-synthetic-1-bounded-BRR-K4-eps0.2-ncal10000-nuse0.02}.}
\label{fig:exp-synthetic-1-bounded-BRR-K4-eps0.2-ncal100000-nuse0.02}
\end{figure}

\begin{figure}[!htb]
\centering
\includegraphics[width=\linewidth]{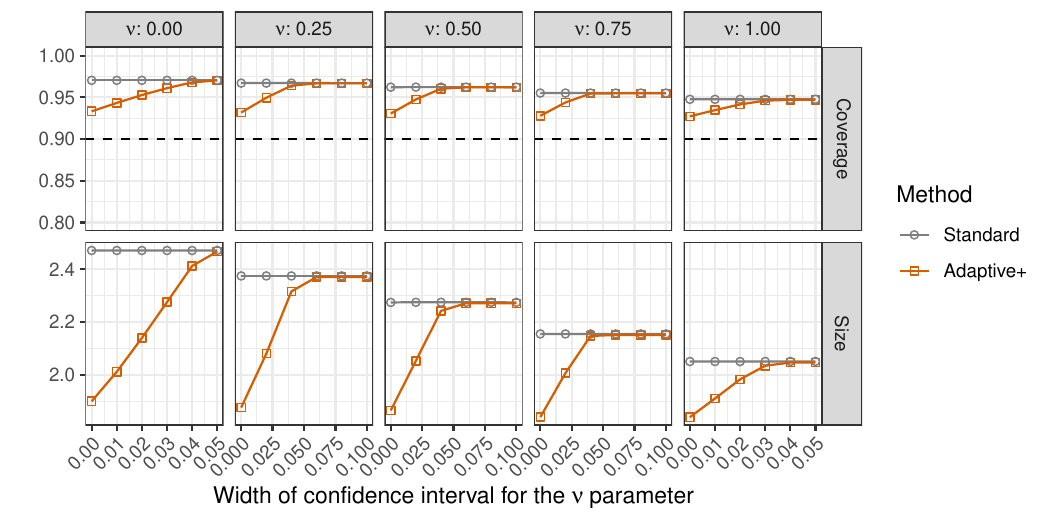}
\caption{Performances of different conformal methods on simulated data with random label contamination from a two-level randomized response model.
The Adaptive+ method is implemented by applying the specialized version of Algorithm~\ref{alg:correction-ci} described in Section~\ref{app:BRR-model}, based on a 99\% symmetric confidence interval $[\hat{\nu}^{\mathrm{low}}, \hat{\nu}^{\mathrm{upp}}]$ for $\nu \in \{0,0.25,0.5,0.75,1\}$ and a degenerate interval $[\epsilon, \epsilon]$ for $\epsilon = 0.2$.
The results are shown as a function of $\hat{\nu}^{\mathrm{upp}}-\hat{\nu}^{\mathrm{low}}$.
The calibration set size is $10,000$.
Other details are as in Figure~\ref{fig:exp-synthetic-1-bounded-BRR-K4-eps0.2-ncal10000-nuse0}.}
\label{fig:exp-synthetic-1-bounded-BRR-K4-eps0.2-ncal10000-epsne0}
\end{figure}

\begin{figure}[!htb]
\centering
\includegraphics[width=\linewidth]{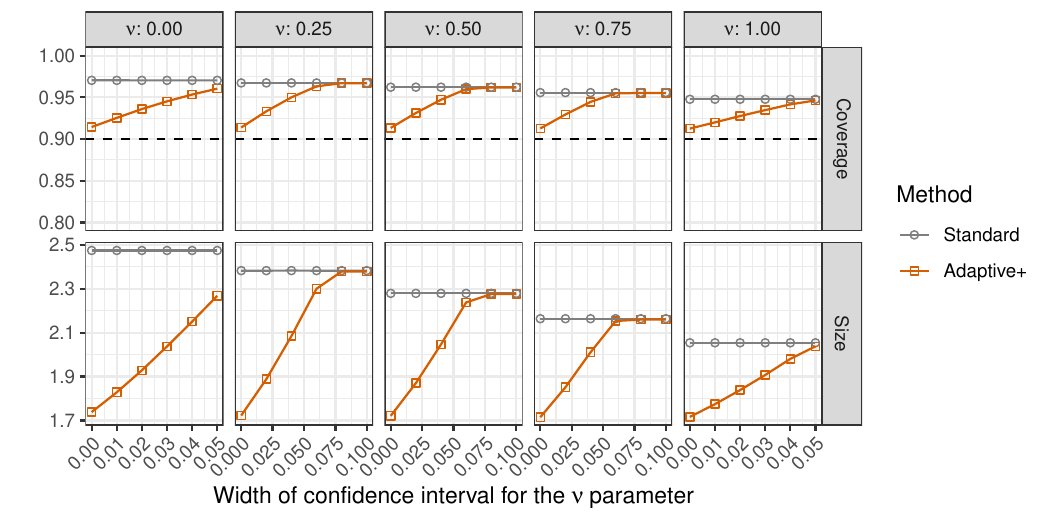}
\caption{Performances of different conformal methods on simulated data with random label contamination from a two-level randomized response model with partly unknown parameters.
The calibration set size is $100,000$.
Other details are as in Figure~\ref{fig:exp-synthetic-1-bounded-BRR-K4-eps0.2-ncal10000-epsne0}.}
\label{fig:exp-synthetic-1-bounded-BRR-K4-eps0.2-ncal100000-epsne0}
\end{figure}

\begin{figure}[!htb]
\centering
\includegraphics[width=\linewidth]{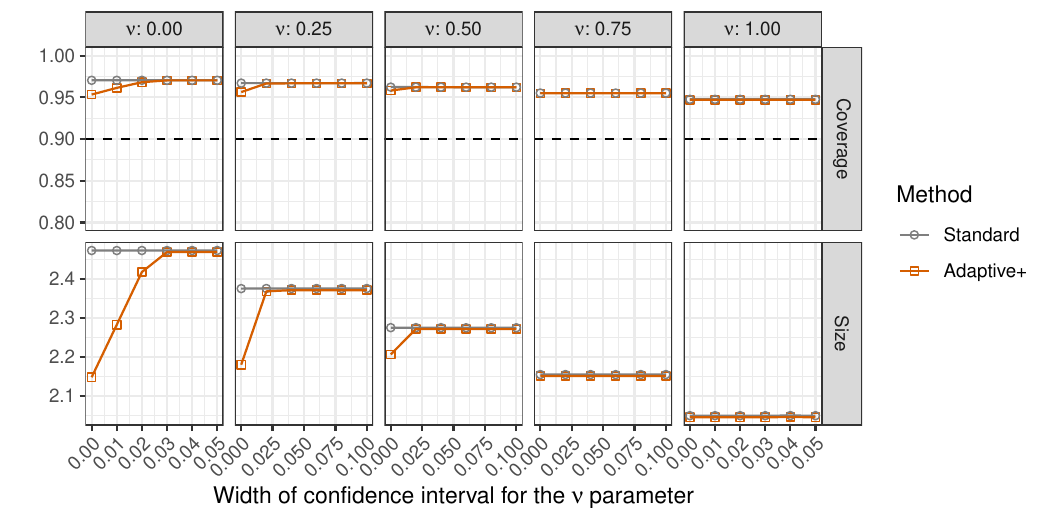}
\caption{Performances of different conformal methods on simulated data with random label contamination from a two-level randomized response model.
The Adaptive+ method is implemented by applying the specialized version of Algorithm~\ref{alg:correction-ci} described in Section~\ref{app:BRR-model}, based on a 99\% symmetric confidence interval $[\hat{\nu}^{\mathrm{low}}, \hat{\nu}^{\mathrm{upp}}]$ for $\nu \in \{0,0.25,0.5,0.75,1\}$ and a degenerate interval $[\epsilon-0.02, \epsilon]$ for $\epsilon = 0.2$.
The results are shown as a function of $\hat{\nu}^{\mathrm{upp}}-\hat{\nu}^{\mathrm{low}}$.
The calibration set size is $10,000$.
Other details are as in Figure~\ref{fig:exp-synthetic-1-bounded-BRR-K4-eps0.2-ncal10000-nuse0}.}
\label{fig:exp-synthetic-1-bounded-BRR-K4-eps0.2-ncal10000-epsne0.02}
\end{figure}

\begin{figure}[!htb]
\centering
\includegraphics[width=\linewidth]{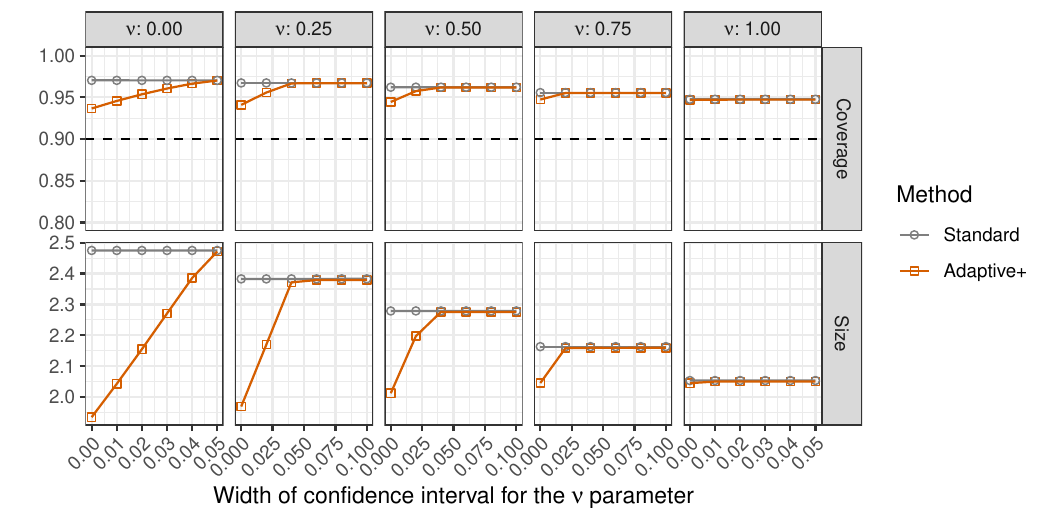}
\caption{Performances of different conformal methods on simulated data with random label contamination from a two-level randomized response model with partly unknown parameters.
The calibration set size is $100,000$.
Other details are as in Figure~\ref{fig:exp-synthetic-1-bounded-BRR-K4-eps0.2-ncal10000-epsne0.02}.}
\label{fig:exp-synthetic-1-bounded-BRR-K4-eps0.2-ncal100000-epsne0.02}
\end{figure}

\FloatBarrier

\subsection{Robustness to model estimation}  \label{app:figures-ci}

\subsubsection{Randomized response model}  \label{app:figures-ci-RR}

\begin{figure}[!htb]
\centering
\includegraphics[width=0.9\linewidth]{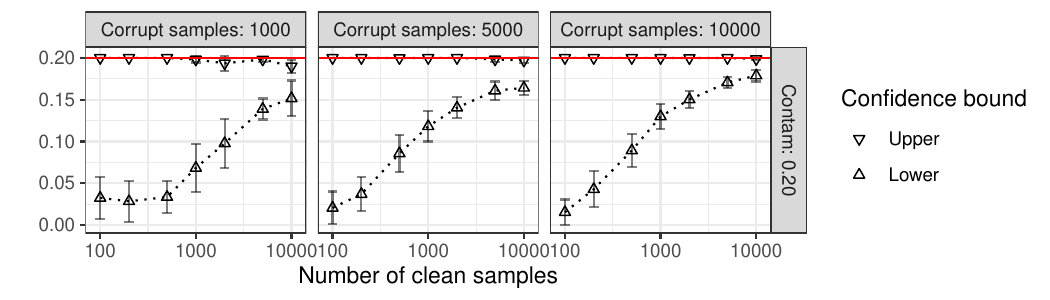}
\caption{Bootstrap confidence intervals for the parameter $\epsilon$ of a randomized response model describing the label contamination process.
The results are shown as a function of the number of clean and contaminated samples used to fit the model.
These confidence intervals are utilized by the {\em Adaptive+ (CI)} method in Figure~\ref{fig:exp-synthetic-1-lab-cond_ci_K2}.}
\label{fig:exp-synthetic-1-lab-cond_ci_bounds_K2}
\end{figure}

\begin{figure}[!htb]
\centering
\includegraphics[width=0.9\linewidth]{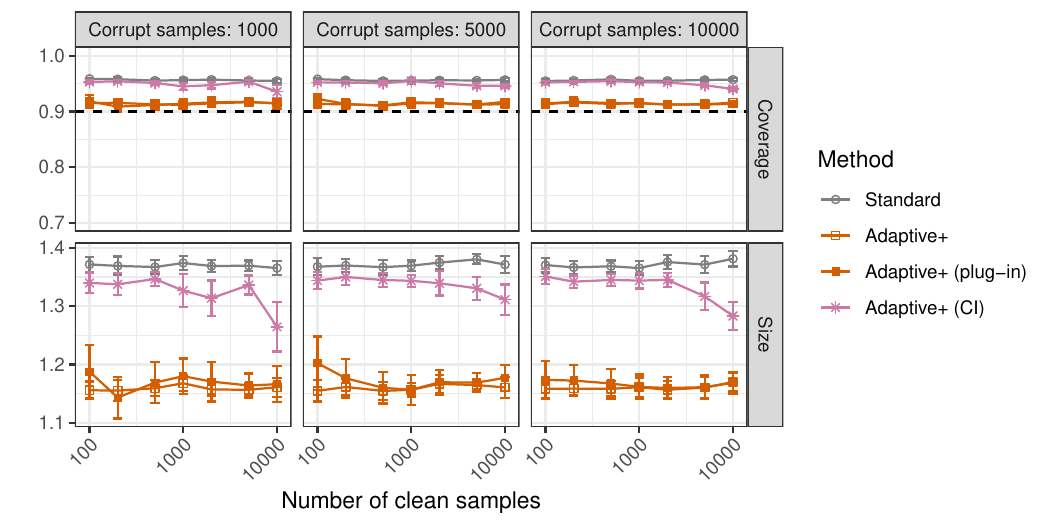}
\caption{Performances of different conformal prediction methods, as a function of the numbers of clean and contaminated samples used to fit a randomized response label contamination model.
The fixed upper bound for $\epsilon=0.2$ utilized by the {\em Adaptive+ (CI)} method, namely Algorithm~\ref{alg:correction-ci}, is equal to $\bar{\varepsilon}=0.25$.
Other details are as in Figure~\ref{fig:exp-synthetic-1-lab-cond_ci_K2}.}
\label{fig:exp-synthetic-1-lab-cond_ci_K2_emax0.25}
\end{figure}

\begin{figure}[!htb]
\centering
\includegraphics[width=0.9\linewidth]{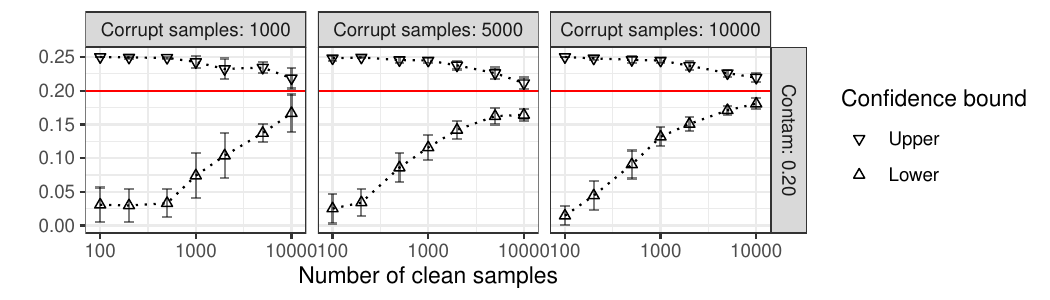}
\caption{Bootstrap confidence intervals for the parameter $\epsilon$ of a randomized response model describing the label contamination process.
The results are shown as a function of the number of clean and contaminated samples used to fit the model.
These confidence intervals are utilized by the {\em Adaptive+ (CI)} method in Figure~\ref{fig:exp-synthetic-1-lab-cond_ci_K2_emax0.25}.}
\label{fig:exp-synthetic-1-lab-cond_ci_bounds_K2_emax0.25}
\end{figure}

\begin{figure}[!htb]
\centering
\includegraphics[width=0.9\linewidth]{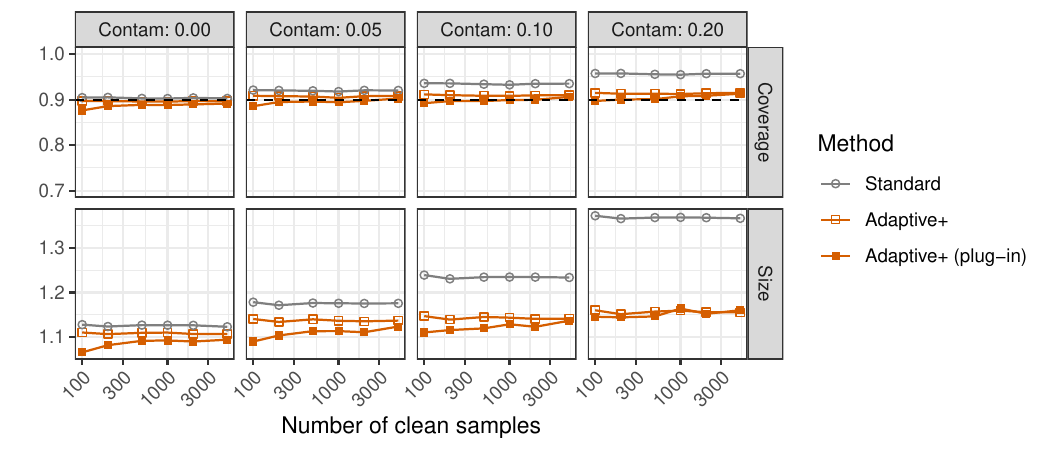}
\caption{Performances of different conformal methods on simulated data with random label contamination following a randomized response model.
The results are shown as a function of the number of clean samples used to fit the noise parameter $\epsilon$.
The results are stratified based on the true value of $\epsilon$.
Other details are as in Figure~\ref{fig:exp-synthetic-1-lab-cond_ci_K2}.}
\label{fig:exp-synthetic-1-lab-cond_point_K2}
\end{figure}

\begin{figure}[!htb]
\centering
\includegraphics[width=0.9\linewidth]{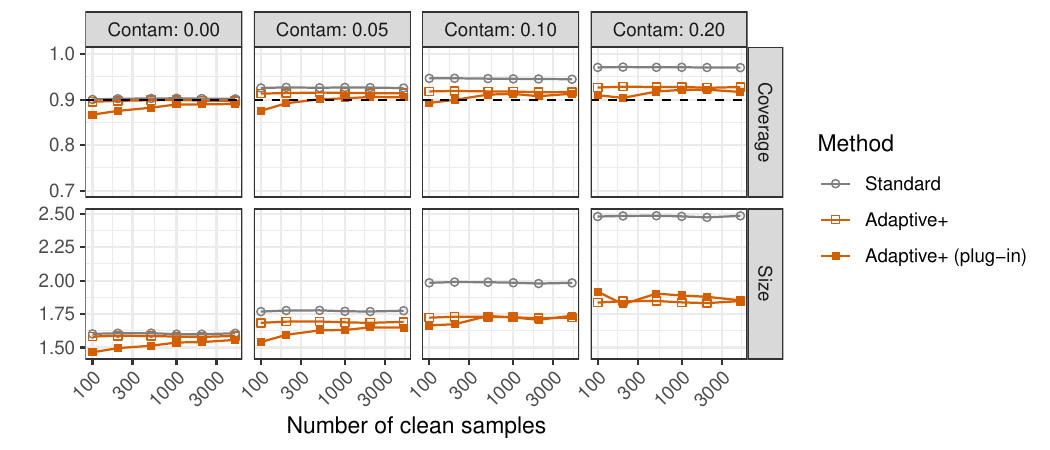}
\caption{Performances of different conformal methods on simulated data with random label contamination  following a randomized response model.
The results are shown as a function of the number of clean samples used to fit the noise parameter $\epsilon$.
The number of possible classes is $K=4$.
Other details are as in Figure~\ref{fig:exp-synthetic-1-lab-cond_point_K2}.}
\label{fig:exp-synthetic-1-lab-cond_point_K4}
\end{figure}

\begin{figure}[!htb]
\centering
\includegraphics[width=0.9\linewidth]{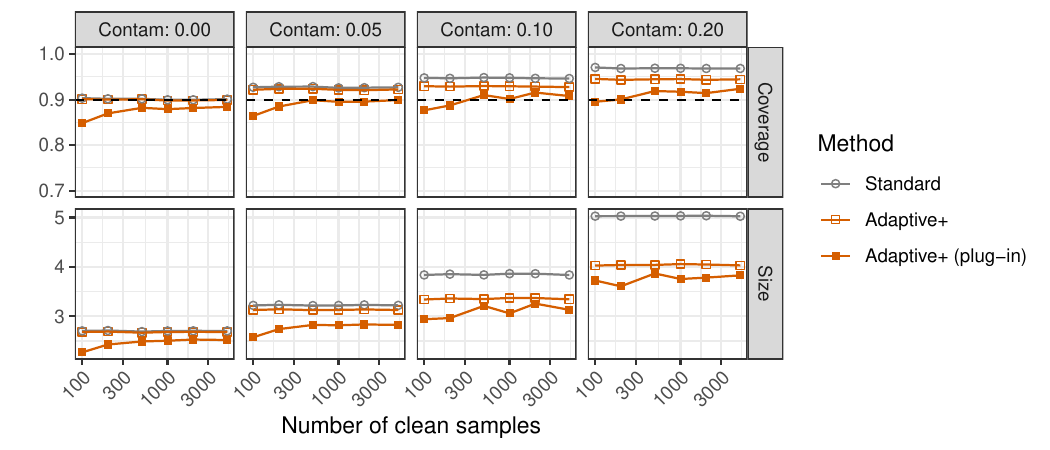}
\caption{Performances of different conformal methods on simulated data with random label contamination following a randomized response model.
The results are shown as a function of the number of clean samples used to fit the noise parameter $\epsilon$.
The number of possible classes is $K=8$.
Other details are as in Figure~\ref{fig:exp-synthetic-1-lab-cond_point_K2}.}
\label{fig:exp-synthetic-1-lab-cond_point_K8}
\end{figure}

\FloatBarrier
\subsubsection{Two-level randomized response model}  \label{app:figures-ci-BRR}

\begin{figure}[!htb]
\centering
\includegraphics[width=0.95\linewidth]{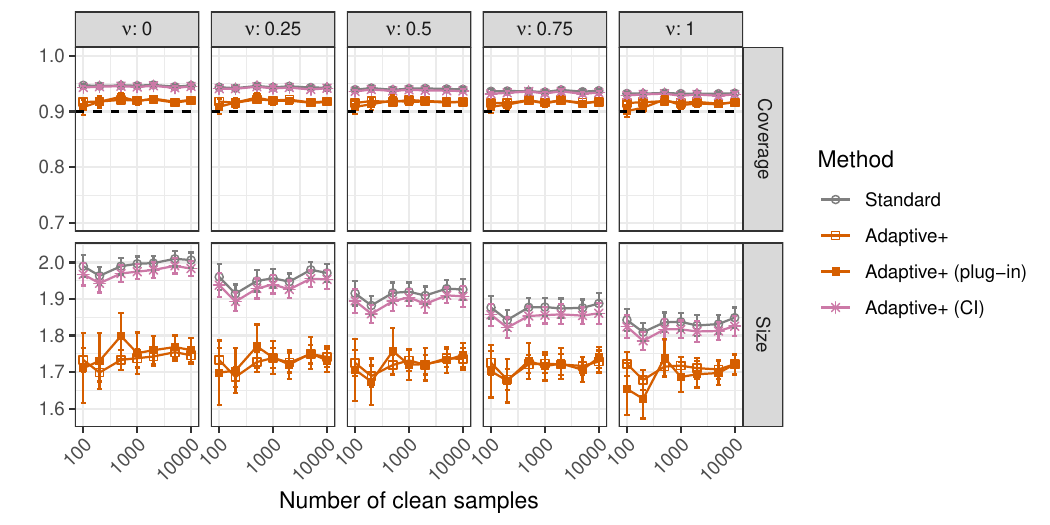}
\caption{Performances of different conformal methods on simulated data with random label contamination from a two-level randomized response model.
The results are shown as a function of the number of clean samples used to fit the parameters $\epsilon$ and $\nu$, for different true values of $\nu$.
The true parameter $\epsilon$ is equal to 0.1.
The number of possible classes is $K=4$.
Other details are as in Figure~\ref{fig:exp-synthetic-1-lab-cond_point_K4}.}
\label{fig:exp-synthetic-1-lab-cond_eps0.1_K4_BRR}
\end{figure}

\begin{figure}[!htb]
\centering
\includegraphics[width=0.95\linewidth]{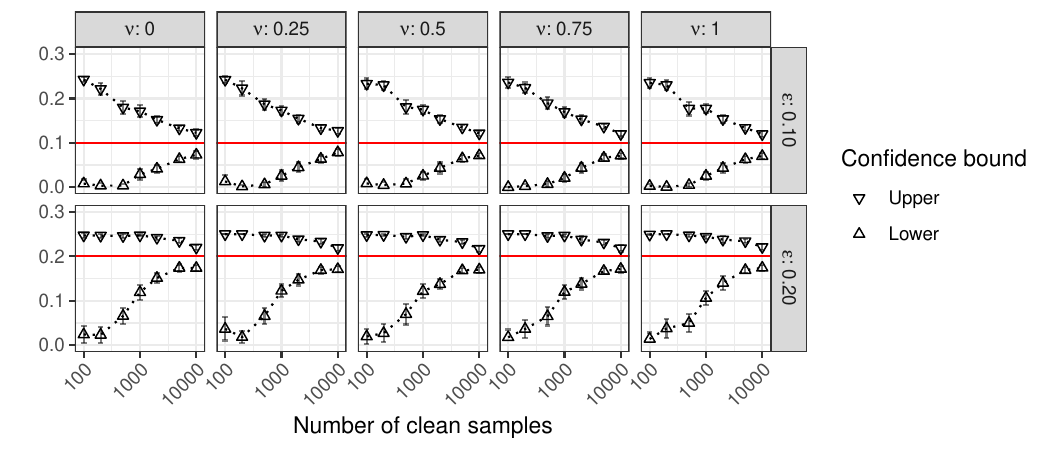}
\caption{Bootstrap confidence intervals for the parameter $\epsilon$ in a two-level randomized response model for the label contamination process.
The results are shown as a function of the number of clean samples used to fit the model, stratified based on the true values of the model parameters $\epsilon$ and $\nu$.
The number of contaminated training samples is $10,000$.
These confidence intervals correspond to those utilized by the {\em Adaptive+ (CI)} method in Figure~\ref{fig:exp-synthetic-1-lab-cond_eps0.1_K4_BRR}.}
\label{fig:exp-synthetic-1-lab-cond_ci_K4_emax0.25_BRR}
\end{figure}

\begin{figure}[!htb]
\centering
\includegraphics[width=0.95\linewidth]{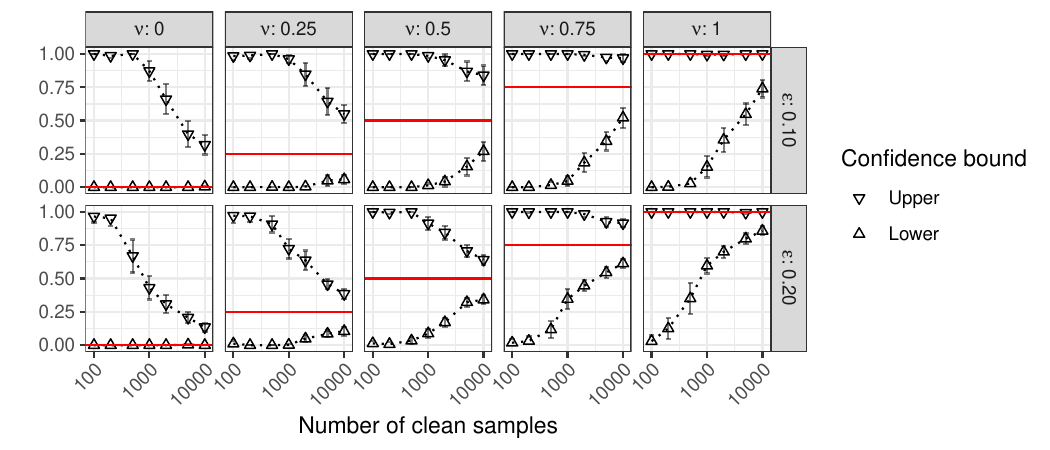}
\caption{Bootstrap confidence intervals for the parameter $\nu$ in a two-level randomized response model for the label contamination process.
Other details are as in Figure~\ref{fig:exp-synthetic-1-lab-cond_ci_K4_emax0.25_BRR}.}
\label{fig:exp-synthetic-1-lab-cond_ci_nu_K4_emax0.25_BRR}
\end{figure}

\begin{figure}[!htb]
\centering
\includegraphics[width=0.95\linewidth]{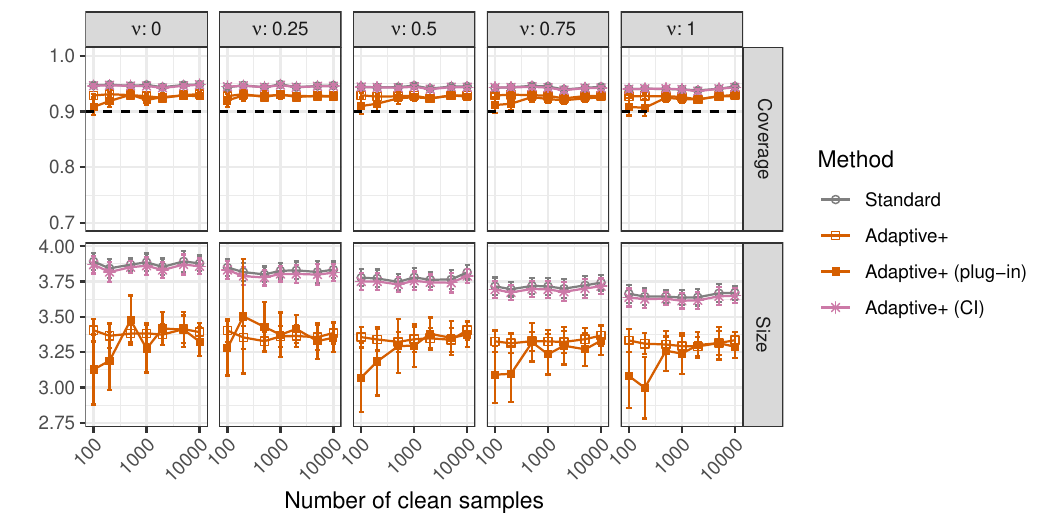}
\caption{Performances of different conformal methods on simulated data with random label contamination from a two-level randomized response model.
The results are shown as a function of the number of clean samples used to fit the parameters $\epsilon$ and $\nu$, for different true values of $\nu$.
The true parameter $\epsilon$ is equal to 0.1.
The number of possible classes is $K=8$.
Other details are as in Figure~\ref{fig:exp-synthetic-1-lab-cond_eps0.1_K4_BRR}.}
\label{fig:exp-synthetic-1-lab-cond_eps0.1_K8_BRR}
\end{figure}

\begin{figure}[!htb]
\centering
\includegraphics[width=0.95\linewidth]{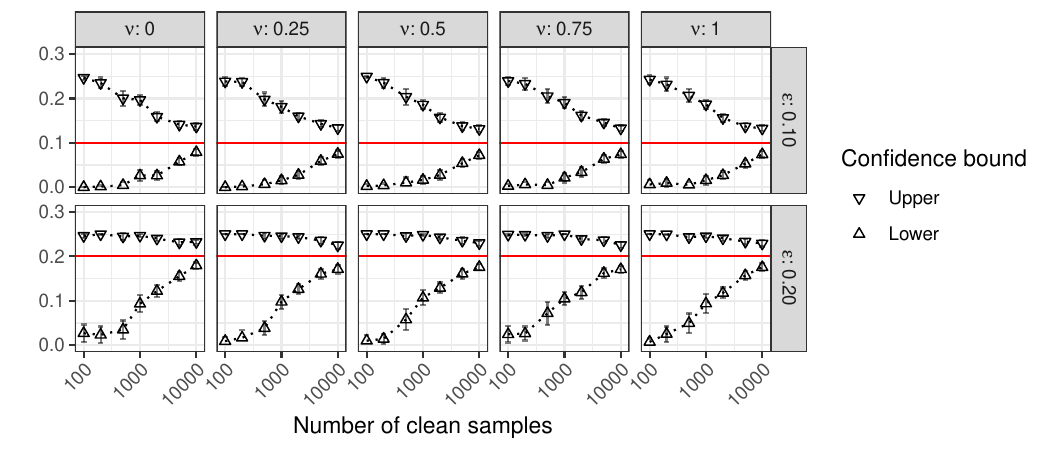}
\caption{Bootstrap confidence intervals for the parameter $\epsilon$ in a two-level randomized response model for the label contamination process.
The results are shown as a function of the number of clean samples used to fit the model, stratified based on the true values of the model parameters $\epsilon$ and $\nu$.
The number of contaminated training samples is $10,000$. The number of possible labels is $K=8$.
These confidence intervals correspond to those utilized by the {\em Adaptive+ (CI)} method in Figure~\ref{fig:exp-synthetic-1-lab-cond_eps0.1_K8_BRR}.}
\label{fig:exp-synthetic-1-lab-cond_ci_K8_emax0.25_BRR}
\end{figure}

\begin{figure}[!htb]
\centering
\includegraphics[width=0.95\linewidth]{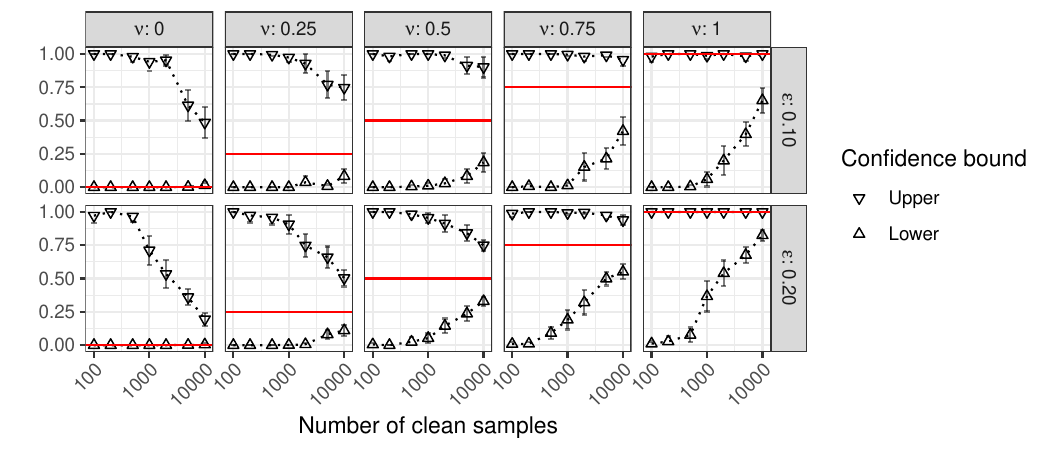}
\caption{Bootstrap confidence intervals for the parameter $\nu$ in a two-level randomized response model for the label contamination process.
The number of possible labels is $K=8$.
Other details are as in Figure~\ref{fig:exp-synthetic-1-lab-cond_ci_K8_emax0.25_BRR}.}
\label{fig:exp-synthetic-1-lab-cond_ci_nu_K8_emax0.25_BRR}
\end{figure}

\FloatBarrier
\subsection{Robustness to model mis-specification}   \label{app:figures-robust}

This section demonstrates the robustness of our adaptive methods to other types of mis-specifications in the label contamination process, going beyond the issue of estimating the noise parameter $\epsilon$ of the randomized response model described in Section~\ref{app:RR-model}.
For this purpose, we generate synthetic data with $K=4$ possible labels as explained in Section~\ref{sec:empirical-known}, but using a label contamination process with a block-like structure.
That is, the contamination process is described by a transition matrix $T \in [0,1]^{K \times K}$, defined as $T_{kl} = \mathbb{P}[\tilde{Y}=k \mid  X, Y=l]$ for all $l,k \in [K]$, given by $T = (1-\epsilon)I_{K} + \epsilon/K \cdot B_{K,2}$, where $B_{K,2}$ is a block-diagonal matrix with $K/2$ constant blocks equal to $J_2$---the $2 \times 2$ matrix of ones.
Then, we apply the optimistic version of Algorithm~\ref{alg:correction} under the mis-specified assumption that $T = (1-\epsilon)I_K + \epsilon/K \cdot J_{K}$, separately using either the correct value of $\epsilon$ or a plug-in empirical estimate obtained from an independent model-fitting data set as in Section~\ref{app:RR-model-estim}.
Consistently with the previous section, we refer to the former method as {\em Adaptive+}  and to the latter as {\em Adaptive+ (plug-in)}.
Figure~\ref{fig:exp-synthetic-1-lab-cond_ci_K4_block} compares the performances of our {\em Adaptive+} and {\em Adaptive+ (plug-in)} to that of the standard conformal inference approach, as a function of the true $\epsilon$ and of the number of clean model-fitting data points used to estimate this parameter.
The number of calibration data points here is set equal to 10,000.
The results show that both methods are quite robust to model mis-specification and can generally produce more informative prediction sets compared to the standard benchmark.

\begin{figure}[!htb]
\centering
\includegraphics[width=0.9\linewidth]{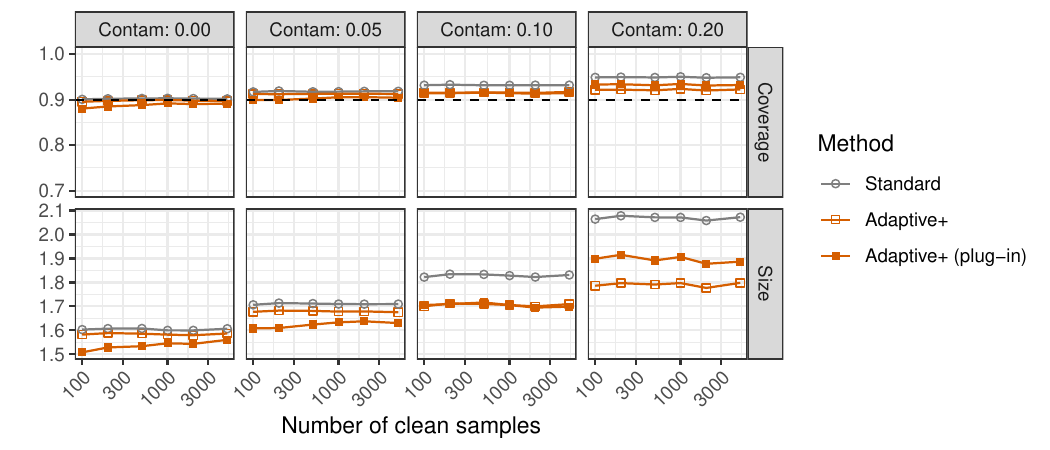}
\caption{Performances of different conformal methods on simulated data, as a function of the number of clean samples used to estimate a mis-specified contamination model.
  The true model matrix $M$ has a block-diagonal structure, and the results are stratified based on the contamination strength.
Other details are as in Figure~\ref{fig:exp-synthetic-1-lab-cond-K4-ncal}.}
\label{fig:exp-synthetic-1-lab-cond_ci_K4_block}
\end{figure}

Finally, Figure~\ref{fig:exp-synthetic-1-lab-cond_ci_K4_random} demonstrates that our methods also enjoy similar robustness under different label contamination processes, focusing specifically on a transition matrix $T = (1-\epsilon)I_{K} + \epsilon/K \cdot U_K$, where $U_K$ is a matrix of i.i.d.~uniform random numbers on $[0,1]$, standardized to have its columns sum to one.

\begin{figure}[!htb]
\centering
\includegraphics[width=0.9\linewidth]{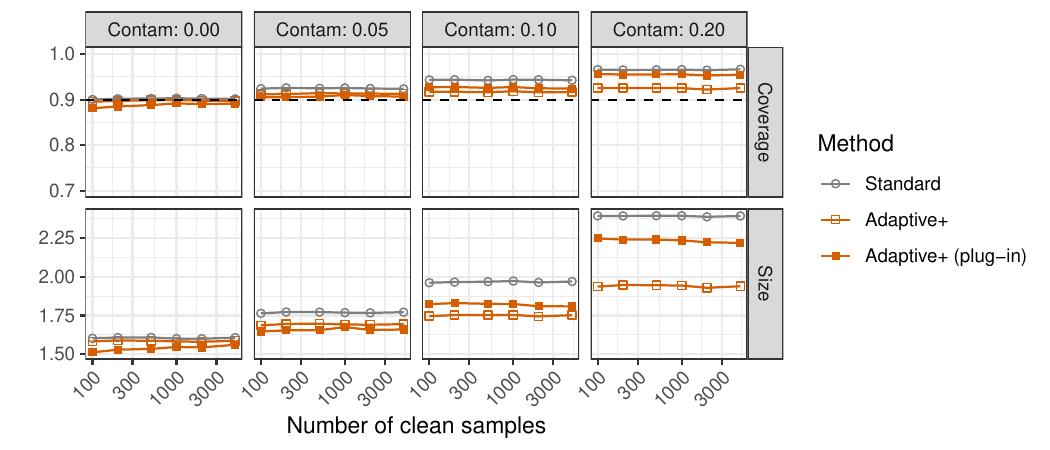}
\caption{Performances of different conformal methods on simulated data, as a function of the number of clean samples used to estimate a mis-specified contamination model.
The contamination model matrix $M$ has a random structure.
Other details are as in Figure~\ref{fig:exp-synthetic-1-lab-cond_ci_K4_block}.}
\label{fig:exp-synthetic-1-lab-cond_ci_K4_random}
\end{figure}

\FloatBarrier
\clearpage

\subsection{Demonstrations with CIFAR-10 image data} \label{app:figures-cifar10}

\begin{figure}[!htb]
\centering
\includegraphics[width=0.85\linewidth]{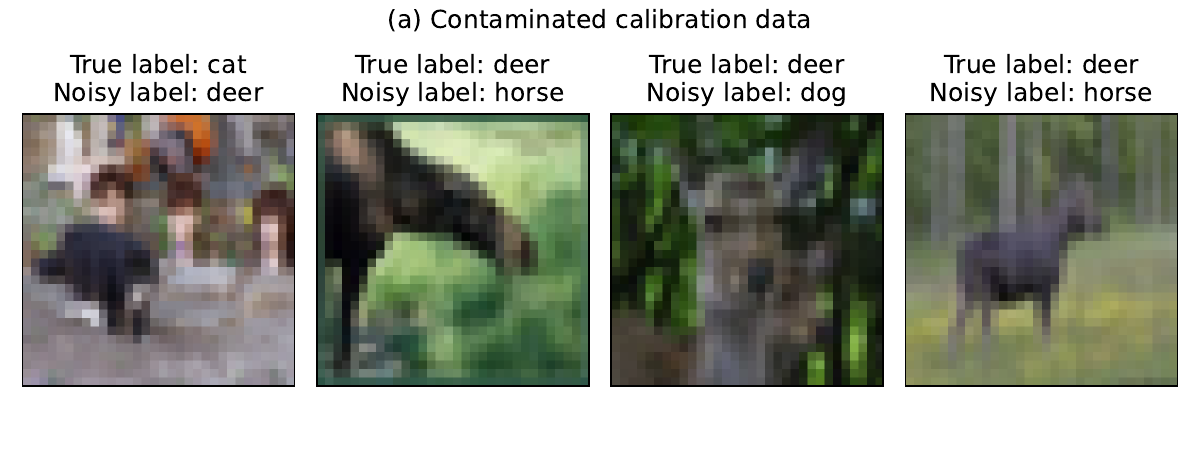}\\[-0.5em]
\includegraphics[width=0.85\linewidth]{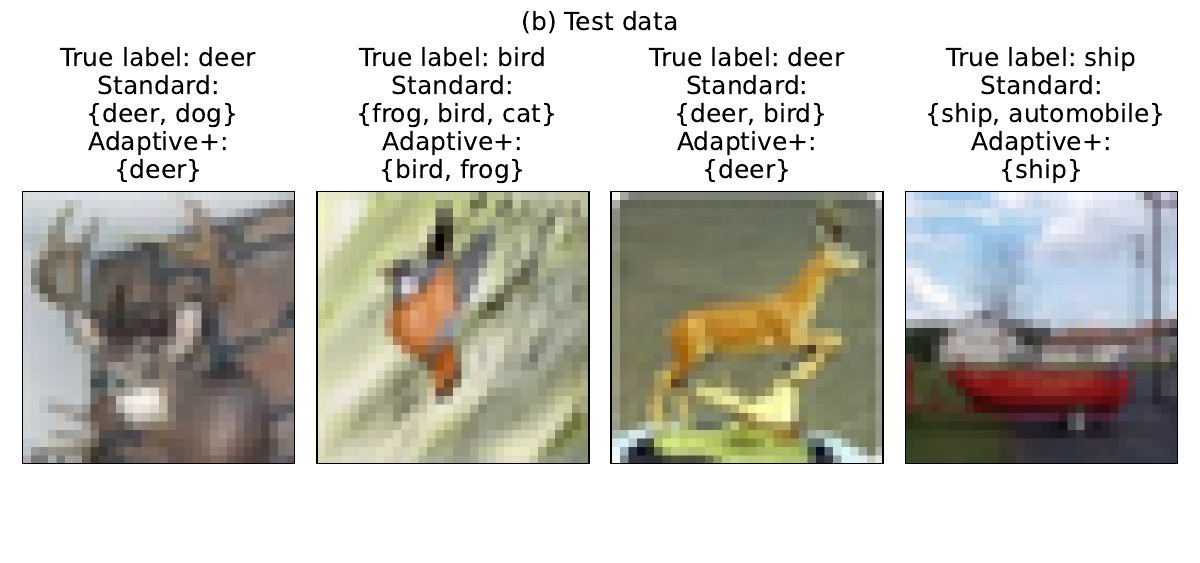}
\caption{Demonstration of some CIFAR-10 images with contaminated labels (a), and of conformal prediction sets (b) obtained with and without correcting for the presence of label contamination in the calibration data, as in the experiments of Figure~\ref{fig:exp-cifar10-marginal}.}
\label{fig:cifar10-demo}
\end{figure}

\begin{figure}[!htb]
\centering
\includegraphics[width=0.85\linewidth]{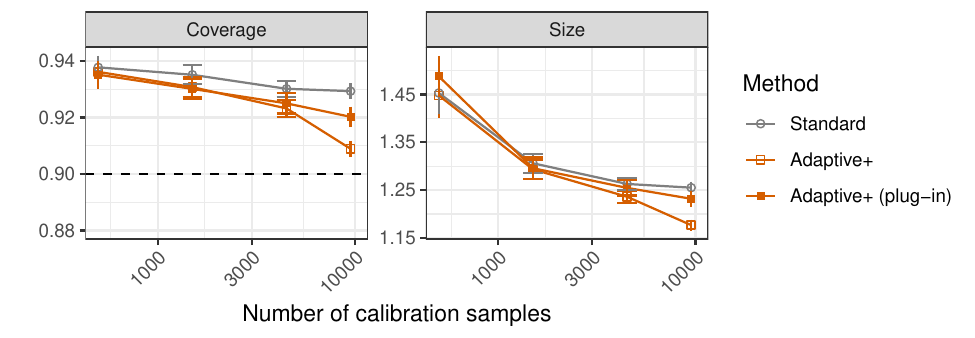}
\caption{Performances of different conformal methods on CIFAR-10 image data with noisy human-assigned labels, as a function of calibration set size.
The horizontal dashed line indicates the nominal 90\% label-conditional coverage level.
Other details are as in Figure~\ref{fig:exp-cifar10-marginal}.}
\label{fig:exp-cifar10-lab-cond}
\end{figure}

\FloatBarrier
\section{Additional details on numerical experiments} \label{app:experiments}

\subsection{Experiments with heteroscedastic decision-tree model} \label{app:details-experiments-tree}

We set $d=50$ and generate each sample of features $X \in \mathbb{R}^{d}$ independently as follows:
$X_1=+1$ w.p. $3/4$, and $X_1=-1$ w.p. $1/4$;
$X_2=+1$ w.p. $3/4$, and $X_2=-2$ w.p. $1/4$;
$X_3=+1$ w.p. $1/4$, and $X_3=-2$ w.p. $1/2$;
$X_4$ is uniformly distributed on $\{1,\ldots,4\}$;
and $X_j \overset{i.i.d.}{\sim} \mathcal{N}(0,1)$ for all $j \in \{5,\ldots,d\}$.
The labels $Y$ belong to one of $K=4$ possible classes, and their conditional distribution given $X=x$ is given by the decision tree shown in Figure~\ref{fig:tree}, which only depends on the first four features.

\begin{figure}[!htb]
  \centering
  \includegraphics[width=0.8\textwidth]{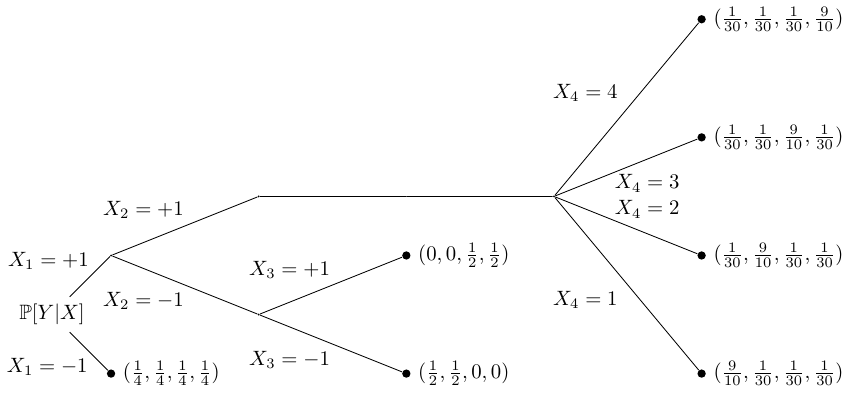}
  \caption{A toy heteroscedastic decision-tree model for $P_{Y \mid X}$ used to generate syntheticin some numerical experiments.}
  \label{fig:tree}
\end{figure}

\end{document}